\documentclass[b5paper,11pt]{book}
\usepackage{miguel}
\usepackage{url}
\usepackage{floatflt}
\usepackage{graphicx}
\usepackage{subfigure}
\usepackage{epsfig}
\usepackage[mathscr]{euscript}
\usepackage{fancyheadings}
\usepackage{enumerate}
\usepackage{paralist}
\usepackage{ifthen}
\usepackage[usenames]{color}
\usepackage{caption}
\usepackage{makeidx}
\usepackage[bookmarks=true, pdftitle={\thesistitle. PhD Thesis, Radboud University Nijmegen, the Netherlands}, pdfauthor={Miguel E. Andres}, pdfcreator={Miguel E. Andres}, colorlinks=true, citecolor=red, citebordercolor={1 1 1}, linktoc=page, linkcolor=blue, urlcolor=blue, hyperfootnotes=false, plainpages=false, pagebackref ]{hyperref}

\usepackage{backref} \renewcommand*{\backref}[1]{} \renewcommand*{\backrefalt}[4]{%
 \ifcase #1 %
(Not cited.)%
\or \textbf {Cited} on page~#2.%
\else \textbf{Cited} on pages~#2.%
\fi} \renewcommand*{\backrefsep}{, } \renewcommand*{\backreftwosep}{ and~} \renewcommand*{\backreflastsep}{ and~} \usepackage{memhfixc}

\title{Quantitative Analysis of Information Leakage \\in Probabilistic and Nondeterministic Systems}
\author{Miguel E. Andr\'es.\\[3pt]%
Institute for Computing and Information Sciences, \\[3pt]%
Radboud University Nijmegen, The Netherlands. \\[3pt]%
\url{mandres@cs.ru.nl}}

\usepackage{atbeginend}
\AfterBegin{itemize}{\addtolength{\itemsep}{-.8ex}}
\AfterBegin{enumerate}{\addtolength{\itemsep}{-.8ex}}
\makeindex

\begin{document}
\pagestyle{empty}

\pagebreak

\begin{titlepage}
\begin{center}
\phantom{x}
\vspace{5cm}
{\Large {\bf Quantitative Analysis of Information Leakage in}}\\
\vspace{3mm}
{\Large {\bf Probabilistic and Nondeterministic Systems}} \\
\vspace{2cm}
{\Large {\bf Miguel E. Andr\'es}}  \\
\vspace{30pt}
\end{center}
\end{titlepage}

\noindent{Copyright \copyright~2011 Miguel E. Andr\'es.}\\
ISBN: 978-94-91211-74-4.\\
IPA dissertation series: 2011-09.\\
\hfill

\noindent
This thesis is typeset using \LaTeX.\\
Translation of the Dutch summary: Peter van Rossum.\\
Cover designed by Marieke Meijer -  \href{www.mariekemeijer.com}{www.mariekemeijer.com}.\\
Thesis printed by Ipskamp Drukers -  \href{www.ipskampdrukkers.nl}{www.ipskampdrukkers.nl}.

\vspace{-0.5cm}

\begin{figure}[!h]
\begin{center}
\subfigure{
\includegraphics[width=75pt]{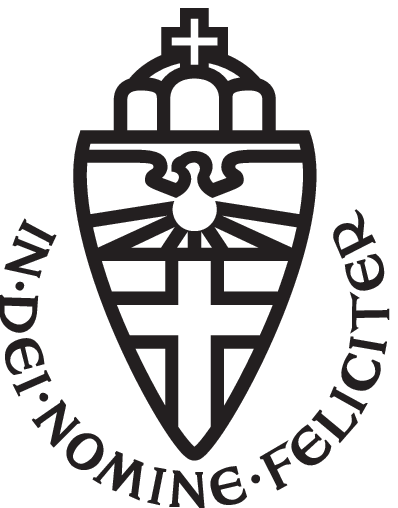}}\hspace{1cm}
\subfigure{\vspace{1.5cm}
\includegraphics[width=125pt]{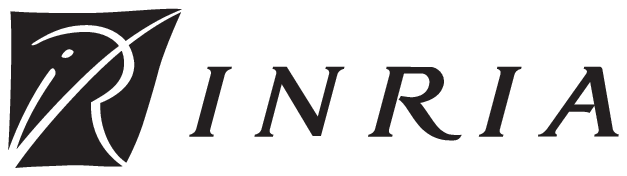}}\\
\hspace{-0.5cm}\subfigure{\hspace{-0.5cm}
\includegraphics[width=125pt]{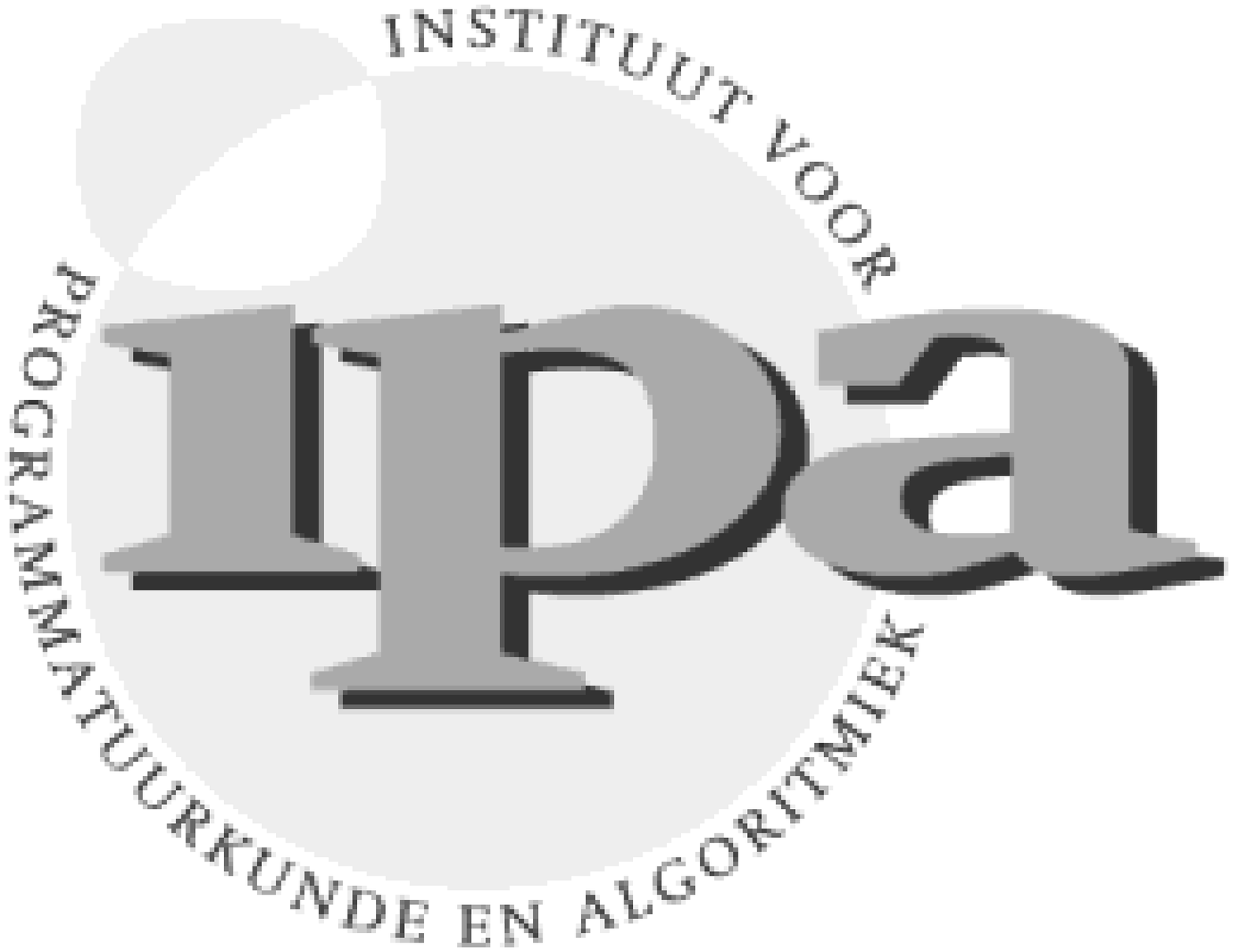}}\hspace{0.75cm}
\subfigure{
\includegraphics[width=100pt]{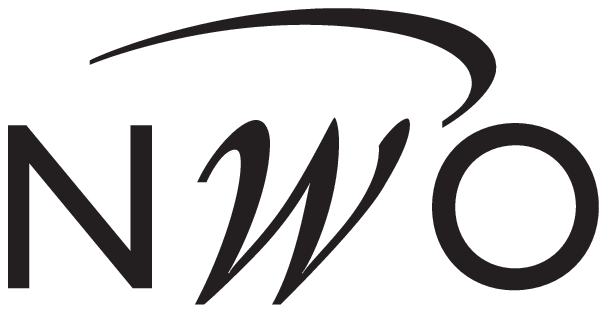}}
\end{center}
\end{figure}
\vspace{-0.25cm}
\noindent The work in this thesis has been carried out at Radboud University and under the auspices of the
research school IPA (Institute for Programming research and Algorithmics). The research funding was provided by the NWO Grant through the open project 612.000.526: Analysis of Anonymity. The author also wishes to acknowledge the French Institute for Research in Computer Science and Control (INRIA) for providing funding for several research visits to the \'Ecole Polytechnique of Paris.

\vfill\pagebreak 

\vspace*{1.5cm}
\begin{center}
{\Large\bf Quantitative Analysis of Information Leakage in}\\
\vspace{2mm}
{\Large\bf Probabilistic and Nondeterministic Systems}\\
\vspace{1cm}
Een wetenschappelijke proeve op het gebied van de Natuurwetenschappen, Wiskunde en Informatica.\\
\vspace{1cm}
{\sc Proefschrift}\\
\vspace{1cm}
ter verkrijging van de graad van doctor\\
aan de Radboud Universiteit Nijmegen\\
op gezag van de rector magnificus, prof. mr. S.C.J.J. Kortmann,\\
volgens besluit van het College van Decanen\\
in het openbaar te verdedigen op vrijdag 1 juli 2011 \\
om 10:30 uur precies\\
\vspace{.6cm}
door\\
\vspace{.6cm}
Miguel E. Andr\'es\\
\vspace{.6cm}
geboren op 02 July 1980,\\
te R\'io Cuarto, C\'ordoba, Argentini\"e.\\
\end{center}

\vfill
\pagebreak

\noindent
\textbf{Promotor:}
\vspace{1pt}

\begin{tabular}{l}
prof. dr. Bart P.F. Jacobs \hspace{0.75cm} \\[5mm]
\end{tabular}

\noindent \textbf{Copromotoren:} \vspace{1pt}

\begin{tabular}{l}
dr. Peter van Rossum\\
dr. Catuscia Palamidessi\hspace{1.15cm} INRIA\\ [5mm]
\end{tabular}

\noindent
\textbf{Manuscriptcommissie:}
\vspace{1pt}

\begin{tabular}{ll}
prof. dr. Joost-Pieter Katoen &  RWTH Aachen University \\
dr. Pedro R. D'Argenio & Universidad Nacional de C\'ordoba\\
prof. dr. Frits W. Vaandrager &  \\
prof. dr. Holger Hermanns & Saarland University\\
dr. Mari{\"e}lle Stoelinga & University of Twente
\end{tabular}
\comment{Decide the order}
 \vfill \pagebreak

\vspace*{1.5cm}
\begin{center}
{\Large\bf Quantitative Analysis of Information Leakage in }\\
\vspace{3mm}
{\Large\bf Probabilistic and Nondeterministic Systems}\\
\vspace{1cm}
A scientific essay in Science.\\
\vspace{1cm}
{\sc Doctoral thesis}\\
\vspace{1cm}
to obtain the degree of doctor\\
from Radboud University Nijmegen\\
on the authority of the Rector Magnificus, Prof. dr. S.C.J.J. Kortmann,\\
according to the decision of the Council of Deans\\
to be defended in public on Friday, 1$^{\text{st}}$  July 2011 \\
at 10:30 hours\\
\vspace{.6cm}
by\\
\vspace{.6cm}
Miguel E. Andr\'es\\
\vspace{.6cm}
born in R\'io Cuarto, C\'ordoba, Argentina\\
on 02 July 1980.\\
\end{center}

\vfill\pagebreak

\noindent
\textbf{Supervisor:}
\vspace{1pt}

\begin{tabular}{l}
prof. dr. Bart P.F. Jacobs \hspace{0.75cm} \\[5mm]
\end{tabular}

\noindent \textbf{Co-supervisors:} \vspace{1pt}

\begin{tabular}{l}
dr. Peter van Rossum\\
dr. Catuscia Palamidessi\hspace{1.15cm} INRIA \\ [5mm]
\end{tabular}

\noindent
\textbf{Doctoral Thesis Committee:}
\vspace{1pt}

\begin{tabular}{ll}
prof. dr. Joost-Pieter Katoen &  RWTH Aachen University \\
dr. Pedro R. D'Argenio & Universidad Nacional de C\'ordoba\\
prof. dr. Frits W. Vaandrager &  \\
prof. dr. Holger Hermanns & Saarland University\\
dr. Mari{\"e}lle Stoelinga & University of Twente
\end{tabular}
\comment{Decide the order}
 \vfill \pagebreak

%
%
%
%
%
%
%

\renewcommand{\thepage}{\roman{page}}
\pagestyle{fancy} \setcounter{page}{1}
{\chapter*{Summary}}
\addcontentsline{toc}{chapter}{Summary}

As we dive into the digital era, there is growing concern about the amount of personal digital information that is being gathered about us. Websites often track people's browsing behavior, health care insurers gather medical data, and many smartphones and navigation systems store or transmit information that makes it possible to track the physical location of their users at any time. Hence, anonymity, and privacy in general, are increasingly at stake. Anonymity protocols counter this concern by offering anonymous communication over the Internet. To ensure the correctness of such protocols, which are often extremely complex, a rigorous framework is needed in which anonymity properties can be expressed, analyzed,  and ultimately verified. Formal methods provide a set of mathematical techniques that allow us to rigorously specify and verify anonymity properties.


This thesis addresses the foundational aspects of formal methods for applications in security and in
particular in anonymity. More concretely, we develop frameworks for the specification of anonymity properties
and propose algorithms for their verification. Since in practice anonymity protocols always leak \emph{some}
information, we focus on quantitative properties which capture the \emph{amount} of information leaked by a
protocol. 



We start our research on anonymity from its very foundations, namely conditional probabilities  -- these are the key ingredient of most quantitative anonymity properties. In Chapter 2 we present cpCTL, the first temporal logic making it possible to specify conditional probabilities. In addition, we present an algorithm to verify cpCTL formulas in a model-checking fashion. This logic, together with the model-checker, allows us to specify and verify quantitative anonymity properties over complex systems where probabilistic and nondeterministic behavior may coexist. 

We then turn our attention to more practical grounds: the constructions of algorithms to compute information leakage. More precisely, in Chapter 3 we present polynomial algorithms to compute the (information-theoretic) leakage of several kinds of fully probabilistic protocols (i.e. protocols without nondeterministic behavior). The techniques presented in this chapter are the first ones enabling the computation of (information-theoretic) leakage in interactive protocols. 

\pagestyle{empty}
In Chapter 4 we attack a well known problem in distributed anonymity protocols, namely full-information
scheduling. To overcome this problem, we propose an alternative definition of schedulers together with several new definitions of anonymity (varying according to the attacker's power), and revise the famous definition of strong-anonymity from the literature. Furthermore, we provide a technique to verify
that a distributed protocol satisfies some of the proposed definitions.

In Chapter 5 we provide (counterexample-based) techniques to debug complex systems, allowing for the detection
of flaws in security protocols. Finally, in Chapter 6 we briefly discuss extensions to the frameworks and techniques
proposed in Chapters 3 and 4.


\vfill
\pagebreak

{\section*{Acknowledgements}}
\addcontentsline{toc}{chapter}{Acknowledgements}

This thesis would not have been possible without the continuous support of many people to whom I will always be grateful.

I am heartily thankful to my supervisor Bart Jacobs. He has closely followed the evolution of my PhD and made sure I always had all the resources a PhD student could possibly need.

I owe my deepest gratitude to my co-supervisor, Peter van Rossum. Four years have passed since he decided to take the risk to hire \emph{me}, an Argentinian guy that he barely knew. I was really lucky to have Peter as my supervisor; he has always been very supportive, flexible, and extremely easygoing with me. I will never forget the football World Cup of $2006$ (not that Argentina did very well); back then I was invited to spend one week in Nijmegen for an \emph{official job interview}. But before I had the time to stress too much about formal talks and difficult questions, I found myself sharing a beer with Peter while watching  Argentina vs the Netherlands (fortunately Argentina did not win --- I still wonder what would have happened otherwise). This was just the first of many nice moments we shared together, including dinners, conversations, and trips. In addition to having fun, we have worked hard together --- indeed we completed one of the most important proofs of this thesis at midnight after a long working day at Peter's house (and also after Mari{\"e}lle finally managed to get little Quinten to sleep $\ddot\smile$).

I cannot allow myself to continue this letter without mentioning Catuscia Palamidessi. Much has changed in my life since I first met her in June 2007. Catuscia came then to visit our group in Nijmegen and we discovered that we had many research interests in common. Soon after, Catuscia invited me to visit her group in Paris and this turned out to be the beginning of a very fruitful collaboration. Since then we have spent countless days (and especially nights) working very hard together, writing many articles, and attending many conferences --- including some in amazing places like Australia, Barbados, and Cyprus. Catuscia is not only an exceptionally bright and passionate scientist, but she is also one of the most thoughtful people I have ever met (placing the interests of her colleagues and PhD students above her own), a wonderful person to work with (turning every work meeting into a relaxed intellectual discussion, enhanced with the finest \emph{caff\`e italiano}), and, above all, an unconditional friend. For these reasons and many more (whose enumeration would require a second volume for this thesis), I am forever indebted to Catuscia.




This work has also greatly benefited from the insightful remarks and suggestions of the members of the reading committee Joost-Pieter Katoen, Pedro D'Argenio, and Frits Vaandrager, whom I wish to thank heartily. To Pedro I am also grateful for his sustained guidance and support in my life as a researcher. Many results in this thesis are a product of joint work, and apart from Peter and Catuscia, I am grateful to my co-authors M\'ario S. Alvim, Pedro R. D'Argenio, Geoffrey Smith and Ana Sokolova, all of whom shared their expertise with me. I am also thankful to Jasper Berendsen, Domingo G\'omez, David Jansen, Mari\"elle Stoelinga, Tingting Han, Sergio Giro, J\'er\'emy Dubreil, and Konstantinos Chatizikokolakis for many fruitful  discussions during my time as a PhD student. Also many thanks to Anne-Lise Laurain for her constant (emotional and technical) support during the writing of my thesis, Alexandra Silva for her insightful comments on the introduction of this work, and Marieke Meijer for devoting her artistic talent to the design of the cover of this thesis.

Special thanks to my paranymphs and dear friends Vicen\c{c}, Igor, Cristina, and Flavio. Together we have shared so much... uncountable lunches in the Refter, coffees in the blue coaches, and trips around the world among many more experiences. But, even more importantly, we have always been there to support each other in difficult times, and this is worth the world to me.  

I wish to thank my colleagues in the DS group for providing such a friendly atmosphere which contributed greatly to the realization of my PhD. I explicitly want to thank Alejandro, Ana, Chris, Christian, Erik, Fabian, Gerhard, Ichiro, Jorik, Ken, \L ukasz, Olha, Pieter, Pim, Roel, Thanh Son, and Wojciech with whom I have shared many coffees, nice conversations, table tennis, and much more. My journey in the DS group would not have been as easy if it was not for Maria and Desiree whose help on administrative issues saved me lots of pain; as for any kind of technical problem or just IT advice, Ronny and Engelbert have been always very helpful.

I also wish to thank all the members of the Com\`ete group for making me feel welcome in such a wonderful and fun group. In particular, I want to thank the Colombian crowd -- Frank, Andr\'es, and Luis -- for the nice nights out to ``La Pe\~na'' and all the fun we had together,  J\'er\'emy for introducing me to the tennis group, Sophia for helping correct my English in this thesis, Kostas for many interesting conversations on the most diverse topics of computer science (and life in general), and M\'ario for finally acknowledging the superiority of the Argentinian soccer over the Brazilian one $\ddot\smile$.

Along the years, I always had a very happy and social life in Nijmegen. Many people, in addition to my paranymphs, have greatly contributed to this. A very big thanks to ``Anita'', for so many happy moments and her unconditional support during my life in Nijmegen. Also many thanks to my dear friend Ren\'ee, every time I hear ``It is difficult to become friends with Dutch people, but once they are your friends they would NEVER let you down'' I have to think of her. Special thanks to Elena and Clara for being there for me when I could not even speak English, Errez for sharing his wisdom in hunting matters $\ddot\smile$, Hilje and Daphne for being there for me when I just arrived to Nijmegen (this was very important to me), also thanks to Klasien, David, and Patricia for many nice nights out and to my dear neighbours -- Marleen, Kim, and Marianne -- for welcoming me in their house and being always so understanding with me. Besides, I would like to thank to the ``Blonde Pater crowd'' including Cristina, Christian, Daniela, Davide, Deniz, Francesco, Jordan, Mariam, Nina, Shankar, and Vicen\c{c} with all of whom I shared many nice cappuccinos, conversations, and nights out. Special thanks to the sweet Nina for being always willing help me, to Francesco for the wonderful guitar nights, to Mariam for her help in Dutch, and to Christian... well, simply for his ``buena onda''.

More than four years have passed since I left Argentina, and many things have certainly changed in my life. However, the support and affection of my lifetime friends have remained immutable. Many thanks to Duro, Seba, Tony, Viole, Gabi, and Mart\'in for all they have contributed to my life. Special thanks to my ``brothers'' Tincho and Lapin, my life would not be the same without them. 

Last but not least, all my affection to my family: dad Miguel, mum Bacho, sister Josefina, and brothers Ignacio and Augusto. Every single success in my life I owe mostly to them. Thanks to my brothers and sister for their constant support and everlasting smiles, which mean to me more than I can express in words. Thanks to my parents for their incessant and selfless sacrifice, thanks to which their children have had all anybody could possible require to be happy and successful. My parents are the greatest role models for me and to them I dedicate this thesis.

\begin{flushright} \noindent\emph{To my dear parents, Miguel and Bacho.}\end{flushright}

%
%
\bigskip
Por \'ultimo, el ``gracias'' mas grande del mundo es para mis queridos padres –- Miguel y Bacho –- y hermanos –- Ignacio, Josefina y Augusto. Porque cada logro conseguido en mi vida, ha sido (en gran parte) gracias a ellos. 
A mis hermanos, les agradezco su apoyo y sonrisas constantes, que significaron y siguen significando para m\'i mucho m\'as de lo que las palabras puedan expresar. A mis padres, su incansable y desinteresado sacrificio, gracias al cual sus hijos han tenido y tienen todas las oportunidades del mundo para ser felices y exitosos. Ellos son, sin lugar a dudas, mi ejemplo de vida, y a ellos dedico esta tesis:
%
\begin{flushright}\noindent\emph{A mis queridos padres, Miguel y Bacho.}\end{flushright}

%
%
%

\vfill

\begin{flushright}
Miguel E. Andr\'es\\
Paris, May 2011.
\end{flushright}
\thispagestyle{empty}
\vfill
\pagebreak

\tableofcontents

\chapter{Introduction}
\pagestyle{fancy} \setcounter{page}{1}
\renewcommand{\thepage}{\arabic{page}}
\label{ch.intro}
%
%
%


\section{Anonymity}

The world \emph{Anonymity} derives from the Greek $\grave{\alpha} \nu\omega\nu\upsilon\mu\acute{\iota}\alpha$, which means
``without a name''. In general, this term is used to express the fact that \rev{the identity of an individual} is not publicly known.

\begin{wrapfigure}{r}{4.25cm}
  \vspace{-1cm}
    \begin{center}
      \includegraphics[width=4.25cm]{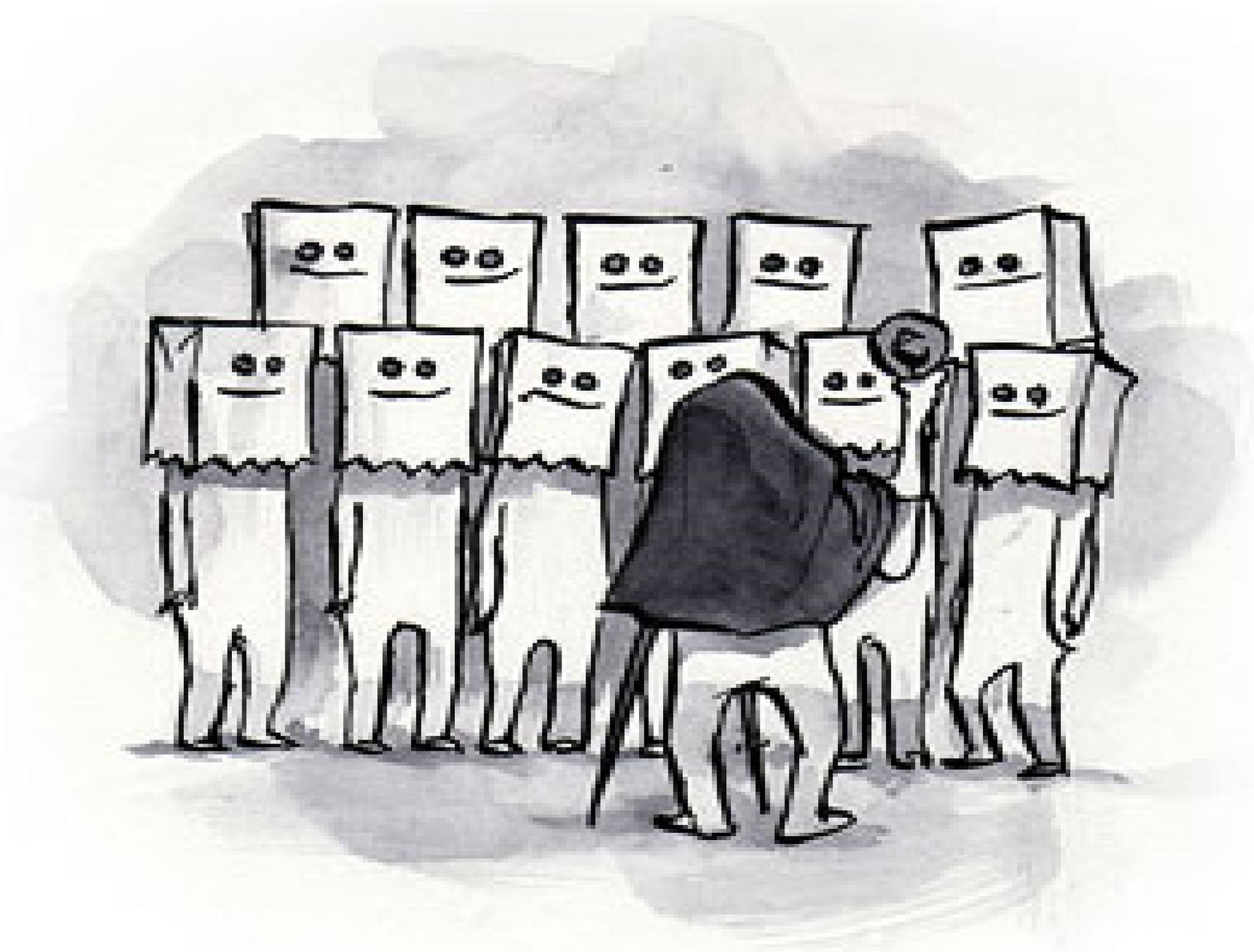}
    \end{center}
  \vspace{-0.75cm}
\end{wrapfigure}
Since the beginning of human society, anonymity has been an
important issue. For instance, people have always felt the need to
be able to express their opinions without being identified,
because of the fear of social and economical retribution,
harassment, or even threats to their lives.

\subsection{The relevance of anonymity nowadays}


With the advent of the Internet, the  issue of anonymity has been
magnified to extreme proportions. On the one hand, the Internet
increases the opportunities of interacting online, communicating
information, expressing opinion in public forums, etc.  On the
other hand, by using the Internet we are disclosing information
about ourselves: every time we visit a website certain data about us may be recorded. In this way,
organizations like multinational corporations can build a
permanent, commercially valuable record of our interests.
Similarly, every email we send goes through multiple control
points and it is most likely scanned by profiling software
belonging to organizations like the National Security Agency of the
USA. Such information can be used against us, ranging from
slightly annoying practices like commercial spam, to more serious
offences like stealing \rev{credit cards' information} for criminal
purposes.
	
Anonymity, however, is not limited to individual issues: it has
considerable social and political implications. In countries
controlled by repressive governments, the Internet
is\comment{arrange the formulation of this first sentence}
becoming increasingly more restricted, with the purpose of
preventing their citizens from accessing uncensored information
and from sending information to the outside world. The role of
anonymizing technologies in this scenario is twofold: (1) they can
help accessing sources of censored information via proxies (2)
they \rev{can help individuals to freely express their ideas (for instance via online forums)}.
\comment{maybe tell the stories about iran and china/google???}

The practice of censoring the Internet is actually not limited to
repressive governments. In fact, a recent research project
conducted by the universities of Harvard, Cambridge, Oxford and
Toronto, studied government censorship in 46 countries and
concluded that 25 of them, including various western countries,
filter to some extent communications concerning political or
religious positions.

Anonymizing technologies, as most technologies, can also be used
for malicious purposes. For instance, they can be used to help
harassment, hate speech, financial scams, disclosure of private information, etc. Because of their nature, they
are actually more controversial than other technologies: people
are concerned that terrorists, pedophiles, or other criminals
could take advantage of them.

Whatever is the use one can make of anonymity, and the personal
view one may have on this topic, it is clearly important to be
able to assess the degree of anonymity of a given system. This is
one of the aims of this thesis.

%

\subsection{Anonymizing technologies nowadays}

The most common use of anonymizing technologies is to  prevent observers from
discovering the source of communications.

\begin{wrapfigure}{r}{4cm}
  \vspace{-1cm}
    \begin{center}
      \includegraphics[width=4cm]{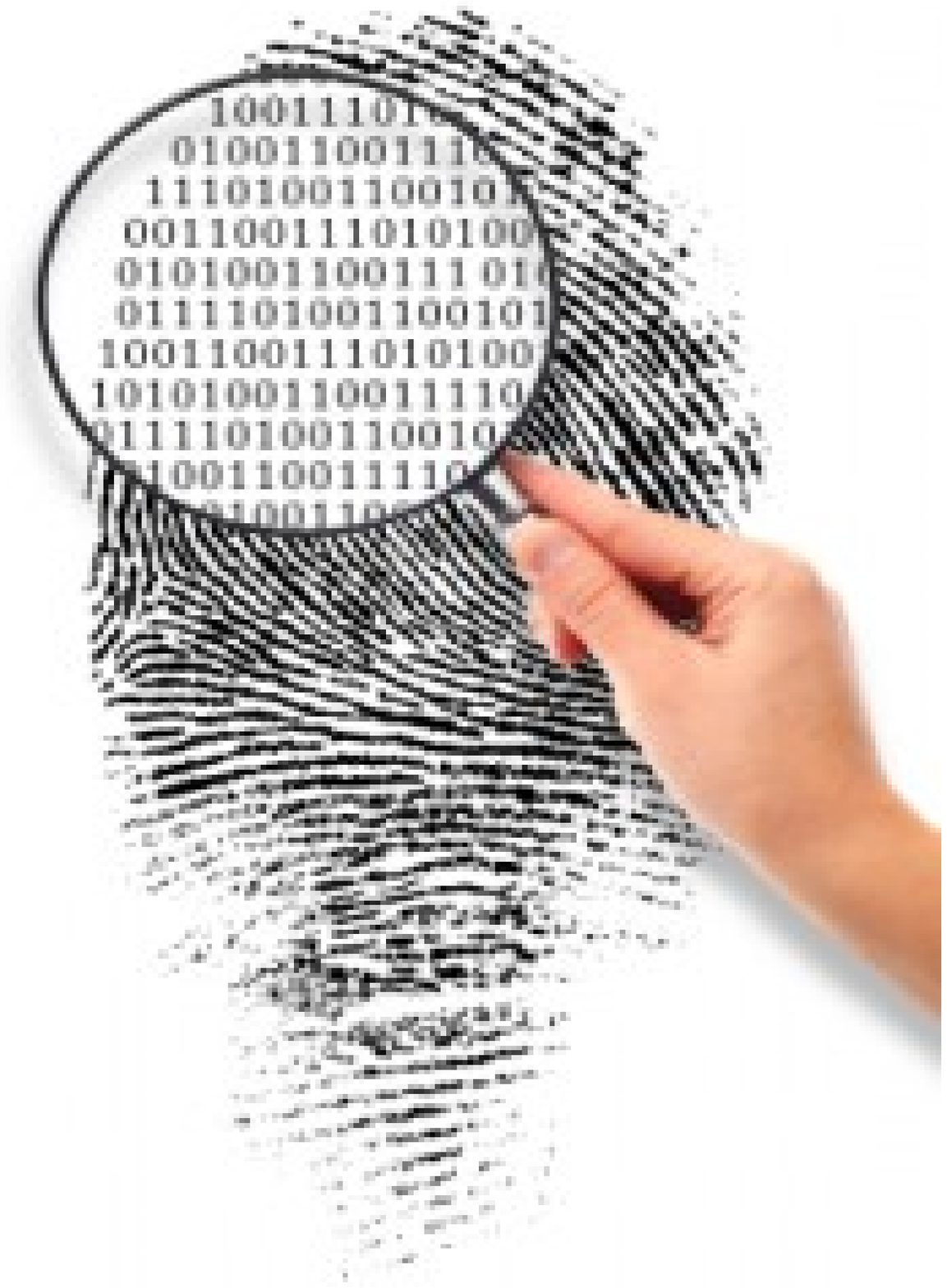}
    \end{center}
  \vspace{-0.9cm}
\end{wrapfigure}
This is not an easy task, since in general users must include in the
message information about themselves. In practice, for Internet
communication, this information is the (unique) IP address of the
computer in use, which  specifies its location in the topology of
the network. This IP number is usually logged along with the host
name (logical name of the sender). Even when the user connects
to the Internet with a temporary IP number assigned to him for a
single session, this number is in general logged by the ISP
(Internet Service Provider), which makes it possible, with the
ISP's \rev{collaboration}, to know who used a certain IP number at a
certain time and thus to find out the identity of the user.

The currently available anonymity tools aim at preventing the
observers of an online communication from learning the IP address of
the participants. Most applications rely on proxies, i.e.
intermediary computers to which messages are forwarded and which
appear then as senders of the communication, thus hiding the
original initiator of the communication. Setting up a proxy server nowadays is easy to
implement and maintain. However, single-hop architectures in which
all users enter and leave through the same proxy, create a single
point of failure which can significantly threaten the security of
the network. Multi-hop architectures have therefore been developed
to increase the performance as well as the security of the system.
In \rev{the so-called} daisy-chaining anonymization for instance,  traffic hops
deliberately via a series of participating nodes (changed for
every new communication) before reaching the intended receiver,
which prevents any single entity from identifying the user.
Anonymouse [Ans], FilterSneak [Fil] and Proxify [Pro] are well-known
free web based proxies, while Anonymizer [Ane] is currently one of
the leading commercial solutions.

%






\subsection{Anonymizing technologies: a bit of history}

Anonymous posting/reply services on the Internet were started
around 1988 and were introduced primarily for use on specific
newsgroups which discussed particularly volatile, sensitive and
personal subjects.
In 1992, anonymity services using  remailers were originated by
Cypherpunk. Global anonymity servers which served the entire
Internet soon sprang up, combining the functions of anonymous
posting as well as anonymous remailing in one service. The new
global services also introduced the concept of pseudonymity
which allowed anonymous mail to be replied.

The first popular anonymizing tool was the Penet remailer
developed by Johan Helsingius of Finland in the early 1990s. The
tool was originally intended to serve only Scandinavia but
Helsingius eventually expanded to worldwide service due to a flood
of international requests.

Based on this tool, in 1995\rev{,} Mikael Berglund made a study on how
anonymity was used. His study was based on scanning all publicly
available newsgroups in a Swedish Usenet News server. He randomly
selected a number of messages from users of the Penet remailer and
classified them by topic. His results are shown in
Table~\ref{table:statAnon}.

In 1993, Cottrell wrote the Mixmaster remailer and two years later
he launched Anonymizer \rev{which} became the first Web-based
anonymity tool. 
%
\begin{figure}[h!]
 \centering
 \qquad\includegraphics[width=08cm]{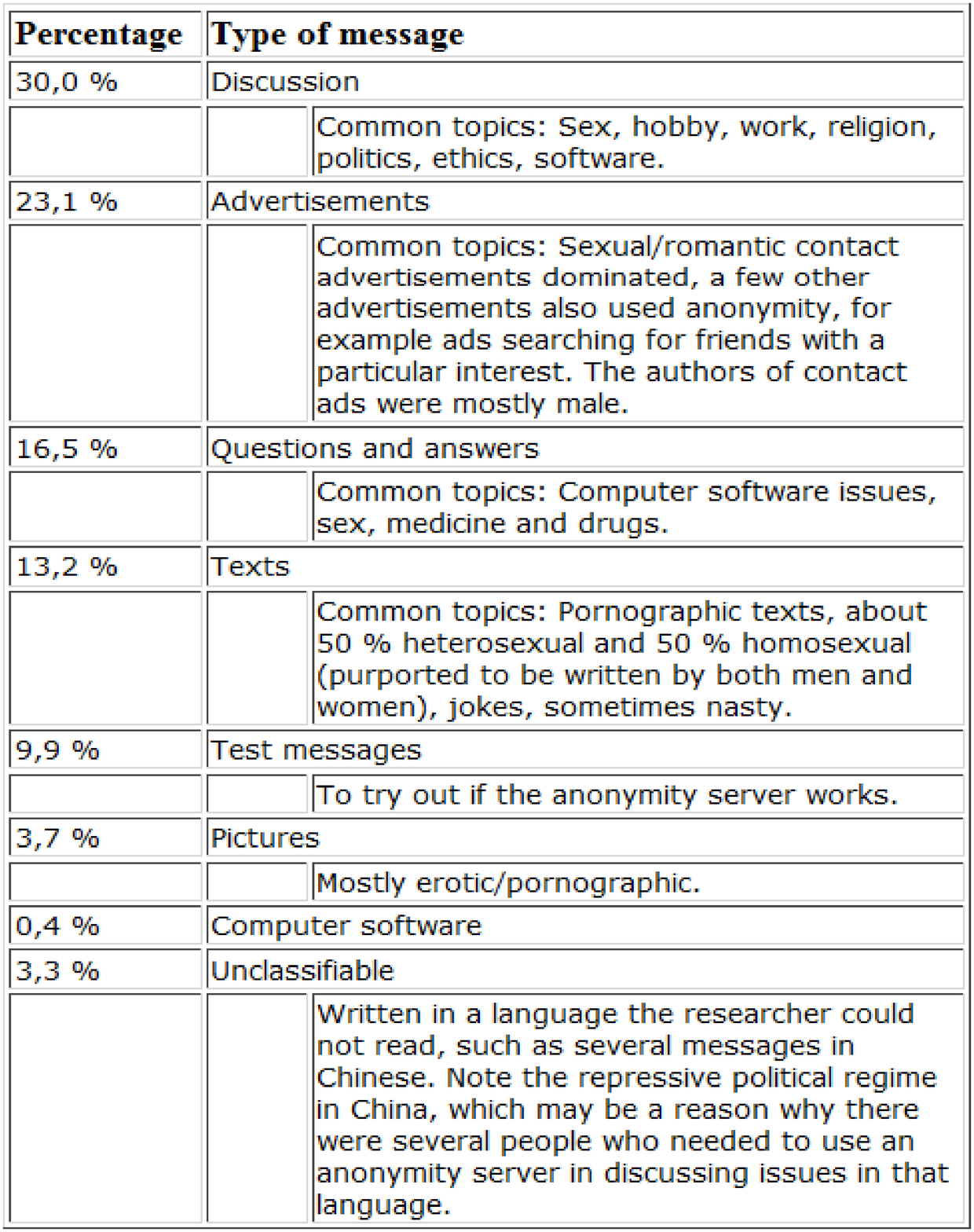}
  \caption{Statistics on the Use of Anonymity -- Penet}\label{table:statAnon}
\end{figure}
\subsection{Anonymity and computer science}

The role of computer science with respect to anonymity is twofold.
On one the hand, the theory of communication helps in the design
and implementation of anonymizing protocols. On the other hand,
like for all software systems, there is the issue of correctness,
i.e., of ensuring that the protocol achieves the expected anonymity
guarantees.

While most of the work on anonymity in the literature belongs to
the \comment{first area and then issue, not so nice...} first
challenge, \rev{this thesis addresses the second one. Ensuring
the correctness of a protocol involves (1) the use of formalisms
to precisely model the behaviour of the protocol, and (2) the
use of formalisms to specify unambiguously the desired
properties. Once the protocol and its desired properties have
been specified, it is possible to employ} verification
techniques to prove formally that the specified model satisfy
such properties. These topics belong to the
branch of computer science called \emph{formal methods}.

\section{Formal methods}

Formal methods are a particular kind of mathematically-based
techniques used in computer science and software engineering for
the specification and verification of software and hardware
systems. These techniques have their foundations on the most
diverse conceptual frameworks: logic calculi, automata theory,
formal languages, program semantics, etc.

\comment{Maybe add references?}

\subsection{The need of formal verification}

As explained in previous sections, internet technologies play an
important role in our lives. However, Internet is not the only
kind of technology we are in contact with: Every day we interact
with embedded systems such as mobile phones, smart cards, GPS
receivers, videogame consoles, digital cameras, DVD players, etc.
%
%
Technology also plays an important role \rev{in critical-life systems, i.e., systems where the malfunction of any component may incur in life losses. Example of such systems can be found in the areas of medicine, aeronautics, nuclear energy generation, etc.}

The malfunction of a technological device can have important
negative consequences ranging from material to life loss. In
the following we list some famous examples of disasters caused by software failure.


\begin{wrapfigure}{r}{5cm}
  \vspace{-0.75cm}
    \begin{center}
      \includegraphics[width=5cm]{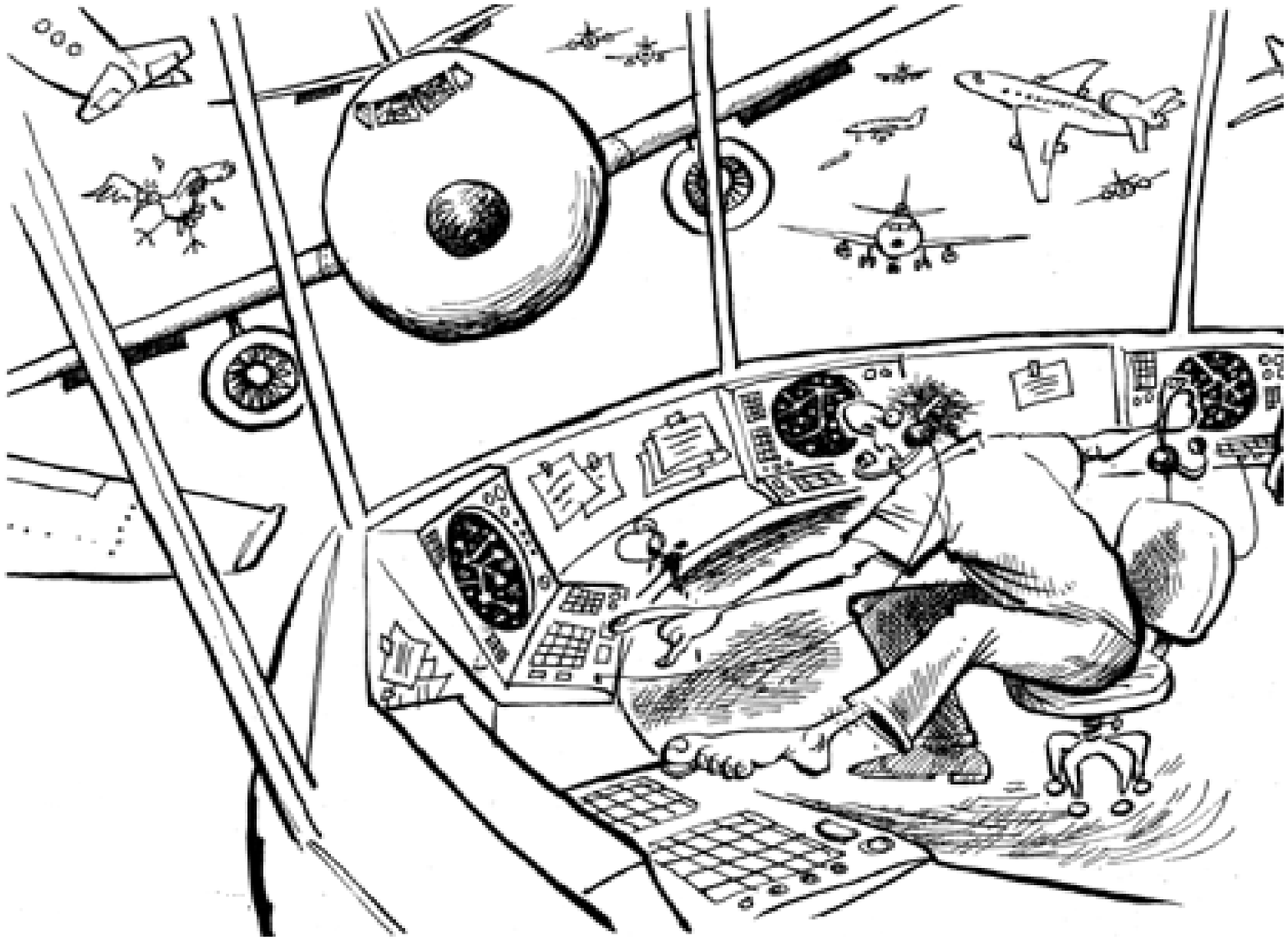}
    \end{center}
  \vspace{-0.5cm}
\end{wrapfigure}
\paragraph{Material loss:} In 2004, the Air Traffic Control
Center of Los Angeles International Airport lost communication
with Airplanes causing the immediate suspension of all operations.
The failure in the radio system was due to a 32-bit countdown
timer that decremented every millisecond.
Due to a bug in the software, when the counter reached zero \rev{the system shut down unexpectedly}. This communication outage disrupted
about 600 flights (including 150 cancellations) impacting over
30.000 passengers \rev{and causing millionaire losses to airway companies involved}.

In 1996, an Ariane 5 rocket launched by the European Space Agency
exploded just forty seconds after lift-off. The rocket was on its
first voyage, after a decade of development costing U\$S 7
billion. The destroyed rocket and its cargo were valued at U\$S
500 million. A board of inquiry investigated the causes of the
explosion and in two weeks issued a report. It turned out that the
cause of the failure was a software error in the inertial
reference system. Specifically a 64 bit floating point number
related to the horizontal velocity of the rocket was converted to
a 16 bit signed integer.



In the early nineties a bug (discovered by a professor of
Lynchburg College, USA) in the floating-point division unit of
the processor Intel Pentium II not only severely damaged Intel's
reputation, but it also forced the replacement of faulty
processors causing a loss of 475 million US dollars for the
company.


\setlength{\abovecaptionskip}{-4pt plus 1pt minus 1pt}

\begin{wrapfigure}{r}{5.1cm}
  \vspace{-0.5cm}
    \begin{center}
      \includegraphics[width=5.1cm]{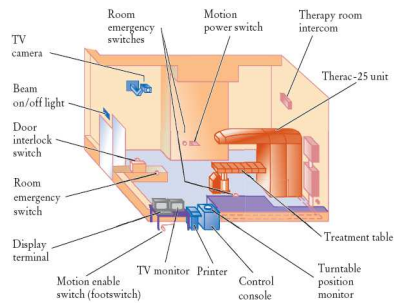}
    \end{center}
  \caption{\!\!Therac-25 Facility.}
  \vspace{-0.25cm}
\end{wrapfigure}
\paragraph{Fatal loss:} A software flaw in the control part of the radiation therapy
machine Therac-25 caused the death of six cancer patients between
1985 and 1987 as they were exposed to an overdose of radiation.

In 1995 the American Airlines Flight 965 connecting Miami and Cali
crashed just five minutes before its scheduled arrival. The
accident led to a total of 159 deaths. Paris Kanellakis, a well
known researcher (creator of the partition refinement algorithm,
broadly used to verify bisimulation), was in the flight together
with his family. Investigations concluded that the accident was 
originated by a sudden turn of the aircraft caused by the autopilot after an instruction
of one of the pilots: the pilot input `R' in the navigational
computer referring to a \rev{location called `Rozo' but the computer erroneously
interpreted it as a location called `Romeo' (due to the spelling similarity and physical proximity of the locations)}.



%



As the use and complexity of technological devices grow quickly,
mechanisms to improve their correctness have become unavoidable.
But, how can we be sure of the correctness of such technologies,
with thousands (and sometimes, millions) of components interacting
in complex ways? One possible answer is by using \emph{formal
verification}, a branch of formal methods.

\subsection{Formal verification}

Formal verification is considered a fundamental area of study in
computer science. In the context of hardware and software systems,
formal verification is the act of proving or disproving the
correctness of the system with respect to a certain property,
using formal methods. In order to achieve this, it is necessary to
construct a \rev{mathematical} model describing all possible behaviors of the
system. In addition, the property must be formally specified
avoiding, in this way, possible ambiguities.


Important formal verification techniques include theorem proving,
simulation, testing, and model checking. In this thesis \rev{we focus on the use} of  this last technique.

\paragraph{Model checking} Model checking is an automated
verification technique that, given a finite model of the system
and a formal property, systematically checks whether the property
holds in the model or not. In addition, if the property is falsified,
debugging information is provided in the form of a
\emph{counterexample}. This situation is represented in
Figure~\ref{fig:ModelCheckingChart}.

\begin{figure}[h!]
 \centering
 \qquad\includegraphics[width=11cm]{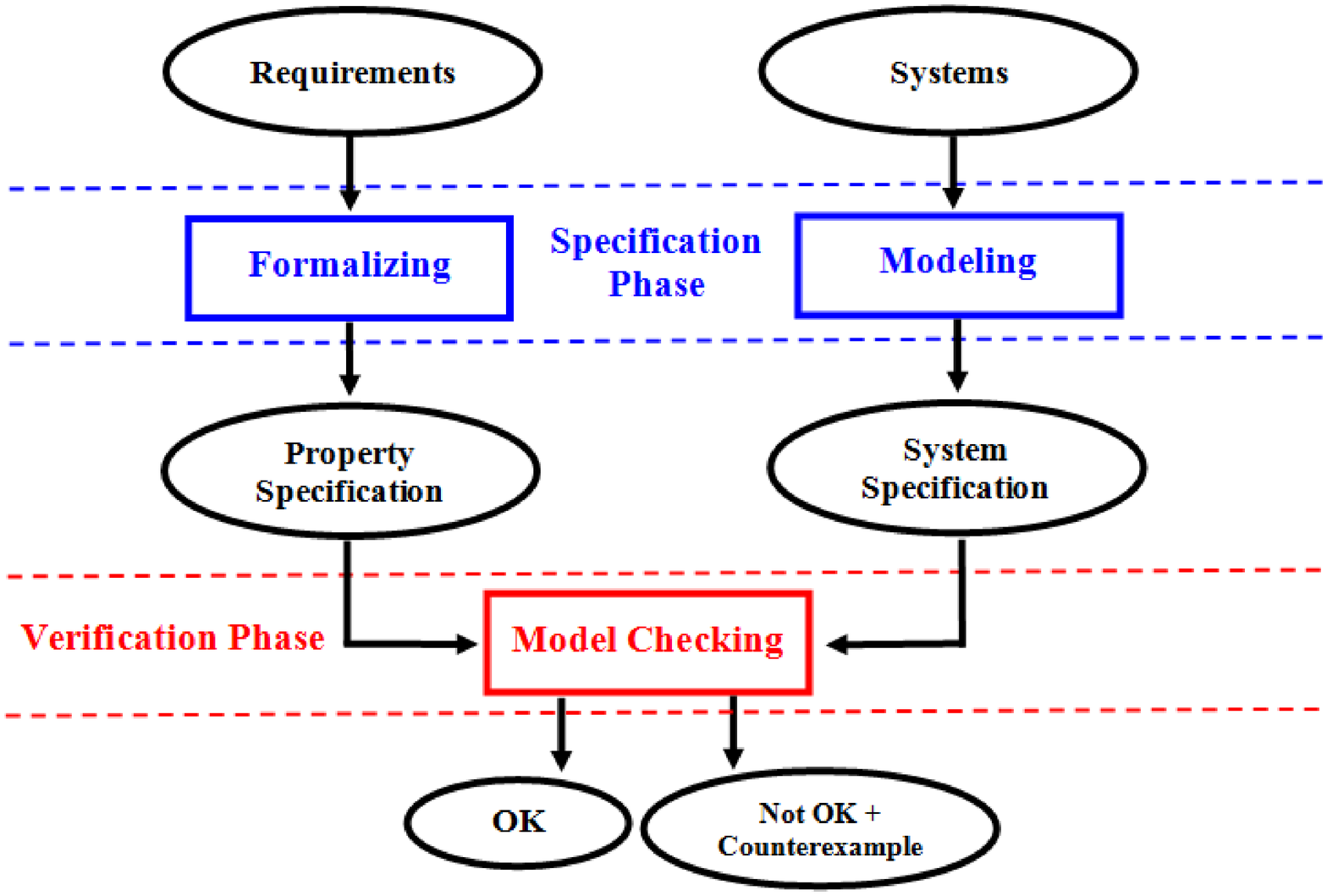}
  \caption{Schematic view of model-checking approach}\label{fig:ModelCheckingChart}
\end{figure}

Usual properties that can be verified are ``Can the system reach
a deadlock state?'', or ``Every sent message is received with
probability at least 0.99?''. Such automated verification is
carried on by a so-called \emph{model checker}, an algorithm that
exhaustively searches the space state of the model looking for
states violating the (correctness) property.

A major strength of model checking is the capability of generating
counterexamples\index{counterexample} which provide
diagnostic information in case the property is violated. Edmund M.
Clarke, one of the pioneers of Model Checking
said~\cite{Clarke:08:modelchecking}: \emph{``It is impossible to
overestimate the importance of the counterexample feature. The
counterexamples are invaluable in debugging complex systems. Some
people use model checking just for this feature''}. In case a
state violating the property under consideration is encountered,
the model checker provides a counterexample describing a possible
execution that leads from the initial state of the system to a
violating state.

%
Other important advantages of model checking are: it is highly
automatic so it requires little interaction and knowledge of
designers, it is rather fast, it can be applied to a large range of
problems, it allows partial specifications.

The main disadvantage of model checking is that the space state of
certain systems, for instance distributed systems, can be rather
large, thus making the verifications inefficient and in some cases
even unfeasible (because of memory limitations). This problem is
known as the \emph{state explosion problem}. Many techniques to
alleviate it have been proposed since the invention of model
checking. Among the most popular ones we mention the use Binary
Decision Diagrams (BDDs), partial order reduction, abstraction,
compositional reasoning, and symmetry reduction. State-of-the-art
model checkers can easily handle up to $10^9$ states with explicit
state representation. For certain specific problems, more
dedicated data structures (like BDDs) can be used thus making it
possible to handle even up to $10^{476}$ states.

The popularity of model checking has grown considerably since
its invention at the beginning of the 80s. Nowadays, model
checking techniques are employed by most or all leading hardware
companies (e.g. INTEL, IBM and MOTOROLA - just to mention few of
them). While model checking is applied less frequently by
software developing companies, there have been several cases in
which it has helped to detect previously unknown defects in
real-world software. %
%
A prominent example is the result of research in Microsoft's SLAM project in which several formal
techniques were used to automatically detect flaws in device
drivers. In 2006, Microsoft released the Static Driver Verifier
as part of Windows Vista, SDV uses the SLAM
software-model-checking engine to detect cases in which device
drivers linked to Vista violate one of a set of interface rules.
Thus SDV helps uncover defects in device drivers, a primary
source of software bugs in Microsoft applications.
Investigations have shown that model checking procedures would
have revealed the exposed defects in, e.g., Intel�s Pentium II
processor and the Therac-25 therapy radiation
machine.\comment{also air control system example}
\comment{maybe a list of popular model checkers nowadays}

\paragraph{Focus of this thesis}

This thesis addresses the foundational aspects of formal methods
for applications in security and in particular in anonymity: We
investigate various issues that have arisen in the area of
anonymity, we develop frameworks for the specification of
anonymity properties, and we propose algorithms for their
verification.

\section{Background}\label{sec:Intro-anonymity}


In this section we give a brief overview of the various approaches
to the foundations of anonymity that have been explored in the
literature. We will focus on anonymity properties, although the
concepts and techniques developed for anonymity apply to a large
extent also to neighbor topics like  \emph{information flow},
\emph{secrecy}, \emph{privacy}. The common denominator of these
problems is the prevention of the leakage of information. More
precisely, we are concerned with situations  in which  there are
certain values (data, identities, actions, etc) that are intended
to be secret, and we want to ensure that an adversary will not be
able to infer the secret values from the information which is
publicly available. Some researchers use the term
\emph{information hiding} to refer to this class of
problems~\cite{Halpern:05:JCS}.

The frameworks for reasoning about anonymity can be classified
into two main categories: the \emph{possibilistic} approaches, and
the \emph{probabilistic} (or \emph{quantitative}) ones.

 \subsubsection{Possibilistic notions}
The term ``possibilistic'' refers to the fact that we do not
consider quantitative aspects. More precisely,  anonymity is
formulated in terms of the possibility or inferring some secrets,
without  worrying about ``how likely'' this is, or ``how much''
we narrow down the secret.

These  approaches have been widely explored in the  literature,
using different conceptual frameworks. Examples include the
proposals based on  epistemic logic
(\cite{Syverson:99:FM,Halpern:05:JCS}), on \qm{function views}
(\cite{Hughes:04:JCS}), and  on process equivalences (see for
instance \cite{Schneider:96:ESORICS,Ryan:01:BOOK}). In the
following we will focus on the latter kind.

In general, possibilistic anonymity means that the observables do
not identify a unique culprit. Often this property relies on
\emph{ nondeterminism}:  for each culprit, the system should be
able to produce alternative executions with different observables,
so that in turn  for a given observable there are many agents that
could be the culprit. More precisely, in its strongest version
this property can be expressed as follows: if in one computation
the identity of the culprit is $i$ and the observable outcome is
$o$, then for every other agent $j$ there must be a computation
where, with culprit $j$, the observable is still $o$.


This kind of approach can be applied also in case of systems that
use randomization. The way this is done is by abstracting the
probabilistic choices into nondeterministic ones. See for example
the Dining Cryptographers example in  \cite{Schneider:96:ESORICS},
where the coin tossing is represented by a nondeterministic
process.

In general the possibilistic approaches have the advantages of
simplicity an efficiency. On the negative side, they lack
precision, and in some cases the approximation can be rather
loose. This is  because every scenario that has a non-null
probability is interpreted as possible. For instance, consider the
case in which a system reveals  the culprit  90 percent of the
times by outputting his identity, while in the remaining 10
percent of the times it outputs the name of some other agent. The
system would not look very anonymous. Yet,
 the possibilistic definition of anonymity  would be satisfied
because all users would appear as possible culprits to the
observer regardless of the output of the system.  In general, in
the \rev{possibilistic} approach the strongest notion of anonymity we can
express is \emph{possible innocence}, which is satisfied when no agent
appear to be the culprit \emph{for sure}: there is always the possibility
that he is innocent (no matter how unlikely it is).

In this thesis we consider only the probabilistic approaches.
Their common feature is that they deal with probabilities in a
concrete way and they are, therefore, much more precise. They have
become very popular in recent times, and there has been a lot of
work dedicated to understanding and  formalizing the notion in a
rigorous way.  In the next section we give a brief  overview of
these efforts.

\subsubsection{Probabilistic notions}

These approaches take  probabilities into account, and are based
on the likelihood that an agent is the culprit, for a given
observable. One notion  of probabilistic anonymity which has been
thoroughly  investigated in the literature is \emph{strong
anonymity}.

\paragraph{Strong anonymity\index{strong anonymity}}
Intuitively the idea behind this notion is that the observables
should not allow to infer any (quantitative) information about the
identity of the culprit. The corresponding notion in the field of
information flow is (quantitative) non-interference.

Once we try to formalize more precisely the above notion we
discover however that there are various possibilities.
Correspondingly, there have been various proposals. We recall here
the three most prominent ones.

\begin{enumerate}
\item {\it Equality of the a posteriori probabilities for different culprits.}
The idea is to consider   a system   strongly anonymous if, given
an observable $o$, the
 \emph{a posteriori} probability that the identity of the culprit is $i$, $\PP(i|o)$,
 is the same as the  \emph{a posteriori}
probability of any other identity $j$. Formally:
\begin{eqnarray}\label{eqn:one}
\PP(i|o)    =   \PP(j|o) \quad \mbox{for all observables $o$, and
all identities $i$ and $j$}
\end{eqnarray}
This is the spirit of  the definition of \emph{strong anonymity}
by Halpern and O'Neill \cite{Halpern:05:JCS}, although their
formalization involves  more sophisticated  epistemic notions.

\item{\it Equality of the a posteriori and a priori probabilities for the same culprit.}
Here the idea is to consider a system strongly anonymous if, for
any observable, the \emph{a posteriori} probability that the
culprit is a certain agent $i$  is the same as its  \emph{a
priori} probability. In other words, the observation does not
increase or decrease the support for suspecting a certain agent.
Formally:
\begin{eqnarray}\label{eqn:two}
\PP (i|o)  =   \PP(i) \quad \mbox{for all observables $o$, and all
identities $i$}
\end{eqnarray}
This is the definition of \emph{anonymity} adopted by Chaum in
\cite{Chaum:88:JC}. He also proved that the Dining Cryptographers
satisfy this property if the coins are fair. Halpern and O'Neill
consider  a similar property in their epistemological setting, and
they call it \emph{conditional anonymity} \cite{Halpern:05:JCS}.

\item{\it Equality of the likelihood of different culprits.}
In this third definition a system is strongly anonymous if, for
any observable $o$  and agent $i$, the \emph{likelihood} of   $i$
being the culprit, namely $\PP(o|i)$, is the same as the
likelihood of any other agent $j$. Formally:
\begin{eqnarray}\label{eqn:three}
\PP(o|i)  =   \PP(o|j) \quad \mbox{for all observables $o$, and
all identities $i$ and $j$}
\end{eqnarray}
This was proposed as definition of \emph{strong anonymity} by
Bhargava and Palamidessi \cite{Bhargava:05:CONCUR}.
\end{enumerate}

In \cite{Beauxis:08:TGC} it has been proved  that definitions
(\ref{eqn:two}) and (\ref{eqn:three}) are equivalent. Definition
(\ref{eqn:three}) has the advantage that it does extend in a
natural way to the case in which the choice of the culprit is
nondeterministic. This could be useful when we do not know the a
priori distribution of the culprit, or when we want to abstract
from it (for instance because we are interested in the worst
case).

Concerning Definition (\ref{eqn:one}), it probably looks at first
sight the most natural, but it actually turns out to be way too
strong.  In fact it is equivalent to (\ref{eqn:two}) and
(\ref{eqn:three}), \emph{plus} the following  condition:
\begin{eqnarray}\label{eqn:four}
\PP(i)   = \PP(j) \qquad \mbox{for all identities $i$ and $j$}
\end{eqnarray}
namely the condition that   the a priori distribution  be uniform.

It is interesting to notice that  (\ref{eqn:one}) can be split in
two orthogonal properties:  (\ref{eqn:three}), which depends only
in the protocol, and (\ref{eqn:four}), which depends only in the
distribution on the secrets.

Unfortunately all the strong anonymity properties discussed above
are too strong, almost never achievable in practice. Hence
researches have started exploring weaker notions. One of the most
renowned properties of this kind  (among the ``simple'' ones based
on conditional probabilities) is that of \emph{probable
innocence}\index{anonymity!probable innocence}.

\paragraph{Probable innocence\index{anonymity!probable innocence}}
The notion of  \emph{probable innocence} was formulated by Rubin
and Reiter  in the context of their work on the Crowds
protocol\index{protocols!crowds} \cite{Reiter:98:TISS}.
Intuitively the idea is that, after the observation,  no agent is
more likely to be the culprit than not to be. Formally:
\[
\PP(i|o) \leq \PP(\neg i|o)  \quad \mbox{for all observations o,
and all identities $i$}
\]
or equivalently
\[
\PP(i|o) \leq \frac{1}{2} \quad \mbox{for all observations o, and
all identities $i$}
\]
In  \cite{Reiter:98:TISS} Rubin and Reiter proved that the Crowds
protocol satisfies probable innocence under a certain assumption
on the number of attackers relatively to the number of honest
users.

All the notions discussed above are rather coarse, in the sense
that they are cut-off notions and do not allow to represent small
variations in the degree of  anonymity. In order to be able to
compare  different protocols in a more precise way, researcher
have started exploring settings to measure the \emph{degree of
anonymity}. The most popular of these approaches are
those based in information theory.

 \subsubsection{Information theory\index{information theory}}
The underlying idea is that   anonymity  systems are interpreted
as \emph{channels} in the information-theoretic sense. The  input
values are the possible identities of the culprit, which,
associated to a probability distribution, form a random variable
$\mathit{Id}$. The outputs are the observables,  and the
transition matrix consists of  the conditional probabilities of
the form $\PP(o|i)$, representing the probability that the system
produces an observable $o$ when the culprit is $i$. A central
notion here is the Shannon entropy, which represents the
uncertainty of a random variable. For the culprit's possible
identity, this is given by:
\[
 H(\mathit{Id}) = - \sum_{i}\PP(i)\log \PP(i) \quad \mbox{(uncertainty a priori)}
 \]
Note that $\mathit{Id}$ and the matrix also determine a
probability distribution on the observables, which can then be
seen as another random variable $\mathit{Ob}$. The
\emph{conditional entropy} $H(\mathit{Id}|\mathit{Ob})$,
representing the uncertainty about the identity of the culprit
\emph{after} the observation, is given by
\[
H(\mathit{Id}|\mathit{Ob}) = - \sum_{o} \PP(o)
\sum_{i}\PP(i|o)\log \PP(i|o) \quad \mbox{(uncertainty a
posteriori)}\] It can be shown that $0 \le
H(\mathit{Id}|\mathit{Ob}) \le H(\mathit{Id})$. We have
$H(\mathit{Id}|\mathit{Ob})=0$ when there is no uncertainty left
about $\mathit{Id}$ after the value of $\mathit{Ob}$ is known.
Namely, when the value of $\mathit{Ob}$ completely determines the
value of $\mathit{Id}$. This is the case of maximum leakage. At
the other extreme, we have
$H(\mathit{Id}|\mathit{Ob})=H(\mathit{Id})$   when $\mathit{Ob}$
gives no information  about $\mathit{Id}$, i.e. when $\mathit{Ob}$
and $\mathit{Id}$ are independent.

The difference between  $H(\mathit{Id})$ and
$H(\mathit{Id}|\mathit{Ob})$ is called \emph{mutual information}
and it is denoted by $I(\mathit{Id};\mathit{Ob})$:
\[
 I(\mathit{Id};\mathit{Ob}) = H(\mathit{Id}) - H(\mathit{Id}|\mathit{Ob})
 \]
The maximum mutual information between $\mathit{Id}$ and
$\mathit{Ob}$ over all possible input distributions
$\PP_\mathit{Id}(\cdot)$ is known as the channel's
\emph{capacity}:
\[
 C = \max_{\PP_\mathit{Id}(\cdot)}I(\mathit{Id};\mathit{Ob})
 \]

In the case of anonymity, the  mutual information represents the
difference between the a priori and the a posteriori uncertainties
about the identity of the culprit. It can therefore be considered
as the  leakage of information due to the system, i.e. the amount
of anonymity which is lost because of  the observables produced by
the system. Similarly, the capacity represents the worst-case
leakage under all possible distributions on the culprit's
possible identities. It can be shown that the capacity is $0$ if and only
if the rows of the matrix are pairwise identical. This corresponds
exactly to the version (\ref{eqn:three}) of strong anonymity.

This view of the degree  of anonymity has been advocated in
various works, including
\cite{Moskowitz:03:WPES,Moskowitz:03:CNIS,Zhu:05:ICDCS,Chatzikokolakis:08:IC}.
In the context of information flow, the same view of leakage in
information theoretic terms has been widely investigated as well.
Without pretending to be exhaustive, we mention
\cite{McLean:90:SSP,Gray:91:SSP,Clark:01:QAPL,Clark:05:QAPL,Lowe:02:CSFW,Boreale:06:ICALP}.

In \cite{Smith:09:FOSSACS} Smith has investigated the use of an
alternative notion of entropy, namely R\'enyi's min entropy
\cite{Renyi:60:Berkeley}, and has proposed to define leakage as
the analogous of mutual information in the setting of R\'enyi's
min entropy. The justification for proposing  this variant is that
it represents better certain attacks called \emph{one-try
attacks}. In general, as K\"opf and Basin illustrate in their
cornerstone paper \cite{Kopf:07:CCS}, one can use the above
information-theoretic approach with many different notions of
entropy,  each representing a different model of attacker, and a
different way of measuring the success of an attack.

A different  information-theoretic approach to leakage has been
proposed in  \cite{Clarkson:09:JCS}:  in that paper, the authors
define as information leakage the difference between the a priori
accuracy of the guess of the attacker, and the a posteriori one,
after the attacker has made his observation. The accuracy of the
guess is defined as the Kullback-Leibler distance between the
\emph{belief} (which is a weight attributed by the attacker to
each input hypothesis) and the true distribution on the
hypotheses. In \cite{Hamadou:10:SSP} a R\'enyi's min entropy
variant of this approach has been explored as well.

We conclude this section by remarking that, in all the approaches
discussed above, the notion of conditional probability plays a
central role.

\section{Contribution and plan of the thesis}
We have seen in Section~\ref{sec:Intro-anonymity} that conditional
probabilities are the key ingredients of all quantitative
definitions of  anonymity. It is therefore desirable to develop
techniques to \rev{analyze and} compute such probabilities.

Our first contribution is cpCTL, a temporal logic allowing us to
specify properties concerned with conditional probabilities in
systems combining probabilistic  and nondeterministic behavior.
This is presented in Chapter 2. cpCTL is essentially pCTL
(probabilistic Computational Tree Logic) \cite{Hansson:94:FACS}
enriched with formulas of the kind $\PP_{\leq a}[\phi|\psi]$,
stating that the probability of $\phi$ given $\psi$ is at most
$a$. We do so by enriching pCTL with formulas of the form
$\PP_{\bowtie a}[\phi|\psi]$. We propose a model checker for
cpCTL. Its design has been quite challenging, due to the fact that
the standard model checking algorithms for pCTL  in MDPs (Markov
Decision Processes) do not extend to conditional probability
formulas. More precisely, in contrast to pCTL, verifying a
conditional probability cannot be reduced to a linear optimization
problem. A related point is that, in contrast to pCTL, the optimal
probabilities are not attained by history independent schedulers.
We attack the problem by proposing the notion of \emph{semi
history independent schedulers}, and we show that these schedulers
do attain optimality with respect to the  conditional
probabilities. Surprisingly, it turns out that we can  further
restrict to deterministic schedulers, and still attain optimality.
Based on these results, we show that it is decidable whether a
cpCTL formula is satisfied in a MDP, and we provide an algorithm
for it. In addition, we define the notion of counterexample for
the logic and sketch an algorithm for counterexample generation.

Unfortunately, the verification of  conditional cpCTL formulae is
not efficient in the presence of nondeterminism. Another issue,
related to nondeterminism within the applications  in the field of
security, is the well known problem of almighty schedulers (see Chapter 4). Such
schedulers have the (unrealistic) ability to peek on internal
secrets of the component and make their scheduling policy
dependent on these secrets, thus leaking the secrets to external
observers. We address these problems in separate chapters.

In Chapter 3 we restrict the framework to purely probabilistic
models where secrets and observables do not interact, and we
consider the problem of computing the leakage and the maximal
leakage  in the information-theoretic approach. These are defined
as mutual information and capacity, respectively. We address
these notions with respect to both the Shannon entropy and the
R\'enyi min entropy. We provide techniques to compute channel
matrices in $O((o \times q)^3)$ time, where $o$ is the number of
observables, and $q$ the number of states. (From the channel
matrices, we can compute mutual information and capacity using
standard techniques.) We also show that, contrarily to what was
stated in literature, the standard  information theoretical
approaches to leakage do not extend to the case in which secrets
and observable interact.

In Chapter 4 we consider the problem of the almighty schedulers.
We define a restricted family of schedulers (\emph{admissible
schedulers}) which cannot base their decisions on secrets, thus
providing more realistic notions of strong anonymity than
arbitrary schedulers. We provide a framework to represent
concurrent systems composed by purely probabilistic components. At
the global level we still have nondeterminism, due to the various
possible ways the component may interact with each other.
Schedulers are then defined as devices that select at every point
of the computation the component(s) moving next. Admissible
schedulers make this choice independently  from the values
of the secrets. In addition, we provide a sufficient (but not
necessary) technique based on automorphisms to prove strong
anonymity for this family of schedulers.

The notion of counterexample has been approached indirectly in
Chapters 2 and 3. In Chapter 5 we come back and fully focus on
this topic. We propose a novel technique to generate
counterexamples for model checking on Markov Chains. Our propose
is to group together violating paths that are likely to provide
similar debugging information thus alleviating the debugging
tasks. We do so by using strongly connected component analysis and
show that it is possible to extend these techniques to Markov
Decision Processes.

\rev{Chapter 6 is an overview chapter. There we briefly describe extensions to  the frameworks presented in Chapters 3
and 4\footnote{For more information about the topics discussed in this chapter we refer the reader to \cite{Alvim:10:CONCUR,Alvim:10:TRa,Alvim:10:LICS,Alvim:10:IFIP-TCS}.}.} First, we consider the case in which secrets and
observables interact, and show that it is still possible to define
an information-theoretic notion of leakage, provided that we
consider a more complex notion of channel, known in literature as
\emph{channel with memory and feedback}. Second, we extend the
systems proposed in Chapter 4 by allowing nondeterminism also
internally to the components. Correspondingly, we define a richer
notion of admissible scheduler and we use it for defining
notion of process equivalences relating to nondeterminism in a
more flexible way than the standard ones in the literature. In
particular, we use these equivalences  for defining notions of
anonymity robust with respect to implementation refinement.

In Figure \ref{fig:chaptersRel} we describe the relation between
the different chapters of the thesis. Chapter 5 is not explicitly
depicted in the figure because it does not fit in any of the branches of cpCTL
(efficiency - security foundations). However, the techniques
developed in Chapter 5 have been applied to the works in both Chapters 2 and 3.

\begin{figure}[h!]
 \centering
 \qquad\includegraphics[width=8cm]{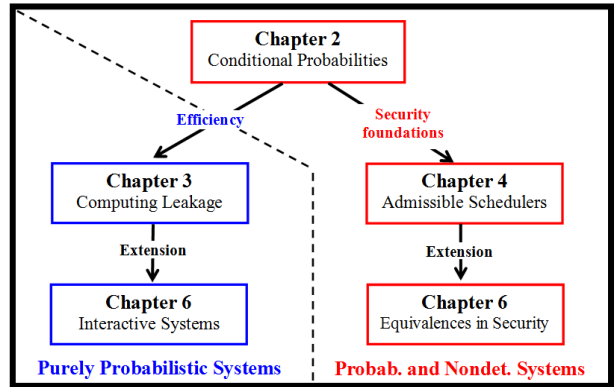}
  \caption{Chapters relation.}\label{fig:chaptersRel}
\end{figure}

%

\rev{We conclude this thesis In Chapter 7, there we present a summary of our main contributions  and discuss further directions.}

\section{Origins of the Chapters and Credits}

\rev{In the following we list, for each chapter, the set of related articles together with their publication venue and corresponding co-authors.}

\begin{itemize}
\item Chapter 2 is mainly based on the article \cite{Andres:08:TACAS} by Peter van Rossum and
myself. The article was presented in TACAS 2008. In addition, this
chapter contains material of an extended version of
\cite{Andres:08:TACAS} that is being prepared for submission to a
journal.
\item Chapter 3 is based on the article \cite{Andres:10:TACAS} by Catuscia Palamidessi, Peter van Rossum,
Geoffrey Smith and myself. The article was presented in TACAS 2010.
\item Chapter 4 is based on
  \begin{itemize}
    \item The article \cite{Andres:10:QEST} by Catuscia Palamidessi, Peter van Rossum, Ana Sokolova and myself. This article was presented in QEST 2010.
    \item The journal article \cite{Andres:10:Admissible} by the same authors.
  \end{itemize}
\item Chapter 5 is based on the article \cite{Andres:08:HVC}  by Pedro D'Argenio, Peter van Rossum, and
  myself. The article was presented in HVC 2008.
\item Chapter 6 is based on
  \begin{itemize}
    \item The article \cite{Alvim:10:LICS} by M\'ario S. Alvim, Catuscia Palamidessi, and myself. This work was presented in LICS 2010 as part of an invited talk by Catuscia Palamidessi.
    \item The article \cite{Alvim:10:CONCUR} by M\'ario S. Alvim, Catuscia Palamidessi, and myself. This work presented in CONCUR 2010.
    \item The journal article \cite{Alvim:10:TRa} by the same authors of the previous works.
    \item The article \cite{Alvim:10:IFIP-TCS} by M\'ario S. Alvim, Catuscia Palamidessi, Peter van Rossum, and myself. This work was presented in IFIP-TCS 2010.
  \end{itemize}
\end{itemize}

The chapters remain close to their published versions, thus there
is inevitably some overlapping between them (in particular in
their introductions where basic notions are explained).

\emph{A short note about authorship:} I am the first author in all
the articles and journal works included in this thesis with the
exception of the ones presented in Chapter 6.

\newpage
\thispagestyle{empty}

\chapter{Conditional Probabilities over Probabilistic and Nondeterministic Systems}
\chaptermark{Conditional probabilistic temporal logic}
\label{chap:cpCTL}

\begin{quote}
\textit{In this chapter we introduce \cpCTL, a logic which extends
the probabilistic temporal logic \pCTL with conditional
probabilities allowing to express statements of the form ``the
probability of $\phi$ given $\psi$ is at most $a$''.  We interpret
\cpCTL over Markov Chains \index{markov chain} and Markov Decision
Processes.  While model checking \cpCTL over Markov Chains
\index{markov chain} can be done with existing techniques, those
techniques do not carry over to Markov Decision Processes. We
study the class of schedulers that suffice to find the maximum and
minimum conditional probabilities, show that the problem is
decidable for Markov Decision Processes and propose a model
checking algorithm. Finally, we present the notion of
counterexamples for \cpCTL model checking and provide a method for
counterexample generation.}
\end{quote}

\section{Introduction}
\label{sec:intro}

  Conditional probabilities\index{conditional probability} are a fundamental concept in probability theory.
  In system validation these appear for instance in
  anonymity, risk assessment, and diagnosability.
  Typical examples here are:
  the probability that a certain message was sent by Alice, given that
  an intruder observes a certain traffic pattern;
  the probability that the dykes break, given that it rains heavily;
  the probability that component A has failed, given error message E.

  In this chapter we introduce \cpCTL (conditional probabilistic CTL), a
  logic which extends strictly the probabilistic temporal logic
  \pCTL \cite{hj_1989_reliability} with new probabilistic
  operators of the form $\CP{}{\leq a}{\phi}{\psi}$. Such formula
  means that the probability of $\phi$ given $\psi$ is at most
  $a$.
  We interpret \cpCTL formulas over Markov Chains \index{markov chain} (MCs) and Markov
  Decision Processes (MDPs).  Model checking \cpCTL over MCs can be
  done with model checking techniques for \pCTLstar, using the
  equality $\CP{}{}{\phi}{\psi} = \Prob{}{}{\phi \land
    \psi}/\allowbreak\Prob{}{}{\psi}$.

  In the case of \MDPs, \cpCTL model checking is significantly more complex.
  Writing $\CP{}{\eta}{\phi}{\psi}$ for the probability
  $\CP{}{}{\phi}{\psi}$ under scheduler $\eta$, model checking
  $\CP{}{\leq a}{\phi}{\psi}$ reduces to computing
  $\CP{+}{}{\phi}{\psi} = \max_\eta \CP{}{\eta}{\phi}{\psi} =
  \allowbreak\max_\eta \allowbreak\Prob{}{\eta}{\phi \land \psi}/\Prob{}{\eta}{\psi}$.
  Thus, we have to maximize a non-linear function.  (Note that in
  general $\CP{+}{}{\phi}{\psi} \not= \Prob{+}{}{\phi
    \land \psi} / \Prob{+}{}{\psi}$.) Therefore, we cannot reuse the
  efficient techniques for \pCTL model checking, since they heavily rely
  on linear optimization techniques~\cite{ba_1995_probabilistic}.

  In particular we show that, differently from what happens in
  \pCTL \cite{ba_1995_probabilistic}, history independent
  schedulers are not sufficient for optimizing conditional reachability
  properties.  This is because in $\cpCTL$ the optimizing schedulers
  are not determined by the local structure of the system. That
  is, the choices made by the scheduler in one branch may
  influence the optimal choices in other branches.  We introduce
  the class of semi history-independent schedulers and show that
  these suffice to attain the optimal conditional probability.
  Moreover, deterministic schedulers still suffice to attain the
  optimal conditional probability. This is surprising
  since many non-linear optimization problems attain their optimal
  value in the interior of a
  convex polytope, which correspond to randomized schedulers in
  our setting.

  Based on these properties, we present an (exponential) algorithm
  for checking whether a given system satisfies a formula in the logic.
  Furthermore, we define the notion of counterexamples for \cpCTL model checking
  and provide a method for counterexample generation.

  To the best of our knowledge, our proposal is the first temporal logic
  dealing with conditional probabilities.

\subsection*{Applications}

\paragraph{Complex Systems.}

One application of the techniques presented in this chapter is in
the area of complex system behavior. We can model the probability
distribution of natural events as probabilistic choices, and the
operator choices as non-deterministic choices. The computation of
maximum and minimum conditional probabilities can then help to
optimize run-time behavior. For instance, suppose that the desired
behavior of the system is expressed as a $\pCTL$ formula $\phi$
and that during run-time we are making an observation about the
system, expressed as a $\pCTL$ formula $\psi$. The techniques
developed in this chapter allow us to compute the maximum
probability of $\phi$ given $\psi$ and to identify the actions
(non-deterministic choices) that have to be taken to achieve this
probability.

\paragraph{Anonymizing Protocols.}

Another application is in the area of anonymizing protocols. The
purpose of these protocols is to hide the identity of the user
performing a certain action. Such a user is usually called the
\emph{culprit}. Examples of these protocols are Onion
Routing~\cite{cl_2005_onion}, Dining Cryptographers
\cite{Chaum:88:JC}, Crowds \cite{Reiter:98:TISS} and voting
protocols~\cite{foo_1992_voting} (just to mention a few). Strong
anonymity\index{strong anonymity} is commonly
formulated~\cite{Chaum:88:JC,Bhargava:05:CONCUR} in terms
of conditional probability: A protocol is considered strongly
anonymous if no information about the culprit's identity can be
inferred from the behavior of the system. Formally, this is
expressed by saying that culprit's identity and the observations,
seen as random variables, are independent from each other. That
is to say, for all users $u$ and all observations of the adversary $o$:
\begin{center}
  \textbf{P}[culprit $=$ $u$ $|$ observation $=$ $o$] $=$
  \textbf{P}[culprit $=$ $u$].
\end{center}

If considering a concurrent setting, it is customary to give the
adversary full control over the network~\cite{dy_1983_public} and
model its capabilities as nondeterministic choices in the system,
while the user behavior and the random choices in the protocol are
modeled as probabilistic choices. Since anonymity should be
guaranteed for all possible attacks of the adversary, the above
equality should hold for all schedulers. That is: the system is
strongly anonymous if for all schedulers $\eta$, all users $u$ and
all adversarial observations $o$:
\begin{center}
  \textbf{P}$_\eta$[culprit $=$ $u$ $|$ observation $=$ $o$]$=$
  \textbf{P}$_\eta$[culprit $=$ $u$]
\end{center}
Since the techniques in this chapter allow us to compute the
maximal and minimal conditional probabilities over all schedulers,
we can use them to prove strong anonymity in presence of
nondeterminism.

Similarly, probable innocence\index{anonymity!probable innocence}
means that a user is not more likely to be innocent than not to be
(where ``innocent'' mans ``not the culprit''). In \cpCTL this can
be expressed as $\CP{}{\leq 0.5}{\text{culprit} =
u\,}{\,\text{observations} = o}$.

\comment{here we should point out the problem with this kind of
almighty schedulers and say something about the non-zero leakage
case}

\paragraph{Organization of the chapter}

In Section~\ref{sec:mdp} we present the necessary background on
MDPs. In Section~\ref{cp} we introduce conditional probabilities
over MDPs and in Section~\ref{sec:cpCTL} we introduce \cpCTL.
Section~\ref{edHIS} introduces the class of semi
history-independent schedulers and Section~\ref{modelchecking}
explains how to compute the maximum and minimum conditional
probabilities. Finally, Section~\ref{sec:cpctl_counterexamples}, we
investigate the notion of counterexamples.


\section{Markov Decision Processes\index{markov decision process}}
\label{sec:mdp}

Markov Decision Processes constitute a formalism that combines
nondeterministic and probabilistic choices.  They are a dominant
model in corporate finance, supply chain optimization, and system
verification and optimization.  While there are many slightly
different variants of this formalism (e.g., action-labeled MDPs
\cite{bel_1957_mdp,FilarVrieze97}, probabilistic automata
\cite{Segala:95:NJC,sv_2004_automata}), we work with the
state-labeled MDPs from~\cite{ba_1995_probabilistic}.

The set of all discrete probability distributions on a set $S$ is
denoted by $\Distr(S)$. The Dirac distribution on an element $s \in S$
is written as $1_s$. We also fix a set $\MP$ of propositions.

\begin{dfn}\label{dfn:nps}
  A \emph{Markov Decision Process} ($\MDP$) is a four-tuple
  $\Pi=(S,s_0,\tau,\allowbreak L)$ where: $S$ is the finite state
  space of the system, $s_0 \in S$ is the initial state,
  $L \colon S \to \wp(\MP)$ is a
  labeling function that associates to each state $s \in S$ a subset
  of propositions, and $\tau \colon S \to \wp(\Distr(S))$ is a function that
  associates to each $s \in S$ a non-empty and finite subset of
  of successor distributions.
\end{dfn}

\rev{In case $|\tau(s)|=1$ for all states $s$ we say that $\Pi$ is a \emph{Markov Chain}.}

\begin{wrapfigure}{r}{4.5cm}
  \vspace{-0.35cm}
    \begin{center}
      \includegraphics[width=4.5cm]{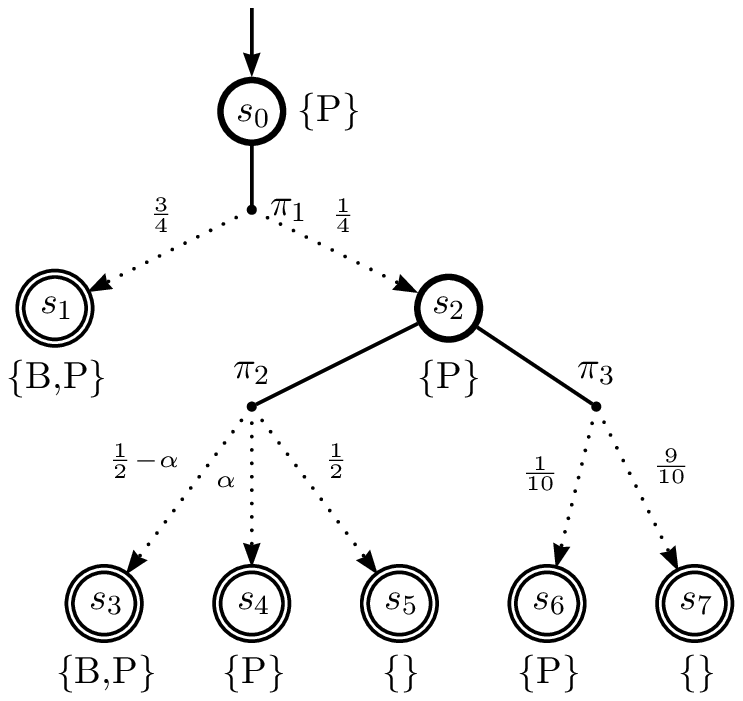}
    \end{center}
  \vspace{-0.4cm}
    \caption{MDP} \label{fig:notMaxNod}
\end{wrapfigure}
\noindent We define the \emph{successor} relation
$\varrho\subseteq S\times S$ by $\varrho \eqdef
\{(s,t)\mid\exists\, \pi\in\tau(s) \qdot \pi(t)>0\}$ and for each
state $s\in S$ we define the sets
$\Paths{s}  \eqdef \{s_0s_1s_2\ldots\in S^\omega \mid
            s_0=s\land\forall n\in \mathbb{N} \qdot \varrho(s_n,s_{n+1})\}$,
and $\FPaths{s} \eqdef \{s_0s_1\ldots s_n\in S^\star \mid
                 s_0=s\land\forall\, 0\leq i < n \qdot \varrho(s_n,s_{n+1})\}$
  of paths and finite paths respectively beginning at $s$.
  Sometimes we will use $\Paths{\Pi}$ to denote $\Paths{s_0}$,
  i.e. the set of paths of $\Pi$. For $\omega \in\Paths{s}$, we
  write the $n$-th state of $\omega$ as $\omega_n$. In addition,
  we write $\sigma_1\sqsubseteq\sigma_2$ if $\sigma_2$ is an extension of
  $\sigma_1$, i.e. $\sigma_2=\sigma_1 \sigma'$ for some $\sigma'$.
  We define the basic cylinder of a finite path $\sigma$ as the
  set of (infinite) paths that extend it, i.e $\Cyl{\sigma} \eqdef
  \{\omega \in \Paths{s}\mid \sigma\sqsubseteq\omega\}$. For a set
  of paths $R$ we write $\Cyl{R}$ for its set of cylinders,
  i.e. $\Cyl{R}\eqdef \bigcup_{\sigma\in R}\Cyl{\sigma}$. As usual,
  we let $\Borel_s\subseteq \wp(\Paths{s})$ be the Borel $\sigma$-algebra on the basic
  cylinders.

\medskip \exampleheader\label{exa:MDP} Figure~\ref{fig:notMaxNod}
shows a $\MDP$.  States with double lines represent absorbing
states (i.e., states $s$ with $\tau(s)=\{1_s\}$) and $\alpha$ is
any constant in the interval $[0,1]$. This $\MDP$ features a
single nondeterministic decision, to be made in state $s_2$.
\medskip

\noindent
Schedulers (also called strategies, adversaries, or policies) resolve the
nondeterministic choices in a
MDP~\cite{pz_1993_verification,var_1985_probabilistic,ba_1995_probabilistic}.

\begin{dfn}\label{dfn:scheduler}
  Let $\Pi=(S,s_0,\tau,L)$ be a $\MDP$ and $s\in
  S$. An \emph{$s$-scheduler}\index{scheduler} $\eta$ for $\Pi$ is a function from $\FPaths{s}$
  to $\Distr(\wp(\Distr(S)))$ such that for all $\sigma\in \FPaths{s}$
  we have $\eta(\sigma)\in \Distr(\tau(\last{\sigma}))$. We denote
  the set of all $s$-schedulers on $\Pi$ by $\Sch_s(\Pi)$. When
  $s = s_0$ we  omit it.
\end{dfn}


\noindent Note that our schedulers are randomized, i.e., in a
finite path $\sigma$ a scheduler chooses an element of
$\tau(\last{\sigma})$ probabilistically.  Under a scheduler
$\eta$, the probability that the next state reached after the path
$\sigma$ is $t$, equals $\sum_{\pi\in
  \tau(\last{\sigma}}\eta(\sigma)(\pi) \cdot \pi(t)$. In this way, a
scheduler induces a probability measure on $\Borel_s$ defined as
follows:

\begin{dfn}
  Let $\Pi$ be a $\MDP$, $s\in S$, and $\eta$ an $s$-scheduler on
  $\Pi$. The probability measure $\PP_{s,\eta}$ is the
  unique measure on $\Borel_s$ such that for all $s_0 s_1\ldots s_n\in
  \FPaths{s}$
  \[
    \PP_{s,\eta}(\Cyl{s_0s_1\ldots s_n})
    \eqdef \prod_{i=0}^{n-1} \ \sum_{\pi\in \tau(s_i)}
    \eta(s_0 s_1\ldots s_i)(\pi) \cdot \pi(s_{i+1}).
  \]
\end{dfn}
Often we will write $\PP_{\eta}(\Delta)$ instead of
$\PP_{s,\eta}(\Delta)$ when $s$ is the initial state and
$\Delta\in\Borel_s$. We now recall the notions of deterministic
and history independent schedulers.

\begin{dfn}
  Let $\Pi$ be a $\MDP$, $s\in S$, and $\eta$ an $s$-scheduler for $\Pi$. We
  say that $\eta$ is \emph{deterministic} if $\eta(\sigma)(\pi)$ is
  either $0$ or $1$ for all $\pi\in\tau(\last{\sigma})$ and all $\sigma
  \in\FPaths{s}$.
  We say that a scheduler is \emph{history independent}
  (HI) if for all finite paths $\sigma_1,\sigma_2$ of $\Pi$ with
  $\last{\sigma_1}=\last{\sigma_2}$ we have
  $\eta(\sigma_1)=\eta(\sigma_2)$.
  %
\end{dfn}

\begin{dfn}
  Let $\Pi$ be a $\MDP$, $s\in S$, and $\Delta \in \Borel_s$. Then the
  \emph{maximal and minimal probabilities of $\Delta$},
  $\PP^+_s(\Delta), \PP^-_s(\Delta)$, are defined as
  $$
    \PP^+_s(\Delta)\eqdef \sup_{\eta\in \Sch_s(\Pi)} \PP_{s,\eta}(\Delta)
    \hspace{0.5cm} \text{and} \hspace{0.5cm}
    \PP^-_s(\Delta)\eqdef \inf_{\eta\in \Sch_s(\Pi)} \PP_{s,\eta}(\Delta).
  $$
  A scheduler that attains $\PP^+_s(\Delta)$ or $\PP^-_s(\Delta)$
  is called an \emph{optimizing} scheduler.
\end{dfn}

\noindent We define the notion of (finite) convex combination of
schedulers.

\begin{dfn}
  Let $\Pi$ be a $\MDP$ and let $s\in S$. An $s$-scheduler $\eta$ is a
  \emph{convex combination} of the $s$-schedulers $\eta_1, \dots,
  \eta_n$ if there are $\alpha_1, \dots, \alpha_n \in [0,1]$ with
  $\alpha_1 + \dots + \alpha_n = 1$ such that for all $\Delta \in
  \Borel_s$, $\PP_{s,\eta}(\Delta) = \alpha_1 \PP_{s,\eta_1}(\Delta) + \dots +
  \alpha_n \PP_{s,\eta_n}(\Delta)$.
\end{dfn}

\noindent
Note that taking the convex combination $\eta$ of $\eta_1$ and
$\eta_2$ as functions, i.e., $\eta(\sigma)(\pi) =
\alpha\eta_1(\sigma)(\pi) + (1-\alpha)\eta_2(\sigma)(\pi)$, does not
imply that $\eta$ is a convex combination of $\eta_1$ and $\eta_2$ in
the sense above.

\section{Conditional Probabilities over MDPs\index{conditional probability!over MDPs}}
\label{cp}

The conditional probability $P(A \mid B\,)$ is the probability of
an event \textit{A}, given the occurrence of another event
\textit{B}. Recall that given a probability space $(\Omega, F, P)$
and two events $A, B\in F$ with $P(B) > 0$, $P(A \mid B)$ is
defined as ${P(A \cap B)}/{P(B)}.$ If $P(B) = 0$, then $P(A \mid
B)$ is undefined. \comment{Make a remark about conditional
probabilities on MCs.}
In particular, given a MDP $\Pi$, a scheduler $\eta$, and a state
$s$, consider the probabilistic space
$(\Paths{s},\Borel_s,\PP_{s,\eta})$. For two sets of paths
$\Delta_1, \Delta_2 \in \Borel_s$ with $\PP_{s,\eta}(\Delta_2) >
0$, the conditional probability of $\Delta_1$ given $\Delta_2$ is
$\PP_{s,\eta}(\Delta_1 \mid \Delta_2) = {\PP_{s,\eta}(\Delta_1
\cap \Delta_2)}/{\PP_{s,\eta}(\Delta_2)}.$ If
$\PP_{s,\eta}(\Delta_2) = 0$, then $\PP_{\eta,s}(\Delta_1 \mid
\Delta_2)$ is undefined. We define the maximum and minimum
conditional probabilities for all $\Delta_2 \in \Borel_s$ as
follows:

\begin{dfn}\label{dfn:min max cond prob} Let $\Pi$ be a $\MDP$. The
  \emph{maximal and minimal conditional probabilities}
  $\PP^+_s(\Delta_1|\Delta_2)$, $\PP^-_s(\Delta_1|\Delta_2)$ of
  sets of paths $\Delta_1,\Delta_2\in\Borel_s$ are defined by
  \begin{align*}
     \PP^+_s(\Delta_1|\Delta_2) & \eqdef \left\lbrace
                                    \begin{array}{ll}
                                      {\displaystyle\sup_{\eta\in \Sch_{\Delta_2}^{>0}}} \PP_{s,\eta}(\Delta_1|\Delta_2) & \mbox{ if } \Sch_{\Delta_2}^{>0}\not= \emptyset,\\
                                      0 & \mbox{ otherwise,}\\
                                    \end{array}
                                \right. \\
    \PP^-_s(\Delta_1|\Delta_2) & \eqdef \left\lbrace
                                    \begin{array}{ll}
                                      {\displaystyle \inf_{\eta\in \Sch_{\Delta_2}^{>0}}} \PP_{s,\eta}(\Delta_1|\Delta_2) & \mbox{ if } \Sch_{\Delta_2}^{>0}\not= \emptyset,\\
                                      1 & \mbox{ otherwise,}\\
                                    \end{array}
                                \right.
  \end{align*}
  where $\Sch_{\Delta_2}^{>0} = \{\eta\in \Sch_s(\Pi) \mid
  \PP_{s,\eta}(\Delta_2)>0\}$.
\end{dfn}


\noindent The following lemma generalizes Lemma~6 of
\cite{ba_1995_probabilistic} to conditional
probabilities\index{conditional probability}.

\begin{lem}\label{lem:relMaxMinProb}
  Given $\Delta_1,\Delta_2\in\Borel_s$, its maximal and minimal
  conditional probabilities are related by:
  $\PP^+_s(\Delta_1|\Delta_2) = 1-
  \PP^-_s(\Paths{s}-\Delta_1|\Delta_2)$.
\end{lem}

\section{Conditional Probabilistic Temporal Logic\index{temporal logic!pCTL}\index{temporal logic!cpCTL}}
\label{sec:cpCTL}

The logic $\cpCTL$ extends $\pCTL$ with formulas of the form
$\CP{}{\bowtie a}{\phi}{\psi}$ where
$\bowtie\in\{<,\leq,>,\geq\}$.  Intuitively, $\CP{}{\leq
  a}{\phi}{\psi}$ holds if the probability of $\phi$ given $\psi$ is at most
  $a$. Similarly for the other comparison operators.

\paragraph{Syntax:}
  The $\cpCTL$ logic is defined as a set of state and path formulas,
  i.e., $\cpCTL\eqdef \Stat\cup \Path$, where $\Stat$ and $\Path$ are
  defined inductively:
\[
   \begin{array}{rcl}
     {\cal P}&\subseteq&\Stat,\\
     \phi,\psi\in \Stat&\Rightarrow&\phi\land\psi,\lnot \phi\in \Stat,\\
     \phi, \psi\in \Path&\Rightarrow&\Prob{}{\bowtie a}{\phi},\CP{}{\bowtie a} {\phi}{\psi} \in \Stat,\\
     \phi,\psi \in \Stat&\Rightarrow& \phi\;\U\psi,\Fin\phi,\G\;\phi\in \Path. \\
   \end{array}
\]
 \noindent Here $\bowtie\in\{<,\leq,>,\geq\}$ and $a\in[0,1]$.

\paragraph{Semantics:}

The satisfiability of state-formulas ($s \models \phi$ for a state
$s$) and path-formulas ($\omega \models \psi$ for a path $\omega$)
is defined as an extension of the satisfiability for \pCTL. Hence,
the satisfiability of the logical, temporal, and \pCTL operators
is defined in the usual way. For the conditional probabilistic
operators we define
 $$
\begin{array}{rcl}
%
%
%
%
s\models \CP{}{\leq a} {\phi}{\psi}&\Leftrightarrow&\PP_{s}^+(\{\omega\in\Paths{s}\mid \omega\models \phi\}|\{\omega\in\Paths{s}\mid \omega\models \psi\})\leq a,\\
s\models \CP{}{\geq a} {\phi}{\psi}&\Leftrightarrow&\PP_{s}^-(\{\omega\in\Paths{s}\mid \omega\models \phi\}|\{\omega\in\Paths{s}\mid \omega\models \psi\})\geq a,\\
\end{array}
 $$
 and similarly for $s\models \CP{}{< a} {\phi}{\psi}$ and $s\models \CP{}{> a} {\phi}{\psi}$. We say that a model $\mdp$ satisfy $\phi$, denoted by $\mdp\models\phi$ if $s_0\models \phi$.

In the following we fix some notation that we will use in the rest of the chapter,
\begin{align*}
  \Prob{+}{s}{\phi} &
  \triangleq\PP^+_{s}(\{\omega\in\Paths{s}\mid \omega\models \phi\}), \\
  \CP{+}{s}{\phi}{\psi} &
  \triangleq \PP^+_s(\{\omega\in\Paths{s}\mid \omega\models \phi\}|
                     \{\omega\in\Paths{s}\mid \omega\models \psi\}), \\
  \CP{}{s,\eta}{\phi}{\psi} &
  \triangleq\PP_{s,\eta}(\{\omega\in\Paths{s}\mid
  \omega\models \phi\}|\{\omega\in\Paths{s}\mid \omega\models
  \psi\}),
\end{align*}
$\CP{-}{s}{\phi}{\psi}$ and $\Prob{-}{s}{\phi}$ are defined
analogously.

\begin{observation}\label{lem:modelchecker}
  As usual, for checking if $s\models\CP{}{\bowtie
    a}{\phi}{\psi}$, we only need to consider the cases
  where $\phi = \Until{\phi_1}{\phi_2}$ and where
  $\psi$ is either $\Until{\psi_1}{\psi_2}$ or $\Glo \psi_1$.  This follows
  from $\Fin \phi\leftrightarrow \Until{\true\,}{\phi}$, $\Glo \phi\leftrightarrow \lnot\Fin \lnot\phi$ and the relations
  $$
    \CP{+}{s}{\lnot\phi}{\psi} =
    1-\CP{-}{s}{\phi}{\psi} \qquad \text{ and } \qquad
    \CP{-}{s}{\lnot\phi}{\psi} = 1-\CP{+}{s}{\phi}{\psi}
  $$
  derived from Lemma~\ref{lem:relMaxMinProb}.
  Since there is no way to relate $\CP{+}{}{\phi}{\psi}$ and
  $\CP{+}{}{\phi}{\lnot\psi}$, we have to provide algorithms
  to compute both $\CP{+}{}{\phi}{\Until{\psi_1}{\psi_2}}$ and $\CP{+}{}{\phi}{\Glo   \psi_1}$. The same remark holds for the minimal conditional
  probabilities $\CP{-}{}{\phi}{\Until{\psi_1}{\psi_2}}$ and $\CP{-}{}{\phi}{\Glo   \psi_1}$. In this chapter we will only focus on the former
  problem, i.e., computing maximum conditional probabilities, the
  minimal case follows using similar techniques.
\end{observation}

\subsection{Expressiveness}

We now show that \cpCTL is strictly more expressive than pCTL. The
notion of \emph{expressiveness} of a temporal logic is based on
the notion of \emph{formula equivalence}. Two temporal logic
formulas $\phi$ and $\psi$ are \emph{equivalent} with respect to a
set $\mc$ of models (denoted by $\phi\equiv_{\mc}\psi$) if for any
model $m\in\mc$ we have $m\models\phi$ if and only if
$m\models\psi$. A temporal logic $\cal L$ is said to be \emph{at
least as expressive as} a temporal logic ${\cal L}^\prime$, over a
set of models $\mc$, if for any formula $\phi\in{\cal L}^\prime$
there is a formula $\psi\in{\cal L}$ that is equivalent to $\phi$
over $\mc$. Two temporal logics are \emph{equally expressive} when
each of them is at least as expressive as the other. Formally:


\begin{dfn} Two temporal logics\index{temporal logic!expressiveness} ${\cal L}$ and ${\cal L}^\prime$ are equally expressive with respect to $\mc$ if
 \begin{align*}
 \forall \phi\in {\cal L}. \left(\exists \psi\in{\cal L}^\prime. \phi\equiv_{\mc}\psi\right) \land \forall \psi\in {\cal L}^\prime. \left(\exists \phi\in{\cal L}. \phi\equiv_{\mc}\psi\right).
 \end{align*}
\end{dfn}

\begin{thm} $\cpCTL$ is more expressive than $\pCTL$ with respect to \MCs and \MDPs.
\end{thm}
\begin{proof} Obviously \cpCTL is at least as expressive as \pCTL,
hence we only need to show that the reverse does not hold. The
result is rather intuitive since the semantics of the conditional
operator for \cpCTL logic is provided by a non-linear equation
whereas there is no \pCTL formula with non-linear semantics.

The following is a formal proof. We plan to show that there is no
\pCTL formula $\psi$ equivalent to $\phi=\CP{}{\leq 0.5}{\Fin
A}{\Fin B}$\rev{, with $A$ and $B$ atomic propositions}. The proof is by cases on the structure of the \pCTL
formula $\psi$. The most interesting case is when $\psi$ is of the
form $\Prob{}{\leq b}{\psi}$, so we will only prove this case. In
addition we restrict our attention to $b$'s such that $0<b<1$ (the
cases $b=0$ and $b=1$ are easy). In Figure \ref{fig:expressCount}
we depict the Markov Chains involved in the proof. We use
$\lnot\psi_1$ to mark the states with an assignment of truth
values (to propositional variables) falsifying $\psi_1$.
\begin{itemize}
\item[Case $\psi=\Prob{}{\leq b}{\Fin\psi_1}$:] $\text{}$\\ If $\psi_1$ is
$\true$ or $\false$ the proof is obvious, so we assume otherwise.
We first note that we either have $\lnot \rev{\psi_1} \Rightarrow
\lnot(B\land\lnot A)$ or $\lnot \rev{\psi_1} \Rightarrow (B\land\lnot A)$.
In the former case, it is easy to see (using $\lnot B\Rightarrow
\psi_1 $) that we have $m_2\models\phi$ and $m_2\not\models\psi$.
In the second case we have $m_1\not\models \phi$ and $m_1\models
\psi$.
\item[Case $\psi=\Prob{}{\leq b}{\psi_1\U\psi_2}$:]$\text{}$\\ We
assume $\psi_1\not=\true$, otherwise we fall into the previous
case. We can easily see that we have $m_3\models\psi$ but
$m_3\not\models\phi$.
\item[Case $\psi=\Prob{}{\leq b}{\G\psi_1}$:] $\text{}$\\ The case
when $\psi_1=\true$ is easy, so we assume $\psi_1\not=\true$. We
can easily see that we have $m_3\models\psi$ but
$m_3\not\models\phi$.
\end{itemize}
\vspace{-0.8cm}
\end{proof}
Note that, since \MCs are a special case of \MDPs, the proof also
holds for the latter class.
{\def\thesubfigure{\thefigure(\alph{subfigure})}
\begin{figure}[h]
 \centering
  \subfigure{\includegraphics[width=4.25cm]{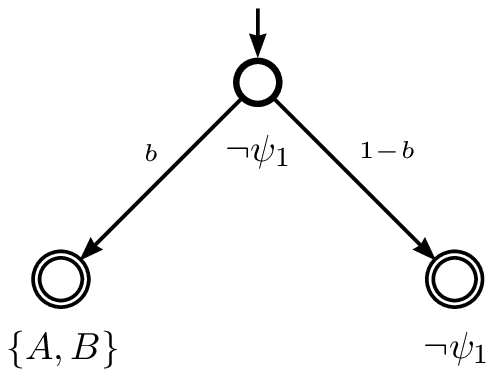}}\quad
  \subfigure{\includegraphics[width=4cm]{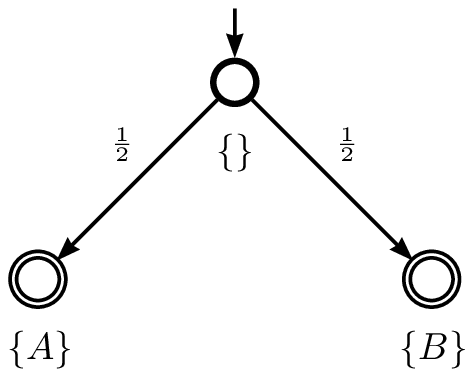}}\qquad
  \subfigure{\includegraphics[width=1.34cm]{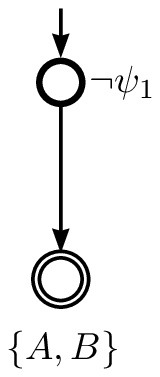}}
\caption{Markov Chains $m_1$, $m_2$, and $m_3$
respectively.}\label{fig:expressCount}
  \vspace{-0.1cm}
\end{figure}}

We note that, in spite of the fact that a \cpCTL formula of the
form $\CP{}{\leq a}{\phi}{\psi}$ cannot be expressed as a \pCTL
formula, if dealing with fully probabilistic systems (i.e. systems
without nondeterministic choices) it is still possible to verify
such conditional probabilities formulas as the quotient of two
\pCTL${\!\!}^\star$ formulas:
$\CP{}{}{\phi}{\psi}=\frac{\Prob{}{}{\phi\land\psi}}{\Prob{}{}{\psi}}$.
However, this observation does not carry over to systems where
probabilistic choices are combined with nondeterministic ones (as
it is the case of Markov Decision Processes). This is due to the
fact that, in general, \rev{it is not the case that}
$\CP{+}{}{\phi}{\psi}\rev{=}\frac{\Prob{+}{}{\phi\land\psi}}{\Prob{+}{}{\psi}}$.

\section{Semi History-Independent and Deterministic Schedulers}
\label{edHIS}

Recall that there exist optimizing (i.e. maximizing and
minimizing) schedulers on $\pCTL$ that are $\HI$ and deterministic
\cite{ba_1995_probabilistic}. We show that, for \cpCTL,
deterministic schedulers still suffice to reach the optimal
conditional probabilities. Because we now have to solve a
non-linear optimization problem, the proof differs from the \pCTL
case in an essential way.
We also show that $\HI$ schedulers do not suffice to attain
optimal conditional probability and introduce the family of semi
history-independent schedulers that do attain it.

\subsection{Semi History-Independent Schedulers\index{scheduler!semi history independent}}

The following example shows that maximizing schedulers are not
necessarily $\HI$.

\begin{wrapfigure}{r}{2.8cm}
\vspace{-0.1cm} \centering
\includegraphics[width=3cm]{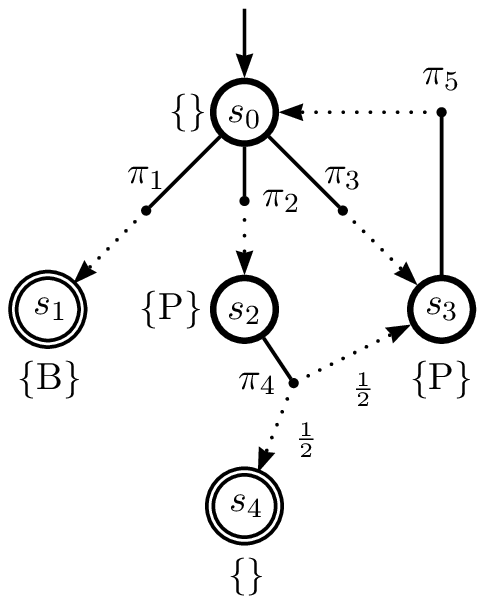}
\caption{\!\!MDP} \label{fig:HI} \vspace{-0.4cm}
\end{wrapfigure}
\medskip\exampleheader
\label{ex:notHI} Let $\Pi$ be the $\MDP$ of Figure~\ref{fig:HI}
and the conditional probability $\CP{}{s_0,\eta}{\Fin B}{\Fin P}$.
There are only three deterministic history independent schedulers,
choosing $\pi_1$, $\pi_2$, or $\pi_3$ in $s_0$. For the first one,
the conditional probability is undefined and for the second and
third it is 0. The scheduler $\eta$ that maximizes
$\CP{}{s_0,\eta}{\Fin B}{\Fin P}$ satisfies $\eta(s_0) = \pi_3$,
$\eta(s_0s_3) = \pi_5$, and $\eta(s_0s_3s_0) = \pi_1$.
Since $\eta$ chooses on $s_0$ first $\pi_2$ and later $\pi_1$, $\eta$
is not history independent.

\medskip\noindent
Fortunately, as we show in Theorem \ref{thm:shiSch}, there exists
a \emph{nearly $\HI$} scheduler that attain optimal conditional
probability. We say that such schedulers are nearly HI because
they always take the same decision \emph{before} the system
reaches a certain condition $\varphi$ and also always take the
same decision \emph{after} $\varphi$. This family of schedulers is
called $\varphi$-semi history independent ($\varphi$-$\nHI$ for
short) and the condition $\varphi$ is called \emph{stopping
condition}. For a pCTL path formula $\phi$ the stopping condition
is a boolean proposition either validating or contradicting
$\phi$. So, the (validating) stopping condition of $\Fin \phi$ is
$\phi$ whereas the (contradicting) stopping condition of $\Glo
\phi$ is $\lnot \phi$. Formally:

\[    \StopCond{\phi} \eqdef
        \begin{cases}
            \lnot\psi_1\lor\psi_2 & \text{if $\phi=\psi_1\U\psi_2$}, \\
            \lnot\psi  & \text{if $\phi=\G\psi$}.
        \end{cases}
\]

Similarly, for a \cpCTL formula $\CP{}{\bowtie a}{\phi}{\psi}$,
the stopping condition is a condition either validating or
contradicting any of its pCTL formulas ($\phi$, $\psi$), i.e.,
$\StopCond{\CP{}{\bowtie
a}{\phi}{\psi}}=\StopCond{\phi}\lor\StopCond{\psi}$.

We now proceed with the formalization of semi history independent
schedulers.
\begin{dfn}[Semi History-Independent Schedulers]
  Let $\Pi$ be a $\MDP$, $\eta$ a scheduler\index{scheduler!semi history independent} for $\Pi$, and
  $\phi\lor\psi \in\Stat$. We say that $\eta$ is a
  \emph{$(\phi\lor\psi)$ semi history-independent scheduler}
  ($(\phi\lor\psi)$-$\nHI$ scheduler for short) if for all
  $\sigma_1,\sigma_2\in\FPaths{s}$ such that
  $\last{\sigma_1}=\last{\sigma_2}$ we have
\begin{align*}
  \sigma_1,\sigma_2\not\models \Fin (\phi\lor\psi)\Rightarrow\eta(\sigma_1)=\eta(\sigma_2), and & \qquad \text{\{HI before stopping
  condition\}}\\
  \sigma_1,\sigma_2\models \Fin\phi\Rightarrow\eta(\sigma_1)=\eta(\sigma_2), and & \qquad \text{\{HI after stopping
  condition\}}\\
  \sigma_1,\sigma_2\models \Fin  \psi\Rightarrow\eta(\sigma_1)=\eta(\sigma_2). & \qquad \text{\{HI after stopping
  condition\}}
\end{align*}
We denote the set of all $\varphi$-$\nHI$ schedulers of
$\Pi$ by $\Sch^\varphi(\Pi)$.
\end{dfn}

We now prove that semi history-independent schedulers suffice to
attain the optimal conditional probabilities for cpCTL formula.
%
%

\begin{thm}\label{thm:shiSch}
  Let $\Pi$ be a $\MDP$, $\phi, \psi \in\Path$, and $\varphi=\StopCond{\phi}\lor
\StopCond{\psi}$. Assume that there exists a scheduler $\eta$ such
that $\Prob{}{\eta}{\psi}>0$. Then:


\[\CP{+}{}{\phi}{\psi}=\sup_{\eta\in
\Sch^\varphi(\Pi)}\CP{}{\eta}{\phi}{\psi}.\]

\end{thm}
\noindent (If there exists no scheduler $\!\eta$ such that
$\Prob{}{\eta}{\psi}\!>\!0$, then the supremum is $0$.)

The proof of this theorem is rather complex. The first step is to
prove that there exists a scheduler $\eta$ HI before the stopping
condition and such that $\CP{}{\eta}{\phi}{\psi}$ is `close' (i.e.
not further than a small value $\epsilon$) to the optimal
conditional probability $\CP{+}{}{\phi}{\psi}$. For this purpose
we introduce some definitions and prove this property first for
long paths (Lemma \ref{lem:initialCut}) and then, step-by-step, in
general (Lemma \ref{lem:extendingCut} and Corollary
\ref{cor:HIbefore}). After that, we create a scheduler that is
also HI after the stopping condition and whose conditional
probability is still close to the optimal one
(Lemma~\ref{lem:HIafter}). From the above results, the theorem
readily follows.

We now introduce some definitions and notation that we will need for the proof.

\begin{dfn}[Cuts] Given a $\MDP$ $\Pi$ we say that a set $K\subseteq
\FPaths{\Pi}$ is a \emph{cut} of $\Pi$ if $K$ is a downward-closed
set of finite paths such that every infinite path passes through it, i.e. \comment{complete path or infinite path???}
\begin{itemize}
    \item $\forall \,\sigma_1 \in K \,.\, \forall\, \sigma_2 \in \FPaths{\Pi}\,.\, \sigma_1 \sqsubseteq \sigma_2 \implies \sigma_2 \in K$, and
    \item $\forall \,\omega \in \Paths{\Pi} \,.\,\exists\, \sigma \in K\,.\, \sigma \sqsubset \omega.$
\end{itemize}

\noindent where $\sigma_1\sqsubseteq\sigma_2$ means that
$\sigma_2$ is an ``extension'' of $\sigma_1$, i.e.
$\sigma_2=\sigma_1 \sigma'$ for some path $\sigma'$. We denote the
set of all cuts of $\Pi$ by $K(\Pi)$.
\end{dfn}

\noindent For $R \subseteq \FPaths{s}$, we say that $\eta$ is
history independent in $R$ if for all $\sigma_1, \sigma_2 \in R$
such that $\last{\sigma_1} = \last{\sigma_2}$ we have that
$\eta(\sigma_1) = \eta(\sigma_2)$. We also define the sets $\Phi$
and $\Psi$ as the set of finite paths validating $\phi$ and $\psi$
respectively, i.e. $\Phi\eqdef\{\sigma\in\FPaths{\Pi}\mid
\sigma\models\phi\}$ and $\Psi\eqdef\{\sigma\in\FPaths{\Pi}\mid
\sigma\models\psi\}$. Finally, given a $\MDP$ $\Pi$, two path
formulas $\phi$, $\psi$, and $\hat{\epsilon}>0$ we define the set
\begin{align*}
{\cal K}\eqdef\lbrace(K,\eta)\in K(\Pi)\times\Sch(\Pi)\mid
\Phi\cup\Psi\subseteq K&\text{ and }\, \eta\text { is HI in
}K\setminus(\Phi\cup\Psi)\\
&\text{ and } \CP{+}{}{\phi}{\psi}-
\CP{}{\eta}{\phi}{\psi}<\hat{\epsilon}\rbrace
\end{align*}
%
\noindent If a scheduler $\eta$ is HI in
$K\setminus(\Phi\cup\Psi)$ then we say that $\eta$ is HI before the stopping
condition.


\begin{lem}[non emptiness of ${\cal K}$]\label{lem:initialCut} There exists $(K,\eta)$ such that $(K,\eta)\in{\cal K}$
and that its complement $K^c\eqdef \FPaths{\Pi}\setminus K$ is finite.
\end{lem}
%
%
%

\begin{proof}

We show that, given formulas $\phi$, $\psi$ and
$\hat{\epsilon}>0$, there exists a cut $K$ and a scheduler $\eta^\star$
such that $K^c$ is finite, $\Phi\cup\Psi\subseteq K$, $\eta^\star$ is HI
in $K\setminus(\Phi\cup\Psi)$, and
$\CP{+}{}{\phi}{\psi}-\CP{}{\eta^\star}{ \phi}{\psi}<\hat{\epsilon}$ .

The proof is by case analysis on the structure of $\phi$ and
$\psi$. We will consider the cases where $\phi$ and $\psi$ are
either ``eventually operators'' ($\diamondsuit$) or ``globally
operators'' ($\Box$), the proof for the until case follows along
the same lines.

\noindent $\bullet$ \textbf{Case $\phi$ is of the form $\F\phi$ and $\psi$ is of the form $\Fin\psi$:\
} \text{}\\ Let us start by defining the the probability of
reaching $\phi$ in at most $N$ steps, as $\Prob{}{\eta}{\leq\!
N,\F\phi}\eqdef
\Prob{}{\eta}{\Cyl{\{\sigma\in\FPaths{\Pi}\mid\allowbreak\sigma\models
\F\phi\land|\sigma|\leq N \}}}$. Note that for all \pCTL
reachability properties $\F\phi$ and schedulers $\eta$ we have
\[\lim_{N\rightarrow \infty}\Prob{}{\eta}{\leq\! N, \F\phi}=\Prob{}{\eta}{\F\phi}.
\]

\noindent We also note that this result also holds for
$\pCTL\!^\star$ formulas of the form $\Fin\phi\land \Fin\psi$.

 Let us now take a scheduler $\eta$ and a number $N$ such
that

\begin{align}
\label{eq:cut1}\CP{+}{}{\Fin\phi}{\Fin\psi}-\CP{}{\eta}{\Fin\phi}{\Fin\psi}&<\epsilon\eqdef \hat{\epsilon}/3, \text{ and}\\
\label{eq:cut2}\Prob{}{\eta}{\Fin\phi\land\Fin\psi}-\Prob{}{\eta}{\leq\!N,\F\phi\land\F\psi}&<\epsilon',\text{ and}\\
\label{eq:cut3}
\Prob{}{\eta}{\Fin\psi}-\Prob{}{\eta}{\leq\!N,\F\psi}&<\epsilon'.
\end{align}

\noindent where $\epsilon'$ is such that
$\epsilon'<\min\left(2\cdot\epsilon\cdot\Prob{}{\eta}{\Fin\psi},\frac{\epsilon
\cdot \Prob{}{\eta}{\Fin\psi}^2}{\Prob{}{\eta}{\Fin\phi \land
\Fin\psi}+2\cdot\epsilon\cdot \Prob{}{\eta}{\Fin\psi}}\right)$.
The reasons for this particular choice for the bound of
$\epsilon'$ will become clear later on in the proof.

\rev{We define $K$ as $\Phi\cup\Psi \cup \FPaths{\leq\!N,\Pi}$, where the latter set is defined as the set of paths with length larger than
$N$, i.e. $\FPaths{\leq\!N,\Pi}\eqdef \{\sigma\in\FPaths{\Pi}\mid N\leq |\sigma|\}$. In addition, we define $\eta^\star$ as a scheduler HI in $\FPaths{\leq\!N,\Pi}$ behaving like
$\eta$ for paths of length less than or equal to $N$ which
additionally minimizes $\Prob{}{}{\Fin\psi}$ after level $N$.} In
order to formally define such a scheduler we let $S_N$ to be the
set of states that can be reached in exactly $N$ steps, i.e.,
$S_N\eqdef\{s\in S\mid \exists\,\sigma\in\FPaths{\Pi}\,:\,
|\sigma|=N\land \last{\sigma}=s\}$. Now for each $s\in S$ we let
$\xi_s$ to be a HI s-scheduler such that
$\Prob{}{s,\xi_s}{\Fin\psi}=\Prob{-}{s}{\Fin\psi}$. Note that such
a scheduler exists, i.e., it is always possible to find a HI
scheduler minimizing a reachability pCTL formula
\cite{ba_1995_probabilistic}.

We now define $\eta^\star$ as
\[
    \eta^\star(\sigma) \eqdef
        \begin{cases}
            \xi_s(\sigma_{|\alpha|} \sigma_{|\alpha|\!+\!1}\cdots\sigma_{|\sigma|}) & \text{if $\alpha\sqsubseteq\sigma$ for some $\alpha\in \FPaths{=\!N,\Pi}$}\\
                    &\qquad\qquad \text{such that $\last{\alpha}=s$}, \\
            \eta(\sigma)  & \text{otherwise}.
        \end{cases}
\]

\noindent where $\FPaths{=\!N,\Pi}$ denotes the set of paths of
$\Pi$ of length $N$. It is easy to see that $\eta^\star$
minimizes $\Prob{}{}{\Fin\psi}$ after level $N$. As for the
history independency of $\eta^\star$ in $K$ there is still one
more technical detail to consider: note there may still be paths
$\alpha_1 s_1 \sigma_1 t$ and $\alpha_2 s_2 \sigma_2 t$ such that
$\alpha_1 s_1,\,\alpha_2 s_2 \in \FPaths{=\!N,\Pi}$ and
$\xi_{s_1}(s_1 \sigma_1 t)\not=\xi_{s_2}(s_2 \sigma_2 t)$. This is
the case when there is more than one distribution in $\tau(t)$
minimizing $\Prob{}{t}{\Fin\psi}$, and $\xi_{s_1}$ happens to
choose a different (minimizing) distribution than $\xi_{s_2}$ for
the state $t$. Thus, the selection of the family of schedulers
$\{\xi_s\}_{s\in {S_N}}$ must be made in such a way that: for all
\rev{$s_1, s_2\in S_N$} we have
$\Prob{}{s_1,\xi_{s_1}}{\Fin\psi}=\Prob{-}{s_1}{\Fin\psi}$,
$\Prob{}{s_2,\xi_{s_2}}{\Fin\psi}=\Prob{-}{s_2}{\Fin\psi}$, and
for all $\sigma_1 t\in\FPaths{s_1},\sigma_2
t\in\FPaths{s_2}\,:\,\xi_{s_1}(\sigma_1 t)=\xi_{s_2}(\sigma_2 t)$.
It is easy to check that such family exists. \rev{We conclude that
$\eta^\star$ is HI in $\FPaths{\leq\!N,\Pi}$ and thus HI in
$K\setminus(\Phi\cup\Psi)$.}

We note that $\Prob{}{\eta^\star}{\F\psi}>0$, this follows from
$0<\Prob{}{\eta}{\F\psi}$, (\ref{eq:cut1}), (\ref{eq:cut3}), and
the definition of $\eta^\star$.

Having defined $\eta^\star$ we proceed to prove that such
scheduler satisfies $\CP{+}{}{\phi}{\psi}-\CP{}{\eta}{
\phi}{\psi}<\hat{\epsilon}$. It is possible to show that:

\begin{align}
\label{eq:interval1}
\Prob{}{\eta}{\leq\! N,\Fin\psi}&\leq&\Prob{}{\eta^\star}{\Fin\psi}& \leq &\Prob{}{\eta}{\Fin\psi},&\\
\label{eq:interval2} \Prob{}{\eta}{\leq \! N, \F\phi\land
\F\psi}&\leq& \Prob{}{\eta^\star}{\Fin\phi\land \Fin\psi} &<&
\Prob{}{\eta}{\Fin\phi\land \Fin\psi} + \epsilon\cdot
\Prob{}{\eta}{\Fin\psi}.&
\end{align}
\comment{how do i modify the alignment within the environment ???}


\noindent (\ref{eq:interval1}) and the first inequality of
(\ref{eq:interval2}) follow straightforwardly from the definition
of $\eta^\star$. For the second inequality of (\ref{eq:interval2})
suppose by contradiction that $\Prob{}{\eta^\star}{\Fin\phi\land
\Fin\psi} \geq \Prob{}{\eta}{\Fin\phi\land \Fin\psi} +
\epsilon\cdot \Prob{}{\eta}{\Fin\psi}$. Then

\[\frac{\Prob{}{\eta^\star}{\Fin\phi \land \Fin\psi}}{\Prob{}{\eta^\star}{\Fin\psi}}\geq\frac{\Prob{}{\eta}{\Fin\phi
\land \Fin\psi}+ \epsilon \cdot
\Prob{}{\eta}{\Fin\psi}}{\Prob{}{\eta}{\Fin\psi}}=\CP{}{\eta}{\Fin\phi}{\Fin\psi}+\epsilon\]
\noindent contradicting (\ref{eq:cut1}).

Now we have all the necessary ingredients to show that
\begin{equation}\label{eq:etsstarisclose}
|\CP{}{\eta}{\Fin\phi}{\Fin\psi}-\CP{}{\eta^\star}{\Fin\phi}{\Fin\psi}|<
2\cdot\epsilon.
\end{equation}

Note that

\[\frac{\Prob{}{\eta}{\Fin\phi\land
\Fin\psi}\!-\!\epsilon'}{\Prob{}{\eta}{\Fin\psi}} \!<\!
\CP{}{\eta^\star}{\Fin\phi}{\Fin\psi} \text{  and  }
\CP{}{\eta^\star}{\Fin\phi}{\Fin\psi}\!<\!
\frac{\Prob{}{\eta}{\Fin\phi\land \Fin\psi}\!+\!\epsilon
\cdot\Prob{}{\eta}{\Fin\psi}}{\Prob{}{\eta}{\Fin\psi}-\epsilon'}.\]

\noindent The first inequality holds because
$\Prob{}{\eta^\star}{\Fin\psi}\leq\Prob{}{\eta}{\Fin\psi}$ and
(combining (\ref{eq:interval2}) and (\ref{eq:cut2}))
$\Prob{}{\eta^\star}{\Fin\phi\land \Fin\psi}>
\Prob{}{\eta}{\Fin\phi\land \Fin\psi}-\epsilon'$. The second
inequality holds because $\Prob{}{\eta^\star}{\Fin\phi\land
\Fin\psi} < \Prob{}{\eta}{\Fin\phi\land \Fin\psi} + \epsilon\cdot
\Prob{}{\eta}{\Fin\psi}$ and (combining (\ref{eq:interval1}) and
(\ref{eq:cut3}))
$\Prob{}{\eta^\star}{\Fin\psi}>\Prob{}{\eta}{\Fin\psi}-\epsilon'$.
It is easy to see that $\CP{}{\eta}{\Fin\phi}{\Fin\psi}$ falls in
the same interval, i.e., both $\CP{}{\eta}{\Fin\phi}{\Fin\psi}$
and $\CP{}{\eta^\star}{\Fin\phi}{\Fin\psi}$ are in the interval

\[\left(\ \frac{\Prob{}{\eta}{\Fin\phi\land \Fin\psi}-\epsilon'}{\Prob{}{\eta}{\Fin\psi}}\quad
,\quad \frac{\Prob{}{\eta}{\Fin\phi\land \Fin\psi}+\epsilon\cdot
\Prob{}{\eta}{\Fin\psi}}{\Prob{}{\eta}{\Fin\psi}-\epsilon'}\right).\]

Thus, we can prove (\ref{eq:etsstarisclose}) by proving

\begin{align*}
\frac{\Prob{}{\eta}{\Fin\phi\land
\Fin\psi}}{\Prob{}{\eta}{\Fin\psi}}-\frac{\Prob{}{\eta}{\Fin\phi\land
\Fin\psi}-\epsilon'}{\Prob{}{\eta}{\Fin\psi}}   &< 2\cdot
\epsilon\text{, and}\\
\frac{\Prob{}{\eta}{\Fin\phi\land
\Fin\psi}+\epsilon\cdot\Prob{}{\eta}{\Fin\psi}}{\Prob{}{\eta}{\Fin\psi}-\epsilon'}-\frac{\Prob{}{\eta}{\Fin\phi\land
\Fin\psi}}{\Prob{}{\eta}{\Fin\psi}} & < 2\cdot \epsilon.
\end{align*}

\noindent The first inequality holds if
and only if $\epsilon' <
2\,\cdot\,\epsilon\,\cdot\,\Prob{}{\eta}{\Fin\psi}$. As for the
second inequality, we have
\[\begin{array}{ll}
&\frac{\Prob{}{\eta}{\Fin\phi\land
\Fin\psi}+\epsilon\cdot\Prob{}{\eta}{\Fin\psi}}{\Prob{}{\eta}{\Fin\psi}-\epsilon'}-\frac{\Prob{}{\eta}{\Fin\phi\land
\Fin\psi}}{\Prob{}{\eta}{\Fin\psi}} < 2\cdot \epsilon \medskip \\
\iff& \Prob{}{\eta}{\Fin\psi}^2 \cdot
\epsilon+\Prob{}{\eta}{\Fin\phi\land \Fin\psi}\cdot \epsilon' <
2\cdot \epsilon \cdot
(\Prob{}{\eta}{\Fin\psi}-\epsilon')\cdot \Prob{}{\eta}{\Fin\psi}\medskip\\

\iff& \Prob{}{\eta}{\Fin\psi}^2 \cdot
\epsilon+\Prob{}{\eta}{\Fin\phi\land \Fin\psi}\cdot \epsilon' <
2\cdot \epsilon \cdot \Prob{}{\eta}{\Fin\psi}^2 -2\cdot \epsilon
\cdot\epsilon'\cdot
\Prob{}{\eta}{\Fin\psi}\medskip\\
\iff& \epsilon'<\frac{\epsilon \cdot
\Prob{}{\eta}{\Fin\psi}^2}{\Prob{}{\eta}{\Fin\phi\land\Fin\psi}+2\cdot\epsilon\cdot
\Prob{}{\eta}{\Fin\psi}}.
\end{array}\]

We conclude, by definition of $\epsilon'$, that both inequalities
hold.

Now, putting (\ref{eq:cut1}) and (\ref{eq:etsstarisclose})
together, we have
$\CP{+}{}{\Fin\phi}{\Fin\psi}-\CP{}{\eta^\star}{\Fin\phi}{\Fin\psi}<
3\cdot \epsilon = \hat{\epsilon}$, which concludes the proof for
this case.

\noindent $\bullet$ \textbf{Case $\phi$ is of the form $\F\phi$ and $\psi$ is of the form $\Glo
\psi$:\ } \text{}\\ We now construct a cut $K$ and a scheduler
$\eta^\star$ such that $K^c$ is finite, $\Phi\cup\Psi\subseteq K$,
$\eta^\star$ is HI in $K\setminus(\Phi\cup\Psi)$, and
$\CP{}{\eta^\star}{\G\lnot\phi}{\G\psi}-\CP{-}{}{
\G\lnot\phi}{\G\psi}<\hat{\epsilon}$. Note that such a cut and
scheduler also satisfy $\CP{+}{}{\F\phi}{\G\psi}-\CP{}{\eta^\star}{
\F\phi}{\G\psi}<\hat{\epsilon}$.


The proof goes similarly to the previous case. We start by
defining the probability of paths of length $N$ always satisfying
$\phi$ as $\Prob{}{\eta}{=\!N,\Glo\phi}\eqdef
\Prob{}{\eta}{\Cyl{\{\sigma\in\FPaths{\Pi}\mid\allowbreak\sigma\models
\G\phi\land|\sigma|= N \}}}$. Note that for all \pCTL formula of
the form $\G\phi$ and schedulers $\eta$ we have
\[\lim_{N\rightarrow \infty}\Prob{}{\eta}{=\! N, \G\phi}=\Prob{}{\eta}{\G\phi}.
\]
\noindent The same result holds for the $\pCTL\!^\star$ formula
$\Glo (\phi\land \psi)$. It is easy to check that for all $N$ and
$\phi$ we have
$\Prob{}{\eta}{=\!N,\G\phi}\geq\Prob{}{\eta}{\G\phi}$. 

Now we take a scheduler $\eta$ and a number $N$ such that:
\begin{align*}
0\leq\CP{}{\eta}{\Glo \lnot\phi}{\Glo \psi}-\CP{-}{}{\Glo \lnot\phi}{\Glo \psi}&<\epsilon\eqdef \hat{\epsilon}/ 3, \text{ and}\\
0\leq\Prob{}{\eta}{=\!N,\Glo(\lnot\phi\land\psi)}-\Prob{}{\eta}{\Glo(\lnot\phi\land\psi)}&<\epsilon',\text{
and}\\
0\leq\Prob{}{\eta}{=\!N,\Glo \psi}-\Prob{}{\eta}{\Glo
\psi}&<\epsilon'.
\end{align*}

\noindent where $\epsilon'$ is such that
$\allowbreak\epsilon'\allowbreak<\allowbreak\min\left(\epsilon\cdot
\Prob{}{\eta}{\Glo \psi},\allowbreak\frac{\epsilon \cdot
\Prob{}{\eta}{\Glo \psi}^2}{\Prob{}{\eta}{\Glo(\lnot \phi \land
\psi)}}\right)$.


\rev{We define $K$ as before, i.e., $K\eqdef \Phi\cup\Psi\cup \FPaths{\leq\!N,\Pi}$. In addition, we can construct (as we did in the previous case) a scheduler $\eta^\star$ behaving as $\eta$ for paths of length at most $N$ and maximizing (instead of minimizing as in the previous case)
$\Prob{}{}{\Glo \psi}$ afterwards. Again, it is easy to check that $\eta^\star$ is HI in $K\setminus(\Phi\cup\Psi)$.}

Then we have
\[\begin{array}{rcccl}
\Prob{}{\eta}{\Glo \psi}&\leq& \Prob{}{\eta^\star}{\Glo \psi}&\leq&\Prob{}{\eta}{=\! N,\Glo \psi},\\
\Prob{}{\eta}{\Glo (\lnot\phi\land \psi)} - \epsilon\cdot
\Prob{}{\eta}{\Glo \psi}&<& \Prob{}{\eta^\star}{\Glo(\lnot
\phi\land \psi)} &\leq&\Prob{}{\eta}{= \! N, \Glo(\lnot\phi\land
\psi)}.
\end{array}\]

In addition, it is easy to check that

\[\begin{array}{rcccl}
-\Prob{}{\eta}{\Glo \psi}\cdot \epsilon &<& \Prob{}{\eta^\star}{\Glo(\lnot\phi\land\psi)} - \Prob{}{\eta}{\Glo(\lnot\phi\land\psi)} &<& \epsilon'\\
0 &\leq& \Prob{}{\eta^\star}{\Glo \psi} - \Prob{}{\eta}{\Glo \psi} &<& \epsilon'.\\
\end{array}\]

Similarly to the previous case we now show that
\begin{equation}\label{eq:baseCaseGlob}
|\CP{}{\eta}{\Glo \lnot\phi}{\Glo \psi}-\CP{}{\eta^\star}{\Glo
\lnot\phi}{\Glo \psi}|< 2\cdot\epsilon.
\end{equation}
 \noindent which together with $\CP{}{\eta}{\Glo \lnot\phi}{\Glo \psi}-\CP{-}{}{\Glo \lnot\phi}{\Glo
 \psi}<\epsilon$ concludes the proof.

 In order to prove (\ref{eq:baseCaseGlob}) we show that

 \[-2\cdot\epsilon < \CP{}{\eta^\star}{\Glo \lnot\phi}{\Glo \psi} - \CP{}{\eta}{\Glo \lnot\phi}{\Glo \psi} < \epsilon\]

 or, equivalently

 \begin{itemize}
 \item[a)] $\Prob{}{\eta^\star}{\Glo(\lnot\phi\land\psi)} \cdot
 \Prob{}{\eta}{\Glo \psi} -
 \Prob{}{\eta}{\Glo(\lnot\phi\land\psi)} \cdot
 \Prob{}{\eta^\star}{\Glo \psi}< \\
 \hspace{2cm} \Prob{}{\eta}{\Glo \psi} \cdot
 \Prob{}{\eta^\star}{\Glo \psi} \cdot \epsilon $, and
 \item[b)] $ 2 \cdot \Prob{}{\eta}{\Glo \psi} \cdot
 \Prob{}{\eta^\star}{\Glo \psi} \cdot \epsilon  <\\
 \hspace{2cm}\Prob{}{\eta^\star}{\Glo(\lnot\phi\land\psi)} \cdot
 \Prob{}{\eta}{\Glo \psi} -
 \Prob{}{\eta}{\Glo(\lnot\phi\land\psi)} \cdot
 \Prob{}{\eta^\star}{\Glo \psi}$.
 \end{itemize}

It is possible to verify that a) is equivalent to
$\epsilon^\prime<\epsilon \cdot \Prob{}{\eta}{\Glo \psi}$ and that
b) is equivalent to $\epsilon^\prime<\frac{\epsilon \cdot
\Prob{}{\eta}{\Glo \psi}^2}{\Prob{}{\eta}{\Glo(\lnot \phi \land
\psi)}}$. The desired result follows by definition of
$\epsilon^\prime$.
\end{proof}

In the proof of the following lemma we step-by-step find pairs
$(K,\eta)$ in ${\cal K}$ with larger $K$ and $\eta$ still close to
the optimal until finally $K$ is equal to the whole of
$\FPaths{\Pi}$.

\begin{lem}[completeness of ${\cal K}$]\label{lem:extendingCut}
There exists a scheduler $\eta$ such that $(\FPaths{\Pi},\eta)\in{\cal K}$.
\end{lem}

\begin{proof}
We prove that if we take a $(K,\eta)\in {\cal K}$ such that
$|K^c|$ is minimal then $K^c=\emptyset$ or, equivalently,
$K=\FPaths{\Pi}$. Note that a pair $(K,\eta)$ with minimal $|K^c|$ exists because, by
the previous lemma, ${\cal K}$ is not empty.

The proof is by contradiction: we suppose $K^c\not=\emptyset$ and
arrive to a contradiction on the minimality of $|K^c|$. Formally,
we show that for all $(K,\eta)\in{\cal{K}}$ such that
$K^c\not=\emptyset$, there exists a cut $K^\star\supset K$ and a
scheduler $\eta^\star$ such that $(K^\star,\eta^\star)\in{\cal
K}$, i.e. such that $\eta^\star$ is HI in
$K^\star\setminus(\Phi\cup\Psi)$ and $\CP{}{\eta}{
\phi}{\psi}\leq\CP{}{\eta^\star}{\phi}{\psi}$.

To improve readability, we prove this result for the case
$\phi$ is of the form $\F\phi$ and $\psi$ is of the form $\F\psi$. However, all the technical details of the proof hold for arbitrary $\phi$ and $\psi$.

Let us start defining the \emph{boundary} of a cut $K$ as
\begin{equation*}
    \delta K \eqdef \{ \sigma_1 \in K\, |\, \forall \, \sigma_2 \in \FPaths{M}\,.\, \sigma_2 \sqsubset \sigma_1 \implies \sigma_2 \notin K \}.
\end{equation*}

Let $\rho$ be a path in $K^c$ such that $\rho t\in \delta K$. Note
that by assumption of $K^c\not=\emptyset$ such $\rho$ exists. Now,
if for all paths $\alpha\in K$ we have
$\last{\alpha}=\last{\rho}\Longrightarrow\eta(\alpha)=\eta(\rho)$
then $\eta$ is also HI in $(K\cup\{\rho\})\setminus(\Phi\cup\Psi)$
so we have $(K\cup\{\rho\},\eta)\in{\cal K}$ as we wanted to show.
Now let us assume otherwise, i.e. that there exists a path
$\alpha\in K\setminus(\Phi\cup\Psi)$ such that
$\last{\alpha}=\last{\rho}$ and $\eta(\alpha)\not=\eta(\rho)$. We
let $s\eqdef \last{\rho}$, $\pi_1\eqdef\eta(\rho)$,
$\pi_2\eqdef\eta(\alpha)$, and $K_s\eqdef \{\sigma\in K\mid
\last{\sigma}=s\}\setminus(\Phi\cup\Psi)$. Note that for all
$\alpha'\in K_s$ we have $\eta(\alpha')=\pi_2$, this follows from
the fact that $\eta$ is HI in $K\setminus(\Phi\cup\Psi)$.

Figure \ref{fig:sHIchart} provides a graphic representation of
this description. The figure shows the set $\FPaths{\Pi}$ of all
finite paths of $\Pi$, the cut $K$ of $\FPaths{\Pi}$, the path
$\rho$ reaching $s$ (in red and dotted border line style), a path
$\alpha$ reaching $s$ in $K$ (in blue and continuous border line
style). The fact that $\eta$ takes different decisions $\rho$ and
$\alpha$ is represented by the different colors and line style of
their respective last states $s$.

\setlength{\abovecaptionskip}{-0pt plus 1pt minus 1pt}
\begin{figure}[h]
\centering
\includegraphics[width=8cm]{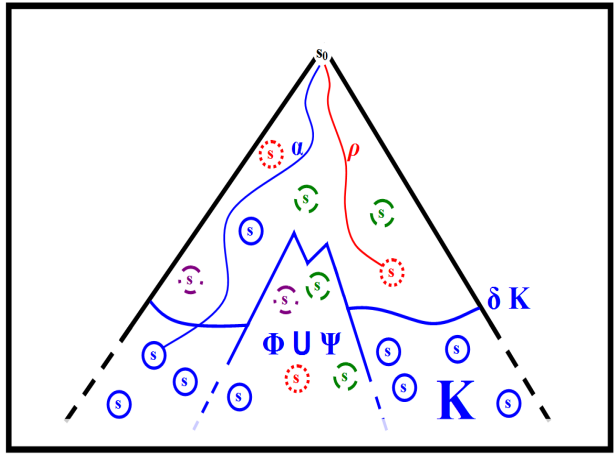}
  \caption{Graphic representation of $\FPaths{\Pi}$, $\Phi\cup\Psi$, $K$, $\delta K$, $\rho$, and $\alpha$.}\label{fig:sHIchart}
\end{figure}
\setlength{\abovecaptionskip}{-4pt plus 1pt minus 1pt}

We now define two schedulers $\eta_1$ and $\eta_2$ such that they
are HI in $(K\cup\{\rho\})\setminus(\Phi\cup \Psi)$. Both $\eta_1$
and $\eta_2$ are the same than $\eta$ everywhere but in $K_s$ and
$\rho$, respectively. The first one selects $\pi_1$ for all
$\alpha\in K_s$ (instead of $\pi_2$ as $\eta$ does), and the
second scheduler selects $\pi_2$ in $\rho$ (instead of $\pi_1$):

\[
    \eta_1(\sigma) =
        \begin{cases}
            \pi_1 & \text{if $\sigma \in K_s$} \\
            \eta(\sigma)  & \text{otherwise}
        \end{cases}
     \qquad \text{and} \qquad
    \eta_2(\sigma) =
        \begin{cases}
            \pi_2 & \text{if $\sigma = \rho$} \\
            \eta(\sigma)  & \text{otherwise}.
        \end{cases}\]

Now we plan to prove that either $\eta_1$ is ``better'' than $\eta$
or $\eta_2$ is ``better'' than $\eta$. In order to prove this result, we will show that:

\begin{equation}\label{eq:betterScheduler}\CP{}{\eta_1}{\F\phi}{\F{\psi}} \leq
\CP{}{\eta_2}{\F\phi}{\F{\psi}}\!\Longleftrightarrow\!
\CP{}{\eta}{\F\phi}{\F{\psi}} \leq\CP{}{\eta_2}{\F\phi}{\F{\psi}}
\end{equation}
and
\begin{equation}\label{eq:betterScheduler2}\CP{}{\eta_2}{\F\phi}{\F{\psi}} \leq
\CP{}{\eta_1}{\F\phi}{\F{\psi}}\!\Longleftrightarrow\!
\CP{}{\eta}{\F\phi}{\F{\psi}} \leq\CP{}{\eta_1}{\F\phi}{\F{\psi}}
\end{equation}

\noindent Therefore, if $\CP{}{\eta_1}{\F\phi}{\F{\psi}} \leq
\CP{}{\eta_2}{\F\phi}{\F{\psi}}$ then we have $(K\cup\{\rho\},
\eta_2)\in {\cal K}$, and otherwise $(K\cup\{\rho\}, \eta_1)\in
{\cal K}$. So, the desired result follows from
(\ref{eq:betterScheduler}) and (\ref{eq:betterScheduler2}). We
will prove (\ref{eq:betterScheduler}), the other case follows the
same way.

In order to prove (\ref{eq:betterScheduler}) we need to analyze
more closely the conditional probability
$\CP{}{}{\Fin\phi}{\Fin\psi}\eqdef \CP{}{}{\Phi}{\Psi}$ for each
of the schedulers $\eta,$  $\eta_1$, and $\eta_2$. For that
purpose we partition the sets $\Phi\cap\Psi$ and $\Psi$ into four \emph{parts}, i.e. disjoint sets. The plan is to partition $\Phi\cap\Psi$ and $\Psi$ in such way that
we can make use of the fact that $\eta$, $\eta_1$, and $\eta_2$
are similar to each other (they only differ in the decision taken
in $K_s$ or $\rho$) obtaining, in this way, that the probabilities
of the parts are the same under these schedulers or differ
only by a factor (this intuition will become clearer later on in the proof), such condition is the key element of our proof of (\ref{eq:betterScheduler}). Let us start by partitioning $\Psi$:

\begin{itemize}
\item[i)]{ We define $\Psi_{\overline{\rho},{\overline{k}_s}}$ as the set of
paths in $\Psi$ neither passing through $K_s$ nor $\rho$, formally
\[\Psi_{\overline{\rho},{\overline{k}_s}}\eqdef
\Psi\setminus(\Cyl{K_s} \cup \Cyl{\rho})\] }
\item[ii)]{We define $\Psi_{\rho,{\overline{k}_s}}$ as the set of paths
in $\Psi$ passing through $\rho$ but not through $K_s$, i.e.:
\[\Psi_{\rho,{\overline{k}_s}}\eqdef \Psi\cap (\Cyl{\rho}\setminus\Cyl{K_s}).\] }
\item[iii)]{We define $\Psi_{\rho,k_s}$ as the set of paths in $\Psi$ passing through $\rho$ and $K_s$, i.e.:
\[\Psi_{\rho,k_s}\eqdef\Psi\cap\Cyl{\rho}\cap\Cyl{K_s}.\] }
\item[iv)]{We define $\Psi_{\overline{\rho},k_s}$ as the set of paths in $\Psi$ passing through $K_s$ but not through $\rho$, i.e.:
\[\Psi_{\overline{\rho},k_s}\eqdef\Psi\cap(\Cyl{K_s}\setminus\Cyl{\rho}).\] }
\end{itemize}

Note that $\Psi=\Psi_{\rho , s}\cup\Psi_{\rho,
\overline{k}_s}\cup\Psi_{\overline{\rho}
,k_s}\cup\Psi_{\overline{\rho} ,\overline{k}_s}.$

Similarly, we can partition the set of paths $\Phi\cap\Psi$ into
four parts obtaining
$\Phi\cap\Psi=(\Phi\cap\Psi)_{\rho , s}\cup(\Phi\cap\Psi)_{\rho,
\overline{k}_s}\cup(\Phi\cap\Psi)_{\overline{\rho}
,k_s}(\Phi\cap\Psi)_{\overline{\rho} ,\overline{k}_s}.$



In the following we analyze the probabilities (under $\eta$) of each part separately.

\begin{itemize}
\item{ The probability of $\Psi_{\rho,\overline{k}_s}$ can be written as $p_\rho\cdot x_{\psi}$, where
$p_{\rho}$ is the probability of $\rho$ and $x_{\psi}$ is the
probability of reaching $\psi$ without passing through $K_s$ given
$\rho$. More formally,  $\Prob{}{\eta}{\Psi_{\rho,\overline{k}_s}}=
\Prob{}{\eta}{\Psi\cap(\Cyl{\rho}\setminus\Cyl{K_s})}\allowbreak=
\Prob{}{\eta}{\Cyl{\rho}}\cdot\CP{}{\eta}{\Psi\cap
(\Cyl{\rho}\setminus\Cyl{K_s})}{\Cyl{\rho}}\eqdef p_\rho\cdot
x_{\psi}. $}
\comment{What if $y_s=1$???}
\item{ The probability of $\Psi_{\rho,k_s}$ can be written as $p_\rho\cdot x_s\cdot\frac{y_{\psi}}{1-y_s}$, where
$x_s$ is the probability of passing through $K_s$ given $\rho$,
$y_{\psi}$ is the probability of, given $\alpha$, reaching $\psi$
without passing through $K_s$ after $\alpha$; and $y_s$ is the
probability of, given $\alpha$, passing through $K_s$ again.
Remember that $\alpha$ is any path in $K_s$. Formally, we have 
\[\begin{array}{lcl} \Prob{}{\eta}{\Psi_{\rho,k_s}}&=
&\Prob{}{\eta}{\Psi\cap\Cyl{\rho}\cap\Cyl{K_s}}\\
&=&\Prob{}{\eta}{\cyl{\rho}}\cdot\CP{}{\eta}{\cyl{K_s}}{\cyl{\rho}}\cdot \CP{}{\eta}{\Psi}{\cyl{K_s}\cap\cyl{\rho}}\\
&=&p_\rho\cdot x_s\cdot \CP{}{\eta}{\Psi}{\cyl{\alpha}}.\\
\end{array}\] 
\noindent Furthermore, 
\[\begin{array}{lcl}
\CP{}{\eta}{\Psi}{\Cyl{\alpha}}&=&\CP{}{\eta}{\Psi}{\overline{K}_{s\,\text{again}}\cap\Cyl{\alpha}} \\
&=&\frac{\CP{}{\eta}{\Psi\cap\overline{K}_{s\,\text{again}}}{\Cyl{\alpha}}}{\CP{}{\eta}{\overline{K}_{s\,\text{again}}}{\Cyl{\alpha}}}\\
&=&\frac{y_\psi}{1-y_s}.\\
\end{array}\]
\noindent where $\overline{K}_{s\,\text{again}}\eqdef \cyl{\alpha} \setminus \{\omega \in \cyl{\alpha \sigma}\mid \alpha \sigma \in K_s\}$. }
\item{ The probability of $\Psi_{\overline{\rho},k_s}$ can be written as $p_{k_s}\cdot\frac{y_{\psi}}{1-y_s}$, where
$p_{k_s}$ is the probability of passing though $K_s$ without
passing through $\rho$. Formally, $\Prob{}{\eta}{\Psi_{\overline{\rho},k_s}}=\Prob{}{\eta}{\Psi\cap(\Cyl{K_s}
\setminus\Cyl{\rho})}=\Prob{}{\eta}{\Cyl{K_s}\setminus\Cyl{\rho}}\cdot
\CP{}{\eta}{\Psi}{\Cyl{\alpha}}\eqdef
p_{s_k}\cdot\frac{y_\psi}{1-y_s} $}
\item{Finally, we write the probability of $\Psi_{\overline{\rho}\overline{k}_s}$ as $p_{\psi}$.}
\end{itemize}

A similar reasoning can be used to analyze the probabilities associated to the parts of $\Phi\cap\Psi$. In this way we obtain that 
(1) $\Prob{}{\eta}{(\Phi\cap\Psi)_{\rho,\overline{k}_s}}=p_\rho\cdot x_{\phi\psi}$, where $x_{\phi\psi}$ is the
probability of reaching $\phi$ and $\psi$ without passing through
$K_s$ given $\rho$, (2) $\Prob{}{\eta}{(\Phi\cap\Psi)_{\rho,k_s}}=p_\rho\cdot x_s\cdot\frac{y_{\phi\psi}}{1-y_s}$, where
$y_{\phi\psi}$ is the probability of reaching $\phi$ and $\psi$
without passing through $K_s$ afterwards given $\alpha$, (3) $\Prob{}{\eta}{(\Phi\cap\Psi)_{\overline{\rho},k_s}}=p_{k_s}\cdot\frac{y_{\phi\psi}}{1-y_s}$, and
(4) $\Prob{}{\eta}{(\Phi\cap\Psi)_{\overline{\rho},\overline{k}_s}}=p_{\phi\psi}$.

In order to help the intuition of the reader, we now provide a graphical representation of the probability (under $\eta$) of the sets $\Phi\cap\Psi$ and $\Psi$ by means of a Markov chain (see 
Figure \ref{fig:MarkovChainEta}). The missing values are defined as $p_{\overline{\phi}\psi}\eqdef
p_{\psi}-p_{\phi\psi}$, $p_\emptyset\eqdef
1-p_{s_k}-p_{\rho}-p_{\psi}$; and similarly for \comment{This
paragraph needs to be improved} $x_{\overline{\phi} \psi}$,
$x_\emptyset$, $y_{\overline{\phi} \psi}$, and $y_\emptyset$.
Furthermore, absorbing states $\phi\psi$ denote states where
$\phi\land\psi$ holds, absorbing states $\overline{\phi}\psi$
denote states where $\lnot\phi\land\psi$ holds, and
$\overline{\psi}$ denote a state where $\lnot\psi$ holds. Finally, the state $\rho$ represents the state of the model where
$\rho$ has been just reached and $\alpha$ a state where any of the
paths $\alpha$ in $K_s$ as been just reached. To see how this Markov Chain is related to the probabilities of
$\Phi\cap\Psi$ and $\Psi$ on the original MDP consider, for example, the probabilities of the set
$\Phi\cap\Psi$. It is easy to show that
\[\begin{array}{lcl}
\Prob{}{\eta}{\Phi\cap\Psi}&=&\Prob{}{\eta}{\Phi_{\rho,k_s}} +
\Prob{}{\eta}{\Phi_{\overline{\rho},k_s}} +
\Prob{}{\eta}{\Phi_{\rho,\overline{k}_s}}+
\Prob{}{\eta}{\Phi_{\overline{\rho},\overline{k}_s}}\\
&=& p_{\phi\psi} + p_\rho\cdot x_{\phi\psi} + p_\rho \cdot
x_s\cdot
\frac{y_{\phi\psi}}{1-y_s}+p_{s_k}\cdot \frac{y_{\phi\psi}}{1-y_s}=\Prob{}{M}{\Fin\phi\psi}.
\end{array}\]

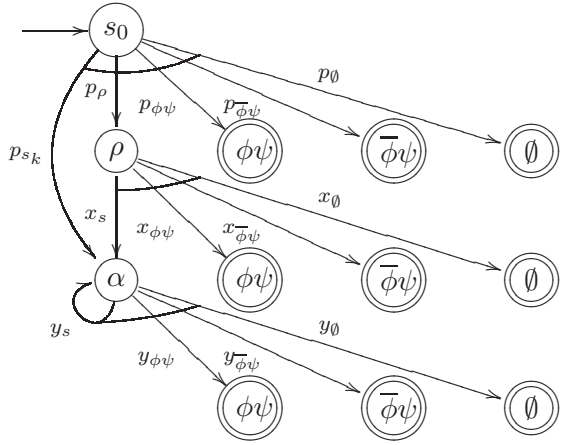
\begin{figure}[!h]
\begin{equation*}
    \xymatrix{%
        {} \ar[r]
        &*++[o][F]{s_0}
            \ar@/_2 pc/[dd]_{p_{s_k}}|(0.15){}="l1"
            \ar[d]_{p_\rho}
            \ar[dr]_{p_{\phi\psi}}
            \ar[drr]_{p_{\overline{\phi}\psi}}
            \ar[drrr]^{p_\emptyset}|(0.2){}="r1"
            \POS {"l1"} \ar@{-}@/_0.3pc/{"r1"}
        \\
        &*++[o][F]{\rho}
            \ar[d]_{x_s}|(0.3){}="l2"
            \ar[dr]_{x_{\phi\psi}}
            \ar[drr]_{x_{\overline{\phi}\psi}}
            \ar[drrr]^{x_\emptyset}|(0.2){}="r2"
            \POS {"l2"} \ar@{-}@/_0.1pc/{"r2"}
        &
        *++[o][F=]{\ \phi\psi}
        &
        *++[o][F=]{\ \overline{\phi}\psi}
        &
        *++[o][F=]{\ \emptyset\ }
        \\
        &
        *++[o][F]{\alpha}
            \ar @(d,l)[]^{y_s}|(0.3){}="l3"
            \ar[dr]_{y_{\phi\psi}}
            \ar[drr]_{y_{\overline{\phi}\psi}}
            \ar[drrr]^{y_\emptyset}|(0.2){}="r3"
            \POS {"l3"} \ar@{-}@/_0.1pc/{"r3"}
        &
        *++[o][F=]{\ \phi\psi}
        &
        *++[o][F=]{\ \overline{\phi}\psi}
        &
        *++[o][F=]{\ \emptyset\ }
        \\
        &
        &
        *++[o][F=]{\ \phi\psi}
        &
        *++[o][F=]{\ \overline{\phi}\psi}
        &
        *++[o][F=]{\ \emptyset\ }
    }
\end{equation*}
 \caption{Graphical representation of how we write the probability of each partition: $M_\eta$.}\label{fig:MarkovChainEta}
 \end{figure}

We note that the values $p_{s_k}$, $p_\rho$, $p_{\phi\psi}$,
$p_{\overline{\phi}\psi}$, and $p_\emptyset$ coincide for $\eta$,
$\eta_1$, and $\eta_2$. Whereas the values $\vec{x} \eqdef (x_s,
x_\psi, x_{\phi\psi}, x_\emptyset)$ coincide for $\eta$ and
$\eta_1$ and the values $\vec{y} \eqdef (y_s, y_\psi,
y_{\phi\psi}, y_\emptyset)$ coincide for $\eta$ and $\eta_2$.
Thus, the variant of $M_{\eta}$ in which $\vec{y}$ is replaced by
$\vec{x}$ describes the probability of each partition under the
scheduler $\eta_1$ instead of $\eta$. Similarly, the variant on
which $\vec{x}$ is replaced by $\vec{y}$ represents the
probability of each partition under the scheduler $\eta_2$. 

Now we have all the ingredients needed to prove
(\ref{eq:betterScheduler}). Our plan is to show that:

\begin{itemize}
\item[1)] $\CP{}{\eta_1}{\Phi}{\Psi} \leq \CP{}{\eta_2}{\Phi}{\Psi}
         \iff (p_\rho + p_{s_k})\cdot d \leq 0 \iff d \leq 0$, and\\
\item[2)] $\CP{}{\eta}{\Phi}{\Psi} \leq \CP{}{\eta_2}{\Phi}{\Psi}\iff (1-y_s)\cdot p_\rho \cdot d \leq 0 \iff d\leq 0.$ \\
\end{itemize}
\noindent where $d$ is the following determinant
    \begin{equation*}
        d =
        \begin{vmatrix}
            p_s+p_{s_k} \ &\  x_s - 1\ & \ y_s - 1 \\
            p_{\phi\psi} & x_{\phi\psi} & y_{\phi\psi} \\
            p_{\psi} & x_{\psi} & y_\psi
        \end{vmatrix}.
    \end{equation*}

We now proceed to prove 1)
\[\begin{array}{c}
\CP{}{\eta_1}{\Phi}{\Psi}-\CP{}{\eta_2}{\Phi}{\Psi}\leq 0\\
\iff\\
\frac{p_{\phi\psi}+ p_\rho\cdot x_{\phi\psi}+ (p_\rho\cdot x_s +
p_{s_k}) \cdot \frac{x_{\phi\psi}}{1-x_s}}{p_\psi + p_\rho\cdot
x_\psi + (p_\rho \cdot x_s + p_{s_k}) \cdot
\frac{x_{\psi}}{1-x_s}}
 -
\frac{p_{\phi\psi}+ p_\rho\cdot y_{\phi\psi}+ (p_\rho\cdot y_s +
p_{s_k}) \cdot \frac{y_{\phi\psi}}{1-y_s}}{p_\psi + p_\rho\cdot
y_\psi + (p_\rho \cdot y_s + p_{s_k}) \cdot
\frac{y_{\psi}}{1-y_s}} \leq 0\\
\iff\\
\frac{p_{\phi\psi} (1-x_s) + p_{\rho} x_{\phi\psi}+ p_{s_k}
x_{\phi\psi}}{p_\psi  (1-x_s) + p_\rho x_\psi + p_{s_k} x_{\psi}}
 -
\frac{p_{\phi\psi} (1-y_s) + p_s y_{\phi\psi}+ p_{s_k}
y_{\phi\psi}}{p_\psi  (1-y_s) + p_\rho y_\psi + p_{s_k} y_{\psi}}
\leq 0\\
\iff\\
\begin{vmatrix}
            p_{\phi\psi} (1-x_s) + p_{\rho} x_{\phi\psi}+ p_{s_k} x_{\phi\psi}\ & \  p_{\phi\psi} (1-y_s)+ p_\rho y_{\phi\psi}+ p_{s_k}  y_{\phi\psi} \\
            p_\psi  (1-x_s) + p_\rho x_\psi + p_{s_k}  x_{\psi} \ &\  p_\psi  (1-y_s) + p_\rho y_\psi + p_{s_k}  y_{\psi}
\end{vmatrix}\leq 0.\\
\end{array}\]
A long but straightforward computation shows that the 2x2 determinant
in the line above is equal to $(p_\rho + p_{s_k}) d$.

The proof of 2) proceeds along the same lines.

\[\begin{array}{c}
\CP{}{\eta}{\Phi}{\Psi}-\CP{}{\eta_2}{\Phi}{\Psi}\leq 0\\
\iff \\
\frac{p_{\phi\psi}+ p_\rho\cdot x_{\phi\psi}+ (p_\rho\cdot x_s +
p_{s_k}) \cdot \frac{y_{\phi\psi}}{1-y_s}}{p_\psi + p_\rho\cdot
x_\psi + (p_\rho \cdot x_s + p_{s_k}) \cdot
\frac{y_{\psi}}{1-y_s}}
 -
\frac{p_{\phi\psi}+ p_\rho\cdot y_{\phi\psi}+ (p_\rho\cdot y_s +
p_{s_k}) \cdot \frac{y_{\phi\psi}}{1-y_s}}{p_\psi + p_\rho\cdot
y_\psi + (p_\rho \cdot y_s + p_{s_k}) \cdot
\frac{y_{\psi}}{1-y_s}} \leq 0\\
\iff \\
\frac{p_{\phi\psi} (1-y_s) + p_{\rho} x_{\phi\psi} (1-y_s)+
(p_\rho x_s + p_{s_k}) y_{\phi\psi}}{p_\psi  (1-y_s) + p_\rho
x_\psi (1-y_s) + (p_\rho x_s + p_{s_k}) y_{\psi}}
 -
\frac{p_{\phi\psi} (1-y_s) + p_\rho y_{\phi\psi}+ p_{s_k}
y_{\phi\psi}}{p_\psi  (1-y_s) + p_\rho y_\psi + p_{s_k} y_{\psi}}
\leq 0\\
\iff\\
\end{array}\]

\[\begin{array}{c}
\begin{vmatrix}
            p_{\phi\psi} (1\!-\!y_s) \!+\! p_{\rho} x_{\phi\psi} (1\!-\!y_s)\!+\! (p_\rho x_s \!+\! p_{s_k}) y_{\phi\psi} \ &\ p_{\phi\psi} (1\!-\!y_s)
\!+\! p_\rho y_{\phi\psi}\!+\! p_{s_k}  y_{\phi\psi} \\
            p_\psi  (1\!-\!y_s) \!+\! p_\rho x_\psi (1\!-\!y_s) \!+\! (p_\rho x_s \!+\! p_{s_k}) y_{\psi} & p_\psi (1\!-\!y_s) \!+\! p_\rho y_\psi \!+\! p_{s_k}  y_{\psi}
\end{vmatrix} \medskip\\
\leq 0
\end{array}\]
and also here a long computation shows that this last 2x2 determinant is
equal to $(1 - y_s)\cdot p_\rho\cdot d$.
\end{proof}

Finally, we have all the ingredients needed to prove that there
exists a scheduler close to the supremum which is HI before the
stopping condition.

\begin{corollary}\label{cor:HIbefore}[HI before stopping condition]
Let $\Pi$ be a $\MDP$, $\phi, \psi \in\Path$. Then for all
$\hat{\epsilon}>0$, there exists a scheduler $\eta^\star$ such
that $\CP{+}{}{\phi}{\psi}-\CP{}{\eta^\star}{\phi}{\psi} <
\hat{\epsilon}$ and $\eta^\star$ is history independent before the stopping condition.
\end{corollary}
\begin{proof} Follows directly from Lemma \ref{lem:initialCut} and Lemma \ref{lem:extendingCut}.
\end{proof}

We now proceed with the construction of a maximizing scheduler and
HI after the stopping condition.

\begin{lem}\label{lem:HIafter}[HI after stopping condition]
Let $\Pi$ be a $\MDP$, $\phi, \psi \in\Path$, and $\varphi=
\StopCond{\phi}\lor \StopCond{\psi}$. Then for all schedulers
$\eta$ there exists a scheduler $\eta^\star$ such that
\begin{itemize}
\item[1)] $\eta^\star$ behaves like $\eta$ before the stopping
condition,
\item[2)] $\eta^\star$ is HI after the stopping condition
$\varphi$, and
\item[3)] $\CP{}{\eta}{\phi}{\psi}\leq\CP{}{\eta^\star}{\phi}{\psi}$.
\end{itemize}
\end{lem}

\begin{proof}
We will prove this result for the case in which $\phi$ is of the form $\F\phi$  and
$\psi$ is of the form $\F\psi$, the proof for the remaining cases follows in the
same way.

Let us start by introducing some notation. We define,
respectively, the set of paths reaching $\phi$, the set of paths not reaching $\phi$,  the set of paths
reaching $\phi$ without reaching $\psi$ before, and the set of
paths reaching $\psi\land\lnot\phi$ without reaching $\phi$ before
as follows
\begin{align*}
\Delta_{\phi}\eqdef&\{\omega\in\Paths{\Pi}\mid \omega\models
\Fin\phi\},\\
\rev{\Delta_{\lnot\phi}\eqdef}&\rev{\{\omega\in\Paths{\Pi}\mid \omega\models
\Glo\lnot\phi\},}\\
\Delta_{\overline{\psi}\phi}\eqdef&\{\omega\in\Paths{\Pi}\mid
\omega\models
\lnot \psi \U \phi\},\\
\Delta_{\overline{\phi}\psi}\eqdef&\{\omega\in\Paths{\Pi}\mid
\omega\models \lnot \phi \U (\psi\land\lnot\phi)\}.
\end{align*}

\noindent Note that the last two sets are disjoint. It is easy to
check that

\[
\Delta_\phi\cap\Delta_\psi=(\Delta_{\overline{\psi}\phi}\cap\Delta_{\psi})\cup(\Delta_{\overline{\phi}\psi}\cap\Delta_\phi),\]
\[\Delta_\psi=\Delta_{\overline{\phi}\psi}\cup
(\Delta_{\overline{\psi}\phi}\cap\Delta_\psi)=[(\Delta_{\overline{\phi}\psi}\cap\Delta_{\phi})\cup(\Delta_{\overline{\phi}\psi}\cap\Delta_{\lnot\phi})]
\cup(\Delta_{\overline{\psi}\phi}\cap\Delta_\psi).\]

Let us now define the minimal set of finite paths ``generating''
(by their basic cylinders) $\Delta_{\overline{\psi}\phi}$ and
$\Delta_{\overline{\phi}\psi}$: $K_{\overline{\psi}\phi}\eqdef
\{\sigma\in\FPaths{\Pi}\mid \last{\sigma}\models \phi\land
\forall\,i<|\sigma|\,:\,\sigma_i\rev{\models\lnot\phi\land \lnot\psi}\}$ and
similarly $K_{\overline{\phi}\psi}\eqdef
\{\sigma\in\FPaths{\Pi}\mid \last{\sigma}\models
(\psi\land\lnot\phi) \land
\forall\,i<|\sigma|\,:\,\sigma_i\rev{\models\lnot\phi\land\lnot\psi\}}$. Note
that $\Delta_{\overline{\psi}\phi}=\Cyl{K_{\overline{\psi}\phi}}$
and $\Delta_{\overline{\phi}\psi}=\Cyl{K_{\overline{\phi}\psi}}$.
Now we can write
\[\CP{}{\eta}{\Fin\phi}{\Fin\psi}=\frac{\Prob{}{\eta}{\Cyl{K_{\overline{\psi}\phi}}\cap\Delta_{\psi}}
+\Prob{}{\eta}{\Cyl{K_{\overline{\phi}\psi}}\cap\Delta_{\phi}}}
{\Prob{}{\eta}{\Cyl{K_{\overline{\psi}\phi}}\cap\Delta_{\psi}}+\Prob{}{\eta}{\Cyl{K_{\overline{\phi}\psi}}\cap\Delta_{\phi}}+\Prob{}{\eta}{\Cyl{K_{\overline{\phi}\psi}}\cap\Delta_{\lnot\phi}}}.\]
The construction of the desired scheduler $\eta^\star$ is in the
spirit of the construction we proposed for the scheduler in Lemma
\ref{lem:initialCut}. We let $S_{\phi}\eqdef\{s\in S\mid
s\models\phi\}$ and $S_{\psi}\eqdef\{s\in S\mid
s\models(\psi\land\lnot\phi)\}$. Note that $S_\phi$ and $S_\psi$
are disjoint. Now we define two families of schedulers
$\{\xi_s\}_{s\in S_\phi}$ and $\{\zeta_s\}_{s\in S_\psi}$ such
that: for all $s_1, s_2\in S_\phi$  
we have
$\Prob{}{s_1,\xi_{s_1}}{\Fin\psi}=\Prob{+}{s_1}{\Fin\psi}$,
$\Prob{}{s_2,\xi_{s_2}}{\Fin\psi}=\Prob{+}{s_2}{\Fin\psi}$, and
for all $\sigma_1 t\in\FPaths{s_1},\,\sigma_2 t\in\FPaths{s_2}$ we
have $\xi_{s_1}(\sigma_1 t)=\xi_{s_2}(\sigma_2 t)$. Similarly for
$\{\zeta_s\}_{s\in S_\psi}$: for all $s_1,s_2\in S_\phi$ we have
$\Prob{}{s_1,\zeta_{s_1}}{\Fin\phi}=\Prob{+}{s_1}{\Fin\phi}$,
$\Prob{}{s_2,\zeta_{s_2}}{\Fin\phi}=\Prob{+}{s_2}{\Fin\phi}$, and
for all $\sigma_1 t\in\FPaths{s_1},\,\sigma_2 t\in\FPaths{s_2}$ we
have $\zeta_{s_1}(\sigma_1 t)=\zeta_{s_2}(\sigma_2 t)$.

We now proceed to define $\eta^\star$:
\[
    \eta^\star(\sigma) \eqdef
        \begin{cases}
            \xi_s(\sigma_{|\alpha|}\cdots\sigma_{|\sigma|}) & \text{if $\alpha\sqsubseteq\sigma$ for some $\alpha\!\in\! K_\phi$ such that $\last{\alpha}\!=\!s$}, \\
            \zeta_s(\sigma_{|\alpha|}\cdots\sigma_{|\sigma|}) & \text{if $\alpha\sqsubseteq\sigma$ for some $\alpha\!\in\! K_\psi$ such that $\last{\alpha}\!=\!s$}, \\
            \eta(\sigma)  & \text{otherwise}.
        \end{cases}
\]
\noindent where $K_\phi\eqdef \{\sigma\in\fpaths \mid \last{\sigma}\in S_{\phi}\}$, and similarly $K_\psi\eqdef \{\sigma\in\fpaths \mid \last{\sigma}\in S_{\psi}\}$.

It is easy to check that $\eta^\star$ satisfies 1) and 2). As for
3) we first note that
$\Prob{}{\eta}{\Cyl{K_{\overline{\psi}\phi}}\cap\Psi}\leq
\Prob{}{\eta^\star}{\Cyl{K_{\overline{\psi}\phi}}\cap\Psi}$,
$\Prob{}{\eta}{\Cyl{K_{\overline{\phi}\psi}}\cap\Delta_{\phi}}\leq\Prob{}{\eta^\star}{\Cyl{K_{\overline{\phi}\psi}}\cap\Delta_{\phi}}$,
and
$\allowbreak\Prob{}{\eta}{\Cyl{K_{\overline{\phi}\psi}}\cap\Delta_{\lnot\phi}}\geq\Prob{}{\eta^\star}{\Cyl{K_{\overline{\phi}\psi}}\cap\Delta_{\lnot\phi}}$.

In addition, we need the following simple remark.

\begin{remark}\label{rem:incr}
Let $f: \mathbb{R}\to\mathbb{R}$ be a function defined as
$f(x)\eqdef\frac{a+x}{b+x}$ where $a$ and $b$ are constants in the
interval $[0,1]$ such that $b\geq a$. Then $f$ is increasing.
\end{remark}

Finally, we have

\begin{align*}
\CP{}{\eta}{\Fin\phi}{\Fin\psi}&=\frac{\Prob{}{\eta}{\Cyl{K_{\overline{\psi}\phi}}\cap\Delta_{\psi}}
+\Prob{}{\eta}{\Cyl{K_{\overline{\phi}\psi}}\cap\Delta_{\phi}}}
{\Prob{}{\eta}{\Cyl{K_{\overline{\psi}\phi}}\cap\Delta_{\psi}}+\Prob{}{\eta}{\Cyl{K_{\overline{\phi}\psi}}\cap\Delta_{\phi}}+\Prob{}{\eta}{\Cyl{K_{\overline{\phi}\psi}}\cap\Delta_{\lnot\phi}}}\\
& \qquad \qquad\{\text{by Remark \ref{rem:incr}}\}\\
&\leq\frac{\Prob{}{\eta^\star}{\Cyl{K_{\overline{\psi}\phi}}\cap\Delta_{\psi}}
+\Prob{}{\eta}{\Cyl{K_{\overline{\phi}\psi}}\cap\Delta_{\phi}}}
{\Prob{}{\eta^\star}{\Cyl{K_{\overline{\psi}\phi}}\cap\Delta_{\psi}}+\Prob{}{\eta}{\Cyl{K_{\overline{\phi}\psi}}\cap\Delta_{\phi}}+\Prob{}{\eta}{\Cyl{K_{\overline{\phi}\psi}}\cap\Delta_{\lnot\phi}}}\\
& \qquad \qquad\{\text{by Remark \ref{rem:incr}}\}\\
\end{align*}

\begin{align*}
&\leq\frac{\Prob{}{\eta^\star}{\Cyl{K_{\overline{\psi}\phi}}\cap\Delta_{\psi}}
+\Prob{}{\eta^\star}{\Cyl{K_{\overline{\phi}\psi}}\cap\Delta_{\phi}}}
{\Prob{}{\eta^\star}{\Cyl{K_{\overline{\psi}\phi}}\cap\Delta_{\psi}}+\Prob{}{\eta^\star}{\Cyl{K_{\overline{\phi}\psi}}\cap\Delta_{\phi}}+\Prob{}{\eta}{\Cyl{K_{\overline{\phi}\psi}}\cap\Delta_{\lnot\phi}}}\\
&\leq\frac{\Prob{}{\eta^\star}{\Cyl{K_{\overline{\psi}\phi}}\cap\Delta_{\psi}}
+\Prob{}{\eta^\star}{\Cyl{K_{\overline{\phi}\psi}}\cap\Delta_{\phi}}}
{\Prob{}{\eta^\star}{\Cyl{K_{\overline{\psi}\phi}}\cap\Delta_{\psi}}+\Prob{}{\eta^\star}{\Cyl{K_{\overline{\phi}\psi}}\cap\Delta_{\phi}}+\Prob{}{\eta^\star}{\Cyl{K_{\overline{\phi}\psi}}\cap\Delta_{\lnot\phi}}}\\
&= \CP{}{\eta^\star}{\Fin\phi}{\Fin\psi}\qedhere
\end{align*} 
\end{proof}

\noindent{\bf Proof of Theorem \ref{thm:shiSch}.~} It follows
straightforwardly from Corollary \ref{cor:HIbefore} and Lemma
\ref{lem:HIafter}. $\qed$

\subsection{Deterministic Schedulers}

We now proceed to show that deterministic schedulers suffice to
attain optimal conditional probabilities.

The following result states that taking the convex combination of
schedulers does not increase the conditional probability
$\CP{}{}{\phi}{\psi}$.

\begin{lem}
  Let $\Pi$ be a $\MDP$, $s$ a state, and $\phi, \psi$ path formulas.
  Suppose that the $s$-scheduler $\eta$ is a convex combination of $\eta_1$
  and $\eta_2$. Then $\CP{}{s,\eta}{\phi}{\psi} \leq \max
  (\CP{}{s,\eta_1}{\phi}{\psi}, \allowbreak \CP{}{s,\eta_2}{\phi}{\psi})$.
\end{lem}

\begin{proof}
To prove this lemma we need to use the following technical result:
The function $f \colon \mathbb{R} \to \mathbb{R}$ defined as below
is monotonous.

\[f(x)\eqdef \frac{x v_1 + (1-x) v_2}{x w_1 + (1-x)
w_2}\] \noindent where $v_1, v_2 \in [0,\infty)$ and $w_1, w_2 \in
(0,\infty)$. This claim follows from the fact that $f'(x) =
\frac{v_1 w_2 - v_2 w_1}{(x w_1 - (1-x) w_2)^2}$ is always $\geq
0$ or always $\leq 0$.

Now, by applying the result above to
  \begin{align*}
   [0,1] \ni \alpha \mapsto
        \frac{\alpha \Prob{}{s,\eta_1}{\phi\land\psi} +
         (1-\alpha) \Prob{}{s,\eta_2}{\phi\land\psi}}
        {\alpha \Prob{}{s,\eta_1}{\psi} + (1-\alpha) \Prob{}{s,\eta_2}{\psi}}
 \end{align*}
 we get that the maximum is reached at $\alpha = 0$ or $\alpha = 1$.
 Because $\eta$ is a convex combination of $\eta_1$ and $\eta_2$,
 $\CP{}{s,\eta}{\phi}{\psi} \leq \CP{}{s,\eta_2}{\phi}{\psi}$ (in the
 first case) or $\CP{}{s,\eta}{\phi}{\psi} \leq
 \CP{}{s,\eta_1}{\phi}{\psi}$ (in the second case).
\end{proof}

\begin{lem}
  Let $\Pi$ be a MDP, $s$ a state, and $\phi$ a path formula. Then
  every $\varphi$-sHI $s$-scheduler on $\Pi$ is a convex combination of deterministic
  $\varphi$-sHI $s$-schedulers.
\end{lem}
\begin{proof}
\rev{The result follows from the fact that sHI schedulers have only finitely many choices to make at each state (at most two) and every choice at a particular state -- either before or after the stopping condition-- is a convex combination of deterministic choices at that state -- either before or after the stopping condition. }
\end{proof}

Finally, combining Theorem \ref{thm:shiSch} and the
previous lemma we obtain:

\begin{thm}
  \label{theor:niceSchedulers}
  Let $\Pi$ be a $\MDP$, $\phi,\psi\in\Path$, and $\varphi=\StopCond{\phi}\lor\StopCond{\psi}$. Then we have
\[\CP{+}{}{\phi}{\psi}=\sup_{\eta\in\text{Sch}^\varphi_d(\Pi)}\CP{}{\eta}{\phi}{\psi},\]
\noindent where $\text{Sch}^\varphi_d(\Pi)$ is the set of
deterministic and $\varphi$-$\nHI$ schedulers of $\Pi$.
\end{thm}

Since the number of deterministic and semi HI schedulers is finite
we know that there exists a scheduler attaining the optimal
conditional probability, i.e.
$\sup_{\eta\in\text{Sch}^\varphi_d(\Pi)}\CP{}{\eta}{\phi}{\psi}=\max_{\eta\in
\text{Sch}^\varphi_d(\Pi)} \CP{}{\eta}{\phi}{\psi}$. Note that
this implies that \cpCTL is decidable.

We conclude this section showing that there exists a deterministic
and semi HI scheduler maximizing the conditional probabilities of
Example \ref{ex:notHI}.

\begin{exa}
  Consider the $\MDP$ and $\cpCTL$ formula of Example~\ref{ex:notHI}.
  According to Theorem~\ref{theor:niceSchedulers} there
  exists a deterministic and $(B\lor P)$-$\nHI$ scheduler that
  maximizes $\CP{}{s_0,\eta}{\Fin B}{\Fin P}$.
  In this case, a maximizing scheduler will take always the
  same decision ($\pi_3$) before the system reaches $s_3$ (a state
  satisfying the until stopping condition $(B\lor P)$) and always the
  same decision ($\pi_1$) after the system reaches $s_3$.
\end{exa}

\section{Model Checking \cpCTL\index{model checking!cpCTL}}
\label{modelchecking}

Model checking \cpCTL means checking if a state $s$ satisfies a
certain state formula $\phi$. We focus on formulas of the form
$\CP{}{\leq a}{\phi}{\psi}$ and show how to compute
$\CP{+}{s}{\phi}{\psi}$. The case $\CP{-}{s}{\phi}{\psi}$ is
similar.

Recall that model checking \pCTL is based on the
Bellman-equations. For instance, $\Prob{+}{s}{\Fin B} \allowbreak
= \allowbreak \max_{\pi
  \in \tau(s)} \sum_{t \in \suc(s)} \pi(t) \cdot \Prob{+}{t}{\Fin B}$ whenever
$s \not\models B$.  So a scheduler $\eta$ that maximizes
$\smash{\Prob{}{s}{\Fin B}}$ chooses $\pi \in \tau(s)$ maximizing
$\smash{\sum_{t \in \suc(s)}} \pi(t) \cdot \allowbreak
  \Prob{+}{t}{\Fin B}$. In a successor state $t$, $\eta$ still behaves
as a scheduler that maximizes $\smash{\Prob{}{t}{\Fin B}}$.
As shown below, such a local Bellman-equation is not true for
conditional probabilities: a scheduler that maximizes a
conditional probability such as $\CP{}{s}{\Fin B}{\Glo P}$ does
not necessarily maximize $\CP{}{t}{\Fin B}{\Glo P}$ for successors
$t$ of $s$.

\begin{exa}\label{ex:decision}
  Consider the $\MDP$ and $\cpCTL$ formula $\CP{}{\leq a}{\Fin B}{\Glo P}$ of Figure~\ref{fig:notMaxNod}.  There are only two
  deterministic schedulers. The first one, $\eta_1$, chooses $\pi_2$
  when the system reaches the state $s_2$ and the second one,
  $\eta_2$, chooses $\pi_3$ when the system reaches $s_2$. For the
  first one $\CP{}{s_0,\eta_1}{\Fin B}{\Glo P}=1-\frac{2\alpha}{7}$, and
  for the second one $\CP{}{s_0,\eta_2}{\Fin B}{\Glo P}=\frac{30}{31}$.
  So $\CP{+}{s_0}{\Fin B}{\Glo P} = \max(1 - \frac{2\alpha}{7},
  \frac{30}{31})$.  Therefore, if $\alpha \geq \frac{7}{62}$ the
  scheduler that maximizes $\CP{}{s_0}{\Fin B}{\Glo P}$ is $\eta_2$
  ($\CP{}{s_0,\eta_2}{\Fin B}{\Glo P} = \CP{+}{s_0}{\Fin B}{\Glo P}$) and
  otherwise it is $\eta_1$ ($\CP{}{s_0,\eta_1}{\Fin B}{\Glo P} =
  \CP{+}{s_0}{\Fin B}{\Glo P}$).

  Furthermore, $\CP{+}{s_1}{\Fin B}{\Glo P}=1$ and $\CP{+}{s_2}{\Fin B}{\Glo     P}=1-2\alpha$; the scheduler that obtains this last maximum is the
  one that chooses $\pi_2$ in $s_2$.

  Thus, if $\alpha\geq \frac{7}{62}$ the scheduler that maximizes the
  conditional probability from $s_0$ is taking a different decision
  than the one that maximize the conditional probability from $s_2$.
  Furthermore, $\max(1 -
  \frac{2\alpha}{7},\frac{30}{31}) = \CP{+}{s_0}{\Fin B}{\Glo P} \neq
  \frac{3}{4} \CP{+}{s_1}{\Fin B}{\Glo P} + \frac{1}{4} \CP{+}{s_2}{\Fin    B}{\Glo P} = 1 - \frac{\alpha}{2}$ for all $\alpha\in(0,1]$, showing that the
  Bellman-equation from above does not generalize to \cpCTL.
\end{exa}

As consequence of this observation, it is not possible to
``locally maximize'' \cpCTL properties (i.e. to obtain the global
maximum $\CP{+}{s_0}{\phi}{\psi}$ by maximizing
$\CP{}{t}{\phi}{\psi}$ for all states $t$). This has a significant
impact in terms of model-checking complexity: as we will show in
the rest of this section, to verify a \cpCTL property it
is necessary to compute and keep track of several conditional
probabilities and the desired maximum value can only be obtained
after all these probabilities have been collected.

\subsection{Model Checking $\CP{}{\leq a}{\phi}{\psi}$}\label{sec:MCacyclic}

An obvious way to compute $\CP{+}{s}{\phi}{\psi}$ is by computing
the pairs ($\Prob{}{s,\eta}{\phi\land\psi}$,
$\Prob{}{s,\eta}{\psi}$) for all deterministic $\nHI$ schedulers $\eta$, and
then taking the maximum quotient
$\Prob{}{s,\eta}{\phi\land\psi}\allowbreak / \allowbreak
\Prob{}{s,\eta}{\psi}$. This follows from the fact that there
exist finitely many deterministic semi history-independent schedulers and
that one of them attains the maximal conditional probability; however,
the number of such schedulers grows exponentially in the size of
the \MDP so computing these pairs for all of them is
computationally expensive. Our plan is to first present the
necessary techniques to naively compute
($\Prob{}{s,\eta}{\phi\land\psi}$, $\Prob{}{s,\eta}{\psi}$) for
all deterministic $\nHI$ schedulers $\eta$ and then present an algorithm that
allows model checking $\CP{}{\leq a}{\phi}{\psi}$ without
collecting such pairs for all $\nHI$ scheduler. 

\subsubsection{1) A naive approach to compute $\CP{+}{}{\phi}{\psi}$}


The algorithm is going to keep track of a list of pairs of probabilities
of the form ($\Prob{}{t,\eta}{\phi\land\psi}$, $\Prob{}{t,\eta}{\psi}$)
for all states $t$ and $\eta$ a deterministic $\nHI$ scheduler.
We start by defining a data
structure to keep track of the these pairs of probabilities.

\begin{dfn}\label{def:expresssion}
  Let $L$ be the set of expressions of the form
  $(p_1,q_1)\lor\cdots\lor(p_n,q_n)$ where $p_i,q_i\in [0,\infty)$ and
  $q_i\geq p_i$, for all $n\in \mathbb{N}^\star$.  On $L$ we consider
  the smallest congruence relation $\equiv_1$ satisfying
  idempotence, commutativity, and associativity, i.e.:
\[\begin{array}{rcl}
   (p_1,q_1)\lor(p_1,q_1)&\equiv_1&(p_1,q_1)\\
   (p_1,q_1)\lor(p_2,q_2)&\equiv_1&(p_2,q_2)\lor(p_1,q_1)\\
   ((p_1,q_1)\lor(p_2,q_2))\lor(p_3,q_3)&\equiv_1&(p_1,q_1)\lor((p_2,q_2)\lor(p_3,q_3))\\
\end{array}\]
\noindent Note that
  $(p_1,q_1) \lor \dots \lor (p_n,q_n) \equiv_1 (p'_1,q'_1) \dots
  (p'_{n'},q'_{n'})$ if and only if $\{(p_1,q_1), \allowbreak\dots, \allowbreak(p_n,q_n)\} =
  \{(p'_1,q'_1),\dots,(p'_{n'},q'_{n'})\}$.

  We let $L_1$ be the set of equivalence classes and denote the
  projection map $L \to L_1$ that maps each expression to its
  equivalence class by $f_1$.
  On $L$ we also define \emph{maximum quotient} $\rmax:L\to
  [0,\infty)$ by

  \[\rmax \left(\bigvee_{i=1}^n (p_i,q_i)\right) \eqdef\max\left(\left\{\frac{p_i}{q_i}|q_i\not=0, i=1,\ldots,n\right\}\cup\{0\}\right)\]

  Note that $\rmax$ induces a map $\rmax_1 \colon L_1 \to
  [0,\infty)$ making the diagram in Figure~\ref{fig:diagramas-delta} (a) commute, i.e.,
  such that $\rmax_1 \circ f_1 = \rmax$.
\end{dfn}
\begin{dfn}\label{def:delta}
  Let $\Pi$ be a $\MDP$. We define the function $\delta:
  S\times\Stat\times\Path\times\allowbreak\Path\to L$ by
  \[\delta(s,\varphi,\phi,\psi) \eqdef \bigvee_{\eta\in
    \Sch_s^\varphi(\Pi)}\left(\Prob{}{s,\eta}{\phi\land
      \psi},\Prob{}{s,\eta}{\psi}\right)\] and we define $\delta_1
  \colon S \times \Stat \times \Path \times \Path \to L_1$ by
  $\delta_1 \eqdef f_1 \circ \delta$.
\end{dfn}
%
\begin{figure}[h]
     \centering
     \subfigure[Commutative diagram]{\qquad\qquad\qquad\includegraphics[width=3.5cm]{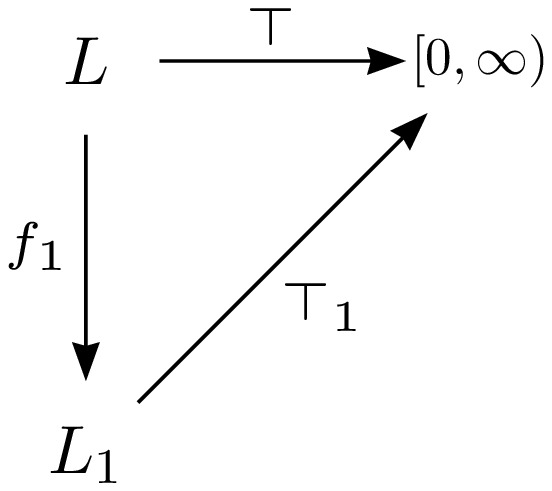}
                                     \qquad}
     \subfigure[$\delta$-values] {\includegraphics[width=4cm]{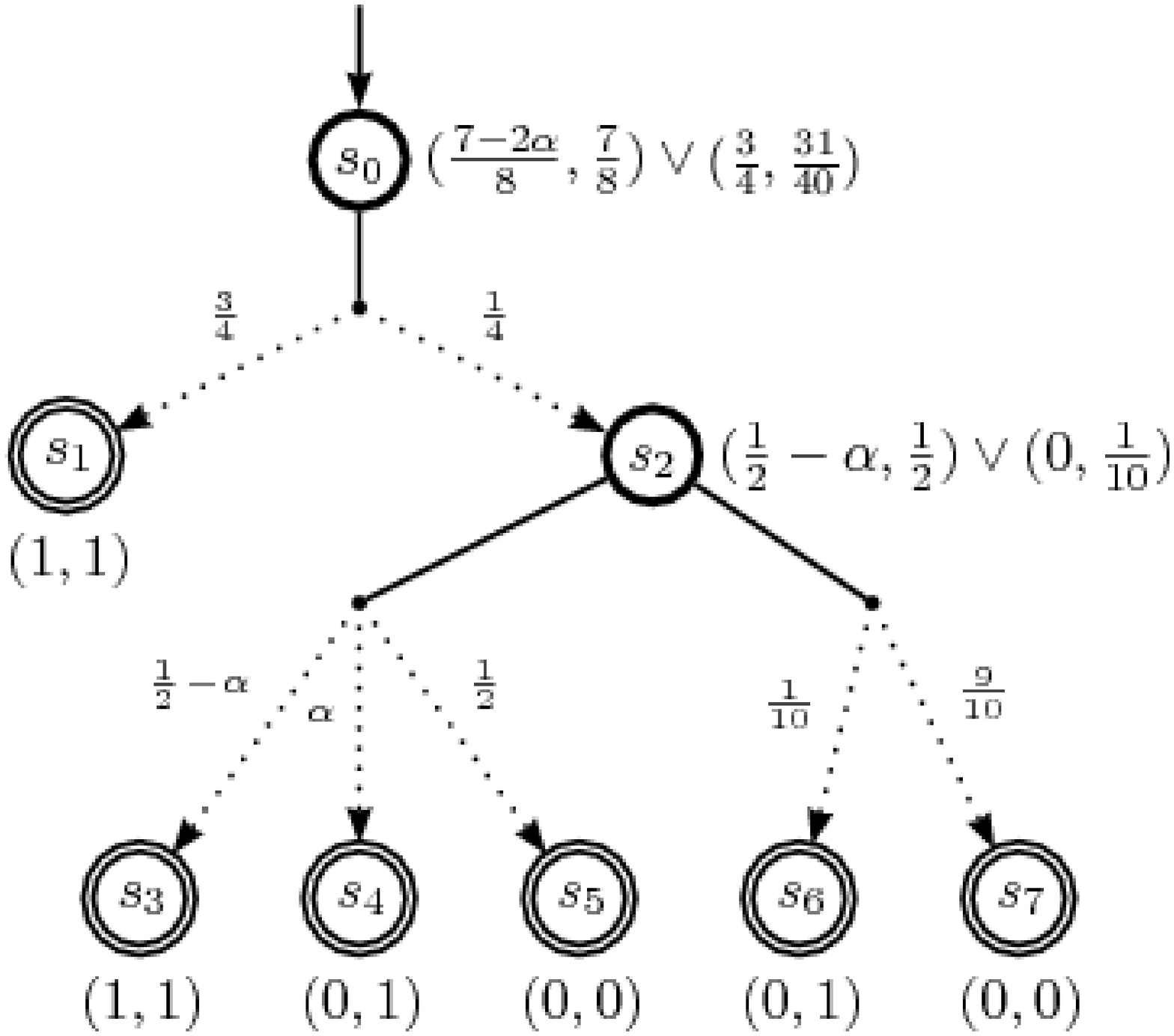}}
\label{fig:diagramas-delta}
\end{figure}

\noindent When no confusion arises, we omit the subscripts 1 and
omit the projection map $f_1$, writing
$(p_1,q_1)\lor\cdots\lor(p_n,q_n)$ for the equivalence class it
generates.

\medskip
\exampleheader In Figure \ref{fig:diagramas-delta} we show the
value $\delta(s,B\lor\lnot P,\Fin B, \Glo P)$ associated to each
state $s$ of the $\MDP$ in Figure~\ref{fig:notMaxNod}.
\medskip

\noindent The following lemma states that it is possible to obtain
maximum conditional probabilities using $\delta$.

\begin{lem} \label{theor:correctness} Given $\Pi=(S,s_0,L,\tau)$ an
  acyclic $\MDP$, and $\phi_1,\phi_2,\psi_1,\psi_2\in\Stat$. Then
\[\CP{+}{s}{\Until{\phi_1}{\phi_2}}{\Until{\psi_1}{\psi_2}}=\rmax\left(\delta^{\Until{}{}}_{s}(\phi_1\U\phi_2\mid\psi_1\U\psi_2)\right)\]
and
\[\CP{+}{s}{\Until{\phi_1}{\phi_2}}{\Glo \psi_1}=\rmax\left(\delta^\Glo_s(\phi_1\U\phi_2\mid\Glo \psi_1)\right),\]

\noindent where
$\delta^{\Until{}{}}_{s}(\phi_1\U\phi_2\mid\psi_1\U\psi_2)\eqdef\delta(s,\StopCond{\phi_1\U\phi_2}\lor\allowbreak\StopCond{\psi_1\U\psi_2},\allowbreak\Until{\phi_1}{\phi_2},\allowbreak\Until{\psi_1}{\psi_2})$
and
$\delta^\Glo_s(\phi_1\U\phi_2\mid\Glo\psi_1)\eqdef\delta(s,\StopCond{\phi_1\U\phi_2}\lor\StopCond{\G\psi_1},\Until{\phi_1}{\phi_2},\allowbreak\Glo
\psi_1)$.
\end{lem}

\begin{proof} The lemma follows straightforwardly from the definitions of $\delta$
and $\top$ and the fact that the maximum conditional probability is indeed
reached by a deterministic $\nHI$ scheduler.
\end{proof}

 \noindent
 Remember that there are finitely many $\nHI$ schedulers. Thus, $\delta$ (and
 therefore $\CP{+}{}{-}{-}$) can in principle be computed by explicitly
 listing them all. However, this is of course an inefficient way to
 compute maximum conditional probabilities.

 We now show how to compute $\CP{+}{}{-}{-}$ in a more efficient
 way. We will first provide an algorithm to compute maximum conditional
probabilities for acyclic \MDPs. We then show how to apply this
algorithm to \MDPs with cycles by mean of a technique, based on
SCC analysis, that allows the transformation of an \MDP with
cycles to an equivalent acyclic \MDP.

\subsubsection{2) An algorithm to compute $\CP{+}{}{\phi}{\psi}$ for Acyclic $\MDPs$}

We will now present a recursive algorithm to compute
$\CP{+}{}{\phi}{\psi}$ for acyclic $\MDPs$ using a variant of
$\delta$ (changing its image). As we mentioned before, to compute
maximum conditional probabilities it is not necessary to consider
all the pairs $(\Prob{}{\eta}{\phi\land
\psi},\Prob{}{\eta}{\psi})$ (with $\eta$ a deterministic and semi
HI scheduler). \rev{In particular, we will show that it is sufficient
to consider only deterministic and semi HI schedulers (see definition of $D$ below) that
behave as an optimizing scheduler (i.e. either maximizing or minimizing a $\pCTL$ formula $\phi$)
after reaching the stopping condition (i.e. a state $s$ satisfying $\StopCond{\varphi}$)}.


We plan to compute a function $\hat{\delta}(-)\subseteq \delta(-)$
such that $\rmax(\delta)=\rmax(\hat{\delta})$. Intuitively,
$\hat{\delta}(-)$ can be thought as
 \[\hat{\delta}(s,\varphi,\phi,\psi) = \bigvee_{\eta\in
    D}\left(\Prob{}{s,\eta}{\phi\land
      \psi},\Prob{}{s,\eta}{\psi}\right)\]
\noindent where \rev{$D$ contains all deterministic and semi HI schedulers $\eta$ such that $\eta$
optimizes $\Prob{}{s,\eta}{\phi}$ for some $s\models \StopCond{\varphi}$  and $\phi\in\pCTL$  formula}.

\rev{This intuition will become evident when we present our recursive algorithm to compute conditional probabilities (see Theorem \ref{thm:algorithm} below). The states $s$ involved in the definition of $D$ correspond to the base case of the algorithm and the formula $\phi$ corresponds to the formula that the algorithm maximizes/minimizes when such $s$ is reached.}

%
%
We will present algorithms to recursively (in $s$) compute
$\hat{\delta}^{\U}_{s}$ and $\hat{\delta}^{\Glo}_{s}$ in acyclic
\MDPs. The base cases of the recursion are the states where the
stopping condition holds. In the recursive case we can express
$\hat{\delta}^\U_{s}$ (respectively $\hat{\delta}^\Glo_{s}$) in
terms of the $\hat{\delta}^\U_{t}$ (respectively
$\hat{\delta}^\Glo_{t}$) of the successors states $t$ of $s$.

We start by formalizing the notion of acyclic $\MDP$. 
\rev{We call a MDP \emph{acyclic} if it contains no cycles other than the trivial ones (i.e., other than selfloops associated to absorbing states).}

\begin{dfn}
\rev{
  A \MDP $\Pi$ is called \emph{acyclic} if for all
  states $s\in S$ and all $\pi \in \tau(s)$ we have $\pi(s)=0$ or
  $\pi(s)=1$, and, furthermore, for all paths $\omega$, if there exist $i, j$ such that $i<j$ and
  $\omega_i=\omega_j$, then we have 
  $\omega_{i}=\omega_{k}$ for all $k>i$.}
\end{dfn}

\noindent In addition, in order to formally define $\hat{\delta}$
we define a new congruence $\equiv_2$.
\begin{dfn}\label{dfn:equiv2} Consider the set of expressions $L$ defined in Definition
  \ref{def:expresssion}.  On $L$ we now consider the smallest
  congruence relation $\equiv_2$ containing $\equiv_1$ and satisfying
\begin{itemize}
\item[(1)] $(p_1,q_1)\lor(p_1,q_  2) \equiv_2
(p_1,\min(q_1,q_2))$, and
\item[(2)]$(p_1,q_1)\lor(p_2,q_1) \equiv_2 (\max(p_1,p_2),q_1)$,
and
\item[(3)]$(p_1+a,q_1+a)\lor(p_1,q_1) \equiv_2 (p_1+a,q_1+a)$,
\end{itemize}
\noindent where $a\in [0,\infty)$.
We write $L_2$ for the set of equivalence classes and denote the projection map $L \to L_2$ by $f_2$.
\end{dfn}

\noindent Since $\equiv_1\subseteq\equiv_2$, this projection maps
factors through $f_1$, say $g \colon L_1\to L_2$ is the unique map
such that $g\circ f_1=f_2.$

\begin{dfn}
  We define $\hat{\delta}: S\times\Stat\times\Path\times\Path\to L_2$ by
  $\hat{\delta} \eqdef f_2 \circ \delta$.
\end{dfn}

Now, in order to prove that $\rmax(\delta)=\rmax(\hat{\delta})$ we
need to define a \emph{scalar multiplication operator} $\rmult$
and an \emph{addition operator} $\radd$ on $L$.
\begin{dfn}
  We define $\rmult:[0,\infty)\times L \to L$ and $\radd:L\times L \to
  L$ by
  \begin{align*}
   c \rmult \bigvee_{i=1}^n (p_i,q_i)&\eqdef\bigvee_{i=1}^n (c\cdot
  p_i,c\cdot q_i)\text{ and }\\
  \bigvee_{i=1}^n (p_i,q_i) \radd \bigvee_{j=1}^m (p^\prime_j,q^\prime_j)&\eqdef\bigvee_{i=1}^n
  \bigvee_{j=1}^m (p_i+p_j^\prime,q_i+q_j^\prime).
   \end{align*}
\end{dfn}
\noindent Note that $\rmult$ and $\radd$ induce maps $\rmult_1
\colon [0,\infty) \times L_1 \to L_1$ and $\radd_1: L_1 \times L_1
\to L_1$ as shown in Figure \ref{fig:diagramas2} below.
As before, we omit the subscript 1 if that will not cause
confusion.

 \begin{figure}[ht]
      \vspace*{-0.2cm}
      \centering
      \subfigure{\includegraphics[width=4cm]{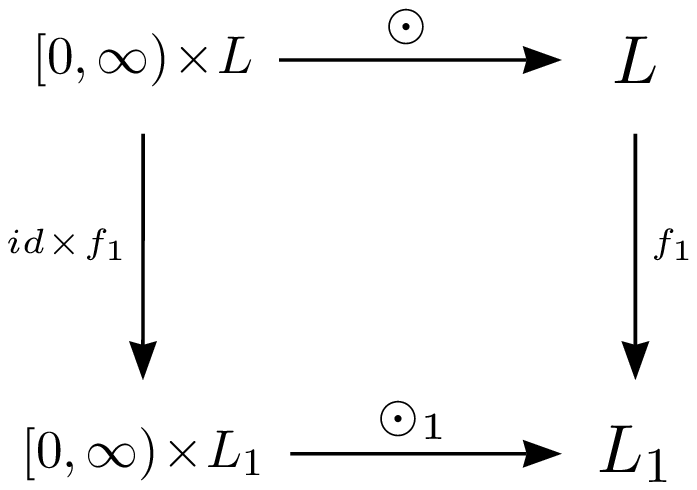}}
      \qquad
      \subfigure{\includegraphics[width=4cm]{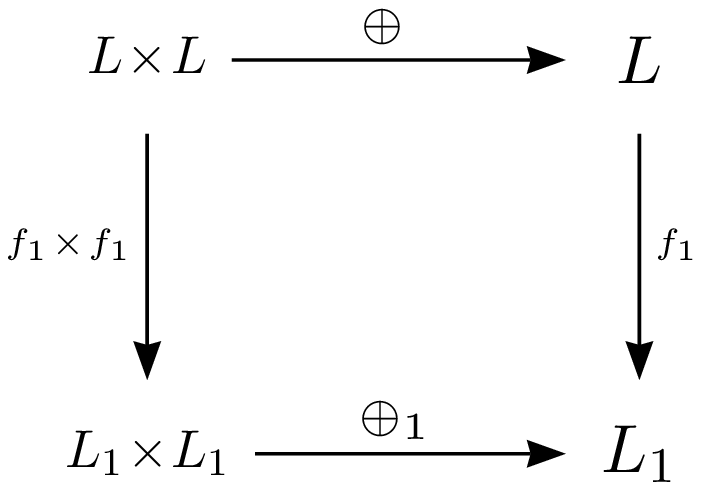}}
 \caption{Commutative diagrams}\label{fig:diagramas2}
 \vspace*{-0.3cm}
 \end{figure}

The following seemingly innocent lemma is readily proven, but it
contains the key to allow us to discard certain pairs of
probabilities. The fact that $\rmax$ induces operations on $L_2$
means that it is correct to ``simplify'' expressions using
$\equiv_2$ when we are interested in the maximum or minimum
quotient.
%

The intuition is as follows. Normally, which decision is best in a
certain state (or rather, at a certain finite path) to optimize
the conditional probability, might depend on probabilities or
choices in a totally different part of the automaton (see
Example~\ref{ex:decision}). Sometimes, however, it is possible to
decide locally what decision the scheduler should take. The
congruence $\equiv_2$ encodes three such cases, each of them
corresponding to one clause in Definition \ref{dfn:equiv2}. (1) If
from a state $t$ the scheduler $\eta$ can either take a transition
after which $\Prob{}{\eta}{\phi \land \psi} = p_1$ and
$\Prob{}{\eta}{\psi} = q_1$ or a transition after which
$\Prob{}{\eta}{\phi \land \psi} = p_1$ and $\Prob{}{\eta}{\psi} =
q_2$, then in order to maximize the conditional probability is
always best to take the decision where $\Prob{}{\eta}{\psi} =
\min(q_1,q_2)$. (2) Similarly, if the scheduler can either take a
transition after which $\Prob{}{\eta}{\phi \land \psi} = p_1$ and
$\Prob{}{\eta}{\psi} = q_1$ or one after which $\Prob{}{\eta}{\phi
\land \psi} = p_2$ and $\Prob{}{\eta}{\psi} = q_1$, then it is
always best to take the decision where $\Prob{}{\eta}{\phi \land
\psi} = \max(p_1,p_2)$. (3) Finally, if $\eta$ has the option to
either take a transition after which $\Prob{}{\eta}{\phi \land
\psi} = p_1 + a$ and $\Prob{}{\eta}{\psi} = q_1 +a$ or one after
which $\Prob{}{\eta}{\phi \land \psi} = p_1$ and
$\Prob{}{\eta}{\psi} = q_1$, for some $a
> 0$, then a maximizing scheduler should always take the first of these two options.

\begin{lem}\label{lem:operators2}
  The operators $\rmult,\ \radd$, and $\rmax$ on $L$ induce
  operators $\rmult_2,\ \radd_2$, and $\rmax_2$ on $L_2$.
\end{lem}

\begin{proof}
The idempotence, commutativity and associativity cases are
trivial; we only treat the other three cases.
\begin{itemize}
\item[($\rmult$\,)]{For (1) we have
\[\begin{array}{ccl}
c\rmult((p,q)\lor(p,q^\prime))&\eqdef&(c\cdot p, c\cdot q)\lor(c\cdot p, c\cdot q^\prime)\\
&\equiv&(c\cdot p, \min(c\cdot q, c\cdot q^\prime)\\
&=&(c\cdot p, c\cdot \min(q, q^\prime))\\
&\eqdef&c\rmult(p,\min(q,q^\prime))\\
\end{array}\]
\noindent Additionally, note that since $q\geq p$ and $q^\prime\geq p$ we
have $\min(q,q^\prime)\geq p$.

For (2) the proof goes like in (1). For (3) we have the following
\[\begin{array}{ccl}
c\rmult((p+a,q+a)\lor(p,q))&\eqdef&(c\cdot p + c\cdot a, c\cdot q+ c\cdot a)\lor (c\cdot p, c\cdot q)\\
&\equiv&(c\cdot p + c\cdot a, c\cdot q+ c\cdot a)\\
&\eqdef&c\rmult(p + a, q + a)\\
\end{array}\]
}
\item[($\radd$\,)] {For (1) we have
\[\begin{array}{ccl}
&&((p,q)\lor(p,q^\prime))\radd \bigvee_{i=1}^n (p_i,q_i)\\
&\eqdef&\bigvee_{i=1}^n (p+p_i,q+q_i)\lor\bigvee_{i=1}^n (p+p_i,q^\prime+q_i)\\
&\equiv&\bigvee_{i=1}^n ((p+p_i,q+q_i)\lor(p+p_i,q^\prime+q_i))\\
&\equiv&\bigvee_{i=1}^n (p+p_i,\min(q+q_i,q^\prime+q_i))\\
&=&\bigvee_{i=1}^n (p+p_i,\min(q,q^\prime)+q_i)\\
&\eqdef&(p,\min(q,q^\prime))\radd\bigvee_{i=1}^n (p_i,q_i)\\
\end{array}\]
For (2) the proof goes like in (1). For (3) we have the following
\[\begin{array}{ccl}
&&((p+a,q+a)\lor(p,q))\radd \bigvee_{i=1}^n (p_i,q_i)\\
&\eqdef&\bigvee_{i=1}^n (p+a+p_i,q+a+q_i)\lor\bigvee_{i=1}^n (p+p_i,q+q_i)\\
&=&\bigvee_{i=1}^n (p+a+p_i,q+a+q_i)\lor(p+p_i,q^\prime+q_i)\\
&\equiv&\bigvee_{i=1}^n (p+a+p_i,q+a+q_i)\\
&\eqdef&(p+a,q+a)\radd\bigvee_{i=1}^n (p_i,+q_i)\\
\end{array}\]
}
\item[($\rmax$\,)]{ For (1) we will start by assuming that
$q,q^\prime\not=0$. Then
\[\begin{array}{ccl}
&&\rmax\left((p,q)\lor(p,q^\prime)\lor\bigvee_{i=1}^n
(p_i,q_i)\right) \\
&\eqdef&\max\left(\{\frac{p}{q}\}\cup\{\frac{p}{q^\prime}\}\cup\{\frac{p_i}{q_i}|\forall_{ 1\leq i\leq n}. q_i\not=0\}\cup\{0\}\right)\\
&=&\max\left(\{\frac{p}{\min(q,q^\prime)}\}\cup\{\frac{p_i}{q_i}|\forall_{ 1\leq i\leq n}. q_i\not=0\}\cup\{0\}\right)\\
&\eqdef&\rmax\left((p, \min(q,q^\prime))\lor \bigvee_{i=1}^n(p_i,q_i)\right)\\
\end{array}\]
\noindent Now assume that $q=0,q^\prime\not=0$ and the case
$q\not=0,q^\prime=0$ is similar. Note that we now have that $p=0$.
Then
\[\begin{array}{ccl}
&&\rmax\left((p,q)\lor(p,q^\prime)\lor\bigvee_{i=1}^n(p_i,q_i)\right)
\\
&\eqdef&\max\left(\{\frac{0}{q^\prime}\}\cup\{\frac{p_i}{q_i}|\forall_{ 1\leq i\leq n}. q_i\not=0\}\cup\{0\}\right)\\
&=&\max\left(\{\frac{p_i}{q_i}|\forall_{ 1\leq i\leq n}. q_i\not=0\}\cup\{0\}\right)\\
&\eqdef&\rmax\left((p,0)\lor\bigvee_{i=1}^n(p_i,q_i)\right)\\
&=&\rmax\left((p,\min(q,q^\prime))\lor\bigvee_{i=1}^n(p_i,q_i)\right)\\
\end{array}\]
\noindent Finally, assume that $q=q^\prime=0$, then also $p=0$, so
\[\begin{array}{ccl}
\rmax\left((p,q)\lor(p,q^\prime)\lor\bigvee_{i=1}^n(p_i,q_i)\right)&\eqdef&\max\left(\{\frac{p_i}{q_i}|\forall_{ 1\leq i\leq n}. q_i\not=0\}\cup\{0\}\right)\\
&\eqdef&\rmax\left((p,0)\lor\bigvee_{i=1}^n(p_i,q_i)\right)\\
&=&\rmax\left((p,\min(q,q^\prime))\lor\bigvee_{i=1}^n(p_i,q_i)\right)\\
\end{array}\]
For (2) the proof goes like in (1). For (3) we first need the
following.

Let $f: \mathbb{R}\to\mathbb{R}$ be a function defined as
$f(x)\eqdef\frac{a+x}{b+x}$ where $a$ and $b$ are constants in the
interval $(0,1]$. Then $f$ is increasing. Let us now assume that
$q\not=0$ or $a\not= 0$. Then
\[\begin{array}{ccl}
&&\rmax\left((p+a,q+a)\lor(p,q))\lor\bigvee_{i=1}^n
(p_i,q_i)\right)\\
&\eqdef&\max\left(\{\frac{p+a}{q+a}\}\cup\{\frac{p}{q}\}\cup\{\frac{p_i}{q_i}|\forall_{ 1\leq i\leq n}. q_i\not=0\}\cup\{0\}\right)\\
&=&\max\left(\{\frac{p+a}{q+a}\}\cup\{\frac{p_i}{q_i}|\forall_{ 1\leq i\leq n}. q_i\not=0\}\cup\{0\}\right)\hspace{0.5cm}\\
& & \{\mbox{By obervation above about } f\}\\
&\eqdef&\rmax\left((p+a,q+a)\bigvee_{i=1}^n(p_i,q_i)\right)\\
\end{array}\]
\noindent Now assume that $q=a=0$. Then
\[\begin{array}{ccl}
&&\rmax\left((p+a,q+a)\lor(p,q))\lor\bigvee_{i=1}^n
(p_i,q_i)\right)\\
&\eqdef&\max\left(\{\frac{p_i}{q_i}|\forall_{ 1\leq i\leq n}. q_i\not=0\}\cup\{0\}\right)\\
&=&\rmax\left((p+a,0)\lor\bigvee_{i=1}^n(p_i,q_i)\right)\\
&\eqdef&\rmax\left((p+a,q+a)\lor\bigvee_{i=1}^n(p_i,q_i)\right)
\end{array}\]}
\end{itemize}\vspace{-0.95cm}\hfill\qedhere
\end{proof}

The fact that $\rmax(\delta)=\rmax(\hat{\delta})$ follows from
previous lemma.

Finally, the following theorem provides recursive equations for
the values of $\hat{\delta}^\U_s$ and $\hat{\delta}^\Glo_s$. If
the \MDPs is acyclic, it can be used to compute these values.

\begin{thm}
  \label{thm:algorithm}
   Let $\Pi$ be a $\MDP$, $s\in S$, and
  $\Until{\phi_1}{\phi_2},\Until{\psi_1}{\psi_2},\Glo \psi_1\in \Path$.
  Then $\hat{\delta}^{\Until{}{}}_s(\phi_1\U\phi_2\mid\psi_1\U\psi_2) =$
  \[\left\lbrace\begin{array}{ll}
    (\Prob{+}{s}{\Until{\psi_1}{\psi_2}},\Prob{+}{s}{\Until{\psi_1}{\psi_2}})& \mbox{if } s\models \phi_2,\\
    (\Prob{+}{s}{\Until{\phi_1}{\phi_2}},1)& \mbox{if } s\models \lnot \phi_2 \!\land\! \psi_2,\\
    (0,\Prob{-}{s}{\Until{\psi_1}{\psi_2}})& \mbox{if } s\models \lnot \phi_1\!\land\! \lnot \phi_2 \!\land\! \lnot \psi_2,\\
    (0,0)& \mbox{if } s\models \phi_1\!\land\! \lnot \phi_2 \!\land\! \lnot \psi_1\!\land\! \lnot \psi_2, \\
    {\displaystyle \bigvee_{\pi\in \tau(s)}}\left(\quad{\displaystyle \probsumD{{t\in \suc(s)}}}\pi(t)\rmult \hat{\delta}^{\Until{}{}}_t(\phi_1\U\phi_2\mid\psi_1\U\psi_2)\right)\!\!\!\!&\mbox{if } s\models \phi_1 \!\land\! \lnot \phi_2 \!\land\! \psi_1\!\land\! \lnot \psi_2,
\end{array}\right.\]
and $\hat{\delta}^{\Glo}_s(\phi_1\U\phi_2\mid\G\psi_1) =$
$$\left\lbrace\begin{array}{ll}
  (\Prob{+}{s}{\Glo \psi_1},\Prob{+}{s}{\Glo \psi_1})& \mbox{if } s\models \phi_2, \\
  (0,0) & \mbox{if } s\models \lnot \phi_2\land \lnot \psi_1, \\
  (0,\Prob{-}{s}{\Glo \psi_1}) & \mbox{if } s\models \lnot \phi_1\land \lnot \phi_2\land \psi_1,\\
   {\displaystyle \bigvee_{\pi\in \tau(s)}}\left(\quad{\displaystyle \probsumD{t\in \suc(s)}}\pi(t)\rmult \hat{\delta}^{\Glo}_t(\phi_1\U\phi_2\mid\G\psi_1) \right) & \mbox{if }  s\models \phi_1\land \lnot \phi_2\land \psi_1.
\end{array}\right.$$
\end{thm}
\begin{proof} We will consider the case $\hat{\delta}^{\Until{}{}}_s$. We will use $\varphi$ to denote
$\phi_1\land\lnot\phi_2\land\psi_1\land\lnot\land\psi_2$, i.e.,
the stopping condition of \cpCTL formula under consideration.
\begin{itemize}
\item[(a)]{ Note that if $s\models \phi_2$, then semi HI schedulers are
exactly the HI schedulers, i.e.,
$\Sch_s^\varphi(\Pi)=\Sch_s^{\HI}(\Pi)$. \vspace{-0.7cm}
\begin{align*}
&\hat{\delta}^{\Until{}{}}_s(\phi_1\U\phi_2\mid\psi_1\U\psi_2)\\
=&\bigvee_{\eta\in
\Sch^{\varphi}_s(\Pi)}\left(\Prob{}{s,\eta}{\Until{\phi_1}{\phi_2}\land
\Until{\psi_1}{\psi_2}},\Prob{}{s,\eta}{\Until{\psi_1}{\psi_2}}\right)\\
&\{s\models\phi_2\}\\
=&\bigvee_{\eta\in \Sch^\varphi_s(\Pi)}\left(\Prob{}{s,\eta}{\Until{\psi_1}{\psi_2}},\Prob{}{s,\eta}{\Until{\psi_1}{\psi_2}}\right)\\
=&(\Prob{+}{s}{\Until{\psi_1}{\psi_2}},\Prob{+}{s}{\Until{\psi_1}{\psi_2}})\hspace{0.5cm}\{\mbox{Case (3)}\}\\
\end{align*}
}
\item[(b)]{\
\vspace{-0.7cm}
\begin{align*}
&\hat{\delta}^{\Until{}{}}_s(\phi_1\U\phi_2\mid\psi_1\U\psi_2)\\
=&\bigvee_{\eta\in
\Sch^{\varphi}_s(\Pi)}\left(\Prob{}{s,\eta}{\Until{\phi_1}{\phi_2}\land
\Until{\psi_1}{\psi_2}},\Prob{}{s,\eta}{\Until{\psi_1}{\psi_2}}\right)\\
&\{s\models\psi_2\}\\
=&\bigvee_{\eta\in \Sch^\varphi_s(\Pi)}\left(\Prob{}{s,\eta}{\Until{\phi_1}{\phi_2}},\Prob{}{s,\eta}{\true}\right)\\
&\{\mbox{Case (2) and definition of }\Prob{}{}{\true} \}\\
=&(\Prob{+}{s}{\Until{\phi_1}{\phi_2}},1)\\
\end{align*}}
\item[(c)]{\
\vspace{-0.7cm}
\begin{align*}
&\hat{\delta}^{\Until{}{}}_s(\phi_1\U\phi_2\mid\psi_1\U\psi_2)\\
=&\bigvee_{\eta\in\Sch^{\varphi}_s(\Pi)}\left(\Prob{}{s,\eta}{\Until{\phi_1}{\phi_2}\land
\Until{\psi_1}{\psi_2}},\Prob{}{s,\eta}{\Until{\psi_1}{\psi_2}}\right)\\
&\{s\models\lnot\phi_1\land\lnot\phi_2\land\lnot\psi_2\}\\
=&\bigvee_{\eta\in \Sch^\varphi_s(\Pi)}\left(\Prob{}{s,\eta}{\false},\Prob{}{s,\eta}{\Until{\psi_1}{\psi_2}}\right)\\
&\{\mbox{Case (1) and definition of }\Prob{}{}{\false} \}\\
=&(0,\Prob{+}{s}{\Until{\psi_1}{\psi_2}})\\
\end{align*}
}

\item[(d)]{\
\vspace{-0.7cm}
\begin{align*}
&\hat{\delta}^{\Until{}{}}_s(\phi_1\U\phi_2\mid\psi_1\U\psi_2)\\
=&\bigvee_{\eta\in
\Sch^{\varphi}_s(\Pi)}\left(\Prob{}{s,\eta}{\Until{\phi_1}{\phi_2}\land
\Until{\psi_1}{\psi_2}},\Prob{}{s,\eta}{\Until{\psi_1}{\psi_2}}\right)\\
&\{\mbox{Since } s\models\lnot\phi_1\land\lnot\phi_2\land\psi_1\land\lnot\psi_2\}\\
=&\bigvee_{\eta\in \Sch^\varphi_s(\Pi)}\left(\Prob{}{s,\eta}{\false},\Prob{}{s,\eta}{\false}\right)= (0,0)\\
%
%
\end{align*}}
\vspace{-0.95cm}
\item[(e)]{\
\vspace{-0.8cm}
\begin{align*}
&\hat{\delta}^{\Until{}{}}_s(\phi_1\U\phi_2\mid\psi_1\U\psi_2)\\
\!\!=\!\!&\bigvee_{\eta\in \Sch^{\varphi}_s(\Pi)}\left(\Prob{}{s,\eta}{\Until{\phi_1}{\phi_2}\land\Until{\psi_1}{\psi_2}},\Prob{}{s,\eta}{\Until{\psi_1}{\psi_2}}\right)\\
\!\!=\!\!&{\displaystyle\bigvee_{\eta\in \Sch^{\varphi}_s(\Pi)}}\left({\displaystyle\sum_{t\in \suc(s)}}\eta(s)(t)\rmult\left(\Prob{}{t,\eta}{\Until{\phi_1}{\phi_2}\land\Until{\psi_1}{\psi_2}},\Prob{}{t,\eta}{\Until{\psi_1}{\psi_2}}\right)\right)\smallskip\\
\!\!=\!\!&\{\mbox{Since }\Pi \mbox{ is acyclic}\}\\
&{\displaystyle\bigvee_{\pi\in \tau(s)}}\!\!\left({\quad\displaystyle\probsumD{t\in \suc(s)}}\!\!\pi(t)\rmult \!\!\!\!\!\!{\displaystyle\bigvee_{\eta_t\in \Sch^\varphi_t(\Pi)}}\!\!\!\!\!\!\!\left(\Prob{}{t,\eta_t}{\Until{\phi_1}{\phi_2}\land\Until{\psi_1}{\psi_2}},\Prob{}{t,\eta_t}{\Until{\psi_1}{\psi_2}}\right)\right)\smallskip\\
\!\!=\!\!&{\displaystyle\bigvee_{\pi\in \tau(s)}}\left(\quad\displaystyle{\probsumD{t\in \suc(s)}}\pi(t)\rmult \hat{\delta}^{\Until{}{}}_t\right)\qedhere\\
\end{align*}
}
\end{itemize}
\end{proof}

\subsubsection{From $\MDPs$ to Acyclic $\MDPs$\index{markov decision process!acyclic}\index{scc analysis}}\label{sec:extension}

Now, we show how to reduce a \MDP with cycles to an acyclic one,
thus generalizing our results to \MDPs with cycles. For that
purpose we first reduce all cycles in $\Pi$ and create a new
acyclic $\MDP$ $\Reduce{\Pi}{}$ such that the probabilities
involved in the computation of $\CP{+}{}{-}{-}$ are preserved.  We
do so by removing every strongly connected component ($\SCC$) $k$
of (the graph of) $\Pi$, keeping only input states and transitions
to output states (in the spirit of \cite{Andres:08:HVC}). We show
that $\CP{+}{}{-}{-}$ on $\Reduce{\Pi}{}$ is equal to the
corresponding value on $\Pi$. For this, we have to make sure that
states satisfying the stopping condition are ignored when removing
$\SCC$s.


\paragraph{(1) Identifying $\SCC$s.}

Our first step is to make states satisfying the stopping condition
absorbing.

\begin{dfn}
  Let $\Pi=(S,s_0, \tau,L)$ be a $\MDP$ and $\varphi\in \Stat$ a state
  formula. We define a new $\MDP$
  $\abs{\Pi}{\varphi}=(S,s_0,\abs{\tau}{\varphi}, L)$ where
  $\abs{\tau}{\varphi}(s)$ is equal to $\tau(s)$ if $s\not \models
  \varphi$ and to $1_s$ otherwise.
\end{dfn}

To recognize cycles in the $\MDP$ we define a graph associated to
it.

\begin{dfn}
  Let $\Pi=(S,s_0,\tau,L)$ be $\MDP$ and $\varphi\in \Stat$.
  We define the digraph $G = G_{\Pi,\varphi} = (S,\to)$ associated to
  $\abs{\Pi}{\varphi}=(S,s_0,\abs{\tau}{\varphi},L)$ where
  $\rightarrow$ satisfies $u\to v \Leftrightarrow
  \exists\pi\in\abs{\tau}{\varphi}(u).\pi(v)>0$.
\end{dfn}


\begin{wrapfigure}{r}{3.5cm}
\vspace{-0.5cm}
\includegraphics[width=3.25cm]{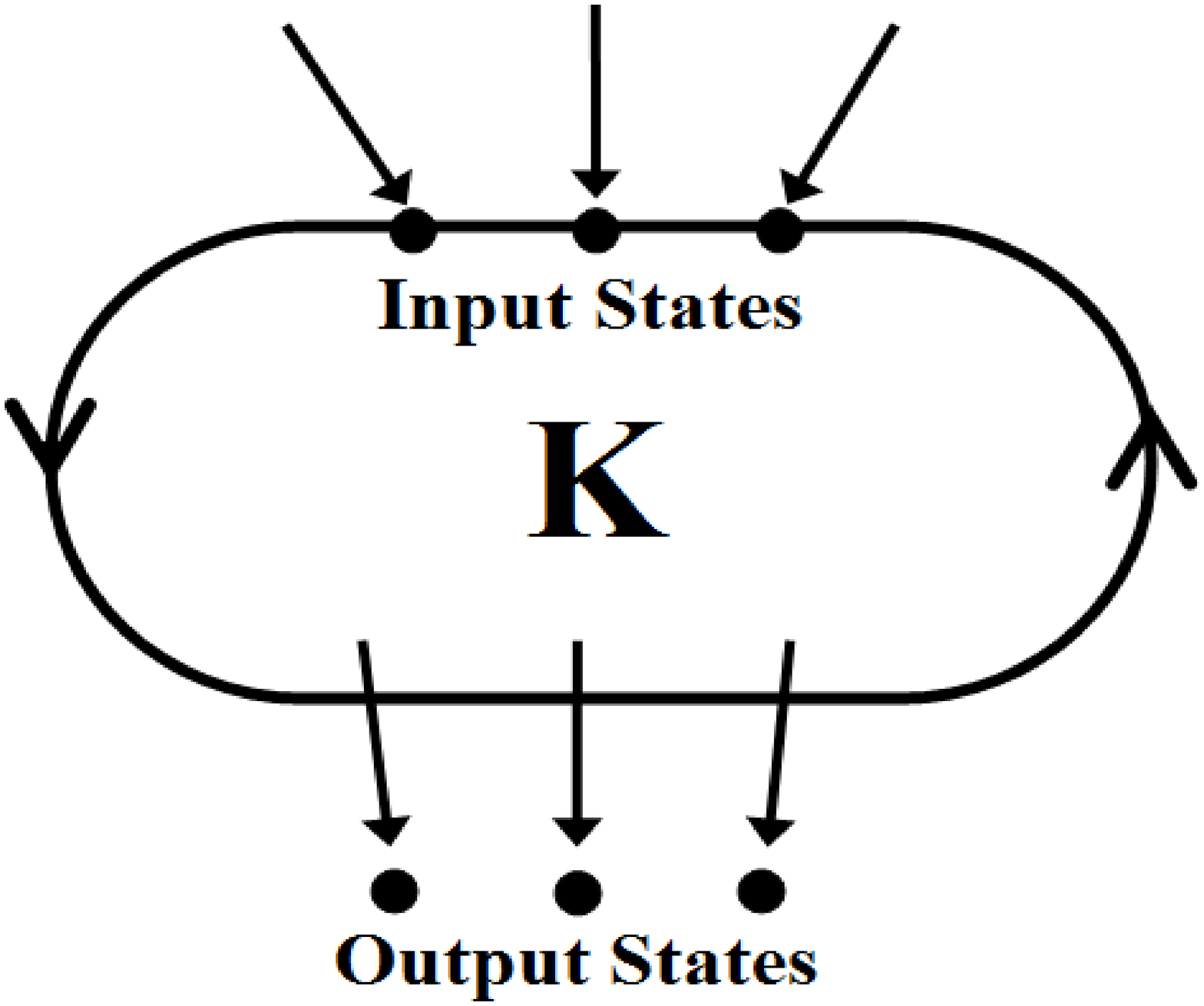}
\vspace{-0.25cm}
\end{wrapfigure}
\noindent Now we let $\SCC = \SCC_{\Pi,\varphi} \subseteq\wp(S)$
be the set of $\SCC$ of $G$. For each $\SCC$ $k$ we define the
sets $Inp_k$ of all states in $k$ that have an incoming transition
of $\Pi$ from a state outside of $k$; we also define the set
$Out_k$ of all states outside of $k$ that have an incoming
transition from a state of $k$.
Formally, for each $k \in \SCC$ we define
\begin{align*}
Inp_{k}&\eqdef\{u\in k\mid\exists\,s\in S\setminus\{k\}\text { such that } 
(s,u)\in\varrho\},\\
Out_{k}&\eqdef\{s\in S\setminus\{k\}\mid\exists\,u\in k \text{ such that } (u,s)\in\varrho\}.
\end{align*}
\noindent where $\varrho$ is the successor relation defined in Section \ref{sec:mdp}.

\noindent We then associate a $\MDP$ $\Pi_k$ to each $\SCC$ $k$ of
$G$. The space of states of $\Pi_k$ is $k\cup Out_k$ and the
transition relation is induced by the transition relation of
$\Pi$.

\begin{dfn} Let $\Pi$ be a $\MDP$ and $k\in \SCC$ be a SCC in $\Pi$.
  We pick an arbitrary element $s_k$ of $Inp_k$ and define the $\MDP$
  $\Pi_k=(S_k,s_k,\tau_k,L)$ where $S_k=\{k\}\cup Out_{k}$ and $\tau_k(s)$
  is equal to $\{1_s\}$ if $s\in Out_k$ and to $\tau(s)$ otherwise.
\end{dfn}

\paragraph{(2) Constructing an acyclic $\MDP$.}

To obtain a reduced acyclic $\MDP$ from the original one we first
define the probability of reaching one state from another
according to a given $\HI$ scheduler in the following way.

\begin{dfn} Let $\Pi=(S,s_0,\tau,L)$ be a $\MDP$, and $\eta$ be a
  $\HI$ scheduler on $\Pi$. Then for each $s,t\in S$ we define the
  function $R$ such that
  $\ReachState{\Pi}{s}{\eta}{t}\triangleq\PP_{s,\eta}(\{\omega\in\Paths{s}\mid\exists\,i.\omega_i=t\})$.
\end{dfn}

\noindent \rev{We note that such reachability values can be efficiently computed  using steady-state analysis techniques~\cite{cassandras_1993_steadystateanalysis}.}

\noindent Now we are able to define an acyclic $\MDP$
$\Reduce{\Pi}{}$ related to $\Pi$ such that
$\CP{+}{\Reduce{\Pi}{}} {-}{-}=\CP{+}{\Pi} {-}{-}$.

\begin{dfn}
  Let $\Pi=(S,s_0,\tau,L)$ be a $\MDP$. Then we define
  $\Reduce{\Pi}{}$ as $(\Reduce{S}{},s_0,\Reduce{\tau}{},L)$ where
\[
    {\Reduce{S}{} = \stackrel{S_{com}}{\overbrace{S\setminus\bigcup_{k\in \SCC}k}}
    \;\cup\; \stackrel{S_{inp}}{\overbrace{\bigcup_{k\in \SCC} Inp_k}}}
\]
  and for all $s\in \Reduce{S}{}$ the set $\Reduce{\tau}{}(s)$ of
  probabilistic distributions on $\Reduce{S}{}$ is given by
\[
    \Reduce{\tau}{}(s) = \left\lbrace
      \begin{array}{ll}
        \tau(s) & \mbox{if } s\in S_{com},\\
        \{\lambda\in\Reduce{S}{}. \ReachState{\Pi_{k_s}}{s}{\eta}{t})\mid\eta\in \Sch^{\HI}_s(\Pi_{k_s})\}&\mbox{if } s\in S_{inp}.\\
      \end{array}
    \right.
\]
  Here $k_s$ is the $\SCC$ associated to $s$.
\end{dfn}


\begin{thm}
  Let $\Pi=(S,s_0,\tau,L)$ be a $\MDP$, and $\CP{}{\leq a}
  {\phi}{\psi} \in \cpCTL$. Then $\Reduce{\Pi}{}$ is an acyclic $\MDP$ and
  $
    \CP{+}{s_0,\Pi}{\phi}{\psi} = \CP{+}{s_0,\Reduce{\Pi}{}}{\phi}{\psi},
  $
  where $\CP{+}{s,\Pi^\prime}{-}{-}$ represents $\CP{+}{s}{-}{-}$ on
  the $\MDP$ $\Pi^\prime$.
\end{thm}

\begin{proof}
The proof follows straightforwardly by the construction of
$\Reduce{\Pi}{}$ and Theorem \ref{theor:niceSchedulers}.\qedhere
\end{proof}

\noindent Finally we can use the technique for acyclic \MDPs on
the reduced $\MDP$ in order to obtain $\CP{+}{s_0}{-}{-}$.
%

\subsection{Complexity}\label{sec:complexity}

As mentioned before, when computing maximum or minimum conditional
probabilities it is not possible to locally optimize. Therefore,
it is necessary to carry on, for each deterministic and HI
scheduler $\eta$, the pair of probabilities
$(\Prob{}{\eta}{\phi\land\psi},\Prob{}{\eta}{\psi})$ from the
leafs (states satisfying the stopping condition) to the initial
state. As the number of HI schedulers in a \MDP grows
exponentially on the state space, our algorithm to verify \cpCTL
formulas has exponential time complexity.

We believe that the complexity of computing optimal conditional probabilities is intrinsically exponential, i.e. computing such probabilities is an NP problem. However, a deeper study on this direction is still missing, 

\paragraph{Conditional probability bounds}

Even if computing exact conditional
probabilities is computationally expensive (exponential time), it
is still possible to efficiently compute upper and lower bounds for such probabilities (polynomial time).

\begin{observation} Let $\Pi$ be a $\MDP$ and
  $\phi, \psi$ two path $\pCTL$ formulas. Then we have
\[\frac{\Prob{-}{}{\phi\land\psi}}{1-\Prob{-}{}{\psi}}\leq\CP{+}{}{\phi}{\psi}\leq \frac{\Prob{+}{}{\phi\land\psi}}{1-\Prob{+}{}\psi}.\]
\end{observation}

\comment{Is this useful to prove strong anonymity????}

\section{Counterexamples for cpCTL\index{counterexample!for cpCTL}}
\label{sec:cpctl_counterexamples}

Counterexamples in model checking provide important diagnostic
information used, among others, for debugging, abstraction-refinement
\cite{clarke_2000_counterexampleguided}, and scheduler synthesis
\cite{LBB01:cav}.
For systems without probability, a counterexample typically consists
of a path violating the property under consideration.
Counterexamples in MCs are sets of paths. E.g, a counterexample
for the formula $\Prob{}{\leq a}{\phi}$ is a set $\Delta$ of
paths, none satisfying $\phi$, and such that the probability mass
of $\Delta$ is greater than
$a$~\cite{hk_2007_counterexamples,Andres:08:HVC,al_2006_search}.

In MDPs, we first have to find the scheduler achieving the optimal
probability. Both for $\pCTL$ and $\cpCTL$, this scheduler can be
derived from the algorithms computing the optimal
probabilities~\cite{Andres:08:HVC}.
Once the optimal scheduler is fixed, the $\MDP$ can be turned into
a Markov Chain \index{markov chain} and the approaches mentioned
before can be used to construct counterexamples for $\pCTL$. For
$\cpCTL$ however, the situation is slightly more complex. It
follows directly from the semantics that:
\begin{equation*}
  s \not\models \CP{}{\leq a}{\phi}{\psi}
  \hspace{0.3cm}\mbox{iff}\hspace{0.3cm}
  \exists \eta\in \Sch_s(\Pi).
  \frac{\PP_{s,\eta}(\{\omega\in\ \Paths{s}|\omega\models\phi\land\psi\})}
  {\PP_{s,\eta}(\{\omega\in\ \Paths{s}|\omega\models\psi\})}
  > a.
\end{equation*}

\begin{lem}\label{lem:counterexample}
  Let $a \in [0,1]$ and consider the formula $\CP{}{\leq
    a}{\phi}{\psi}$.  Let $\Delta_{\phi}\eqdef\{\omega\in \paths\mid
  \omega\models\phi\}$, $\Delta_1\subseteq \Delta_{\phi\land\psi}$,
  and $\Delta_2\subseteq \Delta_{\lnot\psi}$. Then $a <
  {\PP_{\eta}(\Delta_1)}/({1-\PP_{\eta}(\Delta_2)})$ implies $a <
  \CP{}{\eta}{\phi}{\psi}$.
\end{lem}

\begin{proof}
  We first note that
  $$\PP_{\eta}(\Delta_1)\leq\PP_\eta(\Delta_{\phi\land\psi}) \quad \text{    and    } \quad
  \PP_\eta(\Delta_2)\leq\PP_\eta(\Delta_{\lnot\psi}).$$ 
\noindent Then, it is easy to see that 
$$a <
  \frac{\PP_\eta(\Delta_1)}{1-\PP_\eta(\Delta_2)} \leq
  \frac{\PP_\eta(\Delta_{\phi\land\psi})}{1-\PP_\eta(\Delta_{\lnot\psi})}
  = \frac{\PP_\eta(\Delta_{\phi\land\psi})}{\PP_\eta(\Delta_{\psi})} =
  \CP{}{\eta}{\phi}{\psi}.\qedhere$$
\end{proof}

\noindent
This leads to the following notion of counterexample.

\begin{dfn}\label{dfn:counterexample}
  A \emph{counterexample}\index{counterexample!for cpCTL} for $\smash{\CP{}{\leq a}{\phi}{\psi}}$ is a pair
  $(\Delta_1,\Delta_2)$ of measurable sets of paths satisfying
  $\Delta_1 \subseteq \Delta_{\phi \land \psi}$, $\Delta_2 \subseteq
    \Delta_{\lnot\psi}$, and $a < {\PP_\eta(\Delta_1)}/({1 -
      \PP_\eta(\Delta_2)})$, for some scheduler $\eta$.
\end{dfn}

\noindent
Note that such sets $\Delta_1$ and $\Delta_2$ can be computed using
the techniques on Markov Chains mentioned above.

\medskip \exampleheader Consider the evaluation of
$\smash{s_0\models\CP{}{\leq 0.75}{\Fin B}{\Glo P}}$ on the $\MDP$
obtained by taking $\alpha=\frac{1}{10}$ in the MDP depictured in Figure~\ref{fig:notMaxNod}. The corresponding MDP is shown in Figure~2.7(a). In this case the maximizing scheduler, say
$\eta$, chooses $\pi_2$ in $s_2$. In Figure~2.7(b) we show the
Markov Chain \index{markov chain} derived from $\MDP$ using
$\eta$. In this setting we have $\CP{}{s_0, \eta}{\Fin  B}{\Glo
P}=\frac{68}{70}$ and consequently $s_0$ does not satisfy this
formula.

{\def\thesubfigure{\thefigure(\alph{subfigure})}
\begin{figure}[!h]
 \centering
  \subfigure[\MDP]{\label{fig:a}\includegraphics[width=4cm]{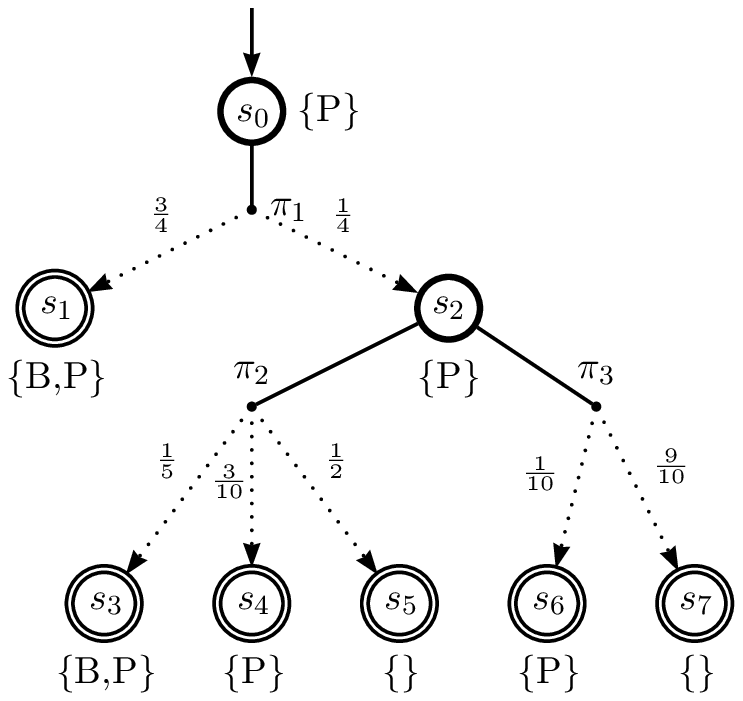}}
  \hspace{1cm} \subfigure[Markov
  Chain]{\label{fig:b}\includegraphics[width=4cm]{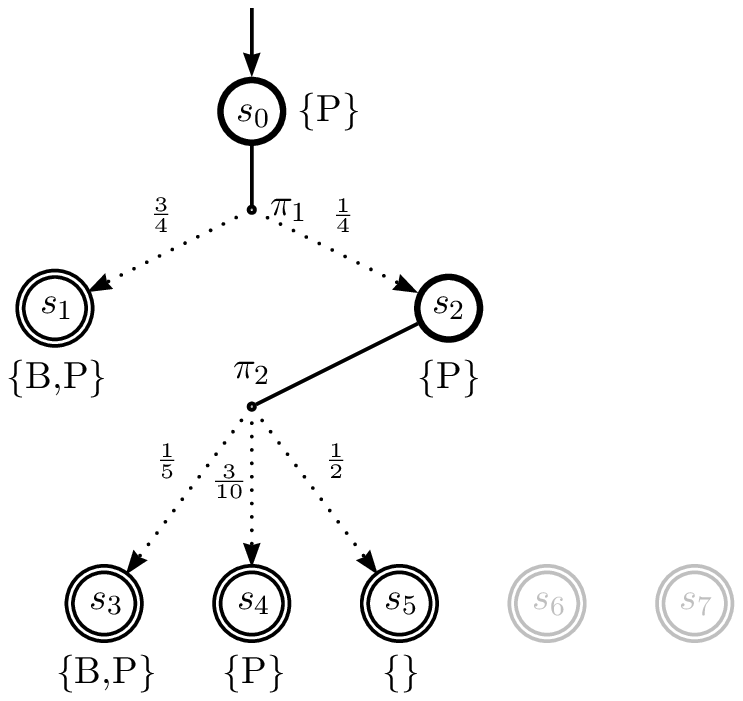}}
  \label{fig:counterexamples}
\end{figure}}

We show this fact with the notion of counterexample of Definition
\ref{dfn:counterexample}. Note that $\Delta_{\Fin B\land \Glo
P}=\Cyl{s_0
  s_1}\cup \Cyl{s_0 s_2 s_3}$ and $\Delta_{\lnot \Glo P}=\Cyl{s_0 s_2
  s_5}$. Using Lemma \ref{lem:counterexample} with $\Delta_1=\Cyl{s_0
  s_1}$ and $\Delta_2=\Cyl{s_0 s_2 s_5}$ we have
$\frac{3}{4}<\frac{\PP_{\eta}(\Delta_1)}{1-\PP_{\eta}(\Delta_2)}=\frac{3/4}{1-1/8}=\frac{6}{7}$.
Consequently $\frac{3}{4}<\CP{}{s_0, \eta}{\Fin B}{\Glo P}$, which
proves that $s_0\not\models\CP{}{\leq 3/4}{\Fin B}{\Glo P}$.

%
%
%
%


\newpage
\thispagestyle{empty}

\chapter{Computing the Leakage of Information Hiding Systems}
\begin{quote}

\textit{In this chapter we address the problem of computing the
information leakage of a system in an efficient way. We propose
two methods: one based on reducing the problem to reachability,
and the other based on techniques from quantitative counterexample
generation. The second approach can be used either for exact or
approximate computation, and provides feedback for debugging.
These methods can be applied also in the case in which the input
distribution is unknown. We then consider the interactive case and
we point out that the definition of associated channel proposed in
literature is not sound. We show however  that the leakage can
still be defined consistently, and  that our methods extend
smoothly. }\end{quote}

\section{Introduction}

By \emph{information hiding}, we refer generally to the problem of
constructing protocols or programs that protect sensitive information
from being deduced by some adversary.
In \emph{anonymity protocols} \cite{Chatzikokolakis:08:IC}, for example,
the concern is to design mechanisms to prevent an observer of network traffic
from deducing who is communicating.
In \emph{secure information flow} \cite{Sabelfeld:03:JSAC}, the concern is to prevent
programs from leaking their secret input to an observer of their
public output. Such leakage could be accidental or malicious.

Recently, there has been particular interest in approaching these issues
\emph{quantitatively}, using concepts of information theory. See for example
\cite{Moskowitz:03:WPES,Clark:05:JLC,Deng:06:FAST,Clarkson:09:JCS,Chatzikokolakis:08:IC}.
The
secret input $S$ and the observable output $O$ of
an information-hiding system  are modeled as random variables
related by a \emph{channel matrix}, whose $(s,o)$ entry specifies $P(o|s)$,
the conditional probability of observing output $o$ given input $s$.
If we define the \emph{vulnerability} of $S$ as the probability that
the adversary could correctly guess the value of $S$ in one try,
then it is natural to measure the information leakage by comparing
the \emph{a priori} vulnerability of $S$ with the \emph{a posteriori}
vulnerability of $S$ after observing   $O$.
We consider two measures of leakage: \emph{additive},
which is the difference between the \emph{a posteriori} and \emph{a priori}
vulnerabilities; and \emph{multiplicative}, which is their quotient
\cite{Smith:09:FOSSACS,Braun:09:MFPS}.

We thus view a protocol or program as a \emph{noisy channel}, and
we calculate the leakage from the channel matrix and the \emph{a
priori} distribution on $S$. But, given an operational
specification of a protocol or program, how do we calculate the
parameters of the noisy channel: the sets of inputs and outputs,
the \emph{a priori} distribution, the channel matrix, and the
associated leakage? These are the main questions we address in
this chapter. We focus on \emph{probabilistic automata}, whose
transitions are labeled with probabilities and \emph{actions},
each of which is classified as secret, observable, or internal.

We first consider the simple case
in which the secret inputs take place
at the beginning of runs, and their probability is fixed.
The interpretation in terms of noisy channel
of this kind of systems is well understood in literature.
The framework of probabilistic automata, however, allows to represent more general situations. Thanks to the nondeterministic choice, indeed,
we can model  the case in which the input distribution is unknown, or variable.
We show that the definition of channel matrix extends smoothly also to this case.
Finally, we turn our attention to the interactive scenario in which inputs can occur again after outputs.
This case has also been  considered in literature, and there has been an attempt
to define the channel matrix in terms of the probabilities of traces \cite{Desharnais:02:LICS}. However
it turns out that the notion of channel is unsound. Fortunately
the leakage is still well defined, and it can be obtained in the same way as the simple case.

We consider two different approaches to computing the channel
matrix. One uses a system of linear equations as in reachability
computations. With this system of equations one can compute the
\emph{joint matrix}, the matrix of probabilities of observing both
$s$ and $o$; the channel matrix is trivially derived from this
joint matrix. The other approach starts with a $0$ channel matrix,
which we call a \emph{partial matrix} at this point. We
iteratively add the contributions in conditional probabilities of
complete paths to this partial matrix, obtaining, in the limit,
the channel matrix itself. We then group paths with the same
secret and the same observable together using ideas from
quantitative counterexample generation\index{counterexample},
namely by using regular expressions and strongly connected
component analysis. In this way, we can add the contribution of
(infinitely) many paths at the same time to the partial matrices.
This second approach also makes it possible to identify which
parts of a protocol contribute most to the leakage, which is
useful for debugging.

Looking ahead, after reviewing some preliminaries (Section~3.2) we
present restrictions on probabilistic automata to ensure that they have
well-defined and finite channel matrices (Section~3.3).
This is followed by the techniques to calculate the channel matrix
efficiently (Section~3.4 and Section~3.5).
We then turn our attention to extensions of our information-hiding system model.
We use nondeterministic choice to model the situation where the
\emph{a priori} distribution on the secret is unknown (Section~3.6).
Finally, we consider interactive systems, in which secret actions and
observable actions can be interleaved arbitrarily (Section~3.7).

\section{Preliminaries}
%

\subsection{Probabilistic automata\index{probabilistic automata}}
This section recalls some basic notions on probabilistic automata.
More details can be found in \cite{Segala:95:PhD}.
A function $\mu \colon Q \to [0,1]$ is a \emph{discrete probability distribution}
on a set~$Q$ if the support of $\mu$ is countable and $\sum_{q \in Q} \mu(q) = 1$.
The set of all discrete probability distributions on $Q$ is denoted by $\mathcal{D}(Q)$.

A \emph{probabilistic automaton} is a quadruple $M = (Q, \Sigma, \qi, \alpha )$ where
$Q$ is a countable set of \emph{states},
$\Sigma$ a finite set of \emph{actions},
$\qi$ the \emph{initial} state,
and $\alpha$ a \emph{transition function} $\alpha: Q \to \wp_f(\distr(\Sigma \times Q))$.
Here $\wp_f(X)$ is the set of all finite subsets of $X$.
If $\alpha(q) = \emptyset$ then $q$ is a \emph{terminal} state. We write $q {\to} \mu$ for $\mu \in \alpha(q), \ q\in Q$.  Moreover, we write $q \smash{\stackrel{a}{\to}} r$ for $q, r \in Q$ whenever $q {\to} \mu$ and $\mu(a,r) > 0$.
A \emph{fully probabilistic automaton} is a probabilistic automaton satisfying $|\alpha(q)|\leq 1$ for all states.
In case $\alpha(q)\not= \emptyset$ we will overload notation and use $\alpha(q)$ to denote the distribution outgoing from $q$.

A \emph{path} in a probabilistic automaton is a sequence $\sigma =
q_0  \stackrel{a_1}{\to} q_1  \stackrel{a_2}{\to} \cdots$ where
$q_i \in Q$, $a_i \in \Sigma$ and $q_i
\smash{\stackrel{a_{i+1}}{\to}} q_{i+1}$. A path can be
\emph{finite} in which case it ends with a state. A path is
\emph{complete} if it is either infinite or finite ending in a
terminal state. Given a path $\sigma$, $\first(\sigma)$ denotes
its first state, and if $\sigma$ is finite then  $\last{\sigma}$
denotes its last state. A \emph{cycle} is a path $\sigma$ such
that $\last{\sigma} = \first(\sigma)$. We denote the set of
actions occurring in a cycle as $\CycleA(M)$. Let $\paths_q(M)$
denote the set of all paths, $\fpaths_q(M)$ the set of all finite
paths, and $\cpaths_q(M)$ the set of all complete paths of an
automaton $M$, starting from the state $q$. We will omit $q$ if
$q=\qi$. Paths are ordered by the prefix relation, which we denote
by $\leq$. The \emph{trace} of a path is the sequence of actions
in $\Sigma^{*} \cup \Sigma^{\infty}$ obtained by removing the
states, hence for the above $\sigma$ we have $\trace(\sigma) =
a_1a_2\ldots$. If $\Sigma'\subseteq \Sigma$, then
$\trace_{\Sigma'}(\sigma)$ is the projection of $\trace(\sigma)$
on the elements of $\Sigma'$. The \emph{length} of a finite path
$\sigma$, denoted by $|\sigma|$, is the number of actions in its
trace.
%

Let $M (Q, \Sigma, \qi, \alpha )$ be a (fully) probabilistic
automaton, $q\in Q$ a state, and let $\sigma \in
\fpaths_{\!\!\!q}(M)$ be a finite path starting in $q$. The
\emph{cone} generated by $\sigma$ is the set of complete paths
$\cone{\sigma} = \{ \sigma^\prime \in \cpaths_q(M) \mid \sigma
\leq \sigma^\prime\}.$ Given a fully probabilistic automaton
$M=(Q, \Sigma, \qi, \alpha)$ and a state $q$, we can calculate the
\emph{ probability value}, denoted by $\PP_q(\sigma)$, of any
finite path $\sigma$ starting in $q$ as follows: $\PP_q(q)   = 1$
and $\PP_q(\sigma \, \stackrel{a}{\to}\, q')   = \PP_q(\sigma)\
\mu(a,q'),   \text{~where~} \last{\sigma} \to \mu$.

Let $\Omega_{q} \eqdef \cpaths_q(M)$ be the sample space, and let $\mathcal F_{q}$ be the smallest $\sigma$-algebra generated by the cones.
Then $\PP$ induces a unique \emph{probability measure} on $\mathcal F_{q}$ (which we will also denote by $\PP_q$) such that $\PP_q(\cone{\sigma}) = \PP_q(\sigma)$ for every finite path $\sigma$ starting in $q$.  For $q=\qi$ we write $\PP$ instead of $\PP_{\qi}$.

Given a probability space $(\Omega, \mathcal{F}, P)$ and two
events $A, B\in F$ with $P(B) > 0$, the \emph{conditional
probability}\index{conditional probability} of $A$ given $B$, $P(A
\mid B)$, is defined as ${P(A \cap B)}/{P(B)}.$

\subsection{Noisy Channels\index{noisy channel}}

This section briefly recalls the notion of
noisy channels from Information Theory \cite{Cover:06:BOOK}.


A \emph{noisy channel}\index{noisy channel} is a tuple $\nc\eqdef(
{\cal X}, {\cal Y}, \cm(\cdot|\cdot))$ where  ${\cal X}=\lbrace
x_1,x_2,\ldots,x_n\rbrace$ is a finite set of \emph{input values},
modeling the \emph{secrets} of the channel, and ${\cal Y} =\lbrace
y_1,y_2,\ldots,y_m\rbrace$ is a finite set of \emph{output
values}, the \emph{observables} of the channel. For $x_i\in {\cal
X}$ and $y_j\in {\cal Y}$, $\cm(y_j|\,x_i)$ is the conditional
probability of obtaining the output $y_j$ given that the input is
$x_i$. These conditional probabilities constitute the so called
\emph{channel matrix}\index{channel matrix}, where $\cm(y_j|x_i)$
is the element at the intersection of the $i$-th row and the
$j$-th column. For any input distribution $P_X$ on $\cal X$, $P_X$
and the channel matrix determine a joint probability $P_\wedge$ on
${\cal X}\times {\cal Y}$, and the corresponding marginal
probability $P_Y$ on $\cal Y$ (and hence a random variable $Y$).
$P_X$ is also called \emph{a priori distribution} and it is often
denoted by $\pi$. The probability of the input given the output is
called \emph{a posteriori  distribution}.

\subsection{Information leakage\index{information leakage}}

\rev{We recall now some notions of information leakage which allow us to quantify the probability of success of a \emph{one-try attacker}, i.e. an attacker that tries to obtain the value of the secret in just one guess. In particular, we consider Smith's definition of \emph{multiplicative leakage}~\cite{Smith:09:FOSSACS}\footnote{The notion proposed by Smith in \cite{Smith:09:FOSSACS} was given in a (equivalent) logarithmic form, and called simply \emph{leakage}. For uniformity's sake we use here the terminology and formulation of \cite{Braun:09:MFPS}.}, and the \emph{additive leakage} definition from Braun et al. \cite{Braun:09:MFPS}.}
%
We assume given a noisy channel $\nc=( {\cal X}, {\cal Y}, \cm(\cdot|\cdot))$ and a random variable $X$ on $\cal X$.
The \emph{a priori vulnerability} of the secrets in $\cal X$ is the probability of guessing the right secret, defined as
$\vul(\emph{X})\eqdef\max_{x\in{\cal X}} P_X(x).$
The rationale behind this definition is that the adversary's best bet is on the secret with highest probability.

The \emph{a posteriori vulnerability} of the secrets in $\cal X$ is the probability of guessing the right secret, after the output has been observed, averaged over the probabilities of the observables. The formal definition is
$\vul(\emph{X}\,|\,\emph{Y})\eqdef\sum_{y\in{\cal Y}}P_Y(y) \max_{x\in{\cal X}} P(x\,|\,y).$
Again, this definition is based on the principle that the adversary will choose the secret with the highest a posteriori probability.

Note that, using Bayes theorem, we can write the a posteriori vulnerability in terms of the channel matrix and the a priori distribution, or in terms of the joint probability:
\begin{eqnarray}\label{eqn:postvulnerability}
 \vul(\emph{X}\, |\, \emph{Y}) \ = \  \sum_{y\in{\cal Y}}\max_{x\in{\cal X}} ( P(y\,|\,x) P_X(x))
\  = \  \sum_{y\in{\cal Y}}\max_{x\in{\cal X}}  P_\wedge(x,y).
\end{eqnarray}

\rev{
The \emph{multiplicative} leakage is then defined as the \emph{quotient} between the a posteriori and a priori vulnerabilities,
$\leakm(\nc,P_X)\eqdef \vul(\emph{X}|\emph{Y})\, / \vul(\emph{X})$. Similarly, the
\emph{additive} leakage is defined as the \emph{difference} between both vulnerabilities,
$\leaka(\nc,P_X)\eqdef\vul(\emph{X}|\emph{Y})-\vul(\emph{X})$.}

\section{Information Hiding Systems}\label{sec:sihs}


To formally analyze the information-hiding properties of protocols and programs, we propose to model them as a particular kind of probabilistic automata,  which we call \emph{Information-Hiding Systems} (IHS).
Intuitively, an IHS is a probabilistic automaton in which the actions are divided in three (disjoint) categories: those which are supposed to remain secret (to an external observer), those which are visible, and those which are internal to the protocol.

First we consider only the case in which the choice of the secret
takes place entirely at the beginning, and is based on a known
distribution.  Furthermore we focus on fully probabilistic
automata. Later in the chapter we will relax these constraints.

\begin{dfn}[Information-Hiding System]\label{def:IHS}
An information-hiding system (IHS) is a quadruple $\ihs=(M,
\Sigma_\secre, \Sigma_\obs, \Sigma_\hidd)$ where
$M=(Q, \Sigma, \qi, \alpha)$ is a fully probabilistic automaton,
$\Sigma=\Sigma_\secre\cup\Sigma_\obs\cup\Sigma_\hidd$ where
$\Sigma_\secre$, $\Sigma_\obs$, and $\Sigma_\hidd$ are pairwise
disjoint sets of secret, observable, and internal actions,
and $\alpha$ satisfies the following restrictions:
\begin{enumerate}
\item $\alpha(\qi)\in \distr(\Sigma_\secre\times Q)$,
\item $\forall s\in \Sigma_\secre \ \exists! q\ .\ \alpha(\qi)(s,q)\not=0,$
\item $\alpha(q)\in \distr(\Sigma_\obs\cup\Sigma_\hidd\times Q)$ for $q\not=\qi$,
\item \rev{$\CycleA(M)\subseteq \Sigma_\hidd$},
\item  $\PP(\cpaths(M)\cap\fpaths(M))=1$.
\end{enumerate}
\end{dfn}

The first two restrictions are on the initial state and mean that only secret actions can happen there ($1$) and each of those actions must have non null probability and occur only once ($2$), Restriction $3$ forbids secret actions to happen in the rest of the automaton, and
Restriction $4$ \rev{specifies that only internal actions can occur inside cycles, this restriction is necessary in order to make sure that the channel associated to the IHS has finitely many inputs and outputs}. Finally, Restriction $5$ means that infinite computations have probability $0$ and therefore we can ignore them.

We now show how to interpret an IHS as a noisy channel.
We call $\trace_{\Sigma_\secre}(\sigma)$ and
$\trace_{\Sigma_\obs}(\sigma)$ the \emph{secret} and
\emph{observable} traces of $\sigma$, respectively. For $s \in
\Sigma_\secre^*$, we define $[s]\eqdef\{\sigma \in \cpaths(M)\mid
\trace_{\Sigma_\secre}(\sigma)=s\}$; similarly for $o \in
\Sigma_\obs^*$, we define $[o]\eqdef\{\sigma \in \cpaths(M)\mid
\trace_{\Sigma_\obs}(\sigma)=o\}$.

\begin{dfn}
Given an IHS $\ihs=(M, \Sigma_\secre, \Sigma_\obs, \Sigma_\hidd)$,
its noisy channel is $(\secre,\obs,\cm)$, where
$\secre\eqdef\Sigma_\secre$,
$\obs\eqdef\trace_{\Sigma_\obs}(\cpaths(M))$, and $\cm(o\mid
s)\eqdef \PP([o] \mid [s])$.
The  a priori  distribution ${\pi }\in\distr(\secre)$ of $\ihs$ is
defined by ${\pi }(s)\eqdef \alpha(\qi)(s,\cdot)$. If $\nc$ is the
noisy channel of $\ihs$, the multiplicative and additive leakage
of $\ihs$ are naturally defined as
$$\leakm(\ihs)\eqdef \leakm({\nc },{\pi })
  \quad\text{and}\quad
  \leaka(\ihs)\eqdef \leaka({\nc },{\pi }).$$
\end{dfn}

\exampleheader
Crowds~\cite{Reiter:98:TISS}\index{protocols!crowds} is a
well-known anonymity protocol, in which a user (called the
\emph{initiator}) wants to send a message to a
\begin{wrapfigure}{r}{4cm}
\centering \vspace{-0.01cm}
\includegraphics[width=4.5cm]{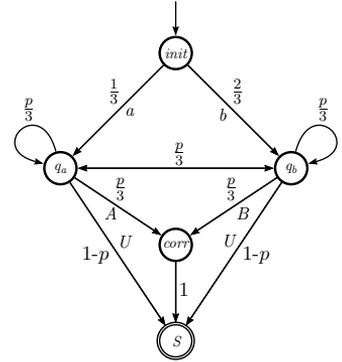}
\caption{Crowds Protocol} \label{fig:crowdsprotocol}
\vspace{-0.35cm}
\end{wrapfigure}
web server without revealing his identity. To achieve this, he
routes the message through a crowd of users participating in the
protocol. Routing is as follows. In the beginning, the initiator
randomly selects a user (called a \emph{forwarder}), possibly
himself, and forwards the request to him.  A forwarder performs a
probabilistic choice.  With probability $p$ (a parameter of the
protocol) he selects a new user and again forwards the message.
With probability $1\!-\!p$ he sends the message directly to the
server. One or more users can be \emph{corrupted} and collaborate
with each other to try to find the identity of the initiator.


We now show how to model Crowds as an $\IHS$ for $2$ honest and $1$ corrupted user.  We assume that the corrupted user  immediately forwards messages to the server, as there is no further information to be gained for him by bouncing the message back.

Figure \ref{fig:crowdsprotocol} shows the automaton\footnote{
For the sake of simplicity, we allow the initiator of the protocol
to send the message to the server also in the first step of the protocol.}.
Actions $a$ and $b$ are secret and represent who initiates the protocol;
actions $A$, $B$, and $U$
are observable; $A$ and $B$ represent who forwards the message to the corrupted user; $U$ represents the fact that the message arrives at the server undetected
by the corrupted user. We assume $U$ to be observable to represent the possibility  that
the message is made publically available at the server's site.

The channel associated to this IHS has $\secre =\{a,b\}$,
$\obs=\{A,B,U\}$, and a priori distribution
$\pi(a)=\frac{1}{3},\pi(b)=\frac{2}{3}$. Its channel matrix is
computed in the next section.


\section{Reachability analysis approach}

This section presents a method
to compute the matrix of joint probabilities $\cm_\lland$ associated to an $\IHS$, defined as
\[\cm_\lland(s,o)\eqdef \PP([s]\cap[o])\mbox{   for all } s\in\secre \mbox{ and }o\in\obs.\]

\comment{joint probability matrix -> matrix of joint
probabilities}
\noindent We omit the subscript $\land$ when no confusion arises.
From $\cm_\lland$ we can derive  the channel matrix by dividing $\cm_\lland(s,o)$ by $\pi(s)$.
The leakage can be computed directly from $\cm_\lland$, using the second form of the a posteriori vulnerability in (\ref{eqn:postvulnerability}).

We write $x_q^\lambda$ for the probability of the set of paths
with trace $\lambda\in(\Sigma_\secre\cup \Sigma_\obs)^\star$
starting from the state $q$  of $M$:
\[x_q^\lambda\eqdef \PP_q([\lambda]_q),\]

\noindent where $[\lambda]_q\eqdef\lbrace
\sigma\in\cpaths_q(M)\mid
\trace_{\Sigma_\secre\cup\Sigma_\obs}(\sigma)=\lambda\rbrace$.
The following key lemma
shows the linear relation between the $x_q^\lambda$'s.
We assume, w.l.o.g., that the $\IHS$ has a unique final state $\qf$.

\begin{lem}\label{lem:SSAeqns} Let $\ihs=(M, \Sigma_\secre, \Sigma_\obs, \Sigma_\hidd)$ be an $\IHS$. For all $\lambda\in(\Sigma_\secre\cup\Sigma_\obs)^\star$ and $q\in Q$ we have
\[
\begin{array}{l}
\smallskip
x_{\qf}^\epsilon  = \  1,\\
\smallskip
x_{\qf}^\lambda   =   \ 0 \quad  \text{for }\lambda\not=\epsilon,\\
\smallskip
x_q^\epsilon \  = \  \sum_{h\in \Sigma_\hidd}\sum_{q^\prime\in \suc(q)}\alpha(q)(h,q^\prime) \cdot x_{q^\prime}^\epsilon\ \quad \text{for } q\not=\qf,\\
x_q^\lambda \  = \  \sum_{q^\prime\in \suc(q)} \alpha(q)(\first(\lambda),q^\prime) \cdot x_{q^\prime}^{\tail{\lambda}}\\
\smallskip
\qquad \ \ \ +
\sum_{h\in \Sigma_\hidd}\alpha(q)(h,q^\prime) \cdot x_{q^\prime}^\lambda
 \qquad \qquad \quad \text{ for }\lambda\not=\epsilon \text{ and }q\not=\qf.\\
\end{array}\]
Furthermore, for $s \in \secre$ and $o \in \obs$ we have
$\PP([s]\cap[o])=x_{\qi}^{s o}$.
\end{lem}

\comment{we dont use  $\cm$ in the Lemma, do we need it?}

Using this lemma, one can compute joint probabilities by solving the system of linear equations in the variables $x_q^\lambda$'s. It is possible that the system has multiple solutions; in that case the required solution is the minimal one.

\begin{exa} Continuing with the Crowds example, we show how to compute joint probabilities. Note that $\qf=\emph{S}$. The linear equations from Lemma~\ref{lem:SSAeqns} are
$$
\begin{array}{lllllllllllllll}
\smallskip
x_{\emph{init}}^{aA}&\!\!=\!\!&\frac{1}{3} \cdot x_{\emph{q}_{a}}^{{A}}, &\!\!\!& x_{\emph{q}_{a}}^{{A}}&\!\!=\!\!&\frac{p}{3} \cdot x_{\emph{q}_{a}}^{{A}} + \frac{p}{3} \cdot x_{\emph{q}_{b}}^{{A}} + \frac{p}{3}\cdot x_{\emph{corr}}^{\epsilon},  & \!\!\! &x_{\emph{corr}}^{{A}}&\!\!=\!\!&x_{\emph{S}}^{{A}}, & \!\!\! & \\
\smallskip
x_{\emph{init}}^{{bA}}&\!\!=\!\!&\frac{2}{3} \cdot x_{\emph{q}_{b}}^{{A}}, &\!\!\!& x_{\emph{q}_{b}}^{{A}}&\!\!=\!\!&\frac{p}{3} \cdot x_{\emph{q}_{a}}^{{A}} + \frac{p}{3} \cdot x_{\emph{q}_{b}}^{{A}} + \frac{p}{3}\cdot x_{\emph{corr}}^{{A}}, && x_{\emph{S}}^{{A}}&\!\!=\!\!&0, \\
\smallskip
x_{\emph{init}}^{{aB}}&\!\!=\!\!&\frac{1}{3} \cdot x_{\emph{q}_{a}}^{{B}}, &\!\!\!& x_{\emph{q}_{a}}^{{B}}&\!\!=\!\!&\frac{p}{3} \cdot x_{\emph{q}_{a}}^{{B}} + \frac{p}{3} \cdot x_{\emph{q}_{b}}^{{B}} + \frac{p}{3}\cdot x_{\emph{corr}}^{{B}},   && x_{\emph{corr}}^{{B}}&\!\!=\!\!&x_{\emph{S}}^{{B}},\\
\smallskip
x_{\emph{init}}^{{bB}}&\!\!=\!\!&\frac{2}{3} \cdot
x_{\emph{q}_{b}}^{{B}}, &\!\!\!&
x_{\emph{q}_{b}}^{{B}}&\!\!=\!\!&\frac{p}{3} \cdot
x_{\emph{q}_\text
{a}}^{{B}} + \frac{p}{3} \cdot x_{\emph{q}_{b}}^{{B}} + \frac{p}{3}\cdot x_{\emph{corr}}^{\epsilon}, && x_{\emph{S}}^{{B}}&\!\!=\!\!&0,\\
\smallskip
x_{\emph{init}}^{{aU}}&\!\!=\!\!&\frac{1}{3} \cdot x_{\emph{q}_{a}}^{{U}}, &\!\!\!& x_{\emph{q}_{a}}^{{U}}&\!\!=\!\!&\frac{p}{3} \cdot x_{\emph{q}_{a}}^{{U}} + \frac{p}{3} \cdot x_{\emph{q}_{b}}^{{U}}+ (1\!-\!p) \cdot x_\emph{S}^{\epsilon},  && x_{\emph{corr}}^{\epsilon}&\!\!=\!\!&x_{\emph{S}}^{\epsilon},  \\
\smallskip
x_{\emph{init}}^{{bU}}&\!\!=\!\!&\frac{2}{3} \cdot
x_{\emph{q}_{b}}^{{U}}, &\!\!\!&
x_{\emph{q}_{b}}^{{U}}&\!\!=\!\!&\frac{p}{3} \cdot
x_{\emph{q}_{a}}^{{U}} + \frac{p}{3} \cdot x_{\emph{q}_b}^{U} +
(1\!-\!p) \cdot x_{\emph{S}}^{\epsilon}, &&
x_{\emph{S}}^{\epsilon}&\!\!=\!\!&1.
\end{array}$$
\end{exa}
%
%
%
%
%

\noindent Let us fix $p=0.9$. By solving the system of linear
equations we obtain

\qquad\begin{wrapfigure}{r}{5cm} \centering \vspace{-0.35cm}
 \begin{tabular}{ c | c  c  c }
    &   \ ${A}$ \ & \ ${B}$ & \ ${U}$\ \\ \hline
    \vspace{-0.45cm}
    \\
    \ ${a}$ \ & $\frac{21}{40}$ & $\frac{9}{40}$ & $\frac{1}{4}$\\
    \vspace{-0.25cm}
    \\
    \ ${b}$ \ & $\frac{9}{40}$ & $\frac{21}{40}$ & $\frac{1}{4}$\\
 \end{tabular}
\vspace{-3.85cm}
\end{wrapfigure}
\[
\begin{array}{l}
\bigskip
\hspace{-4.5cm}x_{\emph{init}}^{{aA}}=\frac{7}{40}, \hspace{0.75cm}  x_{\emph{init}}^{{aB}}=\frac{3}{40}, \hspace{0.75cm} x_{\emph{init}}^{{aU}}=\frac{1}{12},\\
\hspace{-4.5cm}x_{\emph{init}}^{{bA}}=\frac{3}{20},
\hspace{0.75cm} x_{\emph{init}}^{{bB}}=\frac{7}{20},
\hspace{0.8cm} x_{\emph{init}}^{{bU}}=\frac{1}{6}.
\end{array}
\]

\bigskip
\noindent We can now compute the channel matrix by dividing each
$x_{\emph{init}}^{s o}$ by $\pi(s)$. The result is shown in the
figure above.

\subsection{Complexity Analysis}

We now analyze the computational complexity for the computation of
the channel matrix of a simple $\IHS$.  Note that the only
variables (from the system of equations in
Lemma~\ref{lem:SSAeqns}) that are relevant for the computation of
the channel matrix are those $x^\lambda_q$ for which it is
possible to get the trace $\lambda$ starting from state $q$.  As a
rough overestimate, for each state $q$, there are at most
$|\secre| \cdot |\obs|$ $\lambda$'s possible: in the initial state
one can have every secret and every observable, in the other
states no secret is possible and only a suffix of an observable
can occur. This gives at most $|Q| \cdot |\secre| \cdot |\obs|$
variables. Therefore, we can straightforwardly obtain the desired
set of values in $O((|Q|\cdot |\secre| \cdot |\obs|)^3)$ time
(using Gaussian Elimination). Note that using Strassen's methods
the exponent reduces to $2.807$, this consideration applies to
similar results in the rest of the chapter as well.

Because secret actions can happen only at the beginning, the
system of equations has a special form. The variables of the form
$x^{so}_{\qi}$ only depend on variables of the form $x^{o}_q$
(with varying $o$ and $q \neq \qi$) and not on each other. Hence,
we can first solve for all variables of the form $x^{o}_q$ and
then compute the remaining few of the form $x^{so}_{\qi}$.
Required time for the first step is $O((|\obs|\cdot|Q|)^3)$ and
the time for the second step can be ignored.

Finally, in some cases not only do the secret actions happen only
at the beginning of the protocol, but the observable actions
happen only at the end of the protocol, i.e., after taking a
transition with an observable action, the protocol only performs
internal actions (this is, for instance, the case for our model of
Crowds). In this case, one might as well enter a unique terminal
state $\qf$ after an observable action happens. Then the only
relevant variables are of the form $x^{so}_{\qi}$, $x^o_q$, and
$x^\epsilon_{\qf}$; the $x^{so}_{\qi}$ only depends on the
$x^o_q$, the $x^o_q$ only depend on $x^o_{q'}$ (with the same $o$,
but varying $q$'s) and on $x^\epsilon_{\qf}$ and $x^\epsilon_{\qf}
= 1$. Again ignoring the variables $x^{so}_{\qi}$ for complexity
purposes, the system of equations has a block form with $|\obs|$
blocks of (at most) $|Q|$ variables each. Hence the complexity in
this case decreases to $O(|\obs| \cdot |Q|^3)$.

%

\section{The Iterative Approach}\label{sec:IterativeApproach}

We now propose a different approach to compute channel matrices and leakage.  The idea is to iteratively construct the channel matrix of a system by adding probabilities of sets of paths containing paths with the same observable trace $o$ and secret trace $s$ to the $(o|s)$ entry of the matrix.

One reason for this approach is that it allows us to borrow techniques from quantitative counterexample generation. This includes the possibility of using or extending counterexample generation tools to compute channel matrices or leakage. Another reason for this approach is the relationship with debugging. If a (specification of a) system has a high leakage, the iterative approach allows us to determine which parts of the system contribute most to the high leakage, possibly pointing out flaws of the protocol. Finally, if the system under consideration is very large, the iterative approach allows us to only approximate the leakage (by not considering all paths, but only the most relevant ones) under strict guarantees about the accuracy of the approximation. We will focus on the multiplicative leakage; similar results can be obtained for the additive case.

\subsection{Partial matrices}

We start by defining a sequence of matrices converging to the channel matrix by adding the probability of complete paths one by one. We also define partial version of the a posteriori vulnerability and the leakage.
Later, we show how to use techniques from quantitative counterexample generation to add probabilities of many (maybe infinitely many) complete paths all at once.

\begin{dfn}\label{dfn:partialMatrices}
Let $\ihs=(M, \Sigma_\secre, \Sigma_\obs,\Sigma_\hidd)$ be an
$\IHS$, $\pi$ its a priori distribution, and
$\sigma_1,\sigma_2,\ldots$ an enumeration of the set of complete
paths of $M$. We define the \emph{partial matrices}
$\cm^k:\secre\times\obs\rightarrow [0,1]$ as follows
\[\cm^0(o|s)\!\eqdef\! 0, \quad
\cm^{k+1}(o|s) \!\eqdef\! \begin{cases}
                    \cm^k(o|s) + \frac{\PP(\cone{\sigma_{k+1}})}{\pi(s)} & \text {if } \trace_{\Sigma_\obs}(\sigma_{k+1})\!=\!o \\
                    & \text{and }\trace_{\Sigma_\secre}(\sigma_{k+1})\!=\!s,\\
                    \cm^k(o|s) & \text{otherwise}.\\
               \end{cases}
\]

\noindent We define the \emph{partial vulnerability} $\vul^k_{S,O}$ as
$\sum_{o}\max_{s}\cm^k(o|s)\cdot \pi(s)$, and the \emph{partial multiplicative leakage} $\leakm^k(\ihs)$ as
${V^k_{\emph{S,O}}}/{\max_{s}\pi(s)}$.
\end{dfn}


The following lemma states that partial matrices, a posteriori vulnerability, and leakage converge to the correct values.

\begin{lem} Let $\ihs=(M, \Sigma_\secre, \Sigma_\obs,\Sigma_\hidd)$ be an $\IHS$. Then
\begin{enumerate}
\item $\cm^k(o|s)\leq \cm^{k+1}(o|s)$, and $\lim_{k\to\infty}\cm^k(o|s)=\cm(o|s),$
\smallskip
\item $V^k_{\emph{S,O}}\leq V^{k+1}_{\emph{S,O}}$, and $\lim_{k\to\infty}V^k_{\emph{S,O}}=\vul(\emph{S}\ \!|\emph{O})$,
\smallskip
\item $\leakm^{k}(\ihs)\leq\leakm^{k+1}(\ihs)$, and $\lim_{k\to\infty}\leakm^{k}(\ihs)=\leakm(\ihs)$.
\end{enumerate}
\end{lem}

Since rows must sum up to $1$, this technique allow us to compute matrices up to given error $\epsilon$.
We now show how to estimate the error in the approximation of the multiplicative leakage.

\begin{proposition} Let $(M, \Sigma_\secre, \Sigma_\obs,\Sigma_\hidd)$ be an $\IHS$. Then we have
\[\leakm^k(\ihs)\leq \leakm(\ihs)\leq\leakm^k(\ihs)+ \sum_{i=1}^{|\secre|} (1-p^k_i),\]

\noindent where $p^k_i$ denotes the mass probability of the $i$-th row of $\cm^k$, i.e. $p^k_i\eqdef\sum_{o}\cm^k(o|s_i)$.
\end{proposition}

\subsection{On the computation of partial matrices.}

After showing how partial matrices can be used to approximate channel matrices and leakage we now turn our attention to accelerating the convergence. Adding most likely paths first is an obvious way to increase the convergence rate. However, since automata with cycles have infinitely many paths, this (still) gives an infinite amount of path to process. Processing many paths at once (all having the same observable and secret trace) tackles both issues at the same time: it increases the rate of convergence and can deal with infinitely many paths at the same time,

Interestingly enough, these issues also appear in \emph{quantitative counterexample generation}. In that area, several techniques have already been provided to meet the challenges; we show how to apply those techniques in the current context. We consider two techniques: one is to group paths together using regular expressions, the other is to group paths together using strongly connected component analysis.

\paragraph{Regular expressions.}

In~\cite{Daws:04:ICTAC}, regular expressions containing
probability values are used to reason about traces in Markov
Chains. This idea is used in~\cite{Damman:08:QEST} in the context
of counterexample\index{counterexample} generation to group
together paths with the same observable behaviour. The regular
expression there are over pairs $\langle p,q\rangle$ with $p$ a
probability value and $q$ a state, to be able to track both
probabilities and observables. We now use the same idea to group
together paths with the same secret action and the same observable
actions.

We consider regular expressions over triples of the form $\langle
a,p,q\rangle$ with $p \in [0,1]$ a probability value, $a \in
\Sigma$ an action label and $q \in Q$ a state. Regular expressions
represent sets of paths as in~\cite{Damman:08:QEST}. We also take
the probability value of such a regular expression from that
article.

\begin{dfn}\label{dfn:val} The function $\val: {\cal R}(\Sigma)\rightarrow\mathbb{R}$ evaluates regular expressions:
$$
\begin{array}{cclcccl}
\val(\epsilon) &\!\!\eqdef\!\! & 1, &\hspace{0.05cm}& \val(r\concat r^\prime) & \!\!\eqdef\!\! & \val(r) \times  \val(r^\prime),\\
\val(\lre a,p,q \rre) & \!\!\eqdef\!\! & p, &\hspace{0.05cm}&\val(r\kleene)&\!\!\eqdef\!\!&1 \qquad \quad \ \  \mbox{if } \val(r)=1, \\
\val(r\union r^\prime) & \!\!\eqdef\!\! & \val(r) +  \val(r^\prime), && \val(r\kleene)& \!\!\eqdef\!\! &\frac{1}{1-\val(r)} \quad \mbox{if } \val(r)\not=1.\\
\end{array}$$
\end{dfn}

The idea is to obtain regular expressions representing sets of paths of $M$, each regular expression will contribute in the approximation of the channel matrix and leakage. Several algorithms to translate automata into regular expressions have been proposed (see \cite{Neumann:05:TR}). Finally, each term of the regular expression obtained can be processed separately by adding the corresponding probabilities~\cite{Daws:04:ICTAC} to the partial matrix.

\comment{Maybe cite the chapter comparing the different algorithms
to translate automata into regular expressions}

As mentioned before, all paths represented by the regular expression should have the same observable and secret trace in order to be able to add its probability to a single element of the matrix. To ensure that condition we request the regular expression to be normal, i.e., of the form $r_1 + \dots + r_n$ with the $r_i$ containing no $+$'s.

We will now describe this approach by an example.
\begin{exa}\label{exa:regexp} We used JFLAP $7.0$~\cite{JFLAP} to obtain the regular expression $r\eqdef r_1+r_2+\cdots + r_{10}$ equivalent to the automaton in Figure~\ref{fig:crowdsprotocol}.
$$
\begin{array}{lll}
\smallskip
r_1 &\!\!\eqdef\!\! &\lre{b},\frac{2}{3},\emph{q}_{{b}}\rre \cdot \hat{r}^\star \cdot \lre{B},0.3,\emph{corr}\rre \cdot \lre\tau,1,\emph{S}\rre,\\
\smallskip
r_2 &\!\!\eqdef\!\! &\lre{b},\frac{2}{3},\emph{q}_{{b}}\rre \cdot \hat{r}^\star \cdot \lre\tau,0.3,\emph{q}_{{a}}\rre \cdot \lre\tau,0.3,\emph{q}_{{a}}\rre^\star \cdot \lre{A},0.3,\emph{corr}\rre \cdot \lre\tau,1,\emph{S}\rre,\\
\smallskip
r_3&\!\!\eqdef\!\!& \lre{a},\frac{1}{3},\emph{q}_{{a}}\rre \cdot \lre\tau,0.3,\emph{q}_{{a}}\rre^\star \cdot \lre{A},0.3,\emph{corr}\rre \cdot \lre\tau,1,\emph{S}\rre,\\
\smallskip
r_4&\!\!\eqdef\!\!& \lre{b},\frac{2}{3},\emph{q}_{{b}}\rre \cdot \hat{r}^\star \cdot \lre{U},0.1,\emph{S}\rre,\\
r_5 &\!\!\eqdef\!\! &\lre{a},\frac{1}{3},\emph{q}_{{a}}\rre \cdot \lre\tau,0.3,\emph{q}_{{a}}\rre^\star \cdot \lre\tau,0.3,\emph{q}_{{b}}\rre \cdot \hat{r}^\star \cdot \lre{B},0.3,\emph{corr}\rre \cdot \lre\tau,1,\emph{S}\rre,\\
\smallskip
r_6&\!\!\eqdef\!\!& \lre{b},\frac{2}{3},\emph{q}_{{b}}\rre \cdot \hat{r}^\star \cdot \lre\tau,0.3,\emph{q}_{{a}}\rre \cdot \lre\tau,0.3,\emph{q}_{{a}}\rre^\star \cdot\lre{U},0.1,\emph{S}\rre,\\
\smallskip
r_7&\!\!\eqdef\!\!& \lre{a}, \frac{1}{3},\emph{q}_{{a}}\rre \cdot \lre\tau,0.3,\emph{q}_{{a}}\rre^\star \cdot \lre{U},0.1,\emph{S}\rre,\\
\smallskip
r_8&\!\!\eqdef\!\!& \lre{a},\frac{1}{3},\emph{q}_{{a}}\rre \cdot \lre\tau,0.3,\emph{q}_{{a}}\rre^\star \cdot \lre\tau,0.3,\emph{q}_{{b}}\rre \cdot \hat{r}^\star \cdot \lre\tau,0.3,\emph{q}_{{a}}\rre \cdot \lre\tau,0.3,\emph{q}_{{a}}\rre^\star \cdot \\
&&\lre{A},0.3,\emph{corr}\rre \cdot \lre\tau,1,\emph{S}\rre,\\
\smallskip
r_9&\!\!\eqdef\!\!& \lre{a},\frac{1}{3},\emph{q}_{{a}}\rre  \cdot \lre\tau,0.3,\emph{q}_{{a}}\rre^\star \cdot \lre\tau,0.3,\emph{q}_{{b}}\rre \cdot \hat{r}^\star  \cdot \lre{U},0.1,\emph{S}\rre,\\
\smallskip
r_{10}&\!\!\eqdef\!\!& \lre{a},\frac{1}{3},\emph{q}_{{a}}\rre \! \cdot \! \lre\tau,0.3,\emph{q}_{{a}}\rre^\star \! \cdot \! \lre\tau,0.3,\emph{q}_{{b}}\rre \! \cdot \! \hat{r}^\star \! \cdot \! \lre\tau,0.3,\emph{q}_{{a}}\rre  \! \cdot \! \lre\tau,0.3,\emph{q}_{{a}}\rre^\star  \! \cdot \! \lre{U},0.1,\emph{S}\rre,\\
\end{array}
$$

\noindent where $\hat{r}\eqdef (\lre\tau,0.3,\emph{q}_{b}\rre^\star \cdot (\lre\tau,0.3,\emph{q}_{a}\rre \cdot \lre\tau,0.3,\emph{q}_{a}\rre^\star  \cdot \lre\tau,0.3,\emph{q}_{b}\rre)^\star)$. We also note
$$
\begin{array}{rclcrclcrcl}
\smallskip
\val(r_1)&\!\!\!=\!\!\!&\frac{7}{20}\ ({b},{B}), &&\val(r_2)&\!\!\!=\!\!\!&\frac{3}{20}\ ({b},{A}),&&\val(r_3)&\!\!\!=\!\!\!&\frac{1}{7}\ ({a},{A}),\\
\smallskip
\val(r_4)&\!\!\!=\!\!\!&\frac{7}{60}\ ({b},{U})&& \val(r_5)&\!\!\!=\!\!\	!&\frac{3}{40}\ ({a},{B}),&&\val(r_6)&\!\!\!=\!\!\!&\frac{1}{20}\ ({b},{U}),\\
\smallskip
\val(r_7)&\!\!\!=\!\!\!&\frac{1}{21}\
({a},{U}),&&\val(r_8)&\!\!\!=\!\!\!&\frac{9}{280}\ ({a},{A})&&\val(r_9)&\!\!\!=\!\!\!&\frac{1}{40}\ ({a},{U}),\\
\smallskip
\val(r_{10})&\!\!\!=\!\!\!&\frac{3}{280}\ ({a},{U}).
\end{array}$$
\noindent where the symbols between brackets denote the secret and observable traces of each regular expression.

\end{exa}

Now we have all the ingredients needed to define partial matrices using regular expressions.

\begin{dfn}\label{dfn:regexppartialMatrices} Let $\ihs=(M, \Sigma_\secre, \Sigma_\obs,\Sigma_\hidd)$ be an $\IHS$, $\pi$ its a priori distribution, and $r=r_1+r_2+\cdots+r_n$ a regular expression equivalent to $M$ in normal form. We define for $k=0,1,\ldots,n$ the matrices $\cm^k:\rvs\times\rvo\rightarrow [0,1]$ as follows
\vspace{-0.25cm}
\begin{eqnarray*}
\cm^{k}(o|s) &=& \begin{cases}
                    0 & \text {if } k=0,\\
                    \cm^{k-1}(o|s) + \frac{\val(r_k)}{\pi(s)} & \text {if } k\not=0 \text{ and }\trace_{\Sigma_\obs}(r_k)=o \\
                    & \text{and }\trace_{\Sigma_\secre}(r_k)=s,\\
                    \cm^{k-1}(o|s) & \text{otherwise}.\\
               \end{cases}
\end{eqnarray*}
\vspace{-0.25cm}
\end{dfn}
\comment{$r_{k+1}$???}
\noindent Note that in the context of Definition~\ref{dfn:regexppartialMatrices}, we have $P^n = P$.



\paragraph{SCC analysis approach.}
In~\cite{Andres:08:HVC}, paths that only differ in the way they traverse strongly connected components ($\SCC$'s) are grouped together. Note that in our case, such paths have the same secret and observable trace since secret and observable actions cannot occur on cycles. Following~\cite{Andres:08:HVC}, we first abstract away the $\SCC$'s, leaving only probabilistic transitions that go immediately from an entry point of the SCC to an exit point (called input and output states in~\cite{Andres:08:HVC}). This abstraction happens in such a way that the observable behaviour of the automaton does not change.

Instead of going into technical details (which also involves
translating the work~\cite{Andres:08:HVC} from Markov Chains
 to fully probabilistic automata), we describe
the technique by an example.

\begin{wrapfigure}{r}{3cm}
\vspace{-0.01cm} \centering
\includegraphics[width=3.5cm]{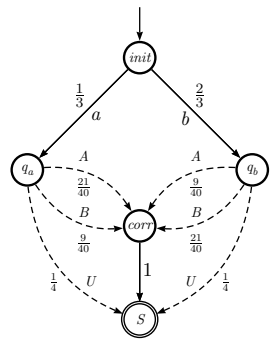}
\caption{Crowds after the $\SCC$ analysis}
\label{fig:CrowdsAcyclic} \vspace{-0.5cm}
\end{wrapfigure}
%



\begin{exa}\label{exa:PMSCC} Figure \ref{fig:CrowdsAcyclic} shows the automaton obtained after abstracting $\SCC$. In the following we show the set of complete paths of the automaton, together with their corresponding probabilities and traces
$$
\begin{array}{lllllll}
\sigma_1&\eqdef&\emph{init}\stackrel{{a}}{\longrightarrow}\emph{q}_{{a}} \stackrel{{A}}{\longrightarrow}\emph{corr} \stackrel{\tau}{\longrightarrow}\emph{S}, & \hspace{0.35cm} & \PP(\sigma_1)=\frac{7}{40}, & \hspace{0.35cm} & ({a},{A}),\\
\sigma_2&\eqdef&\emph{init}\stackrel{{b}}{\longrightarrow}\emph{q}_{{b}} \stackrel{{B}}{\longrightarrow}\emph{corr} \stackrel{\tau}{\longrightarrow}\emph{S}, & & \PP(\sigma_2)=\frac{7}{20}, &  & ({b},{B}),\\
\sigma_3&\eqdef&\emph{init}\stackrel{{a}}{\longrightarrow}\emph{q}_{{a}} \stackrel{{U}}{\longrightarrow}\emph{S}, & & \PP(\sigma_3)=\frac{1}{12}, &  & ({a},{U}),\\
\sigma_4&\eqdef&\emph{init}\stackrel{{b}}{\longrightarrow}\emph{q}_{{b}} \stackrel{{U}}{\longrightarrow}\emph{S}, & & \PP(\sigma_4)=\frac{1}{6}, &  & ({b},{U}),\\
\sigma_5&\eqdef&\emph{init}\stackrel{{a}}{\longrightarrow}\emph{q}_{{a}} \stackrel{{B}}{\longrightarrow}\emph{corr} \stackrel{\tau}{\longrightarrow}\emph{S}, & & \PP(\sigma_5)=\frac{3}{40}, &  & ({a},{B}),\\
\sigma_6&\eqdef&\emph{init}\stackrel{{b}}{\longrightarrow}\emph{q}_{{b}} \stackrel{{A}}{\longrightarrow}\emph{corr} \stackrel{\tau}{\longrightarrow}\emph{S}, & & \PP(\sigma_6)=\frac{3}{20}, &  & ({b},{A}).\\
\end{array}$$

\end{exa}

Note that the $\SCC$ analysis approach groups more paths together (for instance $\sigma_1$ group together the same paths than the regular expressions $r_3$ and $r_8$ in the examples of this section), as a result channel matrix and leakage are obtained faster. On the other hand, regular expressions are more informative providing more precise feedback.

\subsection{Identifying high-leakage sources}

We now describe how to use the techniques presented in this section to identify sources of high leakage of the system.
Remember that the a posteriori vulnerability can be expressed in terms of joint probabilities
\[
\vul(S\mid O) \ = \  \sum_{o} \max_{s} \PP([s]\cap[o]).
\]
%
This suggests that, in case we want to identify parts of the
system generating high leakage, we should look at the sets of
paths $[o_1]\cap[s_1],\ldots, [o_n]\cap[s_n]$ where $\lbrace
o_1,\ldots o_n\rbrace=\obs$ and $s_i\in\arg\left(\max_{s}
\PP([o_i]\cap[s])\right)$. In fact, the multiplicative leakage is
given dividing $\vul(S\mid O)$ by $\vul(S)$, but since $\vul(S)$
is a constant value (i.e., it does not depend on the row) it does
not play a role here. Similarly for the additive case.

The techniques presented in this section allow us to obtain such
sets and, furthermore, to partition them in a convenient way with
the purpose of identifying states/parts of the system that
contribute the most to its high probability. Indeed, this is the
aim of the counterexample generation techniques previously
presented. For further details on how to debug sets of paths and
why these techniques meet that purpose we refer
to~\cite{Aljazzar:08:QEST,Damman:08:QEST,Andres:08:HVC}.

\begin{exa} To illustrate these ideas, consider the path $\sigma_1$ of the previous example; this path has maximum probability for the observable $A$. By inspecting the path we find the transition with high probability $\emph{q}_{a}\stackrel{A}{\to}\emph{corr}$.  This suggests to the debugger that the corrupted user has an excessively high probability of intercepting a message from user $a$ in case he is the initiator.

In case the debugger requires further information on how corrupted users can intercept messages, the regular expression approach provides further/more-detailed information. For instance, we obtain further information by looking at regular expressions $r_3$ and $r_8$ instead of path $\sigma_1$ (in particular it is possible to visualize the different ways the corrupted user can intercept the message of user $a$ when he is the generator of the message).
\end{exa}


\section{Information Hiding Systems with Variable a Priori}

In Section~\ref{sec:sihs} we introduced a notion of $\IHS$ in which the distribution over secrets is fixed. However, when reasoning about security protocols this is often not the case. In general we may assume that an adversary knows the distribution over secrets in each particular instance, but the protocol should not depend on it. In such scenario we want the protocol to be secure, i.e. ensuring low enough leakage, for every possible distribution over secrets. This leads to the definition of maximum leakage.

\begin{dfn} [\cite{Smith:09:FOSSACS,Braun:09:MFPS}] Given a noisy channel $\nc=\allowbreak(\secre,\allowbreak\obs,\allowbreak\cm)$, we define the maximum multiplicative and additive leakage (respectively) as
\[\mleakm(\nc)\eqdef\max_{\pi\in\distr(\secre)}\leakm(\nc,\pi), \qquad \mbox{ and } \qquad \mleaka(\nc)\eqdef\max_{\pi\in\distr(\secre)}\leaka(\nc,\pi).\]
\end{dfn}

\noindent
In order to model this new scenario where the distribution over secrets may change, the selection of the secret is modeled as \emph{nondeterministic choice}. In this way such a distribution remains undefined in the protocol/automaton. We still assume that
the choice of the secret  happens at the beginning, and that we have only one secret per run.
We call such automaton an $\IHS$ \emph{with variable a priori}.


\begin{dfn}
\label{def:IHSvar} An $\IHS$ with variable a priori  is a
quadruple $\ihs=(M, \allowbreak\Sigma_\secre,
\allowbreak\Sigma_\obs, \allowbreak\Sigma_\hidd)$ where
$M=(Q, \Sigma, \qi, \alpha)$ is a probabilistic automaton,
$\Sigma=\Sigma_\secre\cup\Sigma_\obs\cup\Sigma_\hidd$ where
$\Sigma_\secre$, $\Sigma_\obs$, and $\Sigma_\hidd$ are pairwise
disjoint sets of secret, observable, and internal actions,
and $\alpha$ satisfies the following restrictions:
\begin{enumerate}
\item $\alpha(\qi)\subseteq \distr(\Sigma_\secre\times Q)$,
\item $|\alpha(\qi)|=|\secre|\land \forall s\in \Sigma_\secre\ .\ \exists \ q \ .\ \pi(s,q)=1,$ for some $\pi\in\alpha(\qi)$,
\item $\alpha(q)\subseteq \distr(\Sigma_\obs\cup\Sigma_\hidd \times Q)$ and $|\alpha(q)|\leq 1$, for all $q\not=\qi$,
\item $\forall a\in(\Sigma_\secre\cup\Sigma_\obs) \ .\ a\not\in\CycleA(M)$,
\item  $\forall q, s\  \forall \pi\! \in \! \alpha(\qi)\  . \  (\pi(s,q)=1 \Rightarrow \PP(\cpaths_q(M)\cap\paths^*_q(M))=1)$.
\end{enumerate}
\end{dfn}

\comment{restriction $2$ and $5$ are wrong, $\pi$ is a
distribution over $(\Sigma_\secre\times Q)$. I also think the
``unique'' symbol ($!$) is not necessary\\ Additionally, i would
write restriction 2 as $|\alpha(\qi)|=|\secre|\land \forall s\in
\Sigma_\secre\ .\ \exists\ \pi\in\alpha(\qi) .\ \pi(s,q)=1,$ for
some state $q$}

Restrictions  $1$, $2$ and $3$ imply that the secret choice is non deterministic and happens only at the beginning. Additionally, $3$ means that all the other choices are probabilistic. Restriction $4$ ensures that the channel associated to the $\IHS$ has finitely many inputs and outputs.
Finally, $5$ implies that, after we have chosen a secret, every computation terminates except for a set with null probability.



Given an $\IHS$ with variable a priori,  by fixing the a priori distribution we can obtain a standard $\IHS$ in the obvious way:
\begin{dfn} Let $\ihs=((Q,\Sigma,\qi,\alpha), \Sigma_\secre, \Sigma_\obs, \Sigma_\hidd)$ be an $\IHS$ with variable a priori
and $\pi$ a distribution over $\secre$. We define the $\IHS$
associated to $(\ihs,\pi)$ as
$\ihs_{\pi}=((Q,\Sigma,\allowbreak\qi,\alpha^\prime),
\Sigma_\secre, \Sigma_\obs, \Sigma_\hidd)$ with
$\alpha^\prime(q)=\alpha(q)$ for all $q\not=\qi$ and
$\alpha^\prime(\qi)(s,\cdot)=\pi(s)$.
\end{dfn}

\comment{In previous dfn: Define the -> We define the} The
following result says that the conditional probabilities
associated to an $\IHS$ with variable a priori are
\emph{invariant} with respect to the a priori distribution. This
is fundamental in order to interpret the $\IHS$ as a channel.

\begin{proposition} Let $\ihs$ be an $\IHS$ with variable a priori. Then for all $\pi,\pi^\prime\in\distr(\secre)$ such that $\pi(s)\not=0$ and $\pi^\prime(s)\not=0$ for all $s\in\secre$ we have that
$\cm_{\ihs_{\pi}}=\cm_{\ihs_{\pi^\prime}}$.
\end{proposition}
\begin{proof} The secret $s$ appears only once in the tree and only at the beginning of paths,
hence $\PP([s]\cap[o])  =   \alpha^\prime(\qi)(s,\cdot) \, \PP_{\!\! q_s}([o])$ and $\PP([s])=\alpha^\prime(\qi)(s,\cdot)$. Therefore
$\PP([o]\mid [s]) =  \PP_{\!\! q_s}([o]),$
where $q_s$ is the state after performing $s$. While $\alpha^\prime(\qi)(s,\cdot)$ is different in $\ihs_{\pi}$ and $\ihs_{\pi^\prime}$, $\PP_{\!\! q_s}([o])$ is the same, because it only depends on the parts of the paths after the choice of the secret.
\end{proof}

Note that, although in the previous proposition we exclude input distributions with zeros, the concepts of vulnerability and leakage also make sense for these distributions\footnote{We assume that conditional probabilities are extended by continuity on such distributions.}.

This result implies that we can  define the channel matrix of an $\IHS$ $\ihs$ with variable a priori
as the channel matrix of $\ihs_\pi$ for any $\pi$, and we can compute it, or approximate it, using the same techniques of previous sections. 
Similarly we can compute or approximate the leakage for any given $\pi$.

We now turn the attention to the computation of the maximum leakage. The  following result from the literature is crucial for our purposes.
\begin{proposition}[\cite{Braun:09:MFPS}]
Given a channel $\nc$, we have $\arg
\max_{\pi\in\distr(\secre)}\allowbreak\leakm(\nc,\pi)$ is the
uniform distribution, and $\arg
\max_{\pi\in\distr(\secre)}\leaka(\nc,\pi)$ is a \emph{corner
point} distribution, i.e. a distribution $\pi$ such that
$\pi(s)=\frac{1}{\kappa}$ on $\kappa$ elements of $\secre$, and
$\pi(s)= 0$ on all the other elements.
\end{proposition}
As an obvious consequence, we obtain:

\begin{corollary}\label{cor:maxleak}
Given an $\IHS$ $\ihs$ with variable a priori, we have
$\mleakm(\ihs) = \leakm(\ihs_\pi)$, where $\pi$ is the uniform distribution, and
$\mleaka(\ihs) = \leaka(\ihs_{\pi'})$, where $\pi'$ is a corner point distribution.
\end{corollary}

\comment{Point out the multiplicative leakage for our Crowds
example} Corollary \ref{cor:maxleak} gives us a method to compute
the maxima leakages of $\ihs$. In the multiplicative case the
complexity is the same as for computing the
leakage\footnote{Actually we can compute it even faster using an
observation from \cite{Smith:09:FOSSACS} which says that the
leakage on the uniform distribution can be obtained simply by
summing up the maximum elements of each column of the channel
matrix.}. In the additive case we need to find the right corner
point, which can be done by computing the leakages for all corner
points and then comparing them. This method has exponential
complexity (in $|\secre|$) as the size of the set of corner points
is $2^{|\secre|}-1$. We conjecture that this complexity is
intrinsic, i.e. that the problem is NP-hard\footnote{Since the
publication of the article related to this chapter we have proved
that our conjecture is true. The proof will appear, together with
other results, in an extended version of the article}.

\section{Interactive Information Hiding Systems}\label{sec:IIHS}

We now consider extending the  framework to interactive systems,
namely to IHS's in which the secrets and the observables can
alternate in an arbitrary way. The secret part of a run is then an
element of  $\Sigma_\secre^*$, like the observable part is an
element of $\Sigma_\obs^*$. The idea is that such system models an
interactive play between a source of secret information, and a
protocol or program that may produce, each time, some observable
in response. Since each choice is associated to one player of this
``game'', it seems natural to impose that in a choice the actions
are either secret or observable/hidden,  but not both.

The main novelty and challenge of this extension is that part of the secrets come after observable events, and may depend on them.

\begin{dfn}
Interactive $\IHS$'s are defined as $\IHS$'s (Definition \ref{def:IHS}), except that Restrictions $1$ to $3$  are replaced by
$ \alpha(q)\in \distr( \Sigma_\secre\times Q) \cup \distr(
\Sigma-\Sigma_\secre\times Q). $
\end{dfn}

\begin{exa}\label{exa:ebay} Consider an Ebay-like auction protocol with one seller and two possible buyers, one rich and one poor. The seller first publishes the item he wants to sell, which can be either cheap or expensive.
Then the two buyers start bidding. At the end,  the seller looks at the profile of the bid winner
 and decides whether  to sell the item or cancel the transaction.
Figure   \ref{fig:ebay}  illustrates the automaton representing the protocol, for certain given probability distributions.
\end{exa}
%
%

\begin{wrapfigure}{r}{4.9cm}\label{fig:ebay}
\centering
\includegraphics[width=4.9cm]{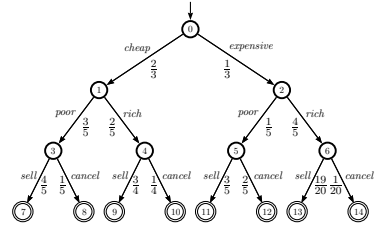}
\caption{Ebay Protocol} \label{fig:ebayprotocol} \vspace{-0.25cm}
\end{wrapfigure}
We assume that the identities of the buyers are secret, while the
price of the item and the seller's decision are observable. We
ignore for simplicity the internal actions which are performed
during the bidding phase. Hence $\Sigma_\obs=\lbrace {cheap} ,
{expensive} , {sell} , {cancel}\rbrace$, $\Sigma_\hidd=\emptyset$,
$\secre=\Sigma_\secre=\lbrace {poor}, {rich} \rbrace$, and $\obs =
\{ {cheap} , {expensive}\}\times\{{sell} , {cancel}\}$. The
distributions on $\secre$ and $\obs$ are defined as usual. For
instance
we   have $\PP([{cheap\ \ sell}])  \ =  \
\PP(\lbrace q_0 {\stackrel{cheap}{\longrightarrow}}  q_1 \stackrel{{poor}}{\longrightarrow} q_3 \stackrel{{sell}}{\longrightarrow} q_7 , q_0\stackrel{{cheap}}{\longrightarrow} q_1 \stackrel{{rich}}{\longrightarrow} q_3 \stackrel{{sell}}{\longrightarrow} q_7\rbrace )
\ =\ \frac{2}{3} \cdot \frac{3}{5} \cdot \frac{4}{5} + \frac{2}{3} \cdot \frac{2}{5} \cdot \frac{3}{4} = \frac{13}{25}.$

Let us now consider how to model the protocol in terms of a noisy
channel. It would seem natural to define the channel associated to
the protocol as the triple $(\secre, \obs, \cm)$ where $\cm(o\mid
s)=\PP([o]\,|\,[s])=\frac{\PP([s]\cap[o])}{\PP([s])}$. This is,
indeed, the approach taken in \cite{Desharnais:02:LICS}. For
instance, with the protocol of Example \ref{exa:ebay}, we would
have:
\begin{eqnarray}\label{eqn:condprob}\
\PP([{cheap\ \ sell}]\,|\,[{poor}])\!=\! \frac{\PP([{poor}]
\!\cap\! {[cheap\ \ sell}])}{\PP([{poor}])} \!=\!
\frac{\frac{2}{3}\cdot\frac{3}{5}\cdot\frac{4}{5}} {\frac{2}{3}
\cdot \frac{3}{5} \!+\! \frac{1}{3} \cdot \frac{1}{5} } \!=\!
\frac{24}{35}.
\end{eqnarray}
However, it turns out that in the interactive case (in particular
when the secrets are not in the initial phase), it does not make
sense to model the protocol in terms of a channel. At least, not a
channel with input $\secre$. In fact, the matrix of a channel is
supposed to be \emph{ invariant } with respect to the input
distribution (like in the case of the $\IHS$'s with variable a
priori considered in previous section), and this is not the case
here. The following is a counterexample.

\begin{exa}\label{exa:isFails}
Consider the same protocol as in Example \ref{exa:ebay}, but
assume now that the distribution over the choice of the buyer is
uniform, i.e.
$\alpha(q_1)({poor},q_3)$ $=
\alpha(q_1)({rich},q_4)$ $=
 \alpha(q_2)({poor},q_5)$ $=
  \alpha(q_2)({rich},q_6)$ $=
  \frac{1}{2}$.
  Then  the conditional probabilities  are different than those for  Example \ref{exa:ebay}.
In particular, in contrast to (\ref{eqn:condprob}), we have
  \[
  \PP([{cheap\ \ sell}]\,|\,[{poor}]) = \frac{\PP([{poor}]\cap{[cheap\ \ sell}])}{\PP([{poor}])} =
\frac{\frac{2}{3}\cdot\frac{1}{2}\cdot\frac{4}{5}} {\frac{2}{3} \cdot \frac{1}{2} + \frac{1}{3} \cdot \frac{1}{2} }
= \frac{8}{15}.
  \]

  \end{exa}

The above observation, i.e. the fact that the conditional probabilities depend on the input distribution, makes it unsound to reason about certain information-theoretic concepts in the standard way. For instance, the \emph{capacity} is defined as the maximum mutual information over all possible input distributions, and the traditional algorithms to compute it are based on the assumption that the channel matrix remains the same while the input distribution variates. This does not make sense anymore in the interactive setting.

However, when the input distribution is fixed, the matrix of the joint probabilities is well defined as $\cm_\lland(s,o) = \PP([s]\cap [o])$,
and  can be computed or approximated using the same methods as for simple $\IHS$'s.
The a priori probability and the channel matrix can then be derived in the standard way:
\vspace{-0.3cm}
\begin{eqnarray*}
\pi(s) = \sum_{o}\cm_\lland(s,o), \qquad \qquad \cm(o\mid s) = \frac{\cm_\lland(s,o)}{\pi(s)}.
\end{eqnarray*}
\vspace{-0.3cm}

Thanks to the formulation (\ref{eqn:postvulnerability}) of the a posteriori vulnerability, the leakage can be computed directly using the joint probabilities.
\begin{exa}
Consider the Ebay protocol $\cal I$ presented in Example \ref{exa:ebay}.
The matrix of the joint probabilities $\cm_\lland(s,o)$ is:
\vspace{-0.5cm}
\begin{center}
\begin{tabular}{   c | c  c  c  c }
    &\ $cheap\ sell$ \ & \ $cheap\ cancel$ \ & \ $expensive \ sell$ \ & $\ expensive \ cancel$\ \\ \hline
    \vspace{-0.25cm}
    \\
    \ $poor$ \ & $\frac{8}{25}$ & $\frac{2}{25}$ & $\frac{1}{25}$ & $\frac{2}{75}$ \\
    \vspace{-0.25cm}
    \\
    \ $rich$ \ & $\frac{1}{5}$ & $\frac{1}{15}$ & $\frac{19}{75}$ & $\frac{1}{75}$ \\
\end{tabular}
\end{center}
Furthermore $\pi({poor})=\frac{7}{15}$ and
$\pi({rich})=\frac{8}{15}$. Hence we have
$\leakm(\ihs)=\frac{51}{40}$ and
$\leaka(\ihs)=\frac{11}{75}$.\comment{In the previous matrix
actions are not italics}

\end{exa}

We note that our techniques to compute channel matrices and
leakage extend smoothly to the case where secrets are not required
to happen at the beginning. However, no assumptions can be made
about the occurrences of secrets (they do not need to occur at the
beginning anymore). This increases the complexity of the
reachability technique to $O((|\secre|\cdot|\obs|\cdot|Q|)^3)$. On
the other hand, complexity bounds for the iterative approach
remain the same.

\section{Related Work}


To the best of our knowledge, this is the first work dealing with
the efficient computation of channel matrices and leakage.
However, for the simple scenario, channel matrices can be computed
using standard model checking techniques. Chatzikokolakis et
al.~\cite{Chatzikokolakis:08:IC} have used Prism~\cite{PRISM} to
model Crowds as a Markov Chain and compute its channel matrix.
Each conditional probability $\cm(o|s)$ is computed as the
probability of reaching a state where $o$ holds starting from
\emph{the} state where $s$ holds. Since for the simple version of
$\IHS$'s secrets occur only once and before observables (as in
Crowds), such a reachability probability equals $\cm(o|s)$. This
procedure leads to $O(|\secre|\cdot |\obs| \cdot
|\overline{Q}|^3)$ time complexity to compute the channel matrix,
where $\overline{Q}$ is the space state of the Markov Chain.

Note that the complexity is expressed in terms of the space state
of a Markov Chain instead of automaton. Since Markov Chains
 do not carry information in transitions they
have a larger state space than an equivalent
\begin{wrapfigure}{r}{4cm}
\vspace{-0.4cm}
\centering
\includegraphics[width=4cm]{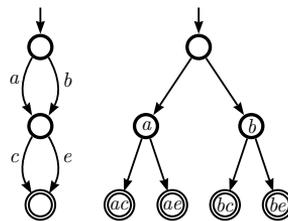}
\label{fig:AutomatavcMC}\caption{Automaton vs Markov Chain}\label{fig:autVSmc}
\vspace{-0.4cm}
\end{wrapfigure}\comment{Change name to states of the Markov Chain to $ac$, $ae, \ldots$}
automaton. Figure~\ref{fig:autVSmc} illustrates this: to model the automaton (left hand side) we need to encode the information in its transitions into states of the Markov Chain (right hand side). Therefore, the probability of seeing observation $a$ and then $c$ in the automaton can be computed as the probability of reaching the state $ac$. The Markov Chain  used for modeling Crowds (in our two honest and one corrupted user configuration) has $27$ states. 

For this reason we conjecture that our complexity $O(|\obs| \cdot
|Q|^3)$ is a considerable improvement over the one on Markov
Chains  $O(|\secre|\cdot |\obs| \cdot |\overline{Q}|^3)$.

With respect to the interactive scenario, standard model checking techniques do not extend because multiple occurrences of the same secret are allowed (for instance in our Ebay example, $\cm({cheap\ \ sell} | {rich})$ cannot be derived from reachability probabilities from the two different states of the automaton where ${rich}$ holds).

\newpage
\thispagestyle{empty}

\chapter{Information Hiding in Probabilistic Concurrent Systems}
\label{ch.ihpcs}

\begin{quote}
 \textit{In this chapter we study the problem of information hiding in systems
 characterized by the coexistence of randomization and
 concurrency. Anonymity and Information Flow are examples of this
 notion.  It is well known that the presence of
 nondeterminism, due to the possible interleavings and
 interactions of the parallel components, can cause unintended
 information leaks. The most established approach to solve this
 problem is to fix the strategy of the scheduler beforehand. In
 this work, we propose a milder restriction on the schedulers, and
 we define the notion of strong (probabilistic) information hiding
 under various notions of observables. Furthermore, we propose a
 method, based on the notion of automorphism, to verify that a
 system satisfies the property of strong information hiding,
 namely strong anonymity or \rev{non}-interference, depending on the
 context. Through the chapter, we use the canonical example of the
 Dining Cryptographers to illustrate our ideas and techniques.}
\end{quote}


\section{Introduction}
The problem of information hiding consists in
trying to prevent the adversary to infer
confidential information from the observables.
Instances of this issue are Anonymity and Information Flow.
In both fields there is a growing interest in the quantitative aspects of the problem, see for instance
\cite{Halpern:05:JCS,Bhargava:05:CONCUR,Zhu:05:ICDCS,Clark:05:QAPL,Clark:05:JLC, Malacaria:07:POPL,Malacaria:08:PLAS,Braun:08:FOSSACS,Clarkson:09:JCS,Chatzikokolakis:08:IC,Chatzikokolakis:08:JCS,Smith:09:FOSSACS}.
This is justified by the fact that often we have some a priori knowledge about the likelihood of the various secrets (which we can usually  express in terms of a probability distribution), and by the fact that protocols often  use randomized actions to obfuscate the link between secret and observable, like in the case of the anonymity protocols of DC Nets \cite{Chaum:88:JC}, Crowds \cite{Reiter:98:TISS}, Onion Routing
\cite{Syverson:97:SSP}, and Freenet \cite{Clarke:00:DIAU}.

In a concurrent setting, like in the case of multi-agent systems, there is also another source of uncertainty,
which derives from the fact that the various entities may interleave and interact in ways that are usually unpredictable,
either because they depend on factors that are too complex to analyze, or because (in the case of specifications) they are
implementation-dependent.

The formal analysis of  systems which exhibit probabilistic and nondeterministic behavior usually involves the use of so-called  \emph{schedulers},
which are functions that, for each path, select only one possible (probabilistic) transition, thus delivering a purely probabilistic
execution tree, where each event has a precise probability.

In the area of security, there is the problem that secret choices, 
like all choices, give rise to different paths. On the other hand, the decision of the scheduler may influence  the observable behavior of the system.
Therefore the security properties are usually violated if we admit as schedulers all possible functions of the paths: 
certain schedulers induce a dependence of the observables on the secrets, and
protocols which would not leak secret information when running in ``real'' systems 
(where the scheduling devices cannot ``see'' the internal secrets of the components and therefore cannot depend on them),
do leak secret information under this more permissive notion of scheduler.
This is a well known problem for which various solutions have already been proposed
\cite{Canetti:06:WODES,Canetti:06:DISC,Chatzikokolakis:10:IC,Chatzikokolakis:09:FOSSACS}.
We will come back to these in the ``Related work'' section.

\subsection{Contribution}
We now list the main contribution of this chapter:

\begin{itemize}
\item
We define a class of partial-information schedulers (which we call \emph{admissible}), schedulers in this class are a restricted version of standard (full-information) schedulers. The restriction is rather flexible and has strong structural properties, thus facilitating the reasoning about security properties. In short, our systems consist of parallel components with certain restrictions on the secret choices and nondeterministic choices. The scheduler selects the next component (or components, in case of synchronization) for the subsequent step independently of the secret choices. We then formalize the notion of quantitative information flow, or degree of anonymity, using this restricted notion of scheduler.

\item
We propose alternative definitions to the property of strong anonymity defined in~\cite{Bhargava:05:CONCUR}. Our proposal differs from the original definition in two aspects: (1) the system should be strongly anonymous for all admissible schedulers instead of  all schedulers (which is a very strong condition, never satisfied in practice), (2) we consider several variants of adversaries, namely (in increasing level of power): external adversaries, internal adversaries, and adversaries in collusion with the scheduler (in a Dolev-Yao fashion). Additionally, we use admissible schedulers to extend the notions of multiplicative and additive leakage (proposed in~\cite{Smith:09:FOSSACS} and~\cite{Braun:09:MFPS}  respectively) to the case of a concurrent system.
\item
We propose a sufficient technique to prove probabilistic strong anonymity, and probabilistic  noninterference,
based on automorphisms.
The idea is the following:
In the purely nondeterministic setting, the strong anonymity of a system is often
proved (or defined) as follows: take two users $A$ and $B$ and a trace
in which user $A$ is `the culprit'. Now find a trace that looks the same
to the adversary, but in which user $B$ is `the culprit'
\cite{Halpern:05:JCS,Garcia:05:FMSE,Mauw:04:ESORICS,Hasuo:07:ESOP}.
This new trace is often most easily obtained by \emph{switching the
behavior of $A$ and $B$}.
Non-interference can be proved in the same way
(where $A$ and $B$ are high information and the trace is the low information).

In this work, we make this technique explicit for anonymity in systems
where probability and nondeterminism coexist, and we need to cope with the restrictions on the
schedulers. We formalize the notion of
\emph{switching behaviors} by using automorphism (it is possible to
switch the behavior of $A$ and $B$ if there exist an automorphism
between them) and then show that the existence of an automorphism implies
strong anonymity.
\item
We illustrate the problem with full-information schedulers in security, our solution providing admissible schedulers, and the application of our prove technique by means of the well known Dining Cryptographers anonymity protocol.
\end{itemize}


\section{Preliminaries}

In this section we gather preliminary notions and results related to
probabilistic automata~\cite{Segala:95:NJC,Segala:95:PhD}, information theory~\cite{Cover:06:BOOK}, and information leakage~\cite{Smith:09:FOSSACS, Braun:09:MFPS}.

\subsection{Probabilistic automata}

A function $\mu \colon Q \to [0,1]$ is a \emph{discrete probability distribution}
on a set~$Q$ if $\sum_{q \in Q} \mu(q) = 1$.
The set of all discrete probability distributions on $Q$ is denoted by $\mathcal{D}(Q)$.

A \emph{probabilistic automaton} is a quadruple $M = (Q, \Sigma, \qi, \theta )$ where
 $Q$ is a countable set of \emph{states},
 $\Sigma$ a finite set of \emph{actions},
 $\qi$ the \emph{initial} state, and
 $\theta$ a \emph{transition function} $\theta: Q \to \mathcal{P}(\distr(\Sigma \times Q))$.
 Here $\pow(X)$ is the set of all subsets of $X$.

If $\theta(q) = \emptyset$, then $q$ is a \emph{terminal} state. We write
$q {\to} \mu$ for $\mu \in \theta(q), \ q\in Q$.
Moreover, we write $q \smash{\stackrel{a}{\to}} r$ for $q, r \in Q$ whenever $q {\to} \mu$ and $\mu(a,r) > 0$.
A \emph{fully probabilistic automaton} is a probabilistic automaton
satisfying $|\theta(q)|\leq 1$ for all states. In case $\theta(q)\not=
\emptyset$ in a fully probabilistic automaton, we will overload notation
and use $\theta(q)$ to denote the distribution outgoing from $q$.
A \emph{path} in a probabilistic automaton
is a sequence $\sigma = q_0  \stackrel{a_1}{\to} q_1 \stackrel{a_2}{\to}
\cdots$ where $q_i \in Q$, $a_i \in \Sigma$ and $q_i
\smash{\stackrel{a_{i+1}}{\to}} q_{i+1}$. A path can be \emph{finite} in
which case it ends with a state. A path is \emph{complete} if it is
either infinite or finite ending in a terminal state. Given a path
$\sigma$, $\first(\sigma)$ denotes its first state, and if $\sigma$ is
finite then  $\last{\sigma}$ denotes its last state. A \emph{cycle} is a
path $\sigma$ such that $\last{\sigma} = \first(\sigma)$. Let $\paths_q(M)$
denote the set of all paths, $\fpaths_{\!\!q}(M)$ the set of all finite paths,
and $\cpaths_q(M)$ the set of all complete paths of an automaton $M$,
starting from the state $q$. We will omit $q$ if $q=\qi$. Paths are
ordered by the prefix relation, which we denote by $\leq$. The
\emph{trace} of a path is the sequence of actions  in $\Sigma^{*} \cup
\Sigma^{\infty}$ obtained by removing the states, hence for the above
path $\sigma$ we have $\trace(\sigma) = a_1a_2\ldots$. If
$\Sigma'\subseteq \Sigma$, then $\trace_{\Sigma'}(\sigma)$ is the
projection of $\trace(\sigma)$ on the elements of $\Sigma'$.
%

Let $M = (Q, \Sigma, \qi, \theta )$ be a (fully) probabilistic automaton, $q\in Q$ a state, and let $\sigma \in \fpaths_{\!\!\!q}(M)$ be a finite path starting in $q$. The \emph{cone} generated by $\sigma$ is the set of complete paths
$\cone{\sigma} = \{ \sigma^\prime \in \cpaths_q(M) \mid \sigma \leq \sigma^\prime\}.$
Given a fully probabilistic automaton $M=(Q, \Sigma, \qi, \theta)$ and a state $q$, we can calculate the \emph{ probability value}, denoted by $\psp_q(\sigma)$, of any finite path $\sigma$ starting in $q$ as follows:
$\psp_q(q)   =   1$ and
$\psp_q(\sigma \, \stackrel{a}{\to}\, q')   =   \psp_q(\sigma)\ \mu(a,q'),   \text{~where~} \last{\sigma} \to \mu$.

Let $\Omega_{q} \eqdef \cpaths_q(M)$ be the sample space, and let
$\mathcal F_{q}$ be the smallest $\sigma$-algebra generated by the
cones. Then $\psp_q$ induces a unique \emph{probability measure} on
$\mathcal F_{q}$ (which we will also denote by $\psp_q$) such that
$\psp_q(\cone{\sigma}) = \psp_q(\sigma)$ for every finite path $\sigma$
starting in $q$.  For $q=\qi$ we write $\psp$ instead of $\psp_{\qi}$.

A (full-information) scheduler for a probabilistic automaton $M$ is a function $\zeta
\colon \fpaths(M) \to (\distr(\Sigma \times Q)\cup
\{\bot\})$
such that for all finite paths $\sigma$, if $\theta(\last{\sigma})\not=\emptyset$ then $\zeta(\sigma)\in\theta(\last{\sigma})$, and $\zeta(\sigma)=\bot$ otherwise.
\noindent Hence, a scheduler $\zeta$ selects one of the available transitions in each state, and determines therefore a fully probabilistic automaton, obtained by pruning from $M$ the alternatives that are not chosen by $\zeta$. Note that a scheduler is history dependent since it can take different decisions for the same state $s$ according to the past evolution of the system.

\subsection{Noisy Channels}

This section briefly recalls the notion of noisy channels from Information Theory \cite{Cover:06:BOOK}.


A \emph{noisy channel} is a tuple ${\cal C}\eqdef( {\cal X}, {\cal Y}, P(\cdot|\cdot))$
where  ${\cal X}=\lbrace x_1,x_2,\ldots,x_n\rbrace$ is a finite set of \emph{input values}, modeling the \emph{secrets} of the channel,
and ${\cal Y} =\lbrace y_1,y_2,\ldots,y_m\rbrace$ is a finite set of \emph{output values}, the \emph{observables} of the channel.
For $x_i\in {\cal X}$ and $y_j\in {\cal Y}$, $\cm(y_j|\,x_i)$ is the conditional probability of obtaining the output $y_j$ given that the input is $x_i$.
These conditional probabilities constitute the so called \emph{channel matrix}, where $\cm(y_j|x_i)$ is the element at the intersection of the
$i$-th row and the $j$-th column.
For any  input distribution $P_X$ on $\cal X$, $P_X$ and  the channel matrix
determine a joint probability $P_\wedge$ on ${\cal X}\times {\cal Y}$, and the corresponding marginal probability $P_Y$ on $\cal Y$ (and hence a random variable $Y$). $P_X$ is also called \emph{a priori distribution} and it is often denoted by $\pi$. The probability of the input given the output is called \emph{a posteriori  distribution}.

\subsection{Information leakage}

We recall here the definitions of \emph{multiplicative leakage} proposed in \cite{Smith:09:FOSSACS}, and of
\emph{additive leakage} proposed in \cite{Braun:09:MFPS}\footnote{The notion proposed by Smith in \cite{Smith:09:FOSSACS} was given in a (equivalent) logarithmic form, and called simply \emph{leakage}. For uniformity sake we use here the terminology and formulation of \cite{Braun:09:MFPS}.}.
We assume given a noisy channel ${\cal{C}}=( {\cal X}, {\cal Y}, P(\cdot|\cdot))$ and a random variable $X$ on $\cal X$.
The \emph{a priori vulnerability} of the secrets in $\cal X$ is the probability of guessing the right secret, defined as
$\vul(\emph{X})\eqdef\max_{x\in{\cal X}} P_X(x).$
The rationale behind this definition is that the adversary's best bet is on the secret with highest probability.
The \emph{a posteriori vulnerability} of the secrets in $\cal X$ is the probability of guessing the right secret, after the output has been observed, averaged over the probabilities of the observables. The formal definition is
$\vul(\emph{X}\,|\,\emph{Y})\eqdef\sum_{y\in{\cal Y}}P_Y(y) \max_{x\in{\cal X}} P(x\,|\,y).$
Again, this definition is based on the principle that the adversary will choose the secret with the highest a posteriori probability.

Note that, using Bayes theorem, we can write the a posteriori vulnerability in terms of the channel matrix and the a priori distribution, or in terms of the joint probability:
\begin{eqnarray*}\label{eqn:postvulnerability}
 \vul(\emph{X}\, |\, \emph{Y}) \! = \!  \sum_{y\in{\cal Y}}\max_{x\in{\cal X}} ( P(y\,|\,x) P_X(x))
\!  = \! \sum_{y\in{\cal Y}}\max_{x\in{\cal X}}  P_\wedge(x,y).
\end{eqnarray*}

The \emph{multiplicative} leakage is
$\leakm(\nc,P_X)\eqdef\frac{\vul(\emph{X}|\emph{Y})}{\vul(\emph{X})}$
whereas the \emph{additive} leakage is
$\leaka(\nc,P_X)\eqdef\vul(\emph{X}|\emph{Y})-\vul(\emph{X})$.

\subsection{Dining Cryptographers}\label{sec:dc}

This problem, described by Chaum in~\cite{Chaum:88:JC}, involves a situation in which three cryptographers
are dining together. At the end of the dinner, each of them is secretly informed by
a central agency (master) whether he should pay the bill, or not. So, either the master
will pay, or one of the cryptographers will be asked to pay. The cryptographers (or some
external observer) would like to find out whether the payer is one of them or the master.
However, if the payer is one of them, they also wish to maintain anonymity over the
identity of the payer.

A possible solution to this problem, described in~\cite{Chaum:88:JC}, is that each cryptographer
tosses a coin, which is visible to himself and his neighbor to the left. Each
cryptographer observes the two coins that he can see and announces \emph{agree} or \emph{disagree}. If a
cryptographer is not paying, he will announce \emph{agree} if the two sides are the same and
\emph{disagree} if they are not. The paying cryptographer will say the opposite.
It can be proved that if the number of disagrees is even, then the master is paying; otherwise,
one of the cryptographers is paying. Furthermore, in case one of the cryptographers is
paying, neither an external observer nor the other two cryptographers can identify, from their individual information, who exactly is paying (provided that the coins are fair).
The Dining Cryptographers (DC) will be a running example through the chapter.
\begin{figure}[h]
\centering
\includegraphics[width=6.5cm]{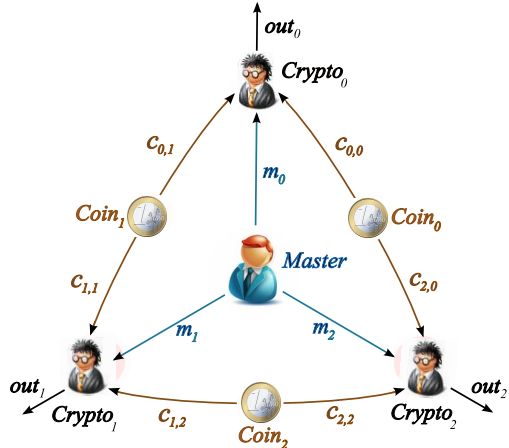}\caption{Chaum's system for the Dining Cryptographers (\cite{Chaum:88:JC})}\label{fig:dincrypt}
\vspace{-0.25cm}
\end{figure}

\section{Systems}\label{sec:systems}


In this section we describe the kind of systems we are dealing
with. We start by introducing a variant of probabilistic automata,
that we call \emph{ tagged probabilistic automata} (TPA). These
systems are parallel compositions of purely probabilistic
processes, that we call \emph{components}. They are equipped with
a unique identifier, that we call \emph{tag}, or {\emph{label}, of
the component. Note that, because of the restriction that the
components are fully deterministic, nondeterminism is generated
only from the interleaving of the parallel components.
Furthermore, because of the uniqueness of the tags, each
transition from a node is associated to a different tag / pair of
two tags (one in case only one component makes a step, and two in
case of a synchronization step among two components).

\subsection{Tagged Probabilistic Automata}

We now formalize the notion of TPA.

\begin{definition}
  A \emph{tagged probabilistic automaton} (TPA) is a tuple $(Q,\allowbreak L,\allowbreak \Sigma,\qi,\theta)$, where
  \begin{itemize}
  \item $Q$ is a set of \emph{states},
  \item $L$ is a set of \emph{tags}, or \emph{labels},
  \item $\Sigma$ is a set of \emph{actions},
  \item $\qi \in Q$ is the \emph{initial state},
  \item $\theta \colon Q \to \mathcal{P}(L \times D(\Sigma \times Q))$ is a \emph{transition function}.
  \end{itemize}
  \noindent with the additional requirement that for every $q \in Q$ and every $\ell \in L$ there is at most one $\mu \in D(\Sigma \times Q)$
  such that $(\ell,\mu) \in \theta(q)$.
\end{definition}

A path for a TPA is a sequence $\sigma=q_0 \stackrel{l_1,a_1}{\longrightarrow}q_1 \stackrel{l_2,a_2}{\longrightarrow}q_2 \cdots$. In this way, the process with identifier $l_{i}$ induces the system to move from $q_{i-1}$ to $q_{i}$ performing  the action  $a_{i}$, and it does so with probability $\mu_{l_{i}}(a_{i},q_{i})$, where $\mu_{l_{i}}$ is \emph{the} distribution associated to the choice made by the component $l_{i}$. Finite paths and complete paths are defined in a similar manner.

In a TPA, the  scheduler's choice  is determined by the choice of the tag.
We will use $\avail(q)$ to denote the tags of the components that are enabled to make a transition. Namely,
\begin{equation}\label{eq:enab}
\avail(q) \eqdef \{\ell\in L\mid \exists\, \mu \!\in\! D(\Sigma \times Q)\,:\, (\ell,\mu) \in \theta(q)\}
\end{equation}

We  assume that the scheduler is forced to select a component among those which are enabled, i.e., that the execution does not stop unless all components are blocked (suspended or terminated). This is in line with the spirit of process algebra, and also with the tradition of Markov Decision Processes, but contrasts with that of the Probabilistic Automata of Lynch and Segala~\cite{Segala:95:NJC}. However, the results in this chapter do not depend on this assumption; we could as well allow schedulers which decide to terminate the execution even though there are transitions which are possible from the last state.

\begin{definition}
  A \emph{scheduler} for a TPA $M = (Q,L,\Sigma,\qi,\theta)$ is a function $\zeta \colon \allowbreak\fpaths(M) \allowbreak\to (L\cup\{\bot\})$ such that for all finite paths $\sigma$, $\zeta(\sigma)\in\avail(\last{\sigma})$ if $\avail(\last{\sigma})\not=\emptyset$ and $\zeta(\sigma)=\bot$ otherwise.
\end{definition}

\subsection{Components}

To specify the components we use a sort of probabilistic version
of CCS~\cite{Milner:89:BOOK,Milner:99:BOOK}. We assume a set of
\emph{secret actions}  $\secr$ with elements $s,s_1,s_2,\cdots$,
and a disjoint set of \emph{observable actions}  $\obse$ with
elements $a, a_1,a_2, \cdots$. Furthermore we have
\emph{communication actions} of the form $c(x)$ (receive $x$ on
channel $c$, where $x$ is a formal parameter), or $\bar{c}\langle
v\rangle$ (send $v$ on channel ${c}$, where $v$ is a value on some
domain $V$). Sometimes we need only to synchronize without
transmitting any value, in which case we will use simply $c$ and
$\bar{c}$. We denote the set of channel names by $C$.

A component $q$ is specified by the following grammar:

\begin{flushleft}
\textbf{Components}
\[
\begin{array}{lrcll}
 &q &\mbox{::=}& 0   \qquad & \text{termination}\\
 &&\mid &a.q  & \text{observable prefix}\\
&&\mid &\sum_{i} p_i: q_i & \text{blind choice}\\
&& \mid &\sum_{i} p_i: s_i . q_i  \qquad & \text{secret choice}\\
&&\mid & \mathit{if}\ x = v \  \mathit{then} \ q_1\  \mathit{else}\  q_2  \  \  \ &\text{conditional}\\
&& \mid &A  \qquad & \text{process call}\\[5mm]
\end{array}
\]
\textbf{Observables}
\[
\begin{array}{lrcll}
  &a &\mbox{::=}& c \mid \bar{c}& \text{simple synchronization}\\
&&\mid &  c(x) \mid \bar{c}\langle v \rangle& \text{synchronization and communication}\\[4mm]
\end{array}
\]
\end{flushleft}

The $p_i$, in the blind and secret choices, represents the probability of the $i$-th branch and must satisfy $0\leq p_i\leq 1$ and
$\sum_{i} p_i=1$.
When no confusion arises, we use simply $+$ for a binary choice.
The process call $A$ is a simple process identifier. For each of them, we assume a
corresponding unique process declaration of the form $A \stackrel{\rm def}{=}q$.
The idea is that, whenever $A$ is  executed, it triggers the execution of $q$.
Note that  $q$ can contain $A$ or another process identifier,
which means that our language allows (mutual) recursion.

Note that each component contains only probabilistic and
sequential constructs. In particular, there is no internal
parallelism. Hence each component corresponds to a purely
probabilistic automaton (apart from the input nondeterminism,
which disappears in the definition of a system), as described by
the operational semantics below. \rev{The main reason to dismiss
the use of internal parallelism is verification: as mentioned in
the Introduction we will present a proof technique for the
different definitions of anonymity proposed in this work. This
result would not be possible without such restriction on the
components (see Example \ref{exa:internal-parallelism}).}

\rev{For an extension of this framework allowing the use of
internal parallelism we refer to \cite{Alvim:10:IFIP-TCS}. There,
the authors combine \emph{global nondeterminism} (arising from the
interleaving of the components) and \emph{local nondeterminism}
(arising from the internal parallelism of the components). The
authors use such (extended) framework for a different purpose than
ours, namely to define a notion of equivalence suitable for
security analysis. No verification mechanisms are provided in
\cite{Alvim:10:IFIP-TCS}.}

\medskip \noindent \emph{Components' semantics:} The operational semantics consists of probabilistic transitions of the form $q{\rightarrow}\mu$
where $q\in Q$ is a process, and $\mu\in\distr(\Sigma\times Q)$ is a distribution on actions and processes.
They are specified by the following rules:

\begin{align*}
\text{PRF1 }\quad&
\begin{tabular}{ c } $ v\in V$
  \\ \vspace{-0.35cm}\\
  \hline \vspace{-0.35cm}
  \\
  $c(x).q\rightarrow \delta{(c(v),q[v/x])}$
\end{tabular}\\
\text{PRF2 } \quad&
\begin{tabular}{ c }
  \\ \vspace{-0.35cm}\\
  \hline \vspace{-0.35cm}
  \\
  $a.q\rightarrow \delta{(a,q)}$
\end{tabular} \quad \text{if } a \neq c(x)\\
\text{INT }\quad &
\begin{tabular}{ c }
  \\ \vspace{-0.35cm}\\
  \hline \vspace{-0.35cm}
  \\
  $\sum_{i}p_i: q_i\rightarrow \probsum{i} p_i \cdot \delta{(\tau,q_i)}$
\end{tabular}\\
\text{SECR}\ \quad&
\begin{tabular}{ c }
  \\ \vspace{-0.35cm}\\
  \hline \vspace{-0.35cm}
  \\
  $\sum_{i}p_i: s_i . q_i\rightarrow \probsum{i}p_i \cdot \delta{(s_i,q_i)}$
\end{tabular}
\end{align*}

\begin{align*}
\text{CND1}\quad&
\begin{tabular}{ l }
  \\ \vspace{-0.35cm}\\
  \hline \vspace{-0.35cm}
  \\
  $\mathit{if}\ v = v \ \mathit{then} \ q_1\ \mathit{else} \ q_2 \rightarrow \delta{(\tau,q_1)}$
\end{tabular}\\
\break \text{CND2 }\ \quad&
\begin{tabular}{ c }
  $v\neq v'$
  \\ \vspace{-0.35cm}\\
  \hline \vspace{-0.35cm}
  \\
 $\mathit{if}\ v = v' \ \mathit{then} \ q_1\ \mathit{else} \ q_2 \rightarrow \delta{(\tau,q_2)}$
\end{tabular}\\
\text{CALL }\quad&
\begin{tabular}{ c }
  $q\rightarrow\mu$
  \\ \vspace{-0.35cm}\\
  \hline \vspace{-0.35cm}
  \\
  $A\rightarrow \mu$
\end{tabular} \quad \text{if } A\eqdef q
\\
\end{align*}

\noindent
$\probsum{i}p_i \cdot \mu_i$ is the distribution $\mu$ such that $\mu(x)=\sum_{i}p_i \mu_i{(x)}$.
We use $\delta(x)$ to represent the delta of Dirac, which assigns probability $1$ to $x$.
The silent action, $\tau$, is a special action different from all the observable and the secret actions.
$q[v/x]$ stands for the process $q$ in which any occurrence of $x$ has been replaced by $v$.
To shorten the notation, in the examples throughout the chapter, we omit writing
explicit termination, i.e., we omit the symbol 0 at the end of a term.

\subsection{Systems} A system consists of $n$ processes (components) in
parallel, restricted at the top-level on the set of channel names $C$: \[ (C) \ q_1 \parallel q_2 \parallel \cdots \parallel q_n.\]
The restriction on $C$ enforces synchronization (and possibly communication)  on the channel names belonging to $C$, in accordance with the CCS spirit.
Since $C$ is the set of all channels, all of them are forced to synchronize. This is to eliminate, at the level of systems,  the nondeterminism
generated by the rule for the receive prefix, PRF1.

\medskip \noindent \emph{Systems' semantics:} The semantics of a system gives rise to a TPA, where the states are
terms representing systems during their evolution.
A transition now is of the form $q\stackrel{\ell}{\rightarrow} \mu$ where $\mu\in(\distr(\Sigma\times Q))$ and $\ell\in L$ is either
the identifier of the component which makes the move, or a two-element set of
identifiers representing the two partners of a synchronization. The following
two rules provide the operational semantics rules in the
case of interleaving and synchronisation/communication, respectively.\\
\paragraph{Interleaving} If  $a_j\not\in C$
\[\begin{tabular}{ c }
  $q_i \rightarrow \probsum{j} p_j \cdot \delta{(a_j,q_{ij})}$
  \\ \vspace{-0.35cm}\\
  \hline \vspace{-0.35cm}
  \\
  $(C)\ q_1\parallel \cdots \parallel q_i\parallel \cdots \parallel q_n\stackrel{  i  }{ \rightarrow}\probsum{j}p_j \cdot \delta{(a_j,(C)\ q_1\parallel \cdots \parallel q_{ij}\parallel \cdots \parallel q_n)}$
\end{tabular}\]

\noindent where $i$ indicates the tag of the component making the step.\\
\paragraph{Synchronization/Communication}

\[\begin{tabular}{ c }
  $q_i \rightarrow \delta{(\bar{c}\langle v \rangle,q_i^\prime)}\qquad q_j \rightarrow \delta{( {c}(v),q_j^\prime)}$
  \\ \vspace{-0.35cm}\\
  \hline \vspace{-0.35cm}
  \\
  $(C)\ q_1\parallel \cdots \parallel q_i\parallel \cdots \parallel q_n \stackrel{\{i,j\}}{\longrightarrow}\delta{(\tau,(C)\ q_1\parallel \cdots \parallel q_{i}^\prime \parallel\cdots\parallel q_j^\prime \parallel \cdots \parallel q_n)}$
\end{tabular} \]\\

\noindent here $\{i,j\}$ is the tag indicating that the components making the step are $i$ and $j$.
For simplicity we write $\stackrel{i,j}{\longrightarrow}$ instead of $\stackrel{\{i,j\}}{\longrightarrow}$.
The rule for synchronization without communication is similar, the only difference is that we do not have $\langle v\rangle$ and $(v)$ in the actions.
Note that $c$ can only be an observable action (neither a secret nor $\tau$), by the assumption  that channel names can only be observable actions. 

We note that both interleaving and synchronization rules generate nondeterminism. The only other source of nondeterminism is PRF1,
 the rule for a  receive prefix $c(x)$.
However the latter is not real  nondeterminism: it is introduced in the semantics of the components but it disappears in the semantics of the systems,
 given that  the channel $c$ is restricted at the top-level. In fact the restriction enforces communication, and when communication takes place, only the branch corresponding to the actual value $v$ transmitted by the corresponding send action is maintained, all the others disappear.





\begin{proposition}\label{prop:TPAcharacter} The operational semantics of a system is a TPA with the following characteristics:
\begin{enumerate}
 \item[(a)] Every step $q\stackrel{\ell}{\rightarrow}\mu$ is either
    \begin{itemize}
     \item [{\rm a blind choice:}] $\mu= \probsum{i} p_i \cdot \delta{(\tau,q_i)}$, or\\[-2mm]
     \item [{\rm a secret choice:}]   $\mu=\probsum{i}p_i \cdot \delta{(s_i,q_i)}$, or\\[-2mm]
     \item [{\rm a delta of Dirac:}]  $\mu=\delta{(\alpha,q^\prime)}$ with $\alpha\in \obse$ or $\alpha=\tau$.
     \end{itemize}
%
%
 \item[(b)]  If $q\stackrel{\ell}\rightarrow \mu$ and $q\stackrel{\ell}\rightarrow \mu^\prime$  then $\mu=\mu^\prime$.
 \end{enumerate}
\end{proposition}
\begin{proof} For (a) , we have that the rules for the components and the rule for synchronization / communication can only produce blind choices, secret choices, or deltas of Dirac. Furthermore, because of the restriction on all channels, the transitions at the system level cannot contain communication actions. Finally, observe that the interleaving rule maintains these properties.

As for (b), we know that at the component level, the only source of nondeterminism is PRF1,
 the rule for a  receive prefix $c(x)$. At the system level, this action is forced to synchronize with a corresponding send action, and, in a component, there can be only one such action available at a time.
Hence the tag determines the  value to be sent, which in turn determines the selection of exactly one branch in the receiving process.
The only other sources of nondeterminism are the interleaving  and the synchronization/communication rules, and they induce a different tag for each alternative.
\end{proof}

\exampleheader \label{exa:DCcomponents}
We now present the components for the Dining Cryptographers using the introduced
syntax. They correspond to Figure~\ref{fig:dincrypt} and to the
automata depicted in Figure~\ref{fig:DCautomata}. As announced before, we
omit the symbol $0$ for explicit termination at the end of each term. The secret
actions $s_i$ represent the choice of the payer.
The operators $\oplus, \ominus$ represent the sum modulo $2$ and the
difference modulo $2$, respectively.
The test $i==n$ returns $1$ (true) if $i= n$, and $0$ otherwise.
The set of restricted channel names is $\mathit{C}\!=\!\lbrace
c_{0,0},c_{0,1}, c_{1,1}, c_{1,2}, c_{2,0},c_{2,2}, m_0,m_1,m_2\rbrace$.

\begin{figure}[!h]
\vspace{-0.25cm}
\begin{eqnarray*}
    \text{Master} &\eqdef & p:
        \outp{m}_0\ang{0} \, .\,
         \outp{m}_1\ang{0} \, .\,
         \outp{m}_2\ang{0} +  (1-p): \smallsumb{i=0}{2} p_i:
        s_{i}\, .\,\\
    && \outp{m}_{0}\ang{i==0} \, .\,
         \outp{m}_{1}\ang{i==1} \, .\,
         \outp{m}_{2}\ang{i==2}\\
    \text{Crypt}_i &\eqdef&
        m_i(\mathit{pay})\, .\,
        c_{i,i}(coin_1)\, .\,
        c_{i,i\oplus 1}(coin_2)\, .\, \outp{out}_i\ang{\mathit{pay}\oplus coin_1\oplus coin_2}\\[2pt]
    \text{Coin}_i &\eqdef & 0.5:
          \bar{c}_{i,i}\ang{0}\, .\,
          \bar{c}_{i\ominus 1,i}\ang{0}\ +\
          0.5: \bar{c}_{i,i}\ang{1}\, . \,
          \bar{c}_{i\ominus 1,i}\ang{1}
          \\[2pt]
    \text{System} & \eqdef & (\mathit{C})\
        \text{Master} \parallel \smallprodb{i=0}{2} \text{Crypt}_i \parallel
            \smallprodb{i=0}{2} \text{Coin}_i
\end{eqnarray*}
    \caption{Dining Cryptographers CCS}  \label{fig:dining}
\end{figure}

The operation $ \mathit{pay}\oplus coin_1\oplus coin_2$ in
Figure~\ref{fig:dining}  is  syntactic sugar, it can be defined using the {\it
if-then-else} operator. Note that, in this way, if a  cryptographer is not paying (pay = 0), then he announces $0$ if the two coins are the same (agree) and $1$ if they are not (disagree).
\begin{figure*}[!h]
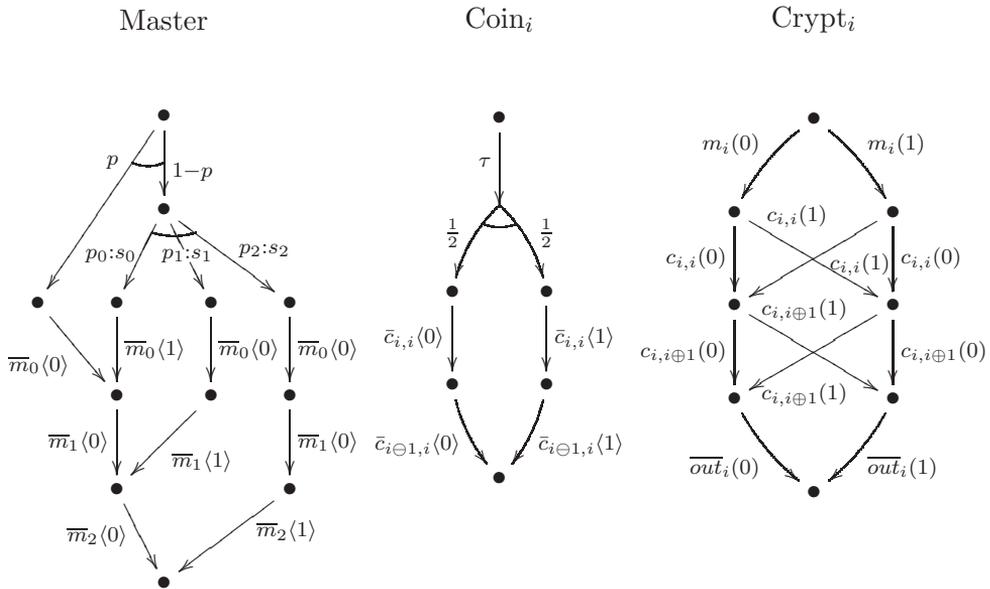

\begin{equation*}
\hspace{-0.4cm}  \input{AdmissibleSchedulers/masterccs.tex}
  \input{AdmissibleSchedulers/coin.tex}
  \input{AdmissibleSchedulers/crypt-new.tex}
\end{equation*}
\centering
\caption{Dining Cryptographers Automata}\label{fig:DCautomata}
\end{figure*} 


\section{Admissible Schedulers}\label{sec:AdmSch}

We now introduce the class of admissible schedulers.

%
%
%
%

Standard (full-information) schedulers have access to all the information about the system and its components, and in particular the secret choices. Hence, such schedulers can leak secrets by making their decisions depend on the secret choice of the system. This is the case with the Dining Cryptographers protocol of Section~\ref{sec:dc}: among all possible schedulers for the protocol, there are several that leak the identity of the payer. In fact the scheduler has the freedom to decide the order of the announcements of the cryptographers (interleaving), so a scheduler could choose to let the payer announce lastly. In this way, the attacker learns the identity of the payer simply by looking at the interleaving of the announcements.

\subsection{The screens intuition}

Let us first describe admissible schedulers informally. As mentioned in the introduction, admissible schedulers can base their decisions only on partial information about the evolution of the system, in particular admissible schedulers cannot base their decisions on information concerned with the internal behavior of components (such as secret choices).

We follow the subsequent intuition: admissible schedulers are entities that have access to a screen with buttons, where each button represents one (current) available option. At each point of the execution the scheduler decides the next step among the available options (by pressing the corresponding button). Then the output (if any) of the selected component becomes available to the scheduler and the screen is refreshed with the new available options (the ones corresponding to the system after making the selected step). We impose that the scheduler can base its decisions only on such information, namely: the screens and outputs he has seen up to that point of the execution (and, of course, the decisions he has made).

\medskip
\exampleheader Consider  $S \eqdef (\{c_1,c_2\}) \ \rev{r}\!
\parallel\! \rev{q} \!\parallel\! \rev{t}$, where
\[\begin{array}{c}\rev{r} \eqdef 0.5 : s_1 . \overline{c}_1 . \overline{c}_2 + 0.5 : s_2 . \overline{c}_1 . \overline{c}_2, \\
\rev{q} \eqdef c_1 . (0.5 : a_1 + 0.5 : b_1),  \quad  \rev{t}
\eqdef c_2 . (0.5 : a_2 + 0.5 : b_2).
\end{array} \]

Figure \ref{fig:screens} shows the sequence of screens corresponding to a particular sequence of choices taken by the scheduler\footnote{The transitions from screens $4$ and $5$ represent $2$ steps each (for simplicity we omit the $\tau$-steps generated by blind choices)}. Interleaving and communication options are represented by yellow and red buttons, respectively. An arrow between two screens represents the transition from one to the other (produced by the scheduler pressing a button), additionally, the decision taken by the scheduler and corresponding outputs are depicted above each arrow.


\begin{figure}[h!]
\centering
\includegraphics[width=8cm]{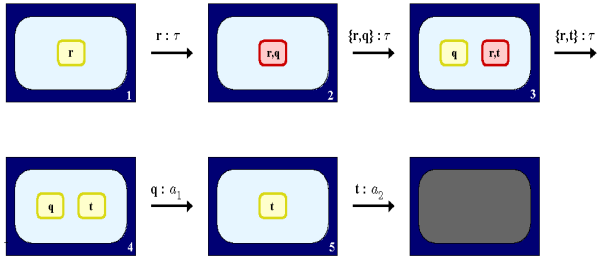}\caption{\rev{Screens intuition}}\label{fig:screens}
\end{figure}

Note that this system has exactly the same problem as the DC
protocol: a full-information scheduler could reveal the secret by
basing the interleaving order (\rev{$q$} first or \rev{$t$} first)
on the secret choice of the component \rev{$r$}. However, the same
does not hold anymore for admissible schedulers (the scheduler
cannot deduce the secret choice by just looking at the screens and
outputs). This is also the case for the DC protocol, i.e.,
admissible schedulers cannot leak the secret of the protocol.

\subsection{The formalization}
Before formally defining admissible schedulers we need to
formalize the ingredients of the screens intuition. The buttons on
the screen (available options) are the enabled options given by
the function $\avail$ (see (\ref{eq:enab}) \rev{in Section 4.3.1}),
the decision made by the scheduler is the tag of the selected
enabled option, observable actions are obtained by sifting the
secret actions to the schedulers by means of the following
function:

\begin{equation*}
\hide(\alpha) \eqdef \left\{\begin{array}{ll}
                \alpha & \mbox{if } \alpha\in\obse \cup \{\tau\},\\[2mm]
                \tau    &\mbox{if  }\alpha \in \secr .
                \end{array}
             \right.
\end{equation*}

\noindent The partial information of a certain evolution of the system is given by the map $t$ defined as follows.

\begin{definition} Let $\qi\stackrel{\ell_1,\alpha_1}{\longrightarrow} \cdots \stackrel{\ell_{n},\alpha_{n}}{\longrightarrow}q_{n+1}$ be a finite path of the system, then we define $t$ as:
$$\begin{array}{l}
\smallskip
t\left(\qi\stackrel{\ell_1,\alpha_1}{\longrightarrow} \cdots
\stackrel{\ell_{n},\alpha_{n}}{\longrightarrow}q_{n+1}\right)
\eqdef \\
\quad \quad \qquad \qquad(\avail(\qi),\ell_1,\hide(\alpha_1))
\cdots
(\avail(q_n),\ell_n,\hide(\alpha_n)) \cdot\avail(q_{n+1}).\\
\end{array}$$
\end{definition}

Finally, we have all the ingredients needed to define admissible schedulers.


\begin{definition}[Admissible schedulers]
A scheduler $\zeta$ is admissible if for all $\sigma,\sigma^\prime\in \fpaths$
\[t(\sigma)=t(\sigma^\prime)\quad \mbox{implies}\quad \zeta(\sigma)=\zeta(\sigma^\prime).\]
\end{definition}

In this way, admissible schedulers are forced to take the same decisions on paths that they cannot tell apart. Note that this is a restriction on the original definition of (full-information) schedulers where $t$ is the identity map over finite paths (and consequently the scheduler is free to choose differently).

In the kind of systems we consider (the TPAs) the only source of
nondeterminism are the interleaving and interactions of the
parallel components. Consequently, in a TPA the notion of
scheduler is quite simple: its role, indeed, is to select, at each
step,  the component or pair of components which will perform the
next transition. In addition, the  TPA model allows us to express
in a simple way the notion of admissibility:  in fact the
transitions available in the last state of $\sigma$ are determined
by the set of components enabled in the last state of $\sigma$,
and  $t(\sigma)$ gives (among other information) such set.
Therefore $t(\sigma)=t(\sigma^\prime)$ implies that the last
states of $\sigma$ and $\sigma'$ have the same possible
transitions, hence it is possible to require that
$\zeta(\sigma)=\zeta(\sigma^\prime)$ without being too restrictive
or too permissive. In more general systems, where the sources of
nondeterminism can be arbitrary, it is difficult to impose that
the scheduler``does not depend on the secret choices'', because
different secret choices in general may give rise to states with
different sets of transitions, and it is unclear whether such
difference should be ruled out as ``inadmissible'', or should be
considered as part of what a ``real'' scheduler can detect.


%

%
%
%
%
%
%
%
%
%
%
%
%
%
%


\section{Information-hiding properties in presence of nondeterminism}

In this section we revise the standard definition of information flow and anonymity in our framework of controlled nondeterminism.

We first consider the notion of adversary. We consider three possible notions of adversaries, increasingly more powerful.

\subsection{Adversaries}
\noindent\emph{External adversaries:}
Clearly, an adversary should be able, by definition, to see at least the observable actions.
For an adversary external to the system $S$, it is natural to assume that these are also the only actions that he is supposed to see.
Therefore, we define the  observation domain, for an external adversary, as the set of the (finite) sequences of observable actions, namely:
\[
\observables_{\rm e} \eqdef  \Sigma^*_O.
\]
Correspondingly, we need a function $t_{\rm e} : \fpaths(S) \rightarrow \observables_{\rm e}$ that extracts the observables from the executions:

\[t_{\rm e}\left(q_0\stackrel{\ell_1,\alpha_1}{\longrightarrow}\cdots \stackrel{\ell_{n},\alpha_n}{\longrightarrow}q_{n+\! 1}\right)\!\eqdef\!\sieve(\alpha_1) \cdots \sieve(\alpha_{n})\]
\noindent where
\[
\sieve(\alpha) \eqdef \left\{\begin{array}{ll}
               \alpha & \mbox{if } \alpha\in\obse,\\[2mm]
               \epsilon    &\mbox{if }  \alpha \in \secr \cup\{ \tau\}.
               \end{array}
            \right.
\]

\noindent\emph{Internal adversaries:} An internal adversary may be able to see, besides the observables, also the intearleaving and synchronizations of the various components, i.e. which component(s)  are active, at each step of the execution. Hence it is natural to define the observation domain, for an internal adversary, as the sequence of pairs of observable action  and tag  (i.e. the identifier(s) of the active component(s)), namely:
\[
\observables_{\rm i} \eqdef ( L \times  (\obse \cup \{\tau\}))^*.
\]
Correspondingly, we need a function $t_{\rm i} : \fpaths(S) \rightarrow \observables_{\rm i}
$ that extracts the  observables from the executions:
\[t_{\rm i}\left(q_0\stackrel{\ell_1,\alpha_1}{\longrightarrow}\cdots \stackrel{\ell_{n},\alpha_n}{\longrightarrow} q_{n+\!1}\right)\eqdef(\ell_1,\sieve(\alpha_1)) \cdots (\ell_{n},\sieve(\alpha_{n})).\]
Note that in this definition we could have equivalently used $\hide$ instead than $\sieve$.

\medskip\noindent\emph{Adversaries in collusion with the
scheduler:} Finally, we consider the case in which the adversary
is in collusion with the scheduler, or possibly the adversary
\emph{is} the scheduler. \rev{To illustrate the difference
between this kind of adversaries and internal adversaries,
consider the scheduler of an operating system. In such scenario an
internal adversary is able to see which process has been scheduled
to run next (process in the ``running state'') whereas an
adversary in collusion with the scheduler can see as much as the
scheduler, thus being able to see (in addition) which processes
are in the ``ready state'' and which processes are in the
``waiting / blocked'' state. We will show later that such
additional information does not help the adversary to leak
information (see Proposition~\ref{prop:internal-adversaries}). The
observation domain of adversaries in collusion with the scheduler
coincides with the one of the scheduler:}
\[
\observables_{\rm s} \eqdef ({\cal P}(L) \times L \times (\obse \cup \{\tau\}))^*.
\]
The corresponding function
\[
t_{\rm s} : \fpaths(S) \rightarrow \observables_{\rm s}
\]
is defined as the one of the scheduler, i.e. $t_{\rm s}= t$.

\subsection{Information leakage}

In Information Flow and Anonymity there is a converging consensus
for formalizing the notion of leakage as the difference or the
ratio between the a priori uncertainty that the adversary has
about the secret, and the a posteriori uncertainty, that is, the
residual uncertainty of the adversary once it has seen the outcome
of the computation. The uncertainty can be measured in different
ways. One popular approach is the information-theoretic one,
according to which  the system is seen as a noisy channel between
the secret inputs and the observable output, and uncertainty
corresponds to the Shannon entropy of the system (see
preliminaries \rev{-- Section 4.2}). In this approach, the leakage
is represented by the so-called mutual information, which
expresses the correlation between the input and the output.

The above approach, however, has been recently criticized by Smith
\cite{Smith:09:FOSSACS}, who has argued that Shannon entropy is
not suitable to represent the security threats in the typical case
in which the adversary is interested in figuring out the secret in
one-try attempt, 
\\ \ }  \commentcat{I reformulated it a bit} and he has proposed
to use R\'enyi's min entropy instead, or equivalently, the average
probability of succeeding. This leads to interpret the uncertainty
in terms of the notion of \emph{vulnerability} defined in the
preliminaries \rev{(Section 4.2)}. The corresponding notion of
leakage, in the pure probabilistic case, has been investigated in
\cite{Smith:09:FOSSACS} (multiplicative case) and in
\cite{Braun:09:MFPS} (additive case).

Here we adopt the vulnerability-based approach to define the
notion of leakage in our probabilistic and nondeterministic
context. The Shannon-entropy-based approach could be extended to
our context as well, because in both cases we only need to specify
how to determine the conditional probabilities which constitute
the channel matrix, and the marginal probabilities that constitute
the input and the output distribution.

We will denote by $S$ the random variable associated to the set of
secrets $\secrets = \secr$, and by $O_{{\rm x}}$ the random variables
associated to the set of observables $\observables_{\rm x}$, where
${\rm x}\in\{{\rm e},{\rm i}, {\rm s}\}$. So,  $\observables_{\rm x}$
represents the observation domains for the various kinds of
adversaries defined above.

As mentioned before, our results require some structural
properties for the system: we assume that there is a single
component in the system containing a secret choice and this
component contains a single secret choice. This hypothesis is
general enough to allow expressing protocols like the Dining
Cryptographers, Crowds, voting protocols, etc., where the secret
is chosen only once.



\begin{assumption}\label{ass:oneprobchoice}
A system contains exactly one component with a syntactic occurrence of a
secret choice, and such a choice does not occur in the scope of a
recursive call.
\end{assumption}

Note that the assumption implies that the choice appears exactly once in the operational semantics of the component.
It would be possible to relax the assumption and allow more than one secret choice in a component, as long as there are no observable actions between the secret choices. For the sake of simplicity in this paper we impose the more restrictive requirement.
As a consequence, we have that the operational semantics of systems satisfies the following property:
\begin{proposition}\label{prop:uniquesecrets}
If $q\stackrel{\ell}\rightarrow \mu$ and $q'\stackrel{\ell'}\rightarrow \mu^\prime$ are both secret choices, then $\ell = \ell'$ and  there exist $p_i$'s, $q_i$'s and $q'_i$'s  such that: \[\mu=\probsum{i} \ p_i \cdot \delta{(s_i,q_i)}\quad \mbox{and}\quad \mu^\prime=\probsum{i} \ p_i \cdot \delta{(s_i,q_i^\prime)}\]
    \noindent i.e., $\mu$ and $\mu^\prime$ differ only for the continuation
    states.
\end{proposition}
\begin{proof}
Because of Assumption \ref{ass:oneprobchoice}, there is only one component that can generate a secret choice, and it generates only one such choice. Due to the different possible interleavings, this choice can appear as an outgoing transition in more than one state of the TPA, but the probabilities are always the same, because the interleaving rule does not change them.
\end{proof}

Given a system, each scheduler $\zeta$ determines a fully probabilistic automaton, and, as a consequence, the probabilities
\[\meas{\zeta}{s,o} \eqdef \meas{\zeta}{\bigcup\left\{\cyl{\sigma}\mid \begin{array}[t]{l}\sigma\in \fpaths(S),  t_{\rm x} (\sigma) = o, {\it secr}(\sigma)=s\end{array}\right\}}
\]
\noindent for each secret $s\in\secrets$ and observable $o\in \observables_{\rm x}$, where ${\rm x}\in\{{\rm e},{\rm i}, {\rm s}\}$. Here ${\it secr}$ is the map from paths to their secret action. From these we can derive,  in standard ways, the marginal probabilities $\meas{\zeta}{s}$, $\meas{\zeta}{o}$,  and the conditional probabilities $\meas{\zeta}{o\mid s}$.


%

Every scheduler leads to a (generally different) noisy channel, whose matrix is determined by the conditional probabilities as follows:
\begin{definition}
Let $x\in\{{\rm e},{\rm i}, {\rm s}\} $.
Given a system and a  scheduler $\zeta$, the  corresponding channel matrix $\nc^{\rm x}_\zeta$ has rows indexed by $s \in \secrets$ and columns indexed by $o\in\observables_{\rm x}$. The value in $(s,o)$ is given by
\[\meas{\zeta}{o\mid s} \eqdef \frac{\meas{\zeta}{s,o}}{\meas{\zeta}{s}}. \] 
\end{definition}

\noindent Given a scheduler $\zeta$, the multiplicative leakage can  be
defined as $\leakm(\nc^{\rm x}_\zeta,P_\zeta)$, while the additive
leakage can be defined as $\leaka(\nc^{\rm x}_\zeta,P_\zeta)$
where $P_\zeta$ is the a priori distribution on the set of secrets
(see preliminaries, \rev{Section 4.2}). However, we want a notion of
leakage independent from the scheduler, and therefore it is
natural to consider the worst case over all possible admissible
schedulers.

\begin{definition}[${\rm x}$-leakage]
Let ${\rm x}\in\{{\rm e},{\rm i}, {\rm s}\} $. Given a system,
the multiplicative leakage is defined as \[{\cal M}\leakm^{\rm x}\eqdef\max_{\zeta\in{\it Adm}}\leakm(\nc^{\rm x}_\zeta,P_\zeta),\] while the additive leakage is defined as
\[{\cal M}\leaka^{\rm x}\eqdef\max_{\zeta\in{\it Adm}}\leaka(\nc^{\rm x}_\zeta,P_\zeta),\]
where {\it Adm} is the class of admissible schedulers defined in the previous section.
\end{definition}

We have that the  classes of observables {\rm e}, {\rm i}, and  {\rm s} determine an increasing degree of leakage:

\begin{proposition}\label{prop:leakage}
Given a system,
for  the multiplicative leakage we have
\begin{enumerate}
\item For every scheduler $\zeta$,  $\leakm(\nc^{\rm e}_\zeta,P_\zeta) \leq \leakm(\nc^{\rm i}_\zeta,P_\zeta) \leq \leakm(\nc^{\rm s}_\zeta,P_\zeta)$
\item ${\cal M}\leakm^{\rm e} \leq {\cal M}\leakm ^{\rm i} \leq{\cal M}\leakm ^{\rm s}$
\end{enumerate}
Similarly for the additive leakage.
\end{proposition}
\begin{proof}\ \ \ \ \
\begin{enumerate}
\item The property follows immediately from the fact that the domain $\observables_{\rm e}$ is an abstraction of $\observables_{\rm i}$, and $\observables_{\rm i}$ is an abstraction of
$\observables_{\rm s}$.
\item Immediate from previous point and from the definition of ${\cal M}\leakm^{\rm x}$ and ${\cal M}\leaka^{\rm x}$. \qedhere
\end{enumerate}
\end{proof}

\subsection{Strong anonymity (revised)}
We consider now the situation in which the  leakage is the minimum for all possible admissible schedules.
In the purely probabilistic case, we know that the minimum possible multiplicative leakage is $1$, and the minimum possible additive one is $0$.
We also know that  this is the case for all possible input distributions if and only if the capacity of the channel matrix is $0$, which corresponds to the case in which
the rows of the matrix are all the same. This corresponds to the notion of strong probabilistic anonymity defined in \cite{Bhargava:05:CONCUR}. In the framework of information flow,
it would correspond to probabilistic non-interference.
Still in \cite{Bhargava:05:CONCUR},
the authors considered also the extension of this notion in presence of nondeterminism, and required the condition to hold under all possible schedulers.
This is too strong in practice, as we have argued in the introduction: in most cases we can build a scheduler that leaks the secret by changing the interleaving order.
 We therefore tune this notion by requiring the condition to hold only under the admissible schedulers.

\begin{definition}[$x$-strongly anonymous]
Let $x\in\{{\rm e},{\rm i}, {\rm s}\} $.
We say that a system is $x$-\emph{strongly-anonymous}
if for all admissible schedulers $\zeta$ we have
$$\meas{\zeta}{o\mid s_1}=\meas{\zeta}{o\mid s_2}$$
\noindent for all $s_1,s_2\in\secr$, and $o\in \observables_{\rm x}$.
\end{definition}

The following corollary is an immediate consequence of Proposition~\ref{prop:leakage}.

\begin{corollary}\label{cor:x-s-a}$\empty$
 \begin{enumerate}
\item If a system is s-strongly-anonymous, then it is also i-strongly-anonymous.
\item If a system is i-strongly-anonymous, then it is also e-strongly-anonymous.
\end{enumerate}
\end{corollary}

The converse of point $(2)$, in the previous corollary, does not hold, as shown by the following example:

\begin{example}
Consider the system $S \eqdef (\{c_1,c_2\})\ P\, ||\, Q\, ||\, T$ where
\[P\eqdef(0.5 : s_1 \,.\, \overline{c}_1) + (0.5 : s_2 \,.\, \overline{c}_2)   \quad Q \eqdef c_1 \,.\, o   \quad T \eqdef c_2 \,.\, o\]

It is easy to check that $S$ is $e$-strongly anonymous but not $i$-strongly anonymous, showing that (as expected) internal adversaries can ``distinguish more'' than external adversaries.
\end{example}

On the contrary, for point $(1)$ of Corollary~\ref{cor:x-s-a},
also the other direction holds:

\begin{proposition}\label{prop:internal-adversaries}
A system is s-strongly-anonymous if and only if it is
i-strongly-anonymous.
\end{proposition}
\begin{proof}
Corollary~\ref{cor:x-s-a} ensures the only-if part. For the if part, we proceed by contradiction. Assume that
the system is i-strongly-anonymous but that
$\meas{\zeta}{o\mid s_1}\neq \meas{\zeta}{o\mid s_2}$ for some admissible
 scheduler $\zeta$ and
observable $o\in \observables_{\rm s}$.
Let $o = (\avail(\qi),\ell_1,\hide(\alpha_1)) \cdots (\avail(q_n),\ell_n,\hide(\alpha_n))$ and let $o'$ be the
projection of $o$ on $\observables_{\rm i}$, i.e.  $o' = (\ell_1,\hide(\alpha_1)) \cdots (\ell_n,\hide(\alpha_n))$.
Since the system is i-strongly-anonymous, $\meas{\zeta}{o'\mid s_1}= \meas{\zeta}{o'\mid s_2}$,
which means that the difference in  probability with respect to $o$ must be due to at least one of the
sets of enabled processes. Let us consider the first set $L$ in $o$ which exhibits a difference in the probabilities,
and let $o''$ be the prefix of $o$ up to the tuple containing $L$. Since the probabilities are determined by the distributions on the probabilistic choices which occur in the individual components,
the probability of each $\ell\in L$ to be available (given the trace  $o''$) is independent of the other labels in $L$.
At least one such $\ell$ must therefore have a different probability, given the trace $o''$, depending on whether the secret choice was $s_1$ or $s_2$. And, because of the assumption on $L$, we can replace the conditioning on trace $o''$ with the conditioning on the projection $o'''$ of $o''$ on $\observables_{\rm i}$.
Consider now an admissible scheduler $\zeta'$ that acts like $\zeta$ up to $o''$, and then selects $\ell$ if and only if it is available.
Since the probability that $\ell$ is not available depends on the choice of $s_1$ or $s_2$, we have
$\meas{\zeta}{o'''\mid s_1}\neq \meas{\zeta}{o'''\mid s_2}$, which contradicts the hypothesis that the system is i-strongly-anonymous.
\end{proof}

\rev{Intuitively, this result means that an $s$-adversary can leak
information if and only if an $i$-adversary can leak information
or, in other words, $s$-adversaries are as powerful as
$i$-adversaries (even when the former can observe more
information).}


\section{Verifying strong anonymity: a proof technique based on automorphisms}
\sectionmark{Verifying strong anonymity}

As mentioned in the introduction, several problems involving
restricted schedulers have been shown undecidable (including
computing maximum / minimum probabilities for the case of standard
model checking~\cite{Giro:07:FORMATS,Giro:09:SBMF}). These
results are discouraging in the aim to find algorithms for
verifying strong anonymity/non-interference using our notion of
admissible schedulers (and most definitions based on restricted
schedulers). Despite the fact that the problem seems to be
undecidable in general, in this section we present a sufficient
(but not necessary) anonymity proof technique: we show that the
existence of automorphisms between each pair of secrets implies
strong anonymity. \rev{We conclude this section illustrating the
applicability of our proof technique by means of the DC
protocol, i.e., we prove that the protocol does not leak
information by constructing automorphisms between pairs of
cryptographers. It is worth mentioning that our proof technique
is general enough to be used for the analysis of information leakage 
of a broad family of protocols, namely any protocol
that can be modeled in our framework.}

\subsection{The proof technique}

In practice proving anonymity often happens in the following way.
Given a trace in which user $A$ is the `culprit', we construct an
observationally equivalent trace in which user $B$ is the
`culprit'~\cite{Halpern:05:JCS,Garcia:05:FMSE,Mauw:04:ESORICS,Hasuo:07:ESOP}.
This new trace is typically obtained by `switching' the behavior
of users $A$ and $B$. We formalize this idea by using the notion
of automorphism, cf. e.g.~\cite{Rutten:00:TCS}.

\begin{definition}[Automorphism] Given a TPA $(Q,L,\Sigma,\qi,\theta)$ we say that a bijection $f:Q\rightarrow Q$ is an \emph{automorphism} if it satisfies $f(\qi)=\qi$ and
\[  q\stackrel{\ell}{\rightarrow} \probsum{i} p_i \cdot \delta{(\alpha_i,q_i)}\Longleftrightarrow f(q)\stackrel{\ell}{\rightarrow}\probsum{i} p_i \cdot \delta{(\alpha_i,f(q_i))}.\]
\end{definition}


%

In order to prove anonymity it is sufficient to prove that the
behaviors of any two 'culprits' can be exchanged without the
adversary noticing. We will express this
by means of the existence of automorphisms that exchange a given
pair of secret $s_i$ and $s_j$.

Before presenting the main theorem of this section we need to
introduce one last definition. Let $S=(C)\ q_1 \allowbreak||
\cdots ||\allowbreak \, q_n$ be a system and $M$ its corresponding
TPA.  We define  $M_\tau$ as the automaton obtained after
``hiding'' all the secret actions of $M$. The idea is to replace
every occurrence of a secret $s$ in $M$ by the silent action
$\tau$. Note that this can be formalized by replacing the secret
choice by a blind choice in the corresponding component $q_i$ of
the system $S$.

%
%


\rev{We now formalize the relation between automorphisms and
strong anonymity. We will first show that the existence of
automorphisms exchanging pairs of secrets implies $s$-strong
anonymity (Theorem \ref{theo:automorphism}). Then, we will show
that the converse does not hold, i.e. $s$-strongly-anonymous
systems are not necessarily automorphic (Example
\ref{exa:no-converse}).}

%
%

%

\begin{theorem}\label{theo:automorphism}
Let $S$ be a system satisfying Assumption~\ref{ass:oneprobchoice} and $M$ its tagged probabilistic automaton. If for every pair of secrets $s_i, s_j \in \secr$ there exists an automorphism $f$ of $M_\tau$ such that for any state $q$ we have
\begin{equation}\label{eq:automorphism} 
q\stackrel{\ell,s_i}{\longrightarrow}_M q^\prime \Longrightarrow
f(q)\stackrel{\ell,s_j}{\longrightarrow}_M f(q^\prime),
\end{equation}

\noindent  then $S$ is $s$-strongly-anonymous.
\end{theorem}
\begin{proof} Assume that for every pair of secrets $s_i$, $s_j$ we have an automorphism $f$ satisfying the hypothesis of the theorem. We have to show that, for every admissible scheduler $\zeta$ we have:
\[\forall\, o\!\in\!\observables_s:\ \meas{\zeta}{o\mid s_1}=\meas{\zeta}{o\mid s_2}.\]

We start by observing that for $s_i$, by Proposition~\ref{prop:uniquesecrets},
there exists a unique $p_i$ such that, for all transitions
\smash{$q\stackrel{l}{\rightarrow}\mu$}, if $\mu$ is a (probabilistic) secret
choice, then $\mu(s_i,-)=p_i$. Similarly for $s_j$, there exists a unique
$p_j$ such that $\mu(s_j,-)=p_j$ for all secret choices $\mu$.

Let us now recall the definition of $\meas{\zeta}{o\mid s}$:
\[\meas{\zeta}{o\mid s}\eqdef\frac{\meas{\zeta}{o\land
s}}{\meas{\zeta}{s}} \] \noindent where
$\meas{\zeta}{o \land s}\eqdef \meas{\zeta}{\{ \pi\!\in\!\cpaths \mid
t_s(\pi)\!=\!o \land {\it secr}(\pi)=s\}}$
\noindent with ${\it secr}(\pi)$ being the (either empty or singleton) sequence of
secret actions of $\pi$, and $\meas{\zeta}{s}\eqdef \meas{\zeta}{\{
\pi\!\in\!\cpaths \mid {\it secr}(\pi)=s\}}.$

\noindent Note that, since a secret appears at most once on a complete path, we have:
\begin{eqnarray*}
\meas{\zeta}{s_i}&=& \meas{\zeta}{\{\pi \stackrel{\ell,s_i}{\longrightarrow} \sigma\in \cpaths \mid \pi,\sigma\}}\\
&=& \hspace{-1.6cm}\sum_{\hspace{1.5cm}\pi \stackrel{\ell,s_i}{\longrightarrow} q_i\in \fpaths}\hspace{-1.25cm}\meas{\zeta}{\pi \stackrel{\ell,s_i}{\longrightarrow} q_i}= \hspace{-1.2cm}\sum_{\hspace{1.15cm}\begin{subarray}{l}\last{\pi} \stackrel{\ell}{\rightarrow} \mu \\ \mu\text{ secret choice} \end{subarray}}\hspace{-1cm}\meas{\zeta}{\pi}\cdot p_i\\
\end{eqnarray*}
\noindent and analogously
\begin{eqnarray*}
\meas{\zeta}{s_j}&=& \meas{\zeta}{\{\pi \stackrel{\ell,s_j}{\longrightarrow} \sigma\in \cpaths\mid \pi,\sigma\}}\\
&=& \hspace{-1.6cm}\sum_{\hspace{1.5cm}\pi \stackrel{\ell,s_j}{\longrightarrow} q_j\in \fpaths}\hspace{-1.25cm}\meas{\zeta}{\pi \stackrel{\ell,s_j}{\longrightarrow} q_j}= \hspace{-1.2cm}\sum_{\hspace{1.15cm}\begin{subarray}{l}\last{\pi} \stackrel{\ell}{\rightarrow} \mu \\ \mu\text{ secret choice} \end{subarray}}\hspace{-1cm}\meas{\zeta}{\pi}\cdot p_j\\
\end{eqnarray*}
\noindent Let us now consider $\meas{\zeta}{o\mid s_i}$ and
$\meas{\zeta}{o\mid s_j}$. We have:
\begin{eqnarray*}
\meas{\zeta}{o \land s_i}&=& \hspace{-0.2cm}\meas{\zeta}{\left
\lbrace\pi\stackrel{\ell,s_i}{\longrightarrow} \sigma\in \cpaths \mid
t_s(\pi \stackrel{\ell,s_i}{\longrightarrow} \sigma)=o\right\rbrace}\\
&=& \hspace{-1cm}\sum_{\hspace{0.5cm}\begin{subarray}{l}\hspace{0.5cm}\pi\\
\last{\pi}\stackrel{\ell}{\rightarrow}\mu \\ \mu\text{ secret
choice}\end{subarray}}\hspace{-0.7cm}\meas{\zeta}{\pi}\cdot p_i \cdot
\hspace{-2.5cm}\sum_{\hspace{2.25cm}\begin{subarray}{l}\hspace{0.6cm}\sigma\\
\pi\stackrel{\ell,s_i}{\longrightarrow}\,\sigma\in
\fpaths \\t_s(\pi\stackrel{\ell,s_i}{\longrightarrow}\,\sigma)=o\land
\last{t_e(\sigma)}\not=\tau \end{subarray}}\hspace{-2.5cm}\meas{\zeta}{\sigma} \\
\end{eqnarray*}
again using that a secret appears at most once on a complete path. Moreover,
note that we have overloaded the notation $\mathbb{P}_{\zeta}$ by using it for
different measures when writing $\meas{\zeta}{\sigma}$, since $\sigma$ need
not start in the initial state $\hat q$.

\noindent
Analogously we have:
\begin{eqnarray*}
\meas{\zeta}{o \land s_j}&=& \hspace{-0.2cm}\meas{\zeta}{\left \lbrace\pi\stackrel{\ell,s_j}{\longrightarrow} \sigma\in \cpaths \mid  t_s(\pi \stackrel{\ell,s_j}{\longrightarrow} \sigma)=o\right\rbrace}\\
&=& \hspace{-1cm}\sum_{\hspace{0.5cm}\begin{subarray}{l}\hspace{0.5cm}\pi\\ \last{\pi}\stackrel{\ell}{\rightarrow}\mu \\ \mu\text{ secret choice}\end{subarray}}\hspace{-0.7cm}\meas{\zeta}{\pi}\cdot p_j \cdot \hspace{-2.5cm}\sum_{\hspace{2.25cm}\begin{subarray}{l}\hspace{0.6cm}\sigma\\ \pi\stackrel{\ell,s_j}{\longrightarrow}\,\sigma\in \fpaths\\ t_s(\pi\stackrel{\ell,s_j}{\longrightarrow}\,\sigma)=o\land \last{t_e(\sigma)}\not=\tau \end{subarray}}\hspace{-2.5cm}\meas{\zeta}{\sigma} \\
\end{eqnarray*}

\vspace{-0.5cm}
\noindent Therefore, we derive
\begin{equation}\label{eq:cond-i}
\meas{\zeta}{o \mid s_i}=\frac{\hspace{-0.25cm}{\displaystyle\sum_{\hspace{0.5cm}\begin{subarray}{l}\hspace{0.5cm}\pi\\ \last{\pi}\stackrel{\ell}{\rightarrow}\mu \\ \mu\text{ secret choice}\end{subarray}}\hspace{-2cm}\sum_{\hspace{2.25cm}\begin{subarray}{l}\hspace{0.6cm}\sigma\\ \pi\stackrel{\ell,s_i}{\longrightarrow}\,\sigma\in \fpaths\\ t_s(\pi\stackrel{\ell,s_i}{\longrightarrow}\,\sigma)=o\land \last{t_e(\sigma)}\not=\tau \end{subarray}}\hspace{-2.5cm}\meas{\zeta}{\pi}  \cdot \meas{\zeta}{\sigma}}}
{{\displaystyle\hspace{-1cm}\sum_{\hspace{1.15cm}\begin{subarray}{l}\last{\pi} \stackrel{\ell}{\rightarrow} \mu \\ \mu\text{ secret choice} \end{subarray}}\hspace{-1cm}\meas{\zeta}{\pi}}}
\end{equation}

\begin{equation}\label{eq:cond-j}
\meas{\zeta}{o \mid s_j}=\frac{\hspace{-0.25cm}{\displaystyle\sum_{\hspace{0.5cm}\begin{subarray}{l}\hspace{0.5cm}\pi\\ \last{\pi}\stackrel{\ell}{\rightarrow}\mu \\ \mu\text{ secret choice}\end{subarray}}\hspace{-2cm}\sum_{\hspace{2.25cm}\begin{subarray}{l}\hspace{0.6cm}\sigma\\ \pi\stackrel{\ell,s_j}{\longrightarrow}\,\sigma\in \fpaths\\ t_s(\pi\stackrel{\ell,s_j}{\longrightarrow}\,\sigma)=o\land \last{t_e(\sigma)}\not=\tau \end{subarray}}\hspace{-2.5cm}\meas{\zeta}{\pi}  \cdot \meas{\zeta}{\sigma}}}
{{\displaystyle\hspace{-1cm}\sum_{\hspace{1.15cm}\begin{subarray}{l}\last{\pi} \stackrel{\ell}{\rightarrow} \mu \\ \mu\text{ secret choice} \end{subarray}}\hspace{-1cm}\meas{\zeta}{\pi}}}
\end{equation}

\noindent Observe that the denominators of both formulae (\ref{eq:cond-i}) and
(\ref{eq:cond-j}) are the same. Also note that, since $f$ is an automorphism,
for every path
$\pi$, $f(\pi)$ obtained by replacing each state in $\pi$ with its image under
$f$ is also a path. Moreover, since $f$ satisfies (\ref{eq:automorphism}), for
every path
$\smash{\pi\stackrel{\ell,s_i}{\longrightarrow}\sigma}$ we have that
$\smash{f(\pi)\stackrel{\ell,s_j}{\longrightarrow}f(\sigma)}$ is also a path.
Furthermore $f$ induces a bijection between the sets
\[
\{(\pi,\sigma)\mid \begin{array}[t]{l}
\last{\pi}\stackrel{\ell'}{\rightarrow}\mu\text{ s.t. }\mu \text{ secret choice},\pi\stackrel{\ell,s_i}{\longrightarrow}\sigma \in\fpaths\\
t_s(\pi\stackrel{\ell,s_i}{\longrightarrow}\sigma )=o,\last{t_e(\sigma)}\not=\tau\ \ \ \} \text{, and}
\end{array}
\]
\vspace{-0.15cm}
\[
\{(\pi,\sigma)\mid \begin{array}[t]{l}
\last{\pi}\stackrel{\ell'}{\rightarrow}\mu\text{ s.t. }\mu \text{ secret choice},\pi\stackrel{\ell,s_j}{\longrightarrow}\sigma \in\fpaths\\
t_s(\pi\stackrel{\ell,s_j}{\longrightarrow}\sigma )=o,\last{t_e(\sigma)}\not=\tau\ \ \ \}
\end{array}
\]

\noindent given by $(\pi,\sigma)\leftrightarrow(f(\pi),f(\sigma))$.

Finally, since $\zeta$ is admissible, $t_s(\pi)=t_s(f(\pi))$, and $f$ is an
automorphism, it is easy to prove by induction that
$\meas{\zeta}{\pi}=\meas{\zeta}{f(\pi)}$. Similarly,
$\meas{\zeta}{\sigma}=\meas{\zeta}{f(\sigma)}$.
Hence the numerators of (\ref{eq:cond-i}) and (\ref{eq:cond-j}) coincide which
concludes the proof.
\end{proof}

Note that, since ${\rm s}$-strong anonymity implies ${\rm i}$-strong anonymity
and ${\rm e}$-strong anonymity, the existence of such an automorphism
implies all the notions of strong anonymity presented in this
work. \rev{We now proceed to show that the converse does not hold,
i.e. strongly anonymous systems are not necessarily automorphic.}

\begin{example}\label{exa:no-converse}
Consider the following (single component) system
\[\begin{array}{c}
\hspace{-0.5cm}0.5: s_1 . (0.5 : (p:a + (1\!-\!p):b) + 0.5 : ((1\!-\!p):a +
p:b))\\ + \\
\hspace{-0.5cm}0.5: s_2 . (0.5 : (q:a + (1\!-\!q):b) + 0.5 : ((1\!-\!q):a +
q:b))\\ \end{array}\]
\noindent It is easy to see that such system is $s$-strongly-anonymous, however if $p\not=q$
and $p\not=1-q$ there does not exist an automorphism for the pair of secrets
$(s_1,s_2)$.
\end{example}

\rev{The following example demonstrates that our proof technique
does not carry over to systems whose components admit internal
parallelism.}
\begin{example}\label{exa:internal-parallelism}  Consider  $S \eqdef (\{c_1,c_2\}) \ r\!
\parallel\! q \!\parallel\! t$, where
\[\begin{array}{c}r \eqdef 0.5 : s_1 . \overline{c}_1 + 0.5 : s_2 . \overline{c}_2, \qquad
q \eqdef c_1 . (a\, |\, b),  \qquad  t \eqdef c_2 . (a\, |\, b).
\end{array} \]

\noindent \rev{where $q_1|q_2$ represents the parallel composition of
$q_1$ and $q_2$. It is easy to show that there exists an
automorphism for $s_1$ and $s_2$. However, admissible schedulers
are able to leak such secrets. This is due to the fact that
component $r$ synchronizes with $q$ and $t$ on different channels,
thus a scheduler of $S$ is not restricted to select the same
transitions on the branches associated to $s_1$ and $s_2$
(remember that schedulers can observe synchronization).}
\end{example}

We now show that the definition of $x$-strong-anonymity is
independent of the particular distribution over secrets, i.e., if
a system is $x$-strongly-anonymous for a particular distribution
over secrets, then it is $x$-strongly-anonymous for all
distributions over secrets. \rev{This result is useful because it
allows us to prove systems to be strongly anonymous even when
their distribution over secrets is not known.}

\begin{theorem} Consider a system $S = (C)\ q_1 \parallel \cdots \parallel  q_i \parallel \cdots \parallel q_n$. Let $q_i$ be the component which contains the secret choice, and assume that it is of the form $\sum_j p_j:s_j\, .\, q_j$. Consider now the system $S' = (C)\ q_1 \parallel \cdots \parallel q'_i \parallel \cdots \parallel q_n$, where $q'_i$ is identical to $q_i$ except for the secret choice, which is replaced by $\sum_j p'_j: s_j\, .\, q_j$. Then we have that:
\begin{enumerate}
\item For every $s_i$, $s_j$ there is an automorphism on $S$ satisfying the assumption of Theorem~\ref{theo:automorphism} if and only if the same holds for $S'$.
\item $S$ is $x$-strongly-anonymous if and only if $S'$ is $x$-strongly-anonymous.
\end{enumerate}
\end{theorem}

\noindent Note: $1)$ does not imply $2)$, because in principle neither $S$ not $S'$ may have the automorphism, and still one of the two could be strongly anonymous.
\begin{proof}
We note that the PAs generated by $S$ and $S'$ coincide except for the probability distribution on the secret choices.
Since the definition of automorphism and the assumption of Theorem ~\ref{theo:automorphism} do not depend on these probability distributions, $(1)$ is immediate.
As for $(2)$, we observe that $x$-strong anonymity only depends on the conditional probabilities $ \meas{\zeta}{o\mid s}$.
By looking at the proof of  Theorem ~\ref{theo:automorphism}, we can see that in the computation of $ \meas{\zeta}{o\mid s}$
the probabilities on the secret choices (i.e. the $p_j$'s) are eliminated. Namely $ \meas{\zeta}{o\mid s}$ does not depend on the  $p_j$'s, which means that
the value of the $p_j$'s has no influence on whether the system is $x$-strong anonymous or not.
\end{proof}


\subsection{An Application: Dining Cryptographers}

Now we show how to apply the proof technique presented in this section to the Dining Cryptographers protocol. Concretely, we show that there exists an automorphism $f$ exchanging the behavior of the Crypt${}_0$ and Crypt${}_1$; by symmetry, the same holds for the other two combinations.

Consider the automorphisms of Master and Coin${}_1$ indicated in Figure \ref{fig:DCautomorphism}. The states that are not explicitly mapped (by a dotted arrow) are mapped to themselves.

\begin{figure}[!h]
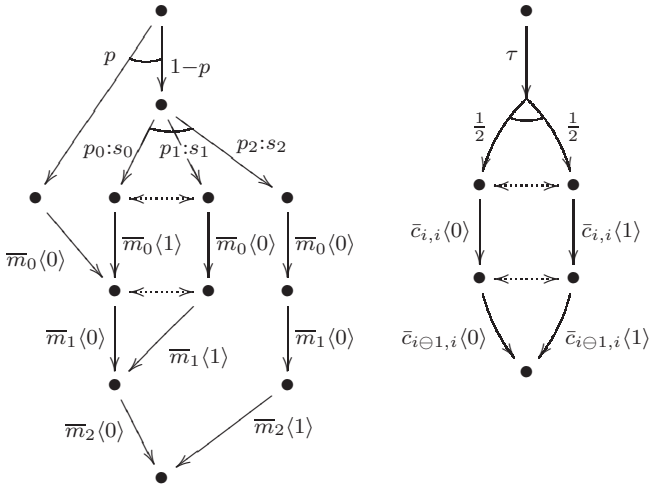

\vspace{-1.6cm}
\begin{equation*}
\input{AdmissibleSchedulers/mastermap-new.tex}
\quad
\input{AdmissibleSchedulers/coinmap-new.tex}
\end{equation*}\caption{Automorphism between Crypt${}_0$ and Crypt${}_1$}\label{fig:DCautomorphism}
\vspace{-.4cm}
\end{figure}
Also consider the identity automorphism on Crypt${}_i$ (for $i =
0, 1, 2$) and on Coin${}_i$ (for $i = 0, 2$). It is easy to check
that the product of these seven automorphisms is an automorphism
for Crypt${}_0$ and Crypt${}_1$.

\section{Related Work}

The problem of the full-information scheduler has already been extensively investigated  in literature. The works \cite{Canetti:06:WODES} and \cite{Canetti:06:DISC}
consider probabilistic automata and introduce a restriction on the scheduler to the purpose of making them suitable to applications in security. Their approach is based on dividing the actions of each component of the system in equivalence classes (\emph{tasks}). The order of execution of different tasks is decided in advance by a so-called \emph{task scheduler}. The remaining nondeterminism within a task is resolved by a second scheduler, which models the standard \emph{adversarial scheduler} of the cryptographic community. This second entity has limited knowledge about the other components: it sees only the information that they communicate during execution. Their notion of task scheduler is similar to our notion of admissible scheduler, but more restricted since the strategy of the task scheduler is decided entirely before the execution of the system.

Another work along these lines is \cite{deAlfaro:01:CONCUR}, which uses partitions on the state-space to obtain partial-information schedulers. However that work considers a synchronous parallel composition, so the setting is rather different from ours.

The works in \cite{Chatzikokolakis:10:IC,Chatzikokolakis:09:FOSSACS} are similar to ours in spirit, but in a sense  \emph{dual} from a technical point of view. Instead of defining a restriction on the class of schedulers, they provide a way to specify that a choice is transparent to the scheduler. They achieve this by introducing labels in process terms, used to represent both the states of the execution tree and the next action or step to be scheduled. They make two states indistinguishable to schedulers, and hence the choice between them private, by associating to them the same label. Furthermore, their ``equivalence classes'' (schedulable actions with the same label) can change dynamically, because the same action can be associated to different labels during the execution.

In~\cite{Alvim:10:IFIP-TCS} we have extended the framework presented in this work (by allowing internal nondeterminism and adding a second type of scheduler to resolve it) with the aim of investigating angelic vs demonic nondeterminism in equivalence-based properties.

The fact that full-information schedulers are unrealistic has also been observed in fields other than security. With the aim to cope with general properties (not only those concerning security), first attempts used restricted schedulers in order to obtain rules for compositional reasoning~\cite{deAlfaro:01:CONCUR}. The justification for those restricted schedulers is the same as for ours, namely, that not all information is available to all entities in the system. Later on, it was shown that model checking is undecidable in its general form for the kind of restricted schedulers presented in \cite{deAlfaro:01:CONCUR}. See~\cite{Giro:07:FORMATS} and, more recently, \cite{Giro:09:SBMF}.

Finally, to the best of our knowledge, this is the first work using automorphisms as a sound proof technique (in our case to prove strong anonymity and non-interference). The closest line of work we are aware of is in the field of model checking. There, isomorphisms can be used to identify symmetries in the system, and such symmetries can then be exploited to alleviate the state space explosion (see for instance~\cite{Kwiatkowska:06:CAV}).

\newpage
\thispagestyle{empty}

\chapter{Significant Diagnostic Counterexample Generation}
\label{ch.counterexamples}

\begin{quote}
\textit{
  In this chapter, we present a novel technique for counterexample generation
  in probabilistic model checking of Markov Chains and Markov Decision
  Processes. (Finite) paths in counterexamples are grouped together
  in witnesses that are likely to provide similar debugging information
  to the user. We list five properties that witnesses should satisfy
  in order to be useful as debugging aid: similarity, accuracy, originality,
  significance, and finiteness. Our witnesses contain paths that behave
  similarly outside strongly connected components.
  Then, we show how to compute these witnesses by reducing the
  problem of generating counterexamples for general properties
  over Markov Decision Processes, in several steps, to the easy
  problem of generating counterexamples for reachability properties
  over acyclic Markov Chains.}
\end{quote}

\section{Introduction}
\label{sec:intro}

%
%

Model checking\index{model checking} is an automated technique
that, given a finite-state model of a system and a property stated
in an appropriate logical formalism, systematically checks the
validity of this property. Model checking is a general approach
and is applied in areas like hardware verification and software
engineering.

Nowadays, the interaction geometry of distributed systems and
network protocols calls for probabilistic, or more generally,
quantitative estimates of, e.g., performance and cost measures.
Randomized algorithms are increasingly utilized to achieve high
performance at the cost of obtaining correct answers only with
high probability. For all this, there is a wide range of models
and applications in computer science requiring quantitative
analysis. Probabilistic model checking\index{model
checking!probabilistic} allows to check whether or not a
probabilistic property is satisfied in a given model, e.g., ``Is
every message sent successfully received with probability greater
or equal than $0.99$?''.

A major strength of model checking is the possibility of generating
diagnostic information in case the property is violated.
This diagnostic information is provided through a
\emph{counterexample}\index{counterexample} showing an execution
of the model that invalidates the property under verification.
Besides the immediate feedback in model checking, counterexamples
are also used in abstraction-refinement
techniques~\cite{clarke_2000_counterexampleguided}, and provide
the foundations for schedule derivation (see,
e.g.,~\cite{BLR_2005_optimalscheduling,Fehnker:02:PhDThesis}).

%
Although counterexample generation was studied from the very
beginning in most model checking techniques, this has not been the
case for probabilistic model checking. Only recently~\cite{ahl_2005_counterexamples,and06thesis,al_2006_search,hk_2007_counterexamples,hk_2007_counterexamplesDTMC,al_2007_counterexamplesMDP} attention was
drawn to this
subject,fifteen years after the first studies on probabilistic model
checking.
Contrarily to other model checking techniques, counterexamples in this
setting are \emph{not} given by a single execution path.  Instead,
they are \emph{sets of executions} of the
system satisfying a certain undesired property whose probability
mass is higher than a given bound. Since counterexamples are used
as a diagnostic tool, previous works on counterexamples have
presented them as sets of \emph{finite} paths with probability large enough. We refer to these sets as \emph{representative
counterexamples}. Elements of representative counterexamples with
high probability have been considered the most informative since
they contribute mostly to the property refutation.

A challenge in counterexample generation for probabilistic model
checking is that (1) representative counterexamples are very large
(often infinite), (2) many of its elements have very low
probability (which implies that they are very distant from the
counterexample), and (3) that elements can be extremely similar to
each other (consequently providing similar diagnostic
information). Even worse, (4) sometimes the finite paths with
highest probability do not indicate the most likely violation of
the property under consideration.

For example, look at the Markov Chain $\mc$ in
Figure~\ref{fig:MC related to MDP}.
The property $\sat{\mc}{\leq 0.5}{\F \psi}$ stating that execution
reaches a state satisfying $\psi$ (i.e., reaches $s_3$ or $s_4$)
with probability lower or equal than $0.5$ is violated (since the
probability of reaching $\psi$ is 1). The left hand side of table
in Figure~\ref{tbl:comparison} lists finite paths reaching $\psi$
ranked according to their probability. Note that finite paths with
highest probability take the left branch in the system, whereas
the right branch in itself has higher probability, illustrating
Problem 4. To adjust the model so that it does satisfy the
property (bug fixing), it is not sufficient to modify the left
hand side of the system alone; no matter how one changes the left
hand side, the probability of reaching $\psi$ remains at least
$0.6$. Furthermore, the first six finite paths provide similar
diagnostic information: they just make extra loops in $s_1$. This
is an example of Problem 3. Additionally, the probability of every
single finite path is far below the bound $0.5$, making it unclear
if a particular path is important; see Problem 2 above. Finally,
the (unique) counterexample for the property $\sat{\mc}{< 1}{\F
\psi}$ consists of infinitely many finite paths (namely all finite
paths of $\mc$); see Problem 1. \setlength{\abovecaptionskip}{2pt}
\begin{figure}
 \hspace{-0.5cm}
 \begin{minipage}[b]{5cm}
   \centering
   \includegraphics[width=4cm]{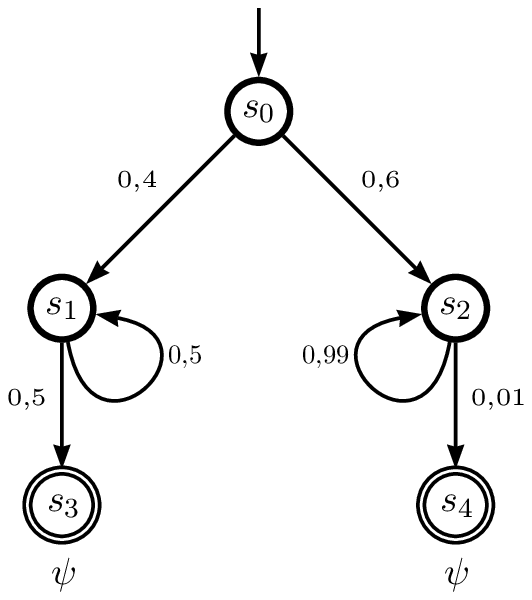}\caption{Markov Chain}\label{fig:MC related to MDP}
 \end{minipage}
 \begin{minipage}[b]{5 cm}
   \hspace{-0.5cm}
   \begin{tabular}{|c||c|c||c|c|}
   \hline \multicolumn{1}{|c||}{}&\multicolumn{2}{|c||}{\textbf{Single paths}}&\multicolumn{2}{|c|}{\textbf{Witnesses}}\\
   \hline \!\!Rank\!\!&F. Path&Prob&\!Witness\!&\!Mass\!\!\\
   \hline
   1&$s_0(s_1)^1s_3$&0.2&[$s_0 s_2 s_4$]&0.6\\
   2&$s_0(s_1)^2s_3$&0.1&[$s_0 s_1 s_3$]&0.4\\
   3&$s_0(s_1)^3s_3$&0.05&&\\
   4&$s_0(s_1)^4s_3$&0.025&&\\
   5&$s_0(s_1)^5s_3$&0.0125&&\\
   6&$s_0(s_1)^6s_3$&0.00625&&\\
   7&$s_0(s_2)^1s_4$&0.006&&\\
   8&$s_0(s_2)^2s_4$&0.0059&&\\
   9&$s_0(s_2)^3s_4$&0.0058&&\\
   \vdots&\vdots&\vdots&&\\
   \hline
   \end{tabular}
   \caption{Comparison Table}\label{tbl:comparison}
 \end{minipage}
\vspace{-0.35cm}
\end{figure}
To overcome these problems, we partition a representative
counterexample
into sets of finite paths that follow a 
similar pattern.
We call these sets \emph{witnesses}.
%
To ensure that witnesses provide valuable diagnostic information,
we desire that the set of witnesses that form a
counterexample satisfies several properties:
two different witnesses should provide different diagnostic
information (solving Problem 3) and elements of a single witness
should provide similar diagnostic information, as a consequence
witnesses have a high probability mass (solving Problems 2 and 4),
and the number of witnesses of a representative counterexample
should be finite (solving Problem 1).

In our setting, witnesses consist of paths that behave the same
outside strongly connected components. In the example of
Figure~\ref{fig:MC related to MDP}, there are two witnesses: the
set of all finite paths going right, represented by [$s_0 s_2
s_4$] whose probability (mass) is $0.6$, and the set of all finite
paths going left, represented by [$s_0 s_1 s_3$] with probability
(mass) $0.4$.





In this chapter, we show how to obtain such sets of witnesses for
bounded probabilistic LTL properties on Markov Decision Processes
($\MDP$).
In fact, we first show how to reduce this problem to finding
witnesses for upper bounded probabilistic reachability properties
on discrete time Markov Chains ($\MCs$).
The major technical matters lie on this last problem to which most
of the chapter is devoted.


In a nutshell, the process to find witnesses for the violation of
$\sat{\mc}{\leq p}{\F \psi}$, with $\mc$ being an $\MC$, is as
follows.
We first eliminate from the original $\MC$ all the
``uninteresting'' parts.  This proceeds as the first steps of the
model checking process: make absorbing all states satisfying
$\psi$, and all states that cannot reach $\psi$, obtaining a new
$\MC$ $\mc_\psi$.  Next reduce this last $\MC$ to an acyclic $\MC$
$\Acyclic{\mc_\psi}$ in which all strongly connected components
have been conveniently abstracted with a single probabilistic
transition.
The original and the acyclic $\MC$s are related by a mapping that,
to each finite path in $\Acyclic{\mc_\psi}$ (that we call \emph{rail}),
assigns a set of finite paths behaving similarly in $\mc$ (that we
call \emph{torrent}).
This map preserves the probability of reaching $\psi$ and hence
relates counterexamples in $\Acyclic{\mc_\psi}$ to counterexamples
in $\mc$.
Finally, counterexamples in $\Acyclic{\mc_\psi}$ are computed by
reducing the problem to a $k$ shortest path problem, as in
\cite{hk_2007_counterexamples}. Because $\Acyclic{\mc_\psi}$ is
acyclic, the complexity is lower than the corresponding problem in
\cite{hk_2007_counterexamples}.

It is worth mentioning that our technique can also be applied to
$\pCTL$ formulas without nested path quantifiers.


%
Looking ahead, Section~\ref{sec:preliminaries} presents the necessary background
on Markov Chains ($\MC$), Markov Decision Processes ($\MDP$), and
Linear Temporal Logic ($\LTL$).
Section~\ref{sec:counterexamples} presents the definition of
counterexamples and discusses the reduction from general $\LTL$
formulas to upper bounded probabilistic reachability properties,
and the extraction of the maximizing $\MC$ in an $\MDP$.
Section~\ref{sec:RepCount-parti-witnesses} discusses desired
properties of counterexamples.  In
Sections~\ref{sec:Torrents-Rails}
and~\ref{sec:significantdiagnosticcounterexamples} we introduce
the fundamentals on rails and torrents, the reduction of the
original $\MC$ to the acyclic one, and our notion of significant
diagnostic counterexamples.
Section~\ref{sec:coumputing-counterexamples} then presents the
techniques to actually compute counterexamples. In
Section~\ref{sec:conclusions} we discuss related work and give
final conclusions.

\section{Preliminaries}
\label{sec:preliminaries}

\rev{We now recall the notions of Markov Decision Processes, Markov Chains, and Linear Temporal Logic.}

\subsection{Markov Decision Processes \index{markov decision process}}

Markov Decision Processes ($\MDPs$) constitute a formalism that
combines nondeterministic and probabilistic choices.  They are an
important model in corporate finance, supply chain optimization,
system verification and optimization.  There are many slightly
different variants of this formalism such as action-labeled
$\MDPs$ \cite{bel_1957_mdp,FilarVrieze97}, probabilistic automata
\cite{Segala:95:NJC,sv_2004_automata}; we work with the
state-labeled $\MDPs$ from~\cite{ba_1995_probabilistic}.

\begin{dfn} Let $S$ be a finite set. A \emph{probability
distribution} on $S$ is a function $p \colon S \to [0,1]$ such
that $\sum_{s\in S}p(s)=1$. We denote the set of all probability
distributions on $S$ by $\Distr(S)$.
Additionally, we define the \emph{Dirac distribution} on an
element $s \in S$ as $1_s$, i.e., $1_s(s)=1$ and $1_s(t)=0$ for
all $t\in S\setminus\lbrace s \rbrace$.
\end{dfn}

\begin{dfn}\label{dfn:mdp}
  A \emph{Markov Decision Process} ($\MDP$) is a quadruple
  $\mdp=(S,s_0,L,\tau)$, where
\vspace{-0.25cm}
\begin{itemize}
  \item $S$ is the finite state space;
  \item $s_0 \in S$ is the initial state;
  \item $L$ is a labeling function that associates to each state $s\in
    S$ a set $L(s)$ of propositional variables that are \emph{valid} in $s$;
  \item $\tau \colon S \to \wp(\Distr(S))$ is a function that
  associates to each $s \in S$ a non-empty and finite subset of
  $\Distr(S)$ of probability distributions.
\end{itemize}
\end{dfn}

\begin{dfn} Let $\mdp=(S,s_0,\tau,L)$ be an $\MDP$. We define a \emph{successor} relation $\delta\subseteq
S\times S$ by $\delta \eqdef \{(s,t)|\exists\, \pi\!\in\!\tau(s)\! \qdot\!
\pi(t)>0\}$ and for each state $s\!\in\! S$ we define the sets
\begin{align*}
 \Paths{\mdp, s}  \eqdef \{t_0t_1t_2\ldots\in S^\omega | t_0=s\land\forall n\in \mathbb{N} \qdot
\delta(t_n,t_{n+1})\}\mbox{ and}\\
 \FPaths{\mdp, s} \eqdef \{t_0t_1\ldots t_n\in S^\star | t_0=s\land\forall\, 0\leq i < n \qdot \delta(t_n,t_{n+1})\}
\end{align*}
  of paths of $\mc$ and finite paths of $\mc$ respectively beginning at $s$. We usually omit $\mdp$ from the notation; we also abbreviate $\Paths{\mdp, s_0}$ as $\Paths{\mdp}$ and $\FPaths{\mdp, s_0}$ as $\FPaths{\mdp}$. For $\omega
  \in\Paths{s}$, we write the $(n\!+\!1)$-st state of $\omega$ as $\omega_n$.
  As usual, we let $\Borel_s\subseteq \wp(\Paths{s})$ be the Borel
  $\sigma$-algebra on the cones $\cyl{t_0 \dots t_n} \eqdef
  \{\omega \in \Paths{s} | \omega_0=t_0 \land \ldots \land
  \omega_n=t_n\}$. Additionally, for a set of finite paths $\Lambda\subseteq\FPaths{s}$, we define $\cyl{\Lambda}\eqdef
\bigcup_{\sigma\in\Lambda}\cyl{\sigma}$.
\end{dfn}
\setlength{\abovecaptionskip}{-6pt plus 1pt minus 1pt}
\begin{figure}
    \begin{center}
      \includegraphics[width=8cm]{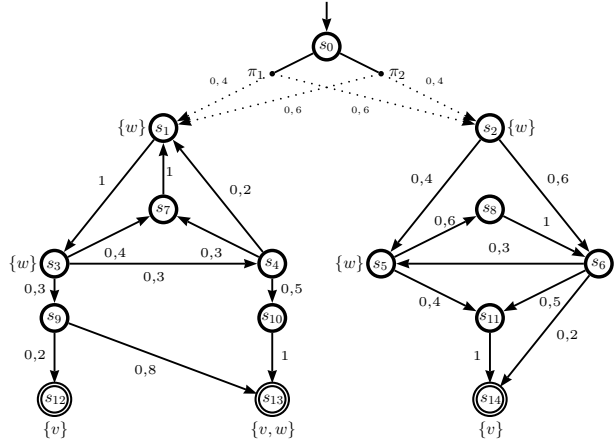}
    \end{center}
    \caption{Markov Decision Process}
    \label{fig:notMaxNodes}
\end{figure}
\label{exa:MDP} Figure~\ref{fig:notMaxNodes} shows an $\MDP$.
Absorbing states (i.e., states $s$ with $\tau(s)=\{1_s\}$) are
represented by double lines.  This $\MDP$ features a single
nondeterministic decision, to be made in state $s_0$, namely
$\pi_1$ and $\pi_2$.

\begin{dfn} Let $\mdp=(S,s_0,\tau,L)$ be an $\MDP$, $s\in S$ and $\Goal\subseteq S$. We define the sets of paths and finite paths
reaching  $\Goal$ from $s$ as
\begin{align*}
\Reach{\mdp,s,\Goal}\triangleq\lbrace \omega\in \Paths{\mdp,s}\mid
\exists_{i\geq 0}. \omega_i\in\Goal\rbrace\mbox{ and}\\
\FReach{\mdp,s,\Goal}\triangleq\lbrace\sigma\in\FPaths{\mdp,s}\mid
\last{\sigma}\in \Goal \land \forall_{i\leq
|\sigma|-1}.\sigma_i\not \in \Goal\rbrace
\end{align*}
respectively. Note that $\FReach{\mdp,s,\Goal}$ consists of those
finite paths $\sigma$ starting on $s$ reaching $\Goal$ exactly
once, at the end of the execution. It is easy to check that these
sets are \emph{prefix free}, i.e. contain finite paths such that
none of them is a prefix of another one.
\end{dfn}

\subsection{Schedulers\index{scheduler}}

Schedulers (also called strategies, adversaries, or policies)
resolve the nondeterministic choices in an
$\MDP$~\cite{pz_1993_verification,var_1985_probabilistic,ba_1995_probabilistic}.

\begin{dfn}\label{dfn:scheduler}
  Let $\mdp=(S,s_0,\tau,L)$ be an $\MDP$. A \emph{scheduler}\index{scheduler} $\eta$ on $\mdp$ is a function from $\FPaths{\mdp}$
  to $\Distr(\wp(\Distr(S)))$ such that for all $\sigma\in \FPaths{\mdp}$
  we have $\eta(\sigma)\in \Distr(\tau(\last{\sigma}))$. We denote
  the set of all schedulers on $\mdp$ by $\Sch(\mdp)$.
\end{dfn}


Note that our schedulers are randomized, i.e., in a
finite path $\sigma$ a scheduler chooses an element of
$\tau(\last{\sigma})$ probabilistically.  Under a scheduler
$\eta$, the probability that the next state reached after the path
$\sigma$ is $t$, equals $\sum_{\pi\in
  \tau(\last{\sigma})}\eta(\sigma)(\pi) \cdot \pi(t)$. In this way, a
scheduler induces a probability measure on $\Borel_s$ as usual.

\begin{dfn}
  Let $\mdp=(S,s_0,\tau,L)$ be an $\MDP$ and $\eta$ a scheduler on
  $\mdp$. We define the probability measure $\mathbb{P}_{\eta}$ as the
  unique measure on $\Borel_{s_{0}}$ such that for all $s_0 s_1\ldots s_n\in
  \FPaths{\mdp}$
\begin{align*}\measure{}{\eta}{\cyl{s_0s_1\ldots s_n}}
    = \prod_{i=0}^{n-1} \sum_{\pi\in \tau(s_i)}
    \eta(s_0 s_1\ldots s_i)(\pi) \cdot \pi(s_{i+1}).
\end{align*}
\end{dfn}
We now recall the notions of deterministic and memoryless
schedulers.

\begin{dfn}
  Let $\mdp$ be an $\MDP$ and $\eta$ a scheduler on $\mdp$. We
  say that $\eta$ is \emph{deterministic}\index{scheduler!deterministic} if $\eta(\sigma)(\pi_i)$ is
  either $0$ or $1$ for all $\pi_i\in\tau(\last{\sigma})$ and all $\sigma
  \in\FPaths{\mdp}$.
  We say that a scheduler is \emph{memoryless}\index{scheduler!memoryless} if for all finite paths $\sigma_1,\sigma_2$ of $\mdp$ with
  $\last{\sigma_1}=\last{\sigma_2}$ we have
  $\eta(\sigma_1)=\eta(\sigma_2)$.
\end{dfn}

\begin{dfn}
  Let $\mdp$ be an $\MDP$ and $\Delta \in \Borel_{s_0}$. Then the
  \emph{maximal probability} $\mathbb{P}^{+}$ and \emph{minimal probability} $\mathbb{P}^{-}$ of $\Delta$ are defined by
\begin{align*}
    \measure{+}{}{\Delta}\eqdef \sup_{\eta\in \Sch(\mdp)}
\measure{}{\eta}{\Delta}
    \hspace{0.5cm} \text{and} \hspace{0.5cm}
    \measure{-}{}{\Delta}\eqdef \inf_{\eta\in \Sch(\mdp)} \measure{}{\eta}{\Delta}.
\end{align*}
\vspace{-0.25cm}

\noindent A scheduler that attains $\measure{+}{}{\Delta}$ or
$\measure{-}{}{\Delta}$ is called a \emph{maximizing} or
\emph{minimizing} scheduler respectively.
\end{dfn}

\subsection{Markov Chains \index{markov chain}}

A (discrete time) \emph{Markov Chain} is an $\MDP$ associating
exactly one     probability distribution to each state. In this
way nondeterministic choices are no longer allowed.

\begin{dfn}[Markov Chain \index{markov chain}] Let $\mdp=(S,s_0,\tau,L)$ be an $\MDP$.
If $|\tau(s)|=1$ for all $s\in S$, then we say that $\mdp$ is a
\emph{Markov Chain} ($\MC$).
\end{dfn}

In order to simplify notation we represent probabilistic
transitions on $\MCs$ by means of a probabilistic matrix $\M$
instead of $\tau$. Additionally, we denote by $\mathbb{P}_{_{\mc,
s}}$ the probability measure induced by a $\MC$ $\mc$ with initial
state $s$ and we abbreviate $\mathbb{P}_{_{\mc, s_{0}}}$ as
$\mathbb{P}_{_{\mc}}$.

\subsection{Linear Temporal Logic\index{temporal logic!LTL}}

Linear temporal logic ($\LTL$) \cite{MP_1991_LTL} is a modal
temporal logic with modalities referring to time. In $\LTL$ is
possible to encode formulas about the future of paths: a condition
will eventually be true, a condition will be true until another
fact becomes true, etc.

\begin{dfn} $\LTL$ is built up from the set of propositional variables $\Var$, the logical connectives $\lnot$,
$\land$, and a temporal modal operator by the following grammar:
\vspace{-0.1cm}
\begin{align*}
\phi ::= \Var \mid \lnot \phi \mid \phi\land \phi \mid
\Until{\phi}{\phi}.
\end{align*}
\vspace{-0.5cm}

\noindent Using these operators we define $\lor,\rightarrow,\F{},$
and $\G{}$ in the standard way.
\end{dfn}

\begin{dfn} Let $\mdp=(S,s_0, \tau,L)$ be an $\MDP$. We define satisfiability for paths $\omega$ in $\mdp$, propositional variables $v\in\Var$, and
$\LTL$ formulas $\phi,\gamma$ inductively by
$$
\begin{array}{lclclcl}
\sat{\omega}{\mdp}{v} &\Leftrightarrow & v\in L(\omega_0)&\hspace{0.45cm}&\sat{\omega}{\mdp}{\phi\land\gamma} &\Leftrightarrow & \sat{\omega}{\mdp}{\phi} \mbox{ and }\sat{\omega}{\mdp}{\!\gamma}\\
\sat{\omega}{\mdp}{\lnot \phi} &\Leftrightarrow &
\mbox{not}(\sat{\omega}{\mdp}{\phi})&\hspace{0.45cm}&
\sat{\omega}{\mdp}{\Until{\phi}{\gamma}} & \Leftrightarrow &
\exists_{i\geq 0}. \sat{\omega_{\downarrow
i}}{\mdp}{\!\gamma}\mbox{ and }\\
&&&&&&\forall_{0\leq j <
i}.\sat{\omega_{\downarrow j}}{\mdp}{\phi}\\
\end{array}
$$
\vspace{-0.25cm}

\noindent where $\omega_{\downarrow i}$ is the $i$-th suffix of
$\omega$. When confusion is unlikely, we omit the subscript $\mdp$
on the satisfiability relation.
\end{dfn}

\begin{dfn} Let $\mdp$ be an $\MDP$. We define the language $\text{Sat}_{_{\mdp}}(\phi)$ associated to an $\LTL$ formula $\phi$
as the set of paths satisfying $\phi$, i.e.
$\text{Sat}_{_{\mdp}}(\phi)\allowbreak\triangleq \lbrace
\omega\in\Paths{\mdp}\mid \omega\models \phi\rbrace.$ Here we also
generally omit the subscript $\mdp$.
\end{dfn}

We now define satisfiability of an $\LTL$ formula $\phi$ on an
$\MDP$ $\mdp$. We say that $\mdp$ satisfies $\phi$ with
probability at most $p$ ($\sat{\mdp}{\leq p}{\phi}$) if the
probability of getting an execution satisfying $\phi$ is at most
$p$.

\begin{dfn} Let $\mdp$ be an $\MDP$, $\phi$ an $\LTL$ formula and $p\in[0,1]$. We define $\sat{}{\leq p}{}$ and $\sat{}{\geq p}{}$by
\vspace{-0.1cm}
\begin{align*}
\sat{\mdp}{\leq p}{\phi} \Leftrightarrow \measure{+}{}{\LangPa{\phi}}\leq p,\\
\sat{\mdp}{\geq p}{\phi} \Leftrightarrow
\measure{-}{}{\LangPa{\phi}}\geq p.
\end{align*}
\vspace{-0.4cm}

\noindent We define $\sat{\mdp}{< p}{\phi}$ and $\sat{\mdp}{>
p}{\phi}$ in a similar way. In case the $\MDP$ is fully
probabilistic, i.e., an $\MC$, the satisfiability problem is
reduced to $\sat{\mdp}{\bowtie p}{\phi} \Leftrightarrow
\measure{}{_{\mdp}}{\LangPa{\phi}}\bowtie p$, where $\bowtie\in
\lbrace <,\leq,>,\geq\rbrace$.

\end{dfn}

\section{Counterexamples\index{counterexample}}\label{sec:counterexamples}

In this section, we define what counterexamples are and how the
problem of finding counterexamples to a general $\LTL$ property
over Markov Decision Processes reduces to finding counterexamples
to reachability problems over Markov Chains.

\begin{dfn}[Counterexamples]\label{dfn:counterexamples}%
  Let $\mdp$ be an $\MDP$ and $\phi$ an $\LTL$ formula. A
  \emph{counterexample}\index{counterexample} to $\sat{\mdp}{\leq p}{\phi}$ is a
  measurable set $\ce\subseteq \LangPa{\phi}$ such that
  $\measure{+}{}{\ce}> p$.
  Counterexamples to $\sat{\mdp}{< p}{\phi}$ are defined similarly.
\end{dfn}

Counterexamples to $\sat{\mdp}{> p}{\phi}$ and $\sat{\mdp}{\geq
  p}{\phi}$ cannot be defined straightforwardly as it is always
possible to find a set $\ce\subseteq \LangPa{\phi}$ such that
$\measure{-}{}{\ce}\leq p$ or $\measure{-}{}{\ce}< p$, note that
the empty set trivially satisfies it.
Therefore, the best way to find counterexamples to lower bounded
probabilities is to find counterexamples to the dual properties
$\sat{\mdp}{< 1-p}{\!\!\!\neg\phi}$ and $\sat{\mdp}{\leq
  1-p}{\!\!\!\neg\phi}$. 
That is, while for upper bounded probabilities, a counterexample
is a set of paths satisfying the property with mass probability
beyond the bound, for lower bounded probabilities the
counterexample is a set of paths that \emph{does not} satisfy the
property with sufficient
probability.

\begin{wrapfigure}{r}{5.2cm}
\includegraphics[width=5cm]{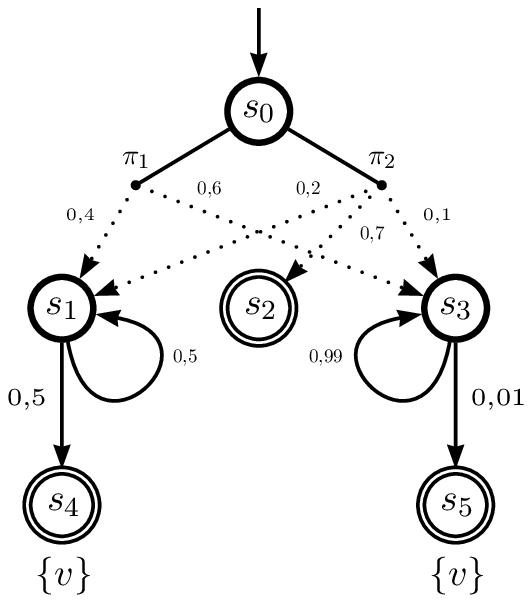}
\caption{} \label{fig:ExaModelCheckerCuant}
\end{wrapfigure}
\begin{example}\label{exa:counterexamples}\ref{exa:counterexamples}
  Consider the $\MDP$ $\mdp$ of Figure~\ref{fig:ExaModelCheckerCuant}
  and the $\LTL$ formula $\F v$. It is easy to check that
  $\nonsat{\mdp}{< 1}{\F v}$. The set $\ce=\LangPa{\F
    v}=\{\rho\!\in\! \Paths{s_0}|\exists_{i\geq 0}.\rho\!=\!s_0
  (s_1)^i (s_4)^\omega\}\cup\{\rho\!\in\!  \Paths{s_0}|\exists_{i\geq
    0}.\rho\!=\!s_0 (s_3)^i (s_5)^\omega\}$ is a counterexample. Note
  that $\measure{}{\eta}{\ce}\!=\!1$ where $\eta$ is any
  deterministic scheduler on $\mdp$ satisfying $\eta(s_0)=\pi_1$.
\end{example}

$\LTL$ formulas are actually checked by reducing the model
checking problem to a reachability
problem~\cite{dealfaro_1997_hybrid}. For checking upper bounded
probabilities, the $\LTL$ formula is translated into an equivalent
deterministic Rabin automaton and composed with the $\MDP$ under
verification. On the obtained $\MDP$, the set of states forming
accepting end components ($\SCC$ that traps accepting
conditions with probability 1) are identified.
The maximum probability of the $\LTL$ property on the original
$\MDP$ is the same as the maximum probability of reaching a state
of an accepting end component in the final $\MDP$.
Hence, from now on we will focus on counterexamples to properties
of the form $\sat{\mdp}{\leq p}{\F\psi}$ or $\sat{\mdp}{<
p}{\F\psi}$, where $\psi$ is a propositional formula, i.e., a
formula without temporal operators.

In the following, it will be useful to identify the set of states
in which a propositional property is valid.

\begin{dfn}%
  Let $\mdp$ be an $\MDP$. We define the state language
  $\text{Sat}_{\mdp}(\psi)$ associated to a propositional formula $\psi$
  as the set of states satisfying $\psi$,
  i.e., $\text{Sat}_{\mdp}(\psi)\triangleq \lbrace s\in S\mid s\models\psi\rbrace$,
  where $\models$ has the obvious satisfaction meaning for states. As
  usual, we generally omit the subscript $\mdp$.
\end{dfn}

We will show now that, in order to find a counterexample to a
property in an $\MDP$ with respect to an upper bound, it suffices
to find a counterexample for the $\MC$ \emph{induced} by the
maximizing scheduler.
The maximizing scheduler turns out to be deterministic and
memoryless~\cite{ba_1995_probabilistic}; consequently the induced
Markov Chain can be easily extracted from the $\MDP$ as follows.

\begin{dfn}%
  Let $\mdp=(S,s_0,\tau,L)$ be an $\MDP$ and $\eta$ a deterministic
  memoryless scheduler. Then we define the $\MC$ induced by $\eta$ as $\mdp_\eta=(S,s_0,\M_\eta,L)$ where
  $\M_{\eta}(s,t)=(\eta(s))(t)$ for all $s,t\in S$.
\end{dfn}

Now we state that finding counterexamples to upper bounded
probabilistic reachability $\LTL$ properties on $\MDPs$ can be
reduced to finding counterexamples to upper bounded probabilistic
reachability $\LTL$ properties on $\MCs$.

\begin{thm}%
  Let $\mdp$ be an $\MDP$, $\psi$ a propositional formula and
  $p\in[0,1]$.  Then, there is a maximizing (deterministic memoryless)
  scheduler $\eta$ such that
  $\sat{\mdp}{\leq p}{\F\psi}\Leftrightarrow \sat{\mdp_\eta}{\leq p}{\F\psi}$.
  Moreover,
  if $\ce$ is a counterexample to $\sat{\mdp_\eta}{\leq p}{\F\psi}$ then $\ce$ is also a counterexample to $\sat{\mdp}{\leq p}{\F\psi}$.
\end{thm}

Note that $\eta$ can be computed by solving a linear minimization
problem \cite{ba_1995_probabilistic}. See Section
\ref{sec:max-schedulers}.

\section{Representative Counterexamples, Partitions and
Witnesses}\label{sec:RepCount-parti-witnesses}

The notion of counterexample from Definition
\ref{dfn:counterexamples} is very broad: just an arbitrary
(measurable) set of paths with high enough mass probability. To be
useful as a debugging tool (and in fact to be able to present the
counterexample to a user), we need counterexamples with specific
properties. We will partition counterexamples (or rather,
representative counterexamples) in witnesses and list five
informal properties that we consider valuable in order to increase
the quality of witnesses as a debugging tool.

We first note that for reachability properties it is sufficient to
consider counterexamples that consist of finite paths.

\begin{dfn}[Representative counterexamples\index{counterexample!representative}] Let $\mdp$ be an $\MDP$, $\psi$
a propositional formula and $p\in[0,1]$. A \emph{representative
counterexample} to $\sat{\mdp}{\leq p}{\F \psi}$ is a set
$\ce\subseteq \FReach{\mdp,\LangSt{\psi}}$ such that
$\measure{+}{}{\cyl{\ce}}>p$. We denote the set of all
representative counterexamples to $\sat{\mdp}{\leq p}{\F \psi}$ by
$\CountSet{\mdp, p,\psi}$.
\end{dfn}

\begin{observation} Let $\mdp$ be an $\MDP$, $\psi$ a propositional formula and $p\in[0,1]$. If $\ce$ is a representative counterexample to $\sat{\mdp}{\leq p}{\F
\psi}$, then $\cyl{\ce}$ is a counterexample to $\sat{\mdp}{\leq
p}{\F\psi}$. Furthermore, there exists a counterexample to
$\sat{\mdp}{\leq p}{\F\psi}$ if and only if there exists a
representative counterexample to $\sat{\mdp}{\leq p}{\F \psi}$.
\end{observation}

Following \cite{hk_2007_counterexamples}, we present the notions
of \emph{minimum counterexample}, \emph{strongest evidence} and
\emph{most indicative counterexamples}.

\begin{dfn}[Minimum counterexample\index{counterexample!minimum}] Let $\mc$ be an $\MC$, $\psi$ a propositional formula and $p\in[0,1]$. We say that $\ce\in
\CountSet{\mc,p,\psi}$ is a \emph{minimum counterexample} if
$|\ce|\leq|\ce^\prime|$, for all
$\ce^\prime\in\CountSet{\mc,p,\psi}$.
\end{dfn}

\begin{dfn}[Strongest evidence] Let $\mc$ be an $\MC$, $\psi$ a propositional formula and $p\in[0,1]$. A \emph{strongest evidence} to
$\smash{\nonsat{\mc}{\leq p}{\F \psi}}$ is a finite path
$\sigma\in \FReach{\mc,\LangSt{\psi}}$ such that
$\measure{}{_{\mc}}{\cyl{\sigma}}\geq\measure{}{_{\mc}}{\cyl{\rho}}$,
for all $\rho \in\FReach{\mc,\LangSt{\psi}}$.
\end{dfn}

\begin{dfn}[Most indicative counterexample\index{counterexample!most indicative}] Let $\mc$ be an $\MC$, $\psi$ a propositional formula and
$p\in[0,1]$. We call $\ce\in \CountSet{\mc,p,\psi}$ a \emph{most
indicative counterexample} if it is minimum and
$\measure{}{_{\mc}}{\cyl{\ce}}\geq
\measure{}{_{\mc}}{\cyl{\ce^\prime}}$, for all minimum
counterexamples $\ce^\prime \in \CountSet{\mc,p,\psi}$.
\end{dfn}

Unfortunately, very often most indicative counterexamples are very
large (even infinite), many of its elements have insignificant
measure and elements can be extremely similar to each other
(consequently providing the same diagnostic information). Even
worse, sometimes the finite paths with highest probability do not
exhibit the way in which the system accumulates higher probability
to reach the undesired property (and consequently where an error
occurs with higher probability). For these reasons, we are of the
opinion that representative counterexamples are still too general
in order to be useful as feedback information. We approach this
problem by refining a representative counterexample into sets of
finite paths following a ``similarity'' criteria (introduced in
Section~\ref{sec:Torrents-Rails}). These sets are called
\emph{witnesses} of the counterexample.


Recall that a set $Y$ of nonempty sets is a partition of $X$ if
the elements of $Y$ cover $X$ and are pairwise
disjoint. We define counterexample partitions in the following
way.

\begin{dfn}[Counterexample partitions and witnesses] Let $\mdp$ be
an $\MDP$, $\psi$ a propositional formula, $p\in[0,1]$, and $\ce$
a representative counterexample to $\sat{\mdp}{\leq p}{\F \psi}$.
A \emph{counterexample partition} $W_\ce$ is a partition of $\ce$.
We call the elements of $W_\ce$ \emph{witnesses}.
\end{dfn}

Since not every partition generates useful witnesses (from the
debugging perspective), we now state five informal properties that
we consider valuable in order to improve the diagnostic
information provided by witnesses. In Section
\ref{sec:coumputing-counterexamples} we show how to partition the
representative counterexample in order to obtain witnesses satisfying
most of these properties.
\begin{description}
\item[\textbf{Similarity:}]{\noindent Elements of a witness should provide similar debugging information.}
\item[\textbf{Accuracy:}]{\noindent Witnesses with higher probability should exhibit evolutions of the system with higher probability of containing errors.}
\item[\textbf{Originality:}]{\noindent Different witnesses should provide different debugging information.}
\item[\textbf{Significance:}]{\noindent Witnesses should be as closed to the counterexample as possible (their mass probability should be as closed as possible to the bound $p$).}
\item[\textbf{Finiteness:}]{\noindent The number of witnesses of a counterexample partition should be finite.}
\end{description}
\section{Rails and Torrents\index{rail}\index{torrent}}\label{sec:Torrents-Rails}

As argued before we consider that representative counterexamples
are excessively general to be useful as feedback information.
Therefore, we group finite paths of a representative
counterexample in witnesses if they are ``similar enough''. We
will consider finite paths that behave the same outside $\SCCs$ of
the system as providing similar feedback information.

In order to formalize this idea, we first reduce the original
$\MC$ $\mc$ to an acyclic $\MC$ preserving reachability
probabilities. We do so by removing all $\SCCs$ $\scc$ of $\mc$
keeping just \emph{input states} of $\scc$. In this way, we get a
new acyclic $\MC$ denoted by $\Acyclic{\mc}$. The probability
matrix of the Markov Chain relates input states of each $\SCC$ to
its \emph{output states} with the reachability probability between
these states in $\mc$. Secondly, we establish a map between finite
paths $\sigma$ in $\Acyclic{\mc}$ (\emph{rails}) and sets of paths
$W_\sigma$ in $\mc$ (\emph{torrents}). Each torrent contains
finite paths that are similar, i.e., behave the same outside
$\SCCs$. We conclude the section showing that the probability of
$\sigma$ is equal to the mass probability of $W_\sigma$.

\subsection*{Reduction to Acyclic Markov Chains \index{markov chain!acyclic}\index{scc analysis}}

Consider an $\MC$ $\mc=(S,s_0, \M,L)$. Recall that a subset
$\scc\subseteq S$ is called \emph{strongly connected}\index{scc}
if for every $s,t\in \scc$ there is a finite path from $s$ to $t$.
Additionally $\scc$ is called a \emph{strongly connected
component} ($\SCC$) if it is a maximally (with respect to
$\subseteq$) strongly connected subset of $S$.

Note that every state is a member of exactly one $\SCC$ of $\mc$; even those states that are not involved in cycles, since the
trivial finite path $s$ connects $s$ to itself. We call \emph{trivial strongly connected components} to the $\SCCs$ containing absorbing states or states not involved in cycles (note that trivial $\SCCs$ are composed by one single state). From now on we
let $\SCC^\star$ be the set of non trivial strongly connected
components of an $\MC$.

A Markov Chain is called \emph{acyclic} if it contains only
trivial $\SCCs$. Note that an acyclic Markov Chain still has
absorbing states.

\begin{dfn}[Input and Output states]
Let $\mc=(S,s_0,\M,L)$ be an $\MC$. Then, for each $\SCC^\star$
$\scc$ of $\mc$, we define the sets $\Inp_{\scc}\subseteq S$ of
all states in $\scc$ that have an incoming transition from a state
outside of $\scc$ and $\Out_{\scc}\subseteq S$ of all states
outside of $\scc$ that have an incoming transition from a state of
$\scc$ in the following way
\begin{figure}[hbt]
\qquad\begin{minipage}[t]{6cm}
  \vspace{-2.5cm}
  \centering
  \begin{align*}
  \Inp_{\scc}\triangleq\{t\in \scc\mid\exists\,s\in
  S\setminus\scc. \M(s,t)>0\},\\
  \Out_{\scc}\triangleq\{s\in S\setminus\scc\mid\exists\,t\in \scc. \M(t,s)>0\}.\\
  \end{align*}
\end{minipage}
\begin{minipage}[t]{5cm}
  \centering
  \includegraphics[width=3cm]{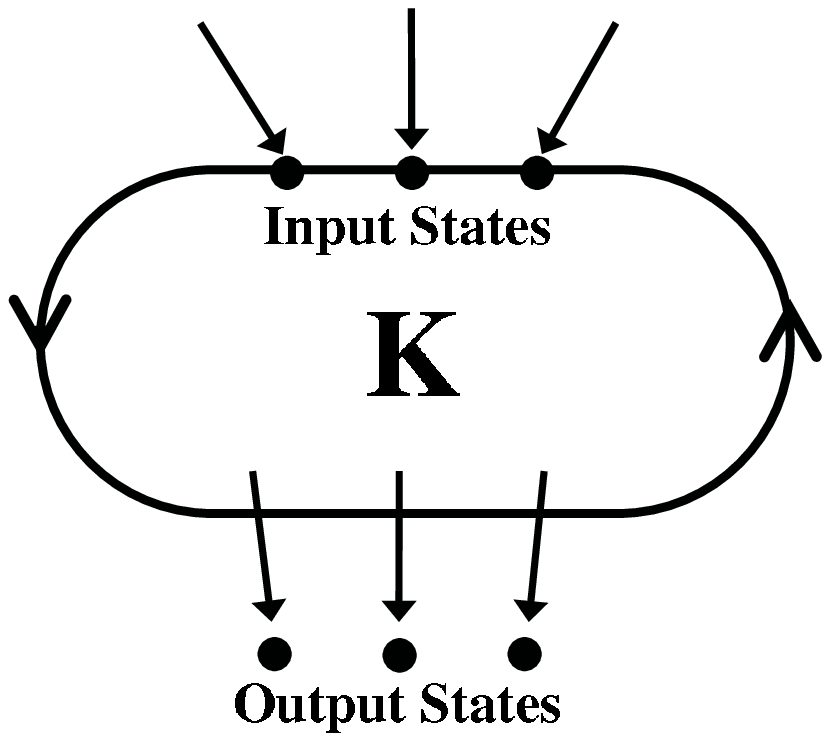}
\end{minipage}
\end{figure}

\noindent We also define for each $\SCC^\star$ $\scc$ an $\MC$
related to $\scc$ as $\mc_{\scc}\allowbreak\triangleq(\scc\cup
\Out_{\scc},\allowbreak s_{\scc},\allowbreak \M_{\scc},L_{\scc})$
where $s_{\scc}$ is any state in $\Inp_{\scc}$,
$L_{\scc}(s)\triangleq L(s)$, and $\M_{\scc}(s,t)$ is equal to
$\M(s,t)$ if $s\in \scc$ and equal to $1_s$ otherwise.
Additionally, for every state $s$ involved in non trivial $\SCCs$
we define $\SCC^+_s$ as $\mc_{\scc}$, where $\scc$ is the
$\SCC^\star$ of $\mc$ such that $s\in \scc$.
\end{dfn}

Now we are able to define an acyclic $\MC$ $\Acyclic{\mc}$ related
to $\mc$.

\begin{dfn}\label{dfn:acyclicMC}
Let $\mc=(S,s_0,\M,L)$ be a  $\MC$. We define
$\Acyclic{\mc}{}\triangleq(S^\prime,s_0,\M^\prime{},L^\prime)$
where
\vspace{-0.25cm}
\begin{itemize}
\item{$S^\prime\triangleq\stackrel{S_{\text{com}}}{\overbrace{S\setminus\bigcup_{\scc\in \SCC^\star} \scc}}\bigcup\stackrel{S_{\text{inp}}}{\overbrace{\bigcup_{\scc\in \SCC^\star} \Inp_{\scc}}},$}
\item $L^\prime\triangleq L_{|_{S^\prime}}$,
\item $ \M^\prime(s,t) \triangleq \left\lbrace
                           \begin{array}{ll}
                             \M(s,t) & \mbox{if } s\in S_{com},\\
                             \measure{}{_{\mc,s}}{\Reach{\SCC^+_s,s,\lbrace t\rbrace}} & \mbox{if } s\in S_{inp}\land t\in \Out_{\SCC^+_s},\\
                             1_s & \mbox{if } s\in S_{inp}\land \Out_{\SCC^+_s}=\emptyset,\\
                             0 & \mbox{otherwise.}
                             \end{array}
              \right.$
\end{itemize}
\end{dfn}

Note that $\Acyclic{\mc}$ is indeed acyclic.

\begin{example} Consider the $\MC$ $\mc$ of Figure \ref{fig:red-a}. The strongly connected components of $\mc$ are $\scc_1\eqdef\lbrace
s_1,s_3,s_4,s_7\rbrace$, $\scc_2\eqdef\lbrace s_5,s_6,s_8\rbrace$
and the singletons $\lbrace s_0 \rbrace$, $\lbrace s_2 \rbrace$,
$\lbrace s_9 \rbrace$, $\lbrace s_{10} \rbrace$, $\lbrace s_{11}
\rbrace$, $\lbrace s_{12} \rbrace$, $\lbrace s_{13} \rbrace$, and
$\lbrace s_{14} \rbrace$. The input states of $\scc_1$ are
$\Inp_{\scc_1}=\lbrace s_1\rbrace$ and its output states are
$\Out_{\scc_1}=\lbrace s_9,s_{10}\rbrace$. For $\scc_2$,
$\Inp_{\scc_2}=\lbrace s_5, s_6\rbrace$ and $\Out_{\scc_2}=\lbrace
s_{11},s_{14}\rbrace$. The reduced acyclic $\MC$ of $\mc$ is shown
in Figure \ref{fig:red-b}.
\end{example}

\begin{figure}[h]
 \centering
  \subfigure[Original $\MC$]{\label{fig:red-a}\includegraphics[width=6cm]{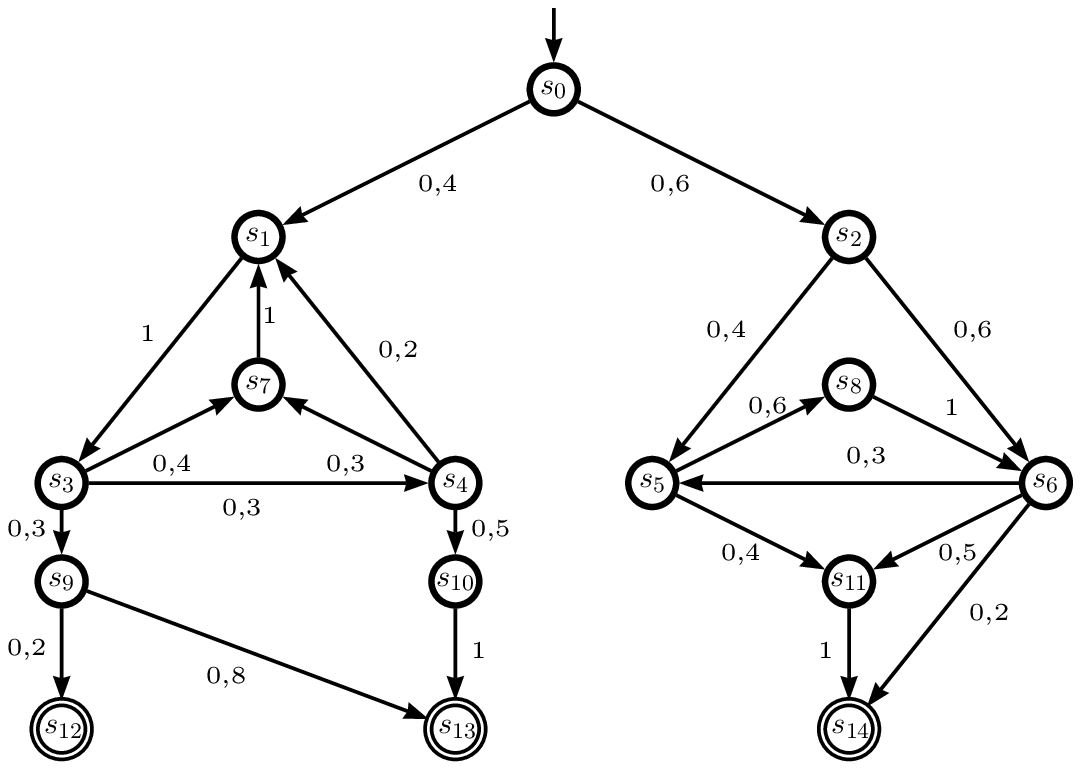}}
  \subfigure[Derived Acyclic $\MC$]{\label{fig:red-b}\includegraphics[width=6cm]{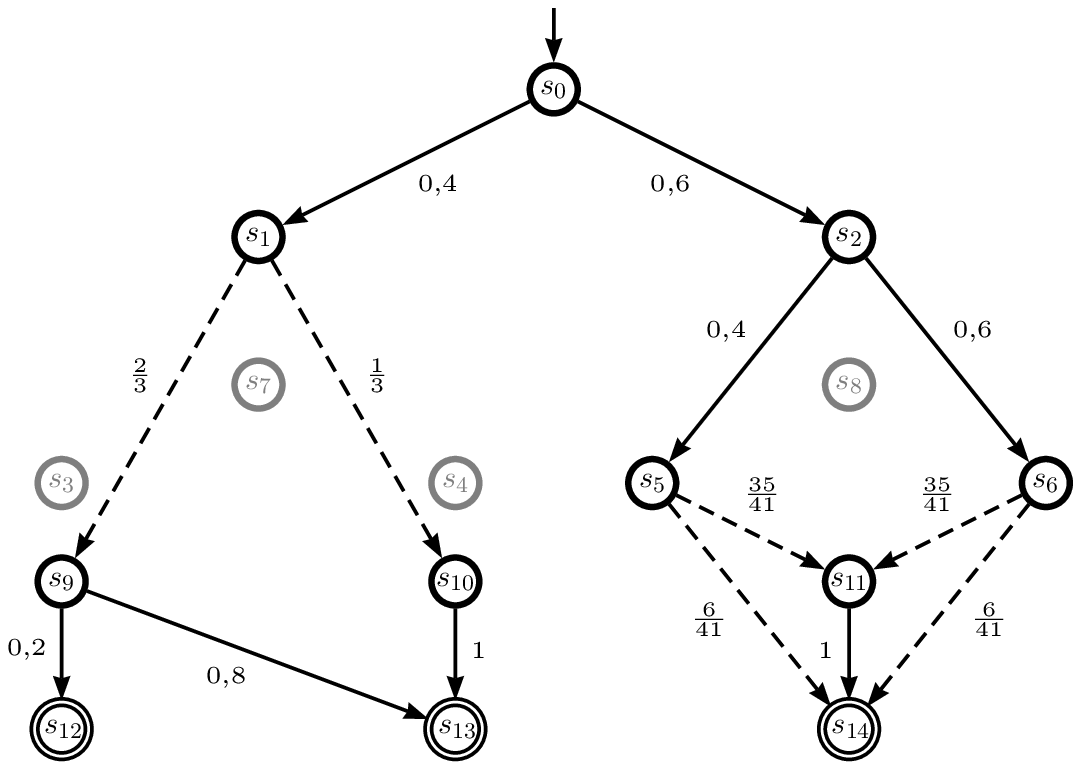}}
  \label{fig:reducingMCs}
  \caption{}
\end{figure}

\subsection*{Rails and Torrents}

We now relate (finite) paths in $\Acyclic{\mc}$ (rails) to sets of
paths in $\mc$ (torrents).

\begin{dfn}[Rails] Let $\mc$ be an $\MC$. A finite path $\sigma\in\FPaths{\Acyclic{\mc}}$ will be called
a \emph{rail of }$\mc$.
\end{dfn}

Consider a rail $\sigma$, i.e., a finite path of $\Acyclic{\mc}$.
We will use $\sigma$ to represent those paths $\omega$ of $\mc$
that behave ``similar to'' $\sigma$ outside $\SCCs$ of $\mc$.
Naively, this means that $\sigma$ is a subsequence of $\omega$.
There are two technical subtleties to deal with: every input state
in $\sigma$ must be the first state in its $\SCC$ in $\omega$
(freshness) and every $\SCC$ visited by $\omega$ must be also
visited by $\sigma$ (inertia) (see Definition~\ref{dfn:sim}). We
need these extra conditions to make sure that no path $\omega$
behaves ``similar to'' two distinct rails (see
Lemma~\ref{lem:disjoint}).

Recall that given a finite sequence $\sigma$ and a (possible
infinite) sequence $\omega$, we say that $\sigma$ is a
\emph{subsequence} of $\omega$, denoted by $\sigma \subS \omega$,
if and only if there exists a strictly increasing function
$f:\lbrace 0,1,\ldots,|\sigma|-1\rbrace\rightarrow \lbrace
0,1,\ldots,|\omega|-1\rbrace$ such that $\forall_{0\leq i <
|\sigma|}. \sigma_i=\omega_{f(i)}$. If $\omega$ is an infinite
sequence, we interpret the codomain of $f$ as $\mathbb{N}$. In
case $f$ is such a function we write $\sigma \subS_f \omega$.

\begin{dfn} Let $\mc=(S,s_0,\M, L)$ be an $\MC$. On $S$ we consider the equivalence relation $\sim_{\!\!\!\!\!_{\mc}}$ satisfying $s\sim_{\!\!\!\!\!_{\mc}} t$ if and only if $s$ and $t$ are in the same
strongly connected component. Again, we usually omit the subscript
$\mc$ from the notation.
\end{dfn}


The following definition refines the notion of subsequence, taking
care of the two technical subtleties noted above.

\begin{dfn}\label{dfn:sim} Let $\mc=(S,s_0,\M,L)$ be an $\MC$, $\omega$ a (finite) path of $\mc$, and $\sigma\in \FPaths{\Acyclic{\mc}}$ a finite path of $\Acyclic{\mc}$. Then we write $\sigma\HsubS \omega$ if there exists $f:\lbrace 0,1,\ldots,|\sigma|-1\rbrace\rightarrow \mathbb{N}$ such that $\sigma\subS_f\omega$ and
\begin{align*}
\forall_{0\leq j < f(i)}: \omega_{f(i)}\not\sim \omega_j & \mbox{; for all }i=0,1,\ldots|\sigma|-1, & \{\mbox{\emph{Freshness property}}\}\\
\forall_{f(i)< j< f(i+1)}:\omega_{f(i)}\sim \omega_{j} & \mbox{;
for all }i=0,1,\ldots|\sigma|-2. & \{\mbox{\emph{Inertia
property}}\}
\end{align*}
\noindent In case $f$ is such a function we write $\sigma\HsubS_f
\omega$.
\end{dfn}

\begin{example} Let $\mc=(S,s_0,\M,L)$ be the $\MC$ of Figure
\ref{fig:red-a} and take $\sigma=s_0 s_2 s_6 s_{14}$. Then for all
$i\in\mathbb{N}$ we have $\sigma\HsubS_{f_i} \omega_i$ where
$\omega_i=s_0 s_2 s_6 (s_5 s_8 s_6)^i s_{14}$ and $f_i(0)
\triangleq 0$, $f_i(1) \triangleq 1$, $f_i(2) \triangleq 2$, and
$f_i(3) \triangleq 3+3 i$. Additionally, $\sigma \not\HsubS s_0
s_2 s_5 s_8 s_6 s_{14}$ since for all $f$ satisfying
$\sigma\subS_f s_0 s_2 s_5 s_8 s_6 s_{14}$ we must have $f(2)=5$;
this implies that $f$ does not satisfy the freshness property.
Finally, note that $\sigma \not\HsubS s_0 s_2 s_6 s_{11} s_{14}$
since for all $f$ satisfying $\sigma\subS_f s_0 s_2 s_6 s_{11}
s_{14}$ we must have $f(2)=2$; this implies that $f$ does not
satisfy the inertia property.
\end{example}

We now give the formal definition of torrents.

\begin{dfn}[Torrents] Let $\mc=(S,s_0,\M, L)$ be an $\MC$ and $\sigma$ a sequence of states in $S$. We define the function
$\text{Torr}$ by
$$\Torrent{\mc,\sigma}\triangleq\lbrace \omega\in\Paths{\mc}\mid \sigma\HsubS \omega\rbrace.$$
\noindent We call $\Torrent{\mc,\sigma}$ the \emph{torrent}
associated to $\sigma$.
\end{dfn}

We now show that torrents are disjoint (Lemma~\ref{lem:disjoint}) and that
the probability of a rail is equal to the probability of its associated torrent (Theorem~\ref{thm:railSameProb}). For this last result, we first show that
torrents can be represented as the disjoint union of cones of finite paths. We call these finite paths \emph{generators} of the torrent (Definition~\ref{dfn:generators}).

\begin{lem}\label{lem:disjoint} Let $\mc$ be an $\MC$. For every $\sigma,\rho\in
\FPaths{\Acyclic{\mc}}$ we have
$$\sigma\not=\rho \Rightarrow\Torrent{\mc,\sigma}\cap \Torrent{\mc,\rho}=\emptyset.$$
\end{lem}


\begin{dfn}[Torrent Generators]\label{dfn:generators} Let $\mc$ be an $\MC$. Then we define for every rail $\sigma\in
\FPaths{\Acyclic{\mc}}$ the set
$$\gen{\mc,\sigma}\triangleq\lbrace \rho\in \FPaths{\mc}\mid
\exists f: \sigma \HsubS_f \rho \land
f(|\sigma|-1)=|\rho|-1\rbrace.$$
\end{dfn}

In the example from the Introduction (see Figure~\ref{fig:MC related to MDP}), $s_0 s_1 s_3$ and $s_0 s_2 s_4$ are rails. Their associated torrents are, respectively, $\{ s_0 s_1^n s_3^\omega \;|\; n \in \Nat^*\}$ and $\{ s_0 s_2^n s_4^\omega \;|\; n \in \Nat^*\}$ (note that $s_3$ and $s_4$ are absorbing states), i.e.~the paths going left and the paths going right. The generators of the first torrent are $\{ s_0 s_1^n s_3 \;|\; n \in \Nat^*\}$ and similarly for the second torrent.


\begin{lem}\label{lem:generators} Let $\mc$ be an $\MC$ and $\sigma\in \FPaths{\Acyclic{\mc}}$ a rail of $\mc$.
Then we have
\begin{align*}
\Torrent{\mc,\sigma}=\biguplus_{\rho\in\gen{\mc,\sigma}}\cyl{\rho}.
\end{align*}
\end{lem}

\begin{proof} \rev{The proof is by cases on the length of $\sigma$. We prove the result for the cases on which $\sigma$ is of the form $\sigma t s$, with $t$ an input state and $s$ an output state, the other cases are simpler. In order to proof this lemma, we define (for each $\sigma s t$ of the above form) the following set of finite paths 

\begin{equation}\label{def:delta}
\Delta_{\sigma t s}\eqdef \{\rho \tail{\pi}\mid \rho\in \gen{\sigma t}\text{ and } \pi\in \FPaths{\SCC^+_t,t,\lbrace s\rbrace}\}
\end{equation}

Checking that $\Torrent{\mc,\sigma}=\biguplus_{\rho\in\Delta_{\sigma t s}}\cyl{\rho}$ is straightforward. We now focus on proving that 

\begin{equation}\label{lem:delta}
\Delta_{\sigma s t}=\gen{\mc,\sigma}.
\end{equation}
\noindent For that purpose we need the following two observations.

\begin{observation}\label{obs:scc-rails}
Let $\mc$ be a $\MC$. Since $\Acyclic{\mc}$ is acyclic we have
$\sigma_i \not \sim \sigma_j$ for every $\sigma\in
\FPaths{\Acyclic{\mc}}$ and $i\not=j$ (with the exception of absorbing
states).
\end{observation}

\begin{observation}\label{obs:surSCC} Let $\sigma,\omega$ and $f$ be such that $\sigma\HsubS_f
\omega$. Then $\forall i: \exists j: \omega_i\sim \sigma_j$. This
follows from $\sigma\subS_f \omega$  and the inertia property.
\end{observation}

We now proceed to prove that $\Delta_{\sigma s t}=\gen{\mc,\sigma}$.
}

\noindent $(\ \supseteq\ )$ Let $\rho_0 \rho_1 \cdots \rho_k \in
\gen{\sigma t s}$ and $n_t$ the lowest subindex of $\rho$ such
that $\rho_{n_t}=t$. Take $\rho\triangleq\rho_0 \rho_1 \cdots
\rho_{n_t}$ and $\pi\triangleq\rho_{n_t} \cdots \rho_k$ (Note that
$\rho_0 \rho_1 \cdots \rho_k=\rho\tail{\pi}$). In order to prove
that $\rho_0 \rho_1 \cdots \rho_k \in \Delta_{\sigma t s}$ we need
to prove that

\begin{enumerate}
\item[(1)]{$\rho\in \gen{\sigma t}$, and}
\item[(2)]{$\pi\in \FPaths{\SCC^+_t,t,\lbrace s\rbrace}$.}
\end{enumerate}

    \begin{itemize}
    \item [(1)] Let $f$ be such that $\sigma t s \HsubS_f \rho_0 \rho_1 \cdots \rho_k$ and $f(|\sigma t s| -1)=k$. Take $g:\lbrace
    0,1,\ldots,|\sigma t|-1\rbrace\rightarrow \mathbb{N}$ be the restriction of $f$. It is easy to check that $\sigma t \HsubS_g
    \rho$.
    Additionally $f(|\sigma t|-1)=n_t$ (otherwise
    $f$ would not satisfy the freshness property for $i=|\sigma
    t|-1$). Then, by definition of $g$, we have $g(|\sigma t|-1)=n_t$.

    \item [(2)] It is clear that $\pi$ is a path from $t$ to $s$.
    Therefore we only have to show that every state of $\pi$ is in
    $\SCC^+_t$. By definition of $\SCC^+_t$, $\pi_0=t\in\SCC^+_t$ and
    $s\in\SCC^+_t$ since $s\in\Out_{\SCC^+_t}$.
    Additionally, since $f$ satisfies inertia property we have
    that $\forall_{f(|\sigma t|-1)<j<f(|\sigma t s|-1)}:\rho_{f(|\sigma t|-1)}\sim
    \rho_j$, since $f(|\sigma t|-1)=n_t$ and
    $\pi\triangleq\rho_{n_t}\cdots\rho_k$ we have $\forall_{0<j<|\pi|-1}:t\sim
    \pi_j$ proving that $\pi_j\in\SCC^+_t$ for
    $j\in\lbrace1,\cdots,|\pi|-2\rbrace$.

    \comment{Alternative prove: Suppose that there exists $i$ such that $\pi_i\not \in
    \SCC^+_t$. But then $\rho_{f(n_t)}\not\sim \rho_{f(n_t)+i}$
    contradicting the Inertia property of $f$.}

    \end{itemize}

\noindent $(\ \subseteq \ )$ Take $\rho\in \gen{\sigma t}$ and
$\tail{\pi}\in \FPaths{\SCC^+_t,t,\lbrace s\rbrace}$. In order to
prove that $\rho\tail{\pi}\in \gen{\sigma t s}$ we need to show
that there exists a function $g$ such that:

\begin{enumerate}
\item[(1)]{$\sigma t s\HsubS_g\rho\tail{\pi}$,}
\item[(2)]{$g(|\sigma t s|-1)=|\rho\tail{\pi}|-1$.}
\end{enumerate}

Since $\rho\in \gen{\sigma t}$ we know that there exists $f$ be
such that $\sigma t \HsubS_f \rho$ and $f(|\sigma t|-1)=|\rho|-1$.
We define $g:\lbrace 0,1,\ldots,|\sigma t s|-1\rbrace\rightarrow
\lbrace 0,1,\ldots,\allowbreak|\rho\tail{\pi}|-1\rbrace$ by
     \begin{align*}
        g(i) & \eqdef \left\lbrace
                                       \begin{array}{ll}
                                         f(i) & \mbox{ if } i<|\sigma t s| - 1,\\
                                         |\rho\tail{\pi}|-1 & \mbox{ if }i=|\sigma t s| -1.\\
                                       \end{array}
                                   \right.
     \end{align*}

   \begin{itemize}

   \item [(1)] It is easy to check that $\sigma t s\subS_g \rho\tail{\pi}$. Now we will show that $g$
    satisfies Freshness and Inertia properties.

    \underline{Freshness property:} We need to show that for all $0\leq i<|\sigma t s|$ we have $\forall_{0\leq j<g(i)}: \rho\tail{\pi}_{g(i)}\not
    \sim\rho\tail{\pi}_j$. For the cases $i\in\lbrace
    0,\ldots,|\sigma t|-1\rbrace$ this holds since $\sigma t\HsubS_f
    \rho$ and definition of $g$.

    Consider $i=|\sigma t s|-1$, in this case we have to prove
    $\forall_{0\leq j< |\rho\tail{\pi}|-1}:
    \rho\tail{\pi}_{|\rho
    \tail{\pi}|-1)}\not \sim \rho\tail{\pi}_j$ or equivalently $\forall_{0\leq j< |\rho\tail{\pi}|-1}:
    s\not \sim \rho\tail{\pi}_j$.

       \begin{itemize}
       \item [Case $j\in\lbrace
       |\rho|,\ldots|\rho\tail{\pi}|-1\rbrace$.]$\\$ From
       $\pi\in\FPaths{\SCC^+_t,t,\lbrace s\rbrace}$ and
       $s\in\Out^+_{\SCC^+_t}$ it is easy to see $\forall_{0\leq
       j<|\tail{\pi}|-1}\allowbreak:\allowbreak s\not\sim \tail{\pi}_j$

       \item [Case $j\in\lbrace 0,\ldots,|\rho|-1\rbrace.$]$\\$
        From $\sigma t s\in \FPaths{\Acyclic{\mc}}$ \rev{and Observation \ref{obs:scc-rails}}
        we have that $\forall_{0\leq j \allowbreak<|\sigma
        t|-1}:s\not\sim \sigma t_j$. Additionally, $\sigma t \HsubS_f
        \rho$ and definition of $g$  \rev{and Observation \ref{obs:surSCC}}
        imply $\forall_{0\leq j <|\rho|}:s\not\sim \rho_j$ or equivalently
        $\forall_{0\leq j <|\rho|}:s\not\sim \rho\tail{\pi}_j$.

       \end{itemize}

    \comment{Alternative prove (by absurd) Suppose that there exists
    $j<g(|\sigma s|-1)$ such that
    $\rho\tail{\pi}_j\sim
    \rho\tail{\pi}_{g(|\sigma s|-1)}$. Since
    $g(|\sigma|-1)=|\rho|-1$ and $\pi\in
    \Paths{\SCC^+_t,t,\lbrace s\rbrace}$ we have that $j<|\rho|-1$.
    Additionally, note that since $\rho_j {\sigma}_{j+1}
    \cdots \rho_{|\rho|-1}$ is a finite
    path from $\rho_j$ to $t$ and $\pi$ is a finite path
    from $t$ to $s$ we have that $t\sim s$ which contradicts
    Observation \ref{obs:scc-rails}.}

    \underline{Inertia property:} Since $\pi\in\FPaths{SCC^+_t,t,\lbrace s\rbrace }$ we
    have $\allowbreak\forall_{0\leq j < |\pi|-1}: t\sim \pi_j$ which implies
    that
    $\forall_{|\rho|-1<j<|\rho\tail{\pi}|-1}:
    \rho\tail{\pi}_{|\rho|-1}\sim
    \rho\tail{\pi}_j$ or equivalently
    $\forall_{g(|\sigma|-1)<\allowbreak j<\allowbreak
    g(|\sigma s|-1)}:
    \rho\tail{\pi}_{g(|\rho|-1)}\sim
    \rho\tail{\pi}_j$ showing that $g$ satisfies
    the inertia property.

   \item [(2)] Follows from the definition of $g$.\qedhere

    \end{itemize}

\end{proof}

\begin{thm}\label{thm:railSameProb} Let $\mc$ be an $\MC$. Then for every rail $\sigma\in\FPaths{\Acyclic{\mc}}$ we have
\begin{align*}
\measure{}{_{\Acyclic{\mc}}}{\cyl{\sigma}}=\measure{}{_{\mc}}{\Torrent{\mc,\sigma}}.
\end{align*}
\end{thm}

\begin{proof} By induction on the structure of $\sigma$.
\begin{itemize}

\item [\emph{Base Case:}] Note that
$\measure{}{\Acyclic{\mc}}{\cyl{s_0}}=\measure{}{\Acyclic{\mc}}{\Paths{\Acyclic{\mc},s_0}}\allowbreak=1$
, and similarly
$1=\measure{}{\mc}{\Paths{\mc,s_0}}=\allowbreak\measure{}{\mc}{\Torrent{s_0}}.$

\item [\emph{Inductive Step:}]

Let $t$ be such that $\last{\sigma}=t$.
Suppose that $t\in S_{\Com}$ and denote by $\Acyclic{\M}$ to
the probability matrix of $\Acyclic{\mc}$. Then
$$\begin{array}{rclr}
\hbox to 0pt{$\measure{}{\Acyclic{\mc}}{\cyl{\sigma s}}$\hss} \\
&=&\measure{}{\Acyclic{\mc}}{\cyl{\sigma}}\cdot \Acyclic{\M}(t,s)\\
&=&\measure{}{\mc}{\Torrent{\sigma}}\cdot \M(t,s)&\\
&&&\llap{$\{\text{Inductive Hypothesis and definition of }\M\}$}\\
&=&\measure{}{\mc}{\biguplus_{\rho\in\gen{\sigma}} \cyl{\rho}}\cdot \M(t,s)&\explan{\text{Lem.~}\ref{lem:generators}}\\
&=&\sum_{\rho\in\gen{\sigma}}\measure{}{\mc}{\cyl{\rho}}\cdot \measure{}{\mc}{\cyl{ts}}\\
&=&\sum_{\rho\in\gen{\sigma}}\measure{}{\mc}{\cyl{\rho \tail{ts}}}&\\
&=&\sum_{\rho\in \gen{\sigma s}}\measure{}{_{\mc}}{\cyl{\rho}}&\\
&=&\measure{}{_{\mc}}{\biguplus_{\rho\in \gen{\sigma s}}\cyl{\rho}}\\
&=&\measure{}{_{\mc}}{\Torrent{\sigma s}}&\explan{\text{Lem.~}\ref{lem:generators}}\\
\end{array}$$
Now suppose that $t\in S_{\Inp}$, then
$$\begin{array}{rclr}
\hbox to 0pt{$\measure{}{_{\Acyclic{\mc}}}{\cyl{\sigma s}}$\hss} \\
&=&\measure{}{_{\Acyclic{\mc}}}{\cyl{\sigma}}\cdot \Acyclic{\M}(t,s)\\
&=&\measure{}{_{\mc}}{\Torrent{\sigma}}\cdot \Acyclic{\M}(t,s)&\explan{\text{HI}}\\
&=&\measure{}{_{\mc}}{\biguplus_{\rho\in\gen{\sigma}}\cyl{\rho}}\cdot\Acyclic{\M}(t,s)&\explan{\text{Lem.~\ref{lem:generators}}}\\
&=&\left(\sum_{\rho\in\gen{\sigma}}\measure{}{_{\mc}}{\cyl{\rho}}\right)\cdot\Acyclic{\M}(t,s)\\
&=&\sum_{\rho\in \gen{\sigma}}\measure{}{_{\mc}}{\cyl{\rho}}\cdot \measure{}{_{\mc,t}}{\Paths{\SCC^+_t,t,\lbrace s\rbrace}}\!\!\!\!\!\!\!\!\!\!\!\!\!\\
&&&\!\!\!\!\!\llap{$\{\text{By definition of }\Acyclic{\M} \text{ and distributivity}\}$}\\
&=&\sum_{\rho\in \gen{\sigma}}\measure{}{_{\mc}}{\cyl{\rho}}\cdot \sum_{\pi\in \FPaths{\SCC^+_t,t,\lbrace s\rbrace}}\measure{}{_{\mc,t}}{\cyl{\pi}}\!\!\!\!\!\!\!\!\!\!\!\!\!\!\!\!\!\\
&=&\sum_{\rho\in \gen{\sigma},\pi\in\FPaths{\SCC^+_t,t,\lbrace s\rbrace}}\measure{}{_{\mc}}{\cyl{\rho \tail{\pi}}}\!\!\!\!\!\!\!\!\!&\explan{\text{Dfn.~} \rev{\mathbb{P}}}\\
&=&\sum_{\rho\in \Delta_{\sigma s}}\measure{}{_{\mc}}{\cyl{\rho}}&\explan{\rev{(\ref{def:delta})}}\\
&=&\sum_{\rho\in \gen{\sigma s}}\measure{}{_{\mc}}{\cyl{\rho}}&\explan{\rev{(\ref{lem:delta})}}\\
&=&\measure{}{_{\mc}}{\biguplus_{\rho\in \gen{\sigma s}}\cyl{\rho}}\\
&=&\measure{}{_{\mc}}{\Torrent{\sigma s}}&
\qedhere\\
\end{array}$$
\end{itemize}


\end{proof}

\section{Significant Diagnostic Counterexamples}\label{sec:significantdiagnosticcounterexamples}

So far we have formalized the notion of paths behaving similarly
(i.e., behaving the same outside $\SCCs$) in an $\MC$ $\mc$ by
removing all $\SCC$ of $\mc$, obtaining $\Acyclic{\mc}$. A
representative counterexample to $\sat{\Acyclic{\mc}}{\leq
p}{\!\F\psi}$ gives rise to a representative counterexample to
$\sat{\mc}{\leq p}{\!\F\psi}$ in the following way: for every
finite path $\sigma$ in the representative counterexample to
$\sat{\Acyclic{\mc}}{\leq p}{\!\F\psi}$ the set $\gen{\mc,\sigma}$
is a witness, then we obtain the desired representative
counterexample to $\sat{\mc}{\leq p}{\!\F\psi}$ by taking the
union of these witnesses.

Before giving a formal definition, there is still one technical
issue to resolve: we need to be sure that by removing $\SCCs$ we are not
discarding useful information. Because torrents are built from
rails, we need to make sure that when we discard $\SCCs$, we do
not discard rails that reach $\psi$.

We achieve this by first making states satisfying $\psi$
absorbing. Additionally, we make absorbing states from which it is
not possible to reach $\psi$. Note that this does not affect
counterexamples.

\begin{dfn}\label{dfn:abs mc}
Let $\mc=(S,s_0,\M,L)$ be an $\MC$ and $\psi$ a propositional
formula. We define the $\MC$
$\mc_\psi\triangleq(S,s_0,\M_\psi,L)$, with
$$ \M_\psi(s,t) \triangleq \left\lbrace
                           \begin{array}{ll}
                             1 & \mbox{if } s\not\in \SatF{\psi} \land s=t,\\
                             1 & \mbox{if } s\in \LangSt{\psi} \land s=t,\\
                             \M(s,t) & \mbox{if } s\in \SatF{\psi}-\LangSt{\psi},\\
                             0 & \mbox{otherwise,}\\
                             \end{array}
              \right. $$
\noindent where $\SatF{\psi}\eqdef \lbrace \allowbreak s\in S\mid
\allowbreak \allowbreak \measure{}{_{\mc,s}}{\Reach{\mc,s,
\allowbreak \LangSt{\psi}}} \allowbreak >0\rbrace$ is the set of
states reaching $\psi$ in $\mc$.
\end{dfn}

The following theorem shows the relation between paths, finite
paths, and probabilities of $\mc$, $\mc_\psi$, and
$\Acyclic{\mc_\psi}$. Most importantly, the probability of a rail
$\sigma$ (in $\Acyclic{\mc_\psi}$) is equal to the probability of
its associated torrent (in $\mc$) (item \ref{i:5} below) and the
probability of $\F\psi$ is not affected by reducing $\mc$ to
$\Acyclic{\mc_\psi}$ (item \ref{i:6} below).

Note that a rail $\sigma$ is always a finite path in
$\Acyclic{\mc_\psi}$, but that we can talk about its associated
torrent $\Torrent{\mc_\psi,\sigma}$ in $\mc_\psi$ and about its
associated torrent $\Torrent{\mc,\sigma}$ in $\mc$. The former
exists for technical convenience; it is the latter that we are
ultimately interested in. The following theorem also shows that
for our purposes, viz.~the definition of the generators of the
torrent and the probability of the torrent, there is no difference
(items~\ref{i:3} and~\ref{i:4} below).

\begin{corollary}\label{thm:pathInAbsSameProb} Let $\mc=(S,s_0,\M,L)$ be an $\MC$ and $\psi$ a propositional formula. Then for every
$\sigma\in\FPaths{\mc_\psi}$
\begin{enumerate}
\item $\FReach{\mc_\psi,s_0,\LangSt{\psi}}=\FReach{\mc,s_0,\LangSt{\psi}}$,
\item $\measure{}{_{\mc_\psi}}{\cyl{\sigma}}=\measure{}{_{\mc}}{\cyl{\sigma}}$,
\item\label{i:3} $\gen{\mc_\psi,\sigma}=\gen{\mc,\sigma}$,
\item\label{i:4} $\measure{}{_{\mc_\psi}}{\Torrent{\mc_\psi,\sigma}}=\measure{}{_{\mc}}{\Torrent{\mc,\sigma}}$,
\item\label{i:5} $\measure{}{\!_{\Acyclic{\mc_{\psi}}}\!}{\cyl{\sigma}}=\measure{}{_{\mc}}{\Torrent{\mc,\sigma}}$,
\item\label{i:6} $\sat{\Acyclic{\mc_{\psi}}}{\leq p}{\F \psi}$ if and only if $\sat{\mc}{\leq p}{\F \psi}$, for any $p\in[0,1]$.
\end{enumerate}
\end{corollary}

\begin{dfn}[Torrent-Counterexamples]\label{dfn:torrentCounterexample} Let $\mc=(S,s_0,\M,L)$ be an $\MC$, $\psi$ a propositional formula, and $p\in[0,1]$ such that $\nonsat{\mc}{\leq p}{\F \psi}$.
Let $\ce$ be a representative counterexample to
$\sat{\Acyclic{\mc_{\psi}}}{\leq p}{\F \psi}$. We define the set
\begin{align*}
\TorRepCount{\ce}\eqdef \lbrace \gen{\mc,\sigma}\mid\sigma\in
\ce\rbrace.
\end{align*}
\noindent We call the set $\TorRepCount{\ce}$ a
\emph{torrent-counterexample} of $\ce$. Note that this set is a
partition of a representative counterexample to $\sat{\mc}{\leq p}{\F \psi}$.
Additionally, we denote by $\TorCountSet{\mc,p,\psi}$ to the set
of all torrent-counterexamples to $\sat{\mc}{\leq p}{\F \psi}$,
i.e., $\lbrace \TorRepCount{\ce}\mid \ce\in
\CountSet{\Acyclic{\mc},p,\psi}\rbrace$.
\end{dfn}

\begin{thm}\label{thm:torrentCounterexample} Let $\mc=(S,s_0,\M,L)$ be an $\MC$, $\psi$ a propositional formula, and $p\in[0,1]$ such that $\nonsat{\mc}{\leq p}{\F \psi}$.
Take $\ce$ a representative counterexample to
$\sat{\Acyclic{\mc_{\psi}}}{\leq p}{\F \psi}$. Then the set of
finite paths $\biguplus_{W\in\TorRepCount{\ce}} W$ is a
representative counterexample to $\sat{\mc}{\leq p}{\F \psi}$.
\end{thm}

Note that for each $\sigma\in \ce$ we get a witness
$\gen{\mc,\sigma}$. Also note that the number of rails is finite, so
there are also only finitely many witnesses.

Following \cite{hk_2007_counterexamples}, we extend the notions of
\emph{minimum counterexamples} and \emph{strongest evidence}.

\begin{dfn}[Minimum torrent-counterexample] Let $\mc$ be an $\MC$,
$\psi$ a propositional formula and $p\in[0,1]$. We say that
$\ce_t\in\TorCountSet{\mc,p,\psi}$ is a \emph{minimum
torrent-counterexample} if $|\ce_t|\leq|\ce_t^\prime|$, for all
$\ce^\prime_t\in\TorCountSet{\mc,p,\psi}$.
\end{dfn}

\begin{dfn}[Strongest torrent-evidence] Let $\mc$ be an $\MC$,
$\psi$ a propositional formula and $p\in[0,1]$. A \emph{strongest
torrent-evidence} to $\smash{\nonsat{\mc}{\leq p}{\F \psi}}$ is a
torrent $\Torrent{\mc,\sigma}$ such that
$\sigma\in\FPaths{\Acyclic{\mc_\psi}}$ and
$\allowbreak\measure{}{\mc}{\Torrent{\mc,\sigma}}\allowbreak\geq
\measure{}{\mc}{\Torrent{\mc,\rho}}$ for all $\rho \in
\FPaths{\Acyclic{\mc_\psi}}$.
\end{dfn}

Now we define our notion of significant diagnostic counterexamples.
It is the generalization of most indicative counterexample from
\cite{hk_2007_counterexamples} to our setting.

\begin{dfn}[Most indicative torrent-counterexample] Let $\mc$ be an $\MC$, $\psi$ a propositional formula and $p\in[0,1]$.
We say that $\ce_t\in\TorCountSet{\mc,p,\psi}$ is a \emph{most
indicative torrent-counterexample} if it is a minimum
torrent-counterexample and
$\measure{}{}{\bigcup_{T\in\ce_t}\cyl{T}}\allowbreak\geq\allowbreak
\measure{}{}{\bigcup_{T\in\ce^\prime_t}\cyl{T}}$ for all minimum
torrent-counterexamples $\ce^\prime_t\in\TorCountSet{\mc,p,\psi}$.
\end{dfn}


Note that in our setting, as in \cite{hk_2007_counterexamples}, a
minimal torrent-counterexample $\ce$ consists of the $|\ce|$
strongest torrent-evidences.

By Theorem \ref{thm:torrentCounterexample} it is possible to
obtain strongest torrent-evidence and most indicative
torrent-counterexamples of an $\MC$ $\mc$ by obtaining strongest
evidence and most indicative counterexamples of
$\Acyclic{\mc_\psi}$ respectively.

\section{Computing Counterexamples}
\label{sec:coumputing-counterexamples}

In this section we show how to compute most indicative torrent-counterexamples.
We also discuss what information to present to the user: how to present
witnesses and how to deal with overly large strongly connected components.

\subsection{Maximizing Schedulers}
\label{sec:max-schedulers}

The calculation of the maximal probability on a reachability
problem can be performed by solving a linear minimization
problem~\cite{ba_1995_probabilistic,dealfaro_1997_thesis}. This
minimization problem is defined on a system of inequalities that
has a variable $x_i$ for each different state $s_i$ and an
inequality $\sum_j\pi(s_j)\cdot x_j \leq x_i$ for each
distribution $\pi\in\tau(s_i)$.
The maximizing (deterministic memoryless) scheduler $\eta$ can be
easily extracted out of such system of inequalities after
obtaining the solution. If $p_0, \dots, p_n$ are the values that
minimize $\sum_ix_i$ in the previous system, then $\eta$ is such
that, for all $s_i$, $\eta(s_i)=\pi$ whenever $\sum_j\pi(s_j)\cdot
p_j = p_i$. In the following we denote $\Prob{}{s_i}{\F
\psi}\eqdef x_i$.
%




\subsection{Computing most indicative torrent-counterexamples}
\label{sec:computing}

We divide the computation of most indicative
torrent-counterexamples to $\sat{\mdp}{\leq p}{\F \psi}$ in three
stages: \emph{pre-processing}, \emph{$\SCC$ analysis}, and
\emph{searching}.

\paragraph{Pre-processing stage.}

We first modify the original $\MC$ $\mc$ by making all states in
$\LangSt{\psi} \cup S \setminus \SatF{\psi}$ absorbing.  In this
way we obtain the $\MC$ $\mc_{\psi}$ from Definition~\ref{dfn:abs
mc}.  Note that we do not have to spend additional computational
resources to compute this set, since $\SatF{\psi} = \lbrace s\in
S\mid \Prob{}{s}{\F\psi} > 0 \rbrace$ and hence all required data
is already available from the $\LTL$ model checking phase.

\paragraph{$\SCC$ analysis stage.}

We remove all $\SCCs$ $\scc$ of $\mc_\psi$ keeping just
\emph{input states} of $\scc$, getting the acyclic $\MC$
$\Acyclic{\mc_\psi}$ according to Definition~\ref{dfn:acyclicMC}.

To compute this, we first need to find the $\SCCs$ of $\mc_\psi$.
There exists several well known algorithms to achieve this:
Kosaraju's, Tarjan's, Gabow's algorithms (among others). We also
have to compute the reachability probability from input states to
output states of every $\SCC$. This can be done by using steady-state analysis
techniques~\cite{cassandras_1993_steadystateanalysis}.



\paragraph{Searching stage.}

To find most indicative torrent-counterexamples in $\mc$, we find
most indicative counterexamples in $\Acyclic{\mc_\psi}$.  For this
we use the same approach as \cite{hk_2007_counterexamples},
turning the MC into a weighted digraph to replace the problem of
finding the finite path with highest probability by a shortest
path problem.  The nodes of the digraph are the states of the
$\MC$ and there is an edge between $s$ and $t$ if $\M(s,t) > 0$.
The weight of such an edge is $-\log(\M(s,t))$.

Finding the most indicative counterexample in $\Acyclic{\mc_\psi}$
is now reduced to finding $k$ shortest paths. As explained in
\cite{hk_2007_counterexamples}, our algorithm has to compute $k$
on the fly. Eppstein's algorithm \cite{epps_98_k-shortest-paths}
produces the $k$ shortest paths in general in $O(m + n \log n +
k)$, where $m$ is the number of nodes and $n$ the number of edges.
In our case, since $\Acyclic{\mc_\psi}$ is acyclic, the complexity
decreases to $O(m + k)$.

\subsection{Debugging issues}\label{sec:debugging}

\paragraph{Representative finite paths.}

What we have computed so far is a most indicative counterexample
to $\sat{\Acyclic{\mc_\psi}}{\leq p}{\F \psi}$. This is a finite
set of rails, i.e., a finite set of paths in $\Acyclic{\mc_\psi}$.
Each of these paths $\sigma$ represents a witness
$\gen{\mc,\sigma}$. Note that this witness itself has usually
infinitely many elements.


In practice, one has to display a witness to the user. The obvious
way would be to show the user the rail $\sigma$. This, however,
may be confusing to the user as $\sigma$ is not a finite path of
the original Markov Decision Process. Instead of presenting the
user with $\sigma$, we therefore show the user the finite path of
$\gen{\mc,\sigma}$ with highest probability.

\begin{dfn} Let $\mc$ be an $\MC$, and $\sigma\in\FPaths{\Acyclic{\mc_\psi}}$ a rail of $\mc$.
We define the \emph{representant of} $\Torrent{\mc,\sigma}$ as
\begin{align*}
\repTorrent{\mc,\sigma}=\repTorrent{\biguplus_{\rho\in\gen{\mc,\sigma}}\cyl{\rho}}\eqdef
\arg \max_{\rho\in\gen{\mc,\sigma}}\measure{}{}{\cyl{\rho}}
\end{align*}
\end{dfn}

Note that given $\repTorrent{\mc,\sigma}$ one can easily recover
$\sigma$. Therefore, no information is lost by presenting torrents
as one of its generators instead of as a rail.

\paragraph{Expanding $\SCC$.}

\begin{wrapfigure}{r}{3cm}
\vspace{-0.3cm}
\includegraphics[width=3cm]{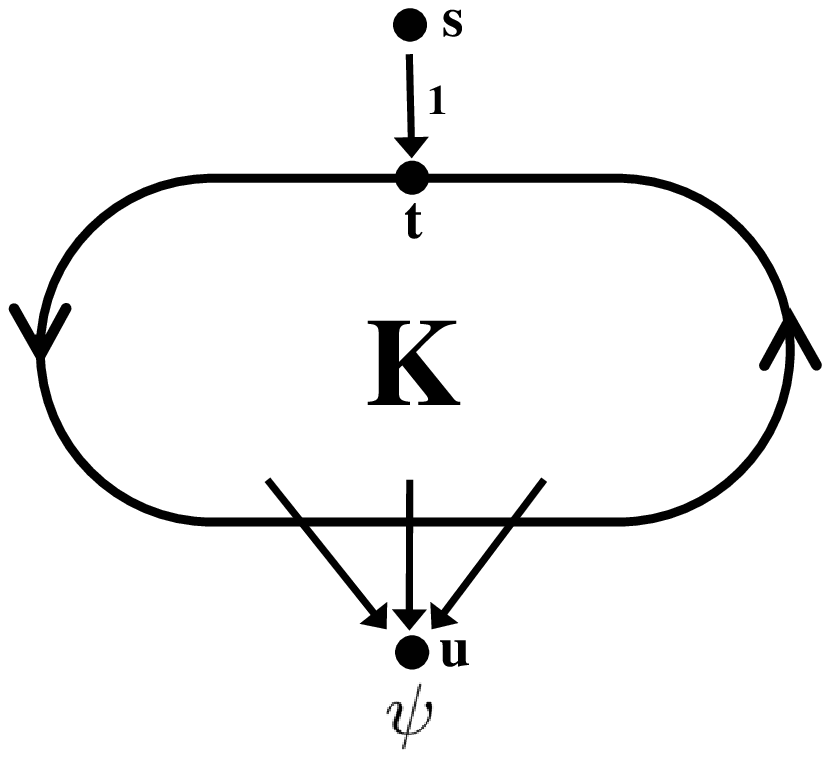}
\caption{} \label{fig:bigSCC} \vspace{-0.35cm}
\end{wrapfigure}
Note that in the Preprocessing stage, we reduced the size of many
$\SCCs$ of the system (and likely even completely removed some) by
making states in $\LangSt{\psi} \cup S \setminus \SatF{\psi}$
absorbing. However, It is possible that the system still contains
some very large strongly connected components. In that case, a
single witness could have a very large probability mass and one
could argue that the information presented to the user is not
detailed enough. For instance, consider the Markov Chain of
Figure~\ref{fig:bigSCC} in which there is a single large $\SCC$
with input state $t$ and output state $u$.

The most indicative torrent-counterexample to the property
$\sat{\mc}{\leq 0.9}{\F \psi}$ is simply $\{\gen{s t u}\}$, i.e.,
a single witness with probability mass 1 associated to the rail $s
t u$. Although this may seem uninformative, we argue that it is
more informative than listing several paths of the form $s t
\cdots u$ with probability summing up to, say, $0.91$. Our single
witness counterexample suggests that the outgoing transition to a state
not reaching $\psi$ was simply forgotten in the design; the listing of paths
still allows the possibility that one of the probabilities in the
whole system is simply wrong.

Nevertheless, if the user needs more information to tackle bugs
inside $\SCCs$, note that there is more
information available at this point. In particular, for every
strongly connected component $\scc$, every input state $s$ of
$\scc$ (even for every state in $\scc$), and every output state
$t$ of $\scc$, the probability of reaching $t$ from $s$ is already
available from the computation of $\Acyclic{\mc_\psi}$ during the
$\SCC$ analysis stage of Section~\ref{sec:computing}.

\section{Related Work}
\label{sec:conclusions}


Recently, some work has been done on counterexample generation
techniques for different variants of probabilistic models
(Discrete Markov Chains  and Continue Markov Chains )
\cite{ahl_2005_counterexamples,al_2006_search,hk_2007_counterexamples,hk_2007_counterexamplesDTMC}.
In our terminology, these works consider witnesses consisting of a
\emph{single} finite path.
We have already discussed in the Introduction that the single path
approach does not meet the properties of accuracy, originality,
significance, and finiteness.

Instead, our witness/torrent approach provides a high level of
abstraction of a counterexample.
By grouping together finite paths that behave the same outside
strongly connected components in a single witness, we can achieve
these properties to a higher extent.
Behaving the same outside strongly connected components is a
reasonable way of formalizing the concept of providing
\emph{similar}
debugging information.
This grouping also makes witnesses significantly different from
each other: each witness comes from a different rail and each rail
provides a different way to reach the undesired property. Then
each witness provides \emph{original} information. Of course, our
witnesses are more \emph{significant} than single finite paths,
because they are sets of finite paths.
This also gives us more \emph{accuracy} than the approach with
single finite paths, as a collection of finite paths behaving the
same and reaching an undesired condition with high probability is
more likely to show how the system reaches this condition than
just a single path.  Finally, because there is a finite number of
rails, there is also a \emph{finite} number of witnesses.

\rev{Another key difference of our work with respect to previous ones is that our
technique allows us to generate counterexamples for probabilistic
systems \emph{with} nondeterminism. However, an independent and concurrent study of counterexample generation for $\MDPs$ was carried out by Aljazzar and Leue~\cite{al_2007_counterexamplesMDP}. There, the authors consider generating counterexamples for a fragment of $\pCTL$, namely upper bounded formulas without nested
temporal operators. The authors present three methods for generating counterexamples and study conditions under which these methods are suitable. 

More recently, Schmalz et al. also investigated quantitative counterexample generation for LTL formulas \cite{Schmalz:09:CONCUR}.  In qualitative probabilistic model checking, a counterexample is presented as a pair $(\alpha, \gamma)$, where $\alpha$ and $\gamma$ are finite words such that all paths that extend $\alpha$ and have infinitely many occurrences of $\gamma$ violate the property under consideration. In quantitative probabilistic model checking, a counterexample is presented as a pair $(W,R)$, where $W$ is a set of such finite words $\alpha$ and $R$ is a set of such finite words $\gamma$. 

Similar $\SCC$ reduction techniques to the one presented in this paper have been studied for different purposes. In \cite{Guen:2002:SM},  the authors focus on the problem of software testing. They use Markov chains to model software behaviour and $\SCC$ analysis to decompose the state space of large Markov chains. More recently, {\'A}brah{\'a}m et al. presented a model checker for Markov chains based on the detection and abstraction of strongly connected components~\cite{Abraham:10:QEST}. Their algorithm has the advantage of offering abstract counterexamples, which can be interactively refined by the user.

Finally, the problem of presenting counterexamples as single paths has also been observed by Han, Katoen, and Damman~\cite{Damman:08:QEST,Han:09:TSE}. There, the authors propose to use regular expressions to group paths together. Thus, in the same way that we group together paths behaving the same outside $\SCC$,  they group together paths associated to the same regular expression.

For a more extensive survey on quantitative counterexample generation for (both discrete and continuous time) Markov chains we refer the reader to chapters $3$, $4$, and $5$ of \cite{Han:2009:THESIS}.}

%

\newpage
\thispagestyle{empty}

\chapter{Interactive Systems and Equivalences for Security}
\label{ch.pif}
\begin{quote}
\textit{\rev{In this overview chapter we briefly discuss extensions to the frameworks presented in Chapters 3 and 4\footnote{For more information about the topics discussed in this chapter we refer the reader to \cite{Alvim:10:CONCUR,Alvim:10:TRa,Alvim:10:LICS,Alvim:10:IFIP-TCS}.}}. First, we consider the case in which secrets and
observables interact (in contrast with the situation in Chapter
3), and show that it is still possible to define an
information-theoretic notion of leakage, provided that we consider
a more complex notion of channel, known in literature as
\emph{channel with memory and feedback}. Second, we extend the
systems proposed in Chapter 4 by allowing nondeterminism also
internally to the components. Correspondingly, we define a richer
notion of admissible scheduler suitable and we use it for defining
notion of process equivalences relating to nondeterminism in a
more flexible way than the standard ones in the literature. In
particular, we use these equivalences  for defining notions of
anonymity robust with respect to implementation refinement.}
\end{quote}



\section{Interactive Information Flow}\label{sec:interactive}

In this section we discuss the applicability of the
information-theoretic approach to interactive systems. These
systems were already considered in \cite{Desharnais:02:LICS}. In
that paper the authors proposed to define the matrix elements
$\PP(b\,|\, a)$ as the measure of the traces with (secret,
observable)-projection $(a,b)$, divided by the measure of the
trace with secret projection $a$. This  follows the definition of
conditional probability in terms of joint and marginal
probability.  However, this approach does not  lead to an
information-theoretic channel. In fact, (by definition) a channel
should be invariant with respect to the input distribution and
such construction is not (as shown by Example ~\ref{exa:isFails}).

In \cite{Alvim:10:CONCUR} and more recently in
\cite{Alvim:10:TRa}, we consider an extension of the theory of
channels which  makes the information-theoretic approach
applicable also the case of interactive systems. It turns out that
a richer notion of channel, known in Information Theory as
\emph{channels with memory and feedback}, serves our purposes. The
dependence of inputs on previous outputs corresponds  to feedback,
and the dependence of outputs on previous inputs and outputs
corresponds to memory.

Let us explain more in detail the difference with the classical
approach. In non-interactive systems, since the secrets always
precede the observables, it is possible to group the sequence of
secrets (and observables) in a single secret (respectively. observable)
string. If we consider only one activation of the system, or if
each use of the system is independent from the other, then we can
model it as a discrete classical channel (memoryless, and without
feedback) from a single input string to a single output string.
When we have interactive systems, however, inputs and outputs may
interleave and influence each other. Considering some sort of
feedback in the channel is a way to capture this richer behavior.
Secrets have a causal influence on observables via the channel,
and, in the presence of interactivity, observables have a causal
influence on secrets via the feedback. This alternating mutual
influence between inputs and outputs can be modeled by repeated
uses of the channels. However, each time the channel is used it
represents a different state of the computation, and the
conditional probabilities of observables on secrets can depend on
this state. The addition of memory to the model allows expressing
the dependency of the channel matrix on such a state (which, as we
will see, can also be represented by the history of inputs and
outputs).

Recent results in Information Theory~\cite{Tatikonda:09:TIT} have
shown that, in channels with memory and feedback, the transmission
rate does not correspond to the maximum mutual information
(capacity), but rather to the maximum of the so-called
\emph{directed information}. Intuitively, this is due to the fact
that mutual information expresses the correlation between the
input and the output, and therefore it includes feedback. However, the
feedback, i.e the way the output influences the next input, should
not be considered part of the information transmitted. Directed
information is essentially mutual information \emph{minus} the
dependence of the next input on previous output. We propose to
adopt directed information and the corresponding notion of
directed capacity to represent leakage.

Our extension is a generalization of the classical model, in the
sense that it can represent both interactive and non-interactive
systems. One important feature of the classical approach is that
the choice of secrets is seen as external to the system, i.e.
determined by the environment. This implies that the probability
distribution on the secrets (input distribution) constitutes the a
priori knowledge and does not count as leakage. In order to
encompass the classical approach, in our extended model we should
preserve this principle, and the most natural way is to consider
the secret choices, at every stage of the computation, as
external. Their probability distributions, which are now in
general conditional probability distributions (depending on the
history of secrets and observables) should be considered as part
of the external knowledge, and should not be counted as leakage.

A second contribution of \cite{Alvim:10:CONCUR} and
\cite{Alvim:10:TRa} is the proof that the channel capacity is a
continuous function of the Kantorovich metric on interactive
systems. This was pointed out also in \cite{Desharnais:02:LICS},
however their construction does not work in our case due to the
fact (as far as we understand) it assumes that the probability of
a secret action (in any point of the computation) is different
from 0. This assumption is not guaranteed in our case and
therefore we had to come out with a different reasoning. The fact
that our proof does not need this assumption shows that the
intuition of \cite{Desharnais:02:LICS} concerning the continuity
of capacity is valid in general.

\subsection{Applications}
Interactive systems can be found in a variety of disparate areas
such as game theory, auction protocols, and zero-knowledge proofs.
We now present two examples of interactive systems.

\begin{itemize}
\item In the area of auction protocols, consider the cocaine
auction protocol \cite{Stajano:99:IH}. The auction is organized as
a succession of rounds of bidding. Round $i$ starts with the
seller announcing the bid price $b_i$ for that round. Buyers have
$t$ seconds to make an offer (i.e. to say \emph{yes}, meaning ``I
am willing to buy at the current bid price $b_i$''). As soon as
one buyer says \emph{yes}, he becomes the winner $w_i$ of that
round and a new round begins. If nobody says anything for $t$
seconds, round $i$ is concluded by timeout and the auction is won
by the winner $w_{i-1}$ of the previous round. The identities of
the buyers in each round constitute the input of the channel,
whereas the bid prices constitute the output of the channel. Note
that inputs and outputs alternate so the system is interactive. It
is also easy to see that inputs depend on past outputs (feedback):
the identity of the winner of each round depends on the previous
bid prices. Furthermore, outputs depend on the previous inputs
(memory): (in some scenarios) the bid price of round $i$ may
depend on the identity of previous winners. For more details on
the modeling of this protocol using channels with memory and
feedback see \cite{Alvim:10:TRa}.

\item In the area of game theory, consider the classic prisoner's
dilemma (the present formulation is due to Albert W.
Tucker~\cite{Poundstone:92:Doubleday}, but it was originally
devised by  Merrill Flood and Melvin Dresher in 1950). Two
suspects are arrested by the police. The police have insufficient
evidence for a conviction, and, having separated both prisoners,
visit each of them to offer the same deal. If one testifies
(defects from the other) for the prosecution against the other and
the other remains silent (cooperates with the other), the betrayer
goes free and the silent accomplice receives the full 10-year
sentence. If both remain silent, both prisoners are sentenced to
only six months in jail for a minor charge. If each betrays the
other, each receives a five-year sentence. Each prisoner must
choose to betray the other or to remain silent. Each one is
assured that the other would not know about the betrayal before
the end of the investigation. In the iterated prisoner's dilemma,
the game is played repeatedly. Thus each player has an opportunity
to punish the other player for previous non-cooperative play. In
this case the strategy (cooperate or defect) of each player is the
input of the channel and the sentence is the output. Once again,
it is easy to see that the system is interactive:  inputs and
outputs alternate. Furthermore, inputs depend on previous outputs
(the strategy depend on the past sentences) and outputs depend on
previous inputs (the sentence of the suspects depend on their
declarations - cooperate or defect).
\end{itemize}

\section{Nondeterminism and Information Flow}\label{sec:nondeterminism}
The \emph{noise} of channel matrices, i.e. the similarity between the
rows of the channel matrix, helps preventing the inference of the \comment{Check this first paragraph}
secret from the observables. In practice noise is created by using
randomization, see for instance the DCNet \cite{Chaum:88:JC} and
the Crowds~\cite{Reiter:98:TISS} protocols.

In the literature about the foundations  of Computer Security, however, the quantitative aspects are often abstracted away, and probabilistic  behavior is replaced by nondeterministic behavior.
Correspondingly, there have been various approaches in which information-hiding properties are expressed in terms of equivalences based on nondeterminism, especially in a concurrent setting. For instance,
\cite{Schneider:96:ESORICS} defines \emph{anonymity} as follows\footnote{The actual definition of \cite{Schneider:96:ESORICS} is more complicated, but the spirit is the same.}: A protocol $S$ is anonymous if, for every pair of culprits $a$ and $b$, $S[^a/ _x]$ and $S[^b/ _x]$ produce the same observable traces.
A  similar definition is given in \cite{Abadi:99:IC} for \emph{secrecy}, with the difference that
$S[^a / _x]$ and $S[^b / _x]$ are required to be bisimilar. In \cite{Delaune:09:JCS}, an electoral system $S$ preserves the \emph{confidentiality of the vote} if
for any voters $v$ and $w$, the observable behavior of $S$ is the same if we swap the votes of $v$ and $w$.
Namely, $S[^a / _v\mid ^b / _w] \sim S[^b / _v \mid ^a / _w]$, where $\sim$ represents bisimilarity.

These proposals are based on the implicit assumption that
\emph{all the nondeterministic executions present in the  specification of $S$ will always be possible under every implementation of $S$}.
Or at least, that the adversary will believe so.
In concurrency, however, as argued in \cite{Chatzikokolakis:09:FOSSACS}, nondeterminism has a rather different meaning:
if a specification $S$ contains some nondeterministic alternatives, typically it is because we want to abstract from specific implementations, such as the scheduling policy.
A specification is considered correct, with respect to some property,  if every alternative satisfies the property. Correspondingly, an implementation is considered correct if all executions are among those possible in the specification, i.e. if the implementation is a refinement of the specification.  There is no expectation that the  implementation will actually make possible all the
alternatives indicated by the specification.

We argue that the use of nondeterminism in concurrency corresponds
to a \emph{demonic} view: the scheduler\index{scheduler!demonic},
i.e. the entity that will decide which alternative to select, may
try to choose the worst alternative. Hence we need to make sure
that ``all alternatives are good'', i.e. satisfy the intended
property. In  the above mentioned approaches to the formalization
of security properties, on the contrary, the interpretation of
nondeterminism is \emph{angelic}: the
scheduler\index{scheduler!angelic} is expected to actually help
the protocol to confuse the adversary and thus protect the secret
information.

There is another issue, orthogonal to the angelic/demonic dichotomy, but relevant for the achievement of security properties:   the scheduler \emph{should not be able to make its choices dependent on the secret}, or else nearly every protocol would be insecure, i.e. the scheduler would always be able to leak the secret to an external observer (for instance by producing different interleavings of the observables, depending on the secret).
This remark has been made several times already, and several approaches have  been proposed to cope with the problem of the ``almighty'' scheduler (aka omniscient, clairvoyant, etc.), see for example  \cite{Canetti:06:WODES,Giro:07:FORMATS, Chatzikokolakis:09:FOSSACS,Andres:10:Admissible,Chatzikokolakis:10:IC}.

The risk of a naive use of nondeterminism to specify a security property, is not only that it may rely on an implicit assumption that the scheduler behaves angelically, but also that it is clairvoyant, i.e. that it
peeks at the secrets (that it is not supposed to be able to see) to achieve its angelic strategy.


Consider the following system, in a CCS-like syntax:
\[S \stackrel{\rm def}{=} (c,\mathit{out})  (A \ \parallel\ \mathit{Corr} \ \parallel\ H_1 \ \parallel \ H_2  ),\]

\noindent with $A \stackrel{\rm def}{=} \overline{c}\langle
\mathit{sec}\rangle$ ,$\mathit{Corr} \stackrel{\rm def}{=}
c(s).\overline{out}\langle s \rangle$, $H_1 \stackrel{\rm def}{=}
c(s).\overline{out}\langle a \rangle$, $H_2 \stackrel{\rm def}{=}
c(s).\overline{out}\langle b \rangle$ and where $\parallel$ is the
parallel operator, $\overline{c}\langle sec\rangle$ is a process
that sends $\mathit{sec}$ on channel $c$, $c(s).P$ is a process
that receives $s$ on channel $c$ and then continues as $P$, and $
(c,\mathit{out}) $ is the restriction operator, enforcing
synchronization  on $c$ and $\mathit{out}$. In this example,
$\mathit{sec}$ represents a secret information.

It is easy to see that we have $S\left[ ^{a}/_{sec}\right]\sim S\left[ ^{b}/_{sec}\right]$. Note that, in order to  simulate the third branch in $S\left[ ^{a}/_{sec}\right]$, the process $S\left[ ^{b}/_{sec}\right]$ needs to select its first branch. Viceversa, in order  to  simulate the third branch in $S\left[ ^{b}/_{sec}\right]$, the process $S\left[ ^{a}/_{sec}\right]$ needs to select its second branch. This means that, in order to achieve bisimulation,  the scheduler needs to know the secret, and change its choice accordingly.

This example shows a system that intuitively is not secure, because the third component, $\mathit{Corr}$, reveals whatever secret it receives. However, according to the equivalence-based  notions of security discussed above, \emph{it is secure}. But it is secure thanks to a scheduler that angelically helps the system to protect the secret, and it does so by making its choices dependent on the secret! In our opinion  these assumptions on the scheduler are excessively strong.


In a recent work \cite{Alvim:10:IFIP-TCS} we address the above
issue by defining a framework in which it is possible to combine
both angelic and demonic nondeterminism in a setting in which also
probabilistic behavior may be present, and in a context in which
the scheduler is restricted (i.e. not clairvoyant). We propose
safe versions of typical equivalence relations (traces and
bisimulation), and we show how to use them to characterize
information-hiding properties.

\newpage
\thispagestyle{empty}

\chapter{Conclusion}
\label{chap:conclusion}
\rev{
\begin{quote}
\textit{In this chapter we summarize the main contributions of this thesis and discuss further directions.}
\end{quote} 

\section{Contributions}

%
%

The goal of this thesis is to develop a formal framework for specifying, analyzing and verifying anonymity protocols and, more in general, information hiding protocols. 

As discussed in the Introduction, conditional probabilities are a key concept in assessing the degree of information protection. In Chapter 2, we have extended the probabilistic temporal logic \pCTL to \cpCTL, in which it is
possible to express conditional probabilities. We have also proved that optimal scheduling decisions can always be
reached by a deterministic and semi history-independent scheduler. This fundamental result, allowed us to define an algorithm to verify \cpCTL formulas. Our algorithm first reduces the \MDP to an acyclic \MDP and then
computes optimal conditional probabilities in the acyclic \MDP. In addition, we have defined a notion of counterexample for conditional formulas and sketched an algorithm for counterexample generation. 

We then turned our attention to more practical grounds. In Chapter 3, we have addressed the problem of
computing the information leakage of a system in an efficient way. We have proposed two methods: one based on
reachability techniques and the other based on quantitative counterexample generation. In addition, we have shown that when the automaton is interactive it is not possible to define its channel in the standard way.  An intriguing problem is how to extend the notion of channel so to capture the dynamic nature of interaction. In Chapter 6 we have briefly discussed how to solve this problem by using more complex information theoretic channels, namely channels with history and feedback.

In Chapter 4, we have attacked a well known problem of concurrent information-hiding protocols, namely full-information
scheduling. To overcome this problem, we have defined a class of partial-information schedulers which can only base their decisions on the information that they have available. In particular they cannot base their decisions on the internal behavior of the components. We have used admissible schedulers to resolve nondeterminism in a realistic way, and to revise some anonymity definitions from the literature. In addition, we have presented a technique to prove the various definitions of anonymity proposed in the chapter. This is particularly interesting considering that many problems related to restricted schedulers have been shown to be undecidable. We have illustrated  the applicability of our proof technique by proving that the well-known DC protocol is anonymous when considering admissible schedulers, in contrast to the situation when considering full-information schedulers.

The last major contribution of this thesis is a novel technique for representing and computing
counterexamples for nondeterministic and probabilistic systems. In Chapter 5, we have shown how to carefully 
partition a counterexample in sets of paths. These sets are intended to provide information related to the violation of the property under consideration, so we call them \emph{witnesses}. Five properties that witnesses should satisfy (in order to provide significant debugging information) are identified in this chapter. The key contribution of this chapter is a technique based on strongly connected component analysis that makes it possible to partition counterexamples into witnesses satisfying the desired properties.

\section{Further directions}

There are several ways of extending the work presented in this thesis. 

As we have shown in Chapter 2, the most important issue when computing conditional probabilities is that optimizing schedulers are not determined by the local structure of the system. As a consequence, it is not possible to reduce the problem of verifying \cpCTL to a linear optimization problem (as it is the case with \pCTL). A natural question arising from this observation, is whether the problem of model checking conditional probabilities is inherently exponential or not. We believe that it is; however we are of the idea that it is also possible to find suitable restrictions (either to the formulas or to the systems under consideration) that would make it possible to model check conditional probabilities in polynomial time.

In a more practical matter, counterexample generation for probabilistic model checking is nowadays a very hot topic for which several applications in the most diverse areas have been identified. During the last few years, many techniques have been proposed for different flavours of logics and models. However, to the best of our knowledge, no practical tool to automatically generate quantitative counterexamples has been implemented. We believe that such a practical tool could be a significant contribution to the field. More concretely, we believe that a tool implementing the regular-expression and k-shortest path techniques introduced by Han et al. in combination with the SCC analysis techniques presented in this thesis would be of great value.

In Chapter 2, we have made a connection between quantitative counterexample generation and information leakage computation. Thanks to this connection, such a tool would also allow us to compute / approximate leakage of large scale protocols. Furthermore, it would make it possible to investigate in more depth how the debugging information provided by the tool can be used to identify flaws of the protocol causing high leakage.

Finally, as for most definitions of partial-information schedulers from the literature, our notions of admissible schedulers may raise undecidability issues. Thus, it would be interesting to investigate whether the notions of anonymity proposed in Chapter 4 are actually verifiable (remember that the proof technique we proposed is sufficient but not necessary). Another interesting direction for future work is to adapt well known isomorphism-checking algorithms and tolls to our setting in order to automatically verify some anonymity properties.

}
%


\newpage

\thispagestyle{empty}
\bibliographystyle{alpha}
\cleardoublepage
\addcontentsline{toc}{chapter}{Bibliography}
\bibliography{cpCTL,thesis2,bib_lics,bib_ihs,bib_adm_sch}
\pagebreak

\cleardoublepage
\notag{\chapter*{Samenvatting}}
\addcontentsline{toc}{chapter}{Samenvatting (Dutch Summary)}

Terwijl we het digitale tijdperk ingaan zijn er immer groeiende zorgen
over de hoeveelheid digitale data die over ons verzameld wordt. 
Websites houden vaak het browse-gedrag van mensen bij,
ziektenkostenverzekeraars verzamelen medische gegegevens en smartphones
en navigatiesystemen versturen informatie die het mogelijk maakt de
fysieke locatie van hun gebruikers te bepalen. Hierdoor staan anonimiteit,
en privacy in het algemeen, steeds meer op het spel. Anonimiteitsprotocollen
proberen iets tegen deze tendens te doen door anonieme communicatie over
het Internet mogelijk te maken. Om de correctheid van dergelijke protocollen,
die vaak extreem complex zijn, te garanderen, is een degelijk framework
vereist waarin anonimiteitseigenschappen kunnen worden uitgedrukt en
geanalyseerd. Formele methoden voorzien in een verzameling wiskundige
technieken die het mogelijk maken anonimiteitseigenschappen rigoreus
te specificeren en te verifi\"eren.

Dit proefschrift gaat over de grondslagen van formele methoden voor 
toepassingen in computerbeveiliging en in het bijzonder anonimiteit.
Concreet, we ontwikkelen frameworks om anonimiteitseigenschappen te
specificeren en algoritmen om ze te verifi\"eren. Omdat in de praktijk
anonimiteitsprotocollen altijd {\em wat} informatie lekken, leggen we de
focus op quantitatieve eigenschappen die de {\em mate} van gelekte
informatie van een protocol beschrijven.

We beginnen het onderzoek naar anonimiteit vanuit de basis, namelijk 
voorwaardelijke kansen. Dit zijn de sleutelingredi\"enten van de meeste
quantitatieve anonimiteitsprotocollen. In Hoofdstuk 2 prenteren we
cpCTL, de eerste temporele logica waarin voorwaardelijke kansen kunnen
worden uitgedrukt. We presenteren ook een algoritme om cpCTL formules
te verifi\"eren met een modelchecker. Samen met een modelchecker maakt
deze logica het mogelijk om quantitatieve anomimiteitseigenschappen van
complexe systemen waarin zowel probabilistisch als nondeterministisch
gedrag voorkomt te specificeren en verifi\"eren.

Vervolgens gaan we meer de praktijk in: de constructie van algoritmen
die de mate van het lekken van informatie meten. Om preciezer te zijn,
Hoofdstuk 3 beschrijft polynomiale algoritmen om de (informatie-theoretische)
information leakage te quantificeren voor verscheidene soorten 
volledig probabilistische protocllen (d.w.z., protocollen zonder
nondeterministisch gedrag). The technieken uit dit hoofdstuk zijn
de eerste die het mogelijk maken de informatie leakage voor
interactieve protocollen te berekenen.

\thispagestyle{empty}

In Hoofdstuk 4 behandelen we een bekend probleem in gedistribueerde
anonimiteitsprotocollen, namelijk schedulers met volledige informatie.
Om dit probleem op te lossen stellen we een alternatieve definitie
van scheduler voor, samen met nieuwe definities voor anonomiteit
(vari\"erend met de capaciteiten van de aanvaller) en herzien de
bekende definitie van sterke anonimiteit uit de literatuur. Bovendien
laten we een techniek zien waarmee gecontroleerd kan worden of een
gedistribueerd protocol aan enkele van deze definities voldoet.

\thispagestyle{empty}

In Hoofdstuk 5 laten we op tegenvoorbeelden gebaseerde technieken zien
die het mogelijk maken complexe systemen te debuggen. Dit maakt het
mogelijk fouten in security protocollen op te sporen.
Tenslotte, in Hoofdstuk 6, beschrijven we kort uitbreidingen van de
frameworks en technieken uit Hoofdstukken 3 en 4.

\vfill 

\cleardoublepage
\addcontentsline{toc}{chapter}{Index}
\printindex
\thispagestyle{empty}
\vfill

\notag{\chapter*{Curriculum Vitae}}
\addcontentsline{toc}{chapter}{Curriculum Vitae}

\vspace{2cm}

\begin{description}
\item [1980] Born on 2 July, R\'io Cuarto, Argentina.
\item [1994--1998] Private Institute Galileo Galilei (High School), R\'io Cuarto, C\'{o}rdoba, Argentina.
\item[1999--2006] Computer Science Licentiate (equivalent to MSc.), Faculty of Mathematics, Astronomy and Physics (Fa.M.A.F.). National University of C\'{o}rdoba (UNC), Argentina.
\item[2006--2010] PhD student in the Digital Security Group, Radboud University Nijmegen, The Netherlands.
\item[2010--] Postdoctoral researcher in the Com\`ete Team, Laboratory of Informatics of the \'Ecole Polytechnique (LIX), France.
\end{description}

%
%
%
%
%
%
%
%

\newcommand*{\promitem}[4]{\noindent \textbf{#1}. \emph{#2}. #3.
#4\medskip}

\clearpage \pagestyle{empty}

\setlength{\columnsep}{2em}
\begin{multicols}{2}
        [\subsection*{Titles in the IPA Dissertation Series since 2005}]

\promitem{E. \'{A}brah\'{a}m}
         {An Assertional Proof System for Multithreaded Java -Theory and Tool Support- }
         {Faculty of Mathematics and Natural Sciences, UL}
         {2005-01}

\promitem{R. Ruimerman}
         {Modeling and Remodeling in Bone Tissue}
         {Faculty of Biomedical Engineering, TU/e}
         {2005-02}

\promitem{C.N. Chong}
         {Experiments in Rights Control - Expression and Enforcement}
         {Faculty of Electrical Engineering, Mathematics \& Computer Science, UT}
         {2005-03}

\promitem{H. Gao}
         {Design and Verification of Lock-free Parallel Algorithms}
         {Faculty of Mathematics and Computing Sciences, RUG}
         {2005-04}

\promitem{H.M.A. van Beek}
         {Specification and Analysis of Internet Applications}
         {Faculty of Mathematics and Computer Science, TU/e}
         {2005-05}

\promitem{M.T. Ionita}
         {Scenario-Based System Architecting - A Systematic Approach to Developing Future-Proof System Architectures}
         {Faculty of Mathematics and Computing Sciences, TU/e}
         {2005-06}

\promitem{G. Lenzini}
         {Integration of Analysis Techniques in Security and Fault-Tolerance}
         {Faculty of Electrical Engineering, Mathematics \& Computer Science, UT}
         {2005-07}

\promitem{I. Kurtev}
         {Adaptability of Model Transformations}
         {Faculty of Electrical Engineering, Mathematics \& Computer Science, UT}
         {2005-08}

\promitem{T. Wolle}
         {Computational Aspects of Treewidth - Lower Bounds and Network Reliability}
         {Faculty of Science, UU}
         {2005-09}

\promitem{O. Tveretina}
         {Decision Procedures for Equality Logic with Uninterpreted Functions}
         {Faculty of Mathematics and Computer Science, TU/e}
         {2005-10}

\promitem{A.M.L. Liekens}
         {Evolution of Finite Populations in Dynamic Environments}
         {Faculty of Biomedical Engineering, TU/e}
         {2005-11}

\promitem{J. Eggermont}
         {Data Mining using Genetic Programming: Classification and Symbolic Regression}
         {Faculty of Mathematics and Natural Sciences, UL}
         {2005-12}

\promitem{B.J. Heeren}
         {Top Quality Type Error Messages}
         {Faculty of Science, UU}
         {2005-13}

\promitem{G.F. Frehse}
         {Compositional Verification of Hybrid Systems using Simulation Relations}
         {Faculty of Science, Mathematics and Computer Science, RU}
         {2005-14}

\promitem{M.R. Mousavi}
         {Structuring Structural Operational Semantics}
         {Faculty of Mathematics and Computer Science, TU/e}
         {2005-15}

\promitem{A. Sokolova}
         {Coalgebraic Analysis of Probabilistic Systems}
         {Faculty of Mathematics and Computer Science, TU/e}
         {2005-16}

\promitem{T. Gelsema}
         {Effective Models for the Structure of pi-Calculus Processes with Replication}
         {Faculty of Mathematics and Natural Sciences, UL}
         {2005-17}

\promitem{P. Zoeteweij}
         {Composing Constraint Solvers}
         {Faculty of Natural Sciences, Mathematics, and Computer Science, UvA}
         {2005-18}

\promitem{J.J. Vinju}
         {Analysis and Transformation of Source Code by Parsing and Rewriting}
         {Faculty of Natural Sciences, Mathematics, and Computer Science, UvA}
         {2005-19}

\promitem{M.Valero Espada}
         {Modal Abstraction and Replication of Processes with Data}
         {Faculty of Sciences, Division of Mathematics and Computer Science, VUA}
         {2005-20}

\promitem{A. Dijkstra}
         {Stepping through Haskell}
         {Faculty of Science, UU}
         {2005-21}

\promitem{Y.W. Law}
         {Key management and link-layer security of wireless sensor networks: energy-efficient attack and defense}
         {Faculty of Electrical Engineering, Mathematics \& Computer Science, UT}
         {2005-22}

\promitem{E. Dolstra}
         {The Purely Functional Software Deployment Model}
         {Faculty of Science, UU}
         {2006-01}

\promitem{R.J. Corin}
         {Analysis Models for Security Protocols}
         {Faculty of Electrical Engineering, Mathematics \& Computer Science, UT}
         {2006-02}

\promitem{P.R.A. Verbaan}
         {The Computational Complexity of Evolving Systems}
         {Faculty of Science, UU}
         {2006-03}

\promitem{K.L. Man and R.R.H. Schiffelers}
         {Formal Specification and Analysis of Hybrid Systems}
         {Faculty of Mathematics and Computer Science and Faculty of Mechanical Engineering, TU/e}
         {2006-04}

\promitem{M. Kyas}
         {Verifying OCL Specifications of UML Models: Tool Support and Compositionality}
         {Faculty of Mathematics and Natural Sciences, UL}
         {2006-05}

\promitem{M. Hendriks}
         {Model Checking Timed Automata - Techniques and Applications}
         {Faculty of Science, Mathematics and Computer Science, RU}
         {2006-06}

\promitem{J. Ketema}
         {B\"ohm-Like Trees for Rewriting}
         {Faculty of Sciences, VUA}
         {2006-07}

\promitem{C.-B. Breunesse}
         {On JML: topics in tool-assisted verification of JML programs}
         {Faculty of Science, Mathematics and Computer Science, RU}
         {2006-08}

\promitem{B. Markvoort}
         {Towards Hybrid Molecular Simulations}
         {Faculty of Biomedical Engineering, TU/e}
         {2006-09}

\promitem{S.G.R. Nijssen}
         {Mining Structured Data}
         {Faculty of Mathematics and Natural Sciences, UL}
         {2006-10}

\promitem{G. Russello}
         {Separation and Adaptation of Concerns in a Shared Data Space}
         {Faculty of Mathematics and Computer Science, TU/e}
         {2006-11}

\promitem{L. Cheung}
         {Reconciling Nondeterministic and Probabilistic Choices}
         {Faculty of Science, Mathematics and Computer Science, RU}
         {2006-12}

\promitem{B. Badban}
         {Verification techniques for Extensions of Equality Logic}
         {Faculty of Sciences, Division of Mathematics and Computer Science, VUA}
         {2006-13}

\promitem{A.J. Mooij}
         {Constructive formal methods and protocol standardization}
         {Faculty of Mathematics and Computer Science, TU/e}
         {2006-14}

\promitem{T. Krilavicius}
         {Hybrid Techniques for Hybrid Systems}
         {Faculty of Electrical Engineering, Mathematics \& Computer Science, UT}
         {2006-15}

\promitem{M.E. Warnier}
         {Language Based Security for Java and JML}
         {Faculty of Science, Mathematics and Computer Science, RU}
         {2006-16}

\promitem{V. Sundramoorthy}
         {At Home In Service Discovery}
         {Faculty of Electrical Engineering, Mathematics \& Computer Science, UT}
         {2006-17}

\promitem{B. Gebremichael}
         {Expressivity of Timed Automata Models}
         {Faculty of Science, Mathematics and Computer Science, RU}
         {2006-18}

\promitem{L.C.M. van Gool}
         {Formalising Interface Specifications}
         {Faculty of Mathematics and Computer Science, TU/e}
         {2006-19}

\promitem{C.J.F. Cremers}
         {Scyther - Semantics and Verification of Security Protocols}
         {Faculty of Mathematics and Computer Science, TU/e}
         {2006-20}

\promitem{J.V. Guillen Scholten}
         {Mobile Channels for Exogenous Coordination of Distributed Systems: Semantics, Implementation and Composition}
         {Faculty of Mathematics and Natural Sciences, UL}
         {2006-21}

\promitem{H.A. de Jong}
         {Flexible Heterogeneous Software Systems}
         {Faculty of Natural Sciences, Mathematics, and Computer Science, UvA}
         {2007-01}

\promitem{N.K. Kavaldjiev}
         {A run-time reconfigurable Network-on-Chip for streaming DSP applications}
         {Faculty of Electrical Engineering, Mathematics \& Computer Science, UT}
         {2007-02}

\promitem{M. van Veelen}
         {Considerations on Modeling for Early Detection of Abnormalities in Locally Autonomous Distributed Systems}
         {Faculty of Mathematics and Computing Sciences, RUG}
         {2007-03}

\promitem{T.D. Vu}
         {Semantics and Applications of Process and Program Algebra}
         {Faculty of Natural Sciences, Mathematics, and Computer Science, UvA}
         {2007-04}

\promitem{L. Brand\'an Briones}
         {Theories for Model-based Testing: Real-time and Coverage}
         {Faculty of Electrical Engineering, Mathematics \& Computer Science, UT}
         {2007-05}

\promitem{I. Loeb}
         {Natural Deduction: Sharing by Presentation}
         {Faculty of Science, Mathematics and Computer Science, RU}
         {2007-06}

\promitem{M.W.A. Streppel}
         {Multifunctional Geometric Data Structures}
         {Faculty of Mathematics and Computer Science, TU/e}
         {2007-07}

\promitem{N. Tr\v{c}ka}
         {Silent Steps in Transition Systems and Markov Chains}
         {Faculty of Mathematics and Computer Science, TU/e}
         {2007-08}

\promitem{R. Brinkman}
         {Searching in encrypted data}
         {Faculty of Electrical Engineering, Mathematics \& Computer Science, UT}
         {2007-09}

\promitem{A. van Weelden}
         {Putting types to good use}
         {Faculty of Science, Mathematics and Computer Science, RU}
         {2007-10}

\promitem{J.A.R. Noppen}
         {Imperfect Information in Software Development Processes}
         {Faculty of Electrical Engineering, Mathematics \& Computer Science, UT}
         {2007-11}

\promitem{R. Boumen}
         {Integration and Test plans for Complex Manufacturing Systems}
         {Faculty of Mechanical Engineering, TU/e}
         {2007-12}

\promitem{A.J. Wijs}
         {What to do Next?: Analysing and Optimising System Behaviour in Time}
         {Faculty of Sciences, Division of Mathematics and Computer Science, VUA}
         {2007-13}

\promitem{C.F.J. Lange}
         {Assessing and Improving the Quality of Modeling: A Series of Empirical Studies about the UML}
         {Faculty of Mathematics and Computer Science, TU/e}
         {2007-14}

\promitem{T. van der Storm}
         {Component-based Configuration, Integration and Delivery}
         {Faculty of Natural Sciences, Mathematics, and Computer Science,UvA}
         {2007-15}

\promitem{B.S. Graaf}
         {Model-Driven Evolution of Software Architectures}
         {Faculty of Electrical Engineering, Mathematics, and Computer Science, TUD}
         {2007-16}

\promitem{A.H.J. Mathijssen}
         {Logical Calculi for Reasoning with Binding}
         {Faculty of Mathematics and Computer Science, TU/e}
         {2007-17}

\promitem{D. Jarnikov}
         {QoS framework for Video Streaming in Home Networks}
         {Faculty of Mathematics and Computer Science, TU/e}
         {2007-18}

\promitem{M. A. Abam}
         {New Data Structures and Algorithms for Mobile Data}
         {Faculty of Mathematics and Computer Science, TU/e}
         {2007-19}

\promitem{W. Pieters}
         {La Volont\'{e} Machinale: Understanding the Electronic Voting Controversy}
         {Faculty of Science, Mathematics and Computer Science, RU}
         {2008-01}

\promitem{A.L. de Groot}
         {Practical Automaton Proofs in PVS}
         {Faculty of Science, Mathematics and Computer Science, RU}
         {2008-02}

\promitem{M. Bruntink}
         {Renovation of Idiomatic Crosscutting Concerns in Embedded Systems}
         {Faculty of Electrical Engineering, Mathematics, and Computer Science, TUD}
         {2008-03}

\promitem{A.M. Marin}
         {An Integrated System to Manage Crosscutting Concerns in Source Code}
         {Faculty of Electrical Engineering, Mathematics, and Computer Science, TUD}
         {2008-04}

\promitem{N.C.W.M. Braspenning}
         {Model-based Integration and Testing of High-tech Multi-disciplinary Systems}
         {Faculty of Mechanical Engineering, TU/e}
         {2008-05}

\promitem{M. Bravenboer}
         {Exercises in Free Syntax: Syntax Definition, Parsing, and Assimilation of Language Conglomerates}
         {Faculty of Science, UU}
         {2008-06}

\promitem{M. Torabi Dashti}
         {Keeping Fairness Alive: Design and Formal Verification of Optimistic Fair Exchange Protocols}
         {Faculty of Sciences, Division of Mathematics and Computer Science, VUA}
         {2008-07}

\promitem{I.S.M. de Jong}
         {Integration and Test Strategies for Complex Manufacturing Machines}
         {Faculty of Mechanical Engineering, TU/e}
         {2008-08}

\promitem{I. Hasuo}
         {Tracing Anonymity with Coalgebras}
         {Faculty of Science, Mathematics and Computer Science, RU}
         {2008-09}

\promitem{L.G.W.A. Cleophas}
         {Tree Algorithms: Two Taxonomies and a Toolkit}
         {Faculty of Mathematics and Computer Science, TU/e}
         {2008-10}

\promitem{I.S. Zapreev}
         {Model Checking Markov Chains: Techniques and Tools}
         {Faculty of Electrical Engineering, Mathematics \& Computer Science, UT}
         {2008-11}

\promitem{M. Farshi}
         {A Theoretical and Experimental Study of Geometric Networks}
     {Faculty of Mathematics and Computer Science, TU/e}
         {2008-12}

\promitem{G. Gulesir}
         {Evolvable Behavior Specifications Using Context-Sensitive Wildcards}
         {Faculty of Electrical Engineering, Mathematics \& Computer Science, UT}
         {2008-13}

\promitem{F.D. Garcia}
         {Formal and Computational Cryptography: Protocols, Hashes and Commitments}
     {Faculty of Science, Mathematics and Computer Science, RU}
         {2008-14}

\promitem{P. E. A. D\"{u}rr}
         {Resource-based Verification for Robust Composition of Aspects}
         {Faculty of Electrical Engineering, Mathematics \& Computer Science, UT}
         {2008-15}

\promitem{E.M. Bortnik}
         {Formal Methods in Support of SMC Design}
         {Faculty of Mechanical Engineering, TU/e}
         {2008-16}

\promitem{R.H. Mak}
         {Design and Performance Analysis of Data-Independent Stream Processing       Systems}
         {Faculty of Mathematics and Computer Science, TU/e}
         {2008-17}

\promitem{M. van der Horst}
         {Scalable Block Processing Algorithms}
         {Faculty of Mathematics and Computer Science, TU/e}
         {2008-18}

\promitem{C.M. Gray}
         {Algorithms for Fat Objects: Decompositions and Applications}
         {Faculty of Mathematics and Computer Science, TU/e}
         {2008-19}

\promitem{J.R. Calam\'{e}}
         {Testing Reactive Systems with Data - Enumerative Methods and Constraint Solving}
         {Faculty of Electrical Engineering, Mathematics \& Computer Science, UT}
         {2008-20}

\promitem{E. Mumford}
         {Drawing Graphs for Cartographic Applications}
         {Faculty of Mathematics and Computer Science, TU/e}
         {2008-21}

\promitem{E.H. de Graaf}
         {Mining Semi-structured Data, Theoretical and Experimental Aspects of Pattern Evaluation}
         {Faculty of Mathematics and Natural Sciences, UL}
         {2008-22}

\promitem{R. Brijder}
         {Models of Natural Computation: Gene Assembly and Membrane Systems}
         {Faculty of Mathematics and Natural Sciences, UL}
         {2008-23}

\promitem{A. Koprowski}
         {Termination of Rewriting and Its Certification}
         {Faculty of Mathematics and Computer Science, TU/e}
         {2008-24}

\promitem{U. Khadim}
         {Process Algebras for Hybrid Systems: Comparison and Development}
         {Faculty of Mathematics and Computer Science, TU/e}
         {2008-25}

\promitem{J. Markovski}
         {Real and Stochastic Time in Process Algebras for Performance Evaluation}
         {Faculty of Mathematics and Computer Science, TU/e}
         {2008-26}

\promitem{H. Kastenberg}
         {Graph-Based Software Specification and Verification}
         {Faculty of Electrical Engineering, Mathematics \& Computer Science, UT}
         {2008-27}

\promitem{I.R. Buhan}
         {Cryptographic Keys from Noisy Data Theory and Applications}
         {Faculty of Electrical Engineering, Mathematics \& Computer Science, UT}
         {2008-28}

\promitem{R.S. Marin-Perianu}
         {Wireless Sensor Networks in Motion: Clustering Algorithms for Service Discovery and Provisioning}
         {Faculty of Electrical Engineering, Mathematics \& Computer Science, UT}
         {2008-29}

\promitem{M.H.G. Verhoef}
         {Modeling and Validating Distributed Embedded Real-Time Control Systems}
         {Faculty of Science, Mathematics and Computer Science, RU}
         {2009-01}

\promitem{M. de Mol}
         {Reasoning about Functional Programs: Sparkle, a proof assistant for Clean}
         {Faculty of Science, Mathematics and Computer Science, RU}
         {2009-02}

\promitem{M. Lormans}
         {Managing Requirements Evolution}
         {Faculty of Electrical Engineering, Mathematics, and Computer Science, TUD}
         {2009-03}

\promitem{M.P.W.J. van Osch}
         {Automated Model-based Testing of Hybrid Systems}
         {Faculty of Mathematics and Computer Science, TU/e}
         {2009-04}

\promitem{H. Sozer}
         {Architecting Fault-Tolerant Software Systems}
         {Faculty of Electrical Engineering, Mathematics \& Computer Science, UT}
         {2009-05}

\promitem{M.J. van Weerdenburg}
         {Efficient Rewriting Techniques}
         {Faculty of Mathematics and Computer Science, TU/e}
         {2009-06}

\promitem{H.H. Hansen}
         {Coalgebraic Modelling: Applications in Automata Theory and Modal Logic}
         {Faculty of Sciences, Division of Mathematics and Computer Science, VUA}
         {2009-07}

\promitem{A. Mesbah}
         {Analysis and Testing of Ajax-based Single-page Web Applications}
         {Faculty of Electrical Engineering, Mathematics, and Computer Science, TUD}
         {2009-08}

\promitem{A.L. Rodriguez Yakushev}
         {Towards Getting Generic Programming Ready for Prime Time}
         {Faculty of Science, UU}
         {2009-9}

\promitem{K.R. Olmos Joffr\'e}
         {Strategies for Context Sensitive Program Transformation}
         {Faculty of Science, UU}
         {2009-10}

\promitem{J.A.G.M. van den Berg}
         {Reasoning about Java programs in PVS using JML}
         {Faculty of Science, Mathematics and Computer Science, RU}
         {2009-11}

\promitem{M.G. Khatib}
         {MEMS-Based Storage Devices. Integration in Energy-Constrained Mobile Systems}
         {Faculty of Electrical Engineering, Mathematics \& Computer Science, UT}
         {2009-12}

\promitem{S.G.M. Cornelissen}
         {Evaluating Dynamic Analysis Techniques for Program Comprehension}
         {Faculty of Electrical Engineering, Mathematics, and Computer Science, TUD}
         {2009-13}

\promitem{D. Bolzoni}
         {Revisiting Anomaly-based Network Intrusion Detection Systems}
         {Faculty of Electrical Engineering, Mathematics \& Computer Science, UT}
         {2009-14}

\promitem{H.L. Jonker}
         {Security Matters: Privacy in Voting and Fairness in Digital Exchange}
         {Faculty of Mathematics and Computer Science, TU/e}
         {2009-15}

\promitem{M.R. Czenko}
         {TuLiP - Reshaping Trust Management}
         {Faculty of Electrical Engineering, Mathematics \& Computer Science, UT}
         {2009-16}

\promitem{T. Chen}
         {Clocks, Dice and Processes}
         {Faculty of Sciences, Division of Mathematics and Computer Science, VUA}
         {2009-17}

\promitem{C. Kaliszyk}
         {Correctness and Availability:
 Building Computer Algebra on top of Proof Assistants and making Proof
Assistants available over the Web}
         {Faculty of Science, Mathematics and Computer Science, RU}
         {2009-18}

\promitem{R.S.S. O'Connor}
         {Incompleteness \& Completeness: Formalizing Logic and Analysis in Type Theory}
         {Faculty of Science, Mathematics and Computer Science, RU}
         {2009-19}

\promitem{B. Ploeger}
         {Improved Verification Methods for Concurrent Systems}
         {Faculty of Mathematics and Computer Science, TU/e}
         {2009-20}

\promitem{T. Han}
         {Diagnosis, Synthesis and Analysis of Probabilistic Models}
         {Faculty of Electrical Engineering, Mathematics \& Computer Science, UT}
         {2009-21}

\promitem{R. Li}
         {Mixed-Integer Evolution Strategies for Parameter Optimization and Their Applications to Medical Image Analysis}
         {Faculty of Mathematics and Natural Sciences, UL}
         {2009-22}

\promitem{J.H.P. Kwisthout}
         {The Computational Complexity of Probabilistic Networks}
         {Faculty of Science, UU}
         {2009-23}

\promitem{T.K. Cocx}
         {Algorithmic Tools for Data-Oriented Law Enforcement}
         {Faculty of Mathematics and Natural Sciences, UL}
         {2009-24}

\promitem{A.I. Baars}
         {Embedded Compilers}
         {Faculty of Science, UU}
         {2009-25}

\promitem{M.A.C. Dekker}
         {Flexible Access Control for Dynamic Collaborative Environments}
         {Faculty of Electrical Engineering, Mathematics \& Computer Science, UT}
         {2009-26}

\promitem{J.F.J. Laros}
         {Metrics and Visualisation for Crime Analysis and Genomics}
         {Faculty of Mathematics and Natural Sciences, UL}
         {2009-27}

\promitem{C.J. Boogerd}
         {Focusing Automatic Code Inspections}
         {Faculty of Electrical Engineering, Mathematics, and Computer Science, TUD}
         {2010-01}

\promitem{M.R. Neuh\"au{\ss}er}
         {Model Checking Nondeterministic and Randomly Timed Systems}
         {Faculty of Electrical Engineering, Mathematics \& Computer Science, UT}
         {2010-02}

\promitem{J. Endrullis}
         {Termination and Productivity}
         {Faculty of Sciences, Division of Mathematics and Computer Science, VUA}
         {2010-03}

\promitem{T. Staijen}
         {Graph-Based Specification and Verification for Aspect-Oriented Languages}
         {Faculty of Electrical Engineering, Mathematics \& Computer Science, UT}
         {2010-04}

\promitem{Y.  Wang}
         {Epistemic Modelling and Protocol Dynamics}
         {Faculty of Science, UvA}
         {2010-05}

\promitem{J.K.  Berendsen}
         {Abstraction, Prices and Probability in Model Checking Timed Automata}
         {Faculty of Science, Mathematics and Computer Science, RU}
         {2010-06}

\promitem{A. Nugroho}
         {The Effects of UML Modeling on the Quality of Software}
         {Faculty of Mathematics and Natural Sciences, UL}
         {2010-07}

\promitem{A. Silva}
         {Kleene Coalgebra}
         {Faculty of Science, Mathematics and Computer Science, RU}
         {2010-08}

\promitem{J.S. de Bruin}
         {Service-Oriented Discovery of Knowledge - Foundations, Implementations and Applications}
         {Faculty of Mathematics and Natural Sciences, UL}
         {2010-09}

\promitem{D. Costa}
         {Formal Models for Component Connectors}
         {Faculty of Sciences, Division of Mathematics and Computer Science, VUA}
         {2010-10}

\promitem{M.M. Jaghoori}
         {Time at Your Service: Schedulability Analysis of Real-Time and Distributed Services}
         {Faculty of Mathematics and Natural Sciences, UL}
         {2010-11}

\promitem{R. Bakhshi}
         {Gossiping Models: Formal Analysis of Epidemic Protocols}
         {Faculty of Sciences, Department of Computer Science, VUA}
         {2011-01}

\promitem{B.J. Arnoldus}
         {An Illumination of the Template Enigma: Software Code Generation with Templates}
         {Faculty of Mathematics and Computer Science, TU/e}
         {2011-02}

\promitem{E. Zambon}
         {Towards Optimal IT Availability Planning: Methods and Tools}
         {Faculty of Electrical Engineering, Mathematics \& Computer Science, UT}
         {2011-03}

\promitem{L. Astefanoaei}
         {An Executable Theory of Multi-Agent Systems Refinement}
         {Faculty of Mathematics and Natural Sciences, UL}
         {2011-04}

\promitem{J. Proen{\c c}a}
         {Synchronous coordination of distributed components}
         {Faculty of Mathematics and Natural Sciences, UL}
         {2011-05}

\promitem{A. Moral\i}
         {IT Architecture-Based Confidentiality Risk Assessment in Networks of Organizations}
         {Faculty of Electrical Engineering, Mathematics \& Computer Science, UT}
         {2011-06}

\promitem{M. van der Bijl}
         {On changing models in Model-Based Testing}
         {Faculty of Electrical Engineering, Mathematics \& Computer Science, UT}
         {2011-07}

\promitem{C. Krause}
         {Reconfigurable Component Connectors}
         {Faculty of Mathematics and Natural Sciences, UL}
         {2011-08}

\promitem{M.E. Andr\'es}
         {Quantitative Analysis of Information Leakage in Probabilistic and Nondeterministic Systems}
         {Faculty of Science, Mathematics and Computer Science, RU}
         {2011-09}

\end{multicols}

\end{document}